\documentclass[letterpaper,unpublished,onecolumn,11pt]{quantumarticle}
\pdfoutput=1
\usepackage[numbers,sort&compress]{natbib}
\usepackage{amsmath, amssymb, amsthm, bm}
\usepackage{mathtools}
\usepackage{graphicx}
\usepackage{nicefrac}
\usepackage{nicematrix}
\usepackage{physics}
\usepackage{standalone}
\usepackage{blkarray}
\usepackage{tikz}
\usepackage{tikz-cd}
\usetikzlibrary{positioning}
\usetikzlibrary{quantikz2}
\usetikzlibrary{external}
\tikzset{external/mode=convert with system call}
\tikzset{external/figure name=main-figure}
\AtBeginEnvironment{tikzcd}{\tikzexternaldisable}
\AtEndEnvironment{tikzcd}{\tikzexternalenable}

\usepackage{subfig}
\usepackage{ifthen}
\usetikzlibrary{fit}
\usetikzlibrary{shapes.geometric}
\tikzset{
dot/.style = {circle, fill, minimum size=#1,
              inner sep=0pt, outer sep=0pt}
dot/.default = 6pt 
}
\usepackage{rotating}
\usepackage{makecell}
\usepackage{tabularx}
\usepackage{booktabs}
\usepackage{multirow}
\usepackage{hyperref}
\hypersetup{colorlinks=true, linkcolor=blue, urlcolor=blue}
\usepackage[capitalize]{cleveref}
\usepackage{stackengine}
\definecolor{niceblue}{RGB}{55, 110, 215}
\definecolor{nicered}{RGB}{214, 39, 40}
\definecolor{purple}{RGB}{148, 103, 189}
\definecolor{nicegreen}{RGB}{98, 159, 116}
\definecolor{gold}{RGB}{239,191,4}

\newtheorem{theorem}{Theorem}[section]
\newtheorem{corollary}{Corollary}[section]
\newtheorem{lemma}{Lemma}[section]
\newtheorem{proposition}{Proposition}[section]

\theoremstyle{definition}
\newtheorem{definition}{Definition}[section]
\newtheorem{example}{Example}[section]
\newtheorem*{remark}{Remark}

\DeclareMathOperator{\im}{Im}
\DeclareMathOperator{\spackle}{spackle}
\DeclareMathOperator{\Hom}{Hom}
\DeclareMathOperator{\backle}{backle}
\DeclareMathOperator{\mbqc}{MBQC}
\DeclareMathOperator{\RANK}{rk}
\DeclareMathOperator\supp{supp}
\DeclareMathOperator\argmin{argmin}

\DeclareRobustCommand{\rvdots}{%
  \vbox{
    \baselineskip4\p@\lineskiplimit\z@
    \kern-\p@
    \hbox{.}\hbox{.}\hbox{.}
  }
}

\newcommand{\akd}[1]{}
\newcommand{\arthur}[1]{}
\newcommand{\mv}[1]{}
\newcommand{\itz}[1]{}

\newcommand{\labeledmeas}[2]{
    \overset{{\color{nicered} #1}}{#2}
}
\newcommand{\labeledwire}[2]{
    \wire[l][1]["\scriptstyle \color{#1} #2"{above,pos=0.2}]{a}
}
\newcommand{\gatebox}[3][]{
    \gategroup[#2,steps=#3,style={dashed,rounded
    corners,fill=blue!20, inner
    xsep=2pt},background,label style={label
    position=above,anchor=north,yshift=0.5cm}]{#1}
}

\newcommand{\mygate}[2]{%
    \gate[#1, style={fill=none}]{#2}%
}

\title{Fault-tolerant transformations of spacetime codes}
\author{Arthur Pesah}
\affiliation{Xanadu Quantum Technologies Inc., Toronto, Ontario, M5G 2C8, Canada}
\affiliation{University College London, Gower Street, London WC1E 6BT, United Kingdom}
\author{Austin K. Daniel}
\affiliation{Xanadu Quantum Technologies Inc., Toronto, Ontario, M5G 2C8, Canada}
\author{Ilan Tzitrin}
\affiliation{Xanadu Quantum Technologies Inc., Toronto, Ontario, M5G 2C8, Canada}
\author{Michael Vasmer}
\affiliation{Xanadu Quantum Technologies Inc., Toronto, Ontario, M5G 2C8, Canada}
\affiliation{Inria Paris, 48 rue Barrault, Paris 75013, France}
\affiliation{Institute for Quantum Computing, Waterloo, Ontario, N2L 3G1, Canada}
\affiliation{Perimeter Institute for Theoretical Physics, Waterloo, Ontario, N2L 2Y5, Canada}

\begin{document}

\begin{abstract}
    Recent advances in quantum error-correction (QEC) have shown that it is often beneficial to understand fault-tolerance as a dynamical process---a circuit with redundant measurements that help correct errors---rather than as a static code equipped with a syndrome extraction circuit.
    Spacetime codes have emerged as a natural framework to understand error correction at the circuit level while leveraging the traditional QEC toolbox.
    Here, we introduce a framework based on chain complexes and chain maps to model spacetime codes and transformations between them. We show that stabilizer codes, quantum circuits, and decoding problems can all be described using chain complexes, and that the equivalence of two spacetime codes can be characterized by specific maps between chain complexes, the \textit{fault-tolerant maps}, that preserve the number of encoded qubits, fault distance, and minimum-weight decoding problem. As an application of this framework, we extend the foliated cluster state construction from stabilizer codes to any spacetime code, showing that any Clifford circuit can be transformed into a measurement-based protocol with the same fault-tolerant properties. To this protocol, we associate a chain complex which encodes the underlying decoding problem, generalizing previous cluster state complex constructions.
    Our method enables the construction of cluster states from non-CSS, subsystem, and Floquet codes, as well as from logical Clifford operations on a given code.
\end{abstract}

\maketitle

\tableofcontents

\section{Introduction}

\subsection{Motivation}

Fault-tolerant quantum computing architectures have traditionally been designed in a ``static" code-centric approach. In this paradigm, constructing a code, building syndrome extraction circuits, and performing logical operations are treated as separate processes.
However, recent advances have shown that it is often beneficial to take a more circuit-centric approach to quantum error correction. For instance, Floquet codes~\cite{hastings2021dynamically,Kesselring_2024, Townsend_Teague_2023, Davydova_2023} have a stabilizer group that evolves through a sequence of non-commuting measurements, making a dynamical description essential.
The order of measurements influences both the fault-tolerant properties of the circuit and the logical operation being applied~\cite{Davydova_2024}. Similarly, the precise scheduling of gauge measurements of a subsystem code~\cite{kribs2005unified,poulin2005stabilizer,kribs2006operator} can have an impact on the performance of the code~\cite{Higgott_2021,gidney2023less,alam2024dynamical,alam2025baconshor}.
Finally, the tools and frameworks developed for analyzing error-correcting circuits have already had a large impact in the field: they have helped design better circuit-level decoders~\cite{higgott2023improved}, more efficient syndrome extraction circuits~\cite{mcewen2023relaxing,gidney2023new,shaw2024lowering}, and noise reduction methods with less overhead than traditional quantum error-correction~\cite{delfosse2024lowcostnoisereductionclifford}.
They have also provided new insights into the study of problematic ``hook errors''~\cite{beverland2024fault}.

Several frameworks have been proposed to formally study fault-tolerant circuits.
The first, originally proposed by Bacon et al.~\cite{bacon2015sparse} and more recently revived by Pryadko \cite{pryadko2020maximumlikelihood} and Gottesman \cite{gottesman2022opportunities}, associates a spacetime subsystem code with any Clifford circuit.
In this code, each physical qubit corresponds to a spacetime location in the circuit.
Gauge operators are defined for each gate, measurement, and input stabilizer, and represent trivial circuit errors. The stabilizers of this subsystem code are then associated either with redundancies between measurements—--often referred to as detectors~\cite{gidney2021stim}—--or with measurements that become redundant when additional ones are introduced at the beginning or end of the circuit (which can be viewed as incomplete detectors).
Given a circuit noise model, the maximum-likelihood decoding problem can also be expressed in terms of this subsystem code \cite{pryadko2020maximumlikelihood}.
In a related framework proposed by Delfosse and Paetznick~\cite{delfosse2023spacetime}, spacetime codes are modeled as stabilizer codes, with only the full detectors as stabilizers.
Overall, viewing a fault-tolerant circuit as a spacetime code allows to understand the fault-tolerant properties of the circuit in terms of the properties of the code, such as the number of encoded qubits throughout the computation, the fault distance of the circuit (the minimal number of physical errors that leads to a logical error on the output), and the circuit-level decoding problem.

A natural next step in understanding spacetime codes is determining when two are equivalent.
In the realm of static codes, studying code equivalence has led to valuable insights.
For instance, defining two code families as equivalent if they are related by a geometrically-local Clifford unitary operator, it has been shown that $D$-dimensional color codes are equivalent to $D$ copies of the toric code~\cite{yoshida2011classification,bombin2012universal,bombin2014structure,kubica2015unfolding}.
This insight has enabled the construction of color code decoders based on toric code decoders~\cite{delfosse2014decoding,kubica2023efficient}.
More general classifications of quantum code equivalence classes have recently been developed~\cite{khesin2024equivalence, cross2025small}.
Establishing a notion of equivalence in the context of error-correcting circuits is particularly valuable, as compiling a fault-tolerant process depends on the hardware architecture: different quantum devices have distinct gate sets, connectivity constraints, and noise models. Finding a device-tailored compilation that preserves fault tolerance is therefore a crucial challenge.

As a simple example to understand this, let us consider the \texttt{SWAP} gate and its compilation into three \texttt{CNOT} gates:
\begin{align*}
    \tikzsetnextfilename{sec1-swap}
    \begin{quantikz}
        & \gate[2]{\text{\texttt{SWAP}}} &  \\
        & \ghost{\text{\texttt{SWAP}}} &
    \end{quantikz}
    \rightarrow
    \tikzsetnextfilename{sec1-swap-cnots}
    \begin{quantikz}
        & \ctrl{1} & \targ{}& \ctrl{1} &  \\
        & \targ{}& \ctrl{-1} & \targ{}&
    \end{quantikz}
\end{align*}
While the circuit on the left propagates any single-qubit $X$ or $Z$ error into a single-qubit error on the other qubit, the circuit on the right can propagate single-qubit errors occurring in between the \texttt{CNOT} gates into two-qubit errors.
If a \texttt{SWAP} gate is therefore present in a fault-tolerant circuit, compiling it into \texttt{CNOT}s can lead to a reduction in the fault distance compared to directly exchanging the qubits. On the other hand, we might want to consider the two circuits below as equivalent:
\begin{align*}
    \tikzsetnextfilename{sec1-mxx-with-double-hadamard}
    \begin{quantikz}
        & \gate[2]{M_{XX}} & \gate{H} & \gate{H} & \gate[2]{M_{XX}} & \\
        &                  & \gate{H} & \gate{H} &                  &
    \end{quantikz}
    \rightarrow
    \tikzsetnextfilename{sec1-mxx-without-double-hadamard}
    \begin{quantikz}
        & \gate[2]{M_{XX}} & \gate[2]{M_{XX}} & \\
        &                  &                  &
    \end{quantikz}
\end{align*}
where $M_{XX}$ and represents a two-qubit $XX$ measurements.
Indeed, while adding Hadamard gates might increase the probability of errors in between the two measurements, it does not fundamentally change the decoding problem or fault distance of the circuit.

An example of recent progress in understanding equivalent compilations of fault-tolerant processes relates to the problem of ``Floquetifying'' codes \cite{Townsend_Teague_2023, delafuente2024xyzrubycode, delafuente2024dynamical, rodatz2024floquetifying,xu2025faulttolerant}.
This problem can be formulated as follows: given a family of stabilizer codes, find a sequence of two-qubit measurements such that the resulting Floquet code maintains the same fault-tolerant properties as the original code family.
Most progress towards this goal has been made by formulating the problem in ZX-calculus~\cite{Coecke_2011,coecke2018picturing} and deriving local diagrammatic rewrite rules that preserve the local structure of detectors and the logical channel action. More recently, Ref.~\cite{rodatz2024floquetifying} has demonstrated the existence of ZX rewrite rules that preserve the fault distance of the circuit.
ZX-calculus has also enabled distance-preserving compilations of surface code circuits optimized for superconducting architectures~\cite{mcewen2023relaxing}.

Another example of adapting fault-tolerant protocols to specific architectures is the para\-digm of measurement-based quantum computing (MBQC), which is particularly relevant to photonic architectures~\cite{bourassa2021blueprintscalable,tzitrin2021faulttolerant,walshe2025linearoptical,aghaeerad2025scaling,bartolucci2021fusionbased,bombin2023fault,alexander2024manufacturable}.
In MBQC, quantum computation is performed by preparing specific resource states, known as \textit{graph states} or \textit{cluster states}, which are then measured in a particular sequence to execute the computation \cite{raussendorf2003measurement,raussendorf2006fault,broadbent2009parallelizing}.
A stabilizer code~\cite{gottesman1997stabilizer} can be transformed into a cluster state through a process called foliation \cite{bolt2016foliated,BrownUniversalFTMBQC2020}, resulting in a measurement-based protocol that acts as a memory channel and preserves the fault-tolerant properties and decoding problem of the original code (in the phenomenological noise model).
Foliation was first introduced for the surface code~\cite{raussendorf2003measurement}, later extended to all CSS codes~\cite{bolt2016foliated}, and more recently generalized to arbitrary stabilizer codes and subsystem codes~\cite{BrownUniversalFTMBQC2020} and certain Floquet codes~\cite{paesani2023highthreshold}.
Foliation was also recently used to map CSS code error correction circuits to mixed-state phases of (decohered) cluster states~\cite{negari2024spacetime}.
Additionally, fault-tolerant MBQC protocols can be understood by associating a chain complex---called the cluster state complex or fault complex---with the protocol~\cite{newman2020generating,bombin2023fault,hillmann2024single}, which allows one to consider protocols beyond foliation~\cite{nickerson2018measurement,newman2020generating}.

While the equivalence between the MBQC and the circuit model of fault-tolerance is well-understood in the context of the surface code---for instance with ZX-calculus~\cite{bombin2024unifying}---, the more general case has not yet been rigorously proven.

\subsection{Main results}

This work advances the understanding of spacetime codes (viewed as chain complexes) by studying their equivalence under transformations that preserve fault-tolerant properties.
Using this understanding, we provide a method to transform any Clifford circuit into an MBQC protocol with equivalent fault-tolerant properties and formulate its decoding problem using cluster state complexes. Specifically, we establish the following:
\begin{enumerate}
    \item Stabilizer codes, subsystem codes and quantum circuits can all be described using chain complexes. While chain complexes have traditionally been used to construct, analyze, and transform CSS codes---and more recently, cluster states---we show that they can be applied more generally.
    \item The equivalence of two spacetime codes can be characterized by a relaxed version of chain maps---that we call \textit{weak chain maps}---between their associated chain complexes. We define a \textit{fault-tolerant map} as a weak chain map that preserves the number of encoded qubits, fault distance, and minimum-weight solution to the decoding problem. We further establish the existence of two specific fault-tolerant maps, which we term reduction rules A and B.
    \item Any Clifford circuit composed of single-qubit Clifford gates, controlled Pauli gates (where the Pauli operator may act on multiple qubits), and single-qubit Pauli measurements can be transformed into a measurement-based protocol with equivalent fault-tolerant properties.
    As examples, we construct cluster states from non-CSS, subsystem, and Floquet codes, as well as from logical Clifford operations on a given code.
    The decoding problem of these protocols can be represented by a novel cluster state complex, which generalizes previous formulations. Notably, our cluster state complexes can be non-bipartite, incorporate
    $Y$ measurements, and include input and output nodes.
\end{enumerate}

\subsection{Reader's guide}

Despite the length of this paper and a certain dependence between sections, some readers could find value in consulting only particular parts. We provide here a few pathways through the paper depending on different readers' interests:
\begin{itemize}
    \item For a reader only interested in \textbf{understanding our chain complex formalism} for representing and manipulating codes beyond CSS stabilizer ones, \cref{sec:qec-from-chain-complexes} should be sufficient. A reader already familiar with chain complexes might want to jump directly to \cref{example:non-css-code} and \cref{example:subsystem-code} to learn how they are used to represent non-CSS and subsystem codes, before going to \cref{sec:fault-tolerant-maps}, in which the new notion of fault-tolerant map between complexes is introduced. The particular fault-tolerant maps that we introduce---rules A and B---can be understood graphically from \cref{fig:reduction-rules} on a first read, while \cref{sec:rule-a} and \cref{sec:rule-b} can be consulted if the formal proof of their properties is of interest.
    \item A reader mainly interested in \textbf{learning about spacetime codes} can start their reading in \cref{sec:spacetime-codes}, in which we review this formalism while synthesizing previous approaches in the literature~\cite{bacon2015sparse,delfosse2023spacetime,gottesman2022opportunities}.
    In particular, we show in \cref{sec:stabilizers-spacetime-code}---and more precisely \cref{prop:stabilizers_backle_spackle}---that the stabilizers of the spacetime code of Delfosse and Paetznick~\cite{delfosse2023spacetime} are also stabilizers of the spacetime subsystem code of Bacon et al.~\cite{bacon2015sparse}, which can be obtained through the forward and backward propagations of measurements operators.
    The link between spacetime codes and chain complexes is only briefly mentioned, and it is therefore not required to have read \cref{sec:qec-from-chain-complexes} for an overall understanding of this section.
    \item For a reader interested in
    \textbf{understanding MBQC through the lens of chain complexes}, \cref{sec:fault-tolerant-cluster-states} is the most relevant section.
    Our main object of study, the cluster state complex, is introduced in \cref{sec:cluster-state-complex} (\cref{def:cluster-state-complex}) and shown to be equivalent to the spacetime code of the corresponding MQBC circuit in \cref{theorem:mbqc-circuit-equivalent-to-cluster-state-complex}.
    While this construction itself can be read independently from the rest of the paper, its motivation and relation to spacetime codes can only be understood after reading \cref{sec:qec-from-chain-complexes} and \cref{sec:spacetime-codes}.
    \item A reader whose objective is to \textbf{construct MBQC protocols from Clifford circuits}---with the goal of simulating or experimentally realizing them---should
    first get familiar with spacetime codes by reading through \cref{sec:spacetime-codes}, and with our MBQC notations by reading through \cref{sec:mbqc} and \cref{sec:spacetime-code-cluster-state}. They can then read through \cref{sec:from-spacetime-codes-to-cluster-states} while skipping the proofs of the different lemmas and propositions. Understanding those proofs require familiarity with the chain complex formalism of \cref{sec:qec-from-chain-complexes}.
    \item For a reader whose objective is to \textbf{construct MBQC protocols from non-CSS codes, subsystem codes or dynamical codes}, understanding \cref{def:cluster-state-complex} and consulting the main theorems of \cref{sec:from-stabilizer-codes-to-cluster-states} or \cref{sec:from-floquet-codes-to-cluster-states} should be sufficient.
    Understanding the proofs requires familiarity with all the preceding sections, and in particular all the lemmas of \cref{sec:from-spacetime-codes-to-cluster-states}.
\end{itemize}

\section{Quantum error correction from chain complexes}
\label{sec:qec-from-chain-complexes}
\subsection{Notation}

To every $n$-qubit Pauli operator $P$, we associate a $2n$-dimensional vector $v(P)$ in the binary symplectic format, where the first $n$ components represent the presence or absence of a Pauli $X$ operator on each qubit, and the last $n$ components the presence or absence of a Pauli $Z$ operator. For instance, the operator $P=X\otimes Y$ can be written as a vector $v(P)=(11|01)$. The \textit{symplectic matrix} $\Omega$ is defined to be
    \begin{align*}
        \Omega = \begin{pmatrix}
            0_n & I_n \\
            I_n & 0_n
        \end{pmatrix},
    \end{align*}
    where $0_n$ and $I_n$ are, respectively, the $n \times n$ zero matrix and identity matrix. One can readily check that two Pauli operators $P_1$ and $P_2$ commute if and only if $v(P_1)^T\Omega v(P_2)=0$, where all operations are performed modulo 2.

\subsection{Homological quantum error correction}

\begin{definition}[Chain and cochain complexes]
    A length-$\ell$ \textit{chain complex} is a collection of $\ell+1$ $\mathbb{F}_2$-vector spaces $C_0,\ldots,C_\ell$, together with $\ell$ linear maps, called \textit{boundary operators}, $\partial_1,\ldots,\partial_\ell$, with $\partial_i: C_i \rightarrow C_{i-1}$, such that they obey the \textit{chain complex condition} $\partial_{i-1} \circ \partial_{i}=0$ for every $i$. We write a chain complex in the following diagrammatic form:
    \begin{align}
        C_\ell \xrightarrow{\partial_\ell} C_{\ell-1} \xrightarrow{\partial_{\ell-1}} \ldots \xrightarrow{\partial_1} C_{0}.
    \end{align}
    We denote such a chain complex by $C_{\bullet}$. The \textit{cochain complex} associated to $C_{\bullet}$ is the collection of vector spaces $C^0,\ldots,C^\ell$, where $C^i=\Hom\left(C_i,\mathbb{Z}_2\right)$ (the set of homomorphisms from $C_i$ to $\mathbb{Z}_2$), along with \textit{coboundary operators} $\partial^1,\ldots,\partial^\ell$, with $\partial^i:C^{i-1} \rightarrow C^i$ defined as $\partial^i(f)=f \circ \partial_i$. We write it as
    \begin{align}
        C^\ell \xleftarrow{\partial^\ell} C^{\ell-1} \xleftarrow{\partial^{\ell-1}} \ldots \xleftarrow{\partial^1} C^{0}
    \end{align}
    and denote it by $C^\bullet$.
\end{definition}

\begin{remark}
    If every space $C_i$ is finite-dimensional, we can identify the spaces $C_i$ and $C^i$.
    Indeed, for every linear form $f \in C^i$, there exists a unique vector $\tilde{f} \in C_i$ such that $f(x)=\tilde{f}^T x$ for every $x \in C_i$. The map $\psi: f \mapsto \tilde{f}$ is an isomorphism between $C^i$ and $C_i$ which can be written in matrix form as
    \begin{align}
        \psi: \begin{pmatrix}
            a_1 & \dots & a_n
        \end{pmatrix} \in C^1
        \mapsto
        \begin{pmatrix}
            a_1 \\ \vdots \\ a_n
        \end{pmatrix} \in C_1.
    \end{align}
    Every coboundary operator can then be associated with a map $\tilde{\partial}^i=\psi \circ \partial^i \circ \psi^{-1}:C_{i-1} \rightarrow C_i$. Writing the boundary and coboundary operators in matrix form, we then have the identity $\tilde{\partial}^i=\partial_i^T$. In the rest of this work, we assume all the chain complexes are defined over finite-dimensional spaces, and we omit the tilde from $\tilde{f}$ and $\tilde{\partial_i}$, implicitly identifying every dual space $C^i$ with $C_i$.
\end{remark}

\begin{definition}[Cycles and boundaries]
    An $i$-th \textit{cycle} of a chain complex $C_\bullet$ is an element of $\ker(\partial_i)$. An $i$-th \textit{boundary} is an element of $\Im(\partial_{i+1})$. Similarly, an $i$-th \textit{cocycle} is an element of $\ker(\partial^{i+1})$, and an $i$-th \textit{coboundary} is an element of $\Im(\partial^i)$.
\end{definition}

\begin{definition}[Homology and cohomology groups]
    The $i$-th \textit{homology group} of a chain complex $C_\bullet$ (for $0 \leq i \leq \ell$) is defined by $H_i(C_\bullet)=\ker(\partial_i)/\Im(\partial_{i+1})$.
    The $i$-th \textit{cohomology group} is the homology group of the corresponding cochain complex, that is, $H^i(C_\bullet)=H_i(C^\bullet)=\ker(\partial^{i+1})/\im(\partial^i)$.
    To resolve the cases $i=0$ and $i=\ell$, we extend the chain complex to the left and right by the zero space and define the corresponding maps $\partial_0$ and $\partial_{l+1}$ as zero maps:
    \begin{align}
        0 \xrightarrow{\partial_{\ell+1}=0} C_\ell \xrightarrow{\partial_\ell} C_{\ell-1} \xrightarrow{\partial_{\ell-1}} \ldots \xrightarrow{\partial_1} C_{0} \xrightarrow{\partial_0=0} 0.
    \end{align}
    For any representative $x$ of $H_i(C_\bullet)$ or $H^i(C_\bullet)$, we denote by $[x]$ its corresponding coset.
\end{definition}

As we will soon see, length-2 chain complexes play a special role in quantum error correction. In this context, it will be convenient to have a particular terminology for elements of the first homology and cohomology groups.

\begin{definition}[Logical errors and logical correlations]
    For any length-2 chain complex
    \begin{align}
        C_2 \xrightarrow{\partial_2} C_1 \xrightarrow{\partial_1} C_0
    \end{align}
    we call \textit{logical errors} the elements of $H_1(C_\bullet)=\ker(\partial_1)/\im(\partial_2)$, and \textit{logical correlations} the elements of $H^1(C_\bullet)=\ker\left(\partial_2^T\right)/\im\left(\partial_1^T\right)$.
    We say that a logical error (logical correlation) is \textit{non-trivial} if it is not zero, that is, if its representatives do not belong to $\im(\partial_2)$ ($\im\left(\partial_1^T\right)$).
\end{definition}
A fundamental theorem in homological algebra---the universal coefficient theorem---relates the homology and cohomology groups \cite{newman2020generating, Hatcher2000}:

\begin{theorem}[Universal coefficient theorem]
    \label{theorem:universal-coefficient}
    The linear map $\phi:H^i(C_\bullet) \rightarrow \Hom\left(H_i(C_\bullet), \mathbb{Z}_2\right)$ defined by $\phi([f])([x])=f(x)$ is a well-defined isomorphism.
\end{theorem}

\begin{corollary} \label{cor:intersection-form}
    The \textrm{intersection form} $\Phi: C^i \times C_i \rightarrow \mathbb{Z}_2$ defined as $\Phi(f,x)=f^Tx$ is invariant on every coset of $H^i(C_\bullet) \times H_i(C_\bullet)$, that is
    \begin{align}
        \Phi(f,x)=\Phi(f',x') \; \text{ if } [f]=[f'] \; \text{ and } \; [x]=[x'].
    \end{align}
    In other words, whether a given cycle and cocycle intersect on an even or odd number of elements only depend on their cosets, not on the specific representatives chosen.
\end{corollary}
\begin{proof}
    We have $\Phi(f,x)=f^T x=\phi([f])([x])$, where $\phi$ is the map defined in \cref{theorem:universal-coefficient}. Therefore, $\Phi(f,x)$ only depends on the cosets of $f$ and $x$.
\end{proof}

\begin{corollary} \label{cor:basis-intersection}
    There exists a basis $\{e_j\}$ of $H_i(C_\bullet)$ and $\{e^j\}$ of $H^i(C_\bullet)$ such that for every pair of representatives $x \in e_j$ and $f \in e^k$, we have $\Phi(f,x) = \delta_{jk}$
\end{corollary}
\begin{proof}
    Let $\{e_j\}$ a basis of $H_i(C_\bullet)$.
    Let $\{f^j\}$ the dual basis of $H_i(C_\bullet)$, that is, the basis of $\Hom\left(H_i(C_\bullet), \mathbb{Z}_2\right)$ such that $f^j(e_k)=\delta_{jk}$ for every $j,k$.
    Since $\phi$ is an isomorphism, for every $j$ there exists a unique $e^j \in H^i(C_\bullet)$ such that $\phi(e^j)=f^j$.
    Since an isomorphism sends a basis to a basis, $\{e^j\}$ is itself a basis of $H^i(C_\bullet)$.
    Moreover, $\Phi(e^j,e_k)=\phi(e^j)(e_k)=f^j(e_k)=\delta_{jk}$.
\end{proof}

In the specific case of length-2 chain complexes, we can interpret the intersection form as the parity of the number of elements in the intersection of a logical error and a logical correlation. \cref{cor:basis-intersection} then implies that there exists a basis of logical errors and logical correlations, such that every basis logical error intersects non-trivially with a unique basis logical correlation.

We next introduce two properties of length-2 chain complexes, which require fixing a basis of $C_1$ in order to assign a weight to each vector. Specifically, for an $\mathbb{F}_2$-vector space with basis $\{b_i\}$, the weight of a vector $x = \sum_i x_i b_i$ is defined as $|x| = \sum_i x_i$ (where here the $x_i$ are interpreted as integers).

\begin{definition}[Distance]
    The \textit{distance} $d(C_\bullet)$ of a length-2 chain complex $C_\bullet$, equipped with a basis $\{b_i\}_i$ for $C_1$, is the minimal weight of a non-trivial logical error, that is
    \begin{align}
        d=\min_{e \in \ker(\partial_1) \setminus \im(\partial_2)} |e|
    \end{align}
    where $|e|=\sum_i e_i$ when writing $e=\sum_i e_i b_i$.
\end{definition}

\begin{definition}[Minimum-weight decoding]
    The \textit{minimum-weight decoding (MWD) function} of a length-2 chain complex equipped with a basis for $C_1$ is the function $x^\star:C_1 \rightarrow C_0$ defined as
    \begin{align}
        x^\star(s)=\argmin |x| \;\text{ s.t. }\; \partial_1(x)=s.
    \end{align}
\end{definition}

\begin{example}[CSS code]
    Let's consider an $n$-qubit Calderbank-Shor-Steane (CSS) code \cite{calderbank1996good, steane1996error} defined by the two parity-check matrices $H_X$ and $H_Z$, of dimensions $m_X \times n$ and $m_Z \times n$, respectively. Then the following diagram defines a valid chain complex:
    \begin{align}
        \mathcal{S}_X \xrightarrow{H_X^T} \mathcal{Q} \xrightarrow{H_Z} \mathcal{S}_Z,
    \end{align}
    where $\mathcal{Q}=\mathbb{Z}_2^n$, $\mathcal{S}_X=\mathbb{Z}_2^{m_X}$ and $\mathcal{S}_Z=\mathbb{Z}_2^{m_Z}$.
    Indeed, the chain complex condition $H_Z H_X^T=0$ is equivalent to the commutation between all the $X$ and $Z$ stabilizers of the code.
    Moreover, any length-2 chain complex can be used to define a CSS code, by interpreting the two boundary operators as parity-check matrices for $X$ and $Z$ stabilizers, making it a convenient tool to classify and generate CSS codes.
    We note that in this chain complex, logical errors (as we defined them above) are only the $X$ logical operators, while logical correlations are the $Z$ logical operators.
    The intersection property of \cref{cor:basis-intersection} tells us that there exists a basis of $X$ logical operators $\{\bar{X}_1,\ldots,\bar{X_k}\}$ and a basis of $Z$ logical operators $\{\bar{Z}_1,\ldots,\bar{Z_k}\}$ such that the intersection between $\bar{X}_i$ and $\bar{Z}_j$ is odd if and only if $i=j$, independently of the representatives chosen.
    The distance of this chain complex correspond to the $X$-distance of the CSS code and the minimum-weight decoding function to the decoding of $X$ errors.
    The asymmetry between $X$ and $Z$ in our definition is resolved when considering a different chain complex formulation of stabilizer codes, which we present in the next example.
\end{example}

\begin{example}[Non-CSS code] \label{example:non-css-code}
    While less commonly used in quantum error correction, non-CSS Pauli stabilizer codes~\cite{gottesman1997stabilizer} can also be written as chain complexes.
    Let us consider a generic stabilizer code defined by a parity-check matrix $H$ of dimension $m \times 2n$ written in the binary symplectic format. Then the following diagram defines a valid chain complex:
    \begin{align}
        \mathcal{S} \xrightarrow{H^T} \mathcal{E} \xrightarrow{H \Omega} \mathcal{S},
    \end{align}
    where $\mathcal{S}=\mathbb{Z}_2^m$, and $\mathcal{E}=\mathbb{Z}_2^{2n}$. Indeed, similarly to the CSS case, the chain complex condition $H \Omega H^T=0$ is equivalent to the commutation of the stabilizers. However, contrary to the CSS case, not all length-2 chain complexes define a stabilizer code in this way, as we have the additional constraints that $\partial_1=\partial_2^T \Omega$ and $\mathcal{E}$ has even dimension.
\end{example}

\begin{example}[Subsystem code] \label{example:subsystem-code}
    Let us consider an $n$-qubit Pauli subsystem code~\cite{kribs2005unified,poulin2005stabilizer,kribs2006operator} defined by a \textit{gauge matrix} $H_G$ of dimension $m_G \times 2n$.
    We denote its parity-check matrix by $H_S$, which has dimension $m_S \times 2n$.
    The following diagram then defines a valid chain complex, that we call the \textit{gauge complex of the subsystem code}:
    \begin{align}
        \mathcal{G} \xrightarrow{H_G^T} \mathcal{E} \xrightarrow{H_S \Omega} \mathcal{S},
    \end{align}
    where $\mathcal{G}=\mathbb{Z}_2^{m_G}$, $\mathcal{S}=\mathbb{Z}_2^{m_S}$, $\mathcal{Q}$ is a $\mathbb{Z}_2$-vector space of dimension $n$, and $\mathcal{E}=\mathbb{Z}_2^{2n}$.
    Indeed, the chain complex condition $H_S \Omega H_G^T = 0$ is equivalent to the commutation of gauge operators and stabilizers.
    Not every chain complex gives rise to a subsystem code, since we have the additional constraints that $C_1$ is even and $\Im\left(\Omega \partial_0^T\right) \subset \Im (\partial_1)$ (every stabilizer is itself a gauge operator).
    Elements of the kernel of $H_S\Omega$ correspond to the undetectable errors of the code.
    The homology group $\ker(H_S\Omega)/\Im\left(H_G^T\right)$ (group of logical errors) of this chain complex corresponds to the dressed logical operators $\mathcal{N}(\mathcal{S})/\mathcal{G}$ of the code, while the cohomology group $\ker(H_G)/\Im\left(\Omega H_S^T\right)$ (group of logical correlations) corresponds to the bare logical operators $\mathcal{N}(\mathcal{G})/\mathcal{S}$.
    Notably, we can specialize this definition to CSS subsystem and stabilizer codes~\cite{Liu2024subsystemcsscodes}. For a CSS subsystem code, all the spaces and maps split into $X$ and $Z$ parts:
    \begin{align}
        \mathcal{G}_X\oplus \mathcal{G}_Z \xrightarrow{\left(H_G^X\right)^T \oplus \left(H_G^Z\right)^T} \mathcal{Q} \oplus \mathcal{Q} \xrightarrow{H_S^Z \oplus H_S^X} \mathcal{S}_Z \oplus \mathcal{S}_X,
    \end{align}
    which is the direct sum of the two complexes
    \begin{align}
        \mathcal{G}_X \xrightarrow{\left(H_G^X\right)^T} \mathcal{Q} \xrightarrow{H_S^Z} \mathcal{S}_Z
    \end{align}
    and
    \begin{align}
        \mathcal{G}_Z \xrightarrow{\left(H_G^Z\right)^T} \mathcal{Q} \xrightarrow{H_S^X} \mathcal{S}_X.
    \end{align}
    In the special case of a CSS stabilizer code, where $\mathcal{G}_X=\mathcal{S}_X$ and $\mathcal{G}_Z=\mathcal{S}_Z$, we recover two copies of the usual chain complex for CSS codes:
    \begin{align}
        \mathcal{S}_X \xrightarrow{\left(H_S^X\right)^T} \mathcal{Q} \xrightarrow{H_S^Z} \mathcal{S}_Z
    \end{align}
    and
    \begin{align}
        \mathcal{S}_Z \xrightarrow{\left(H_S^Z\right)^T} \mathcal{Q} \xrightarrow{H_S^X} \mathcal{S}_X.
    \end{align}
\end{example}

\subsection{Graphical representation of a chain complex}
\label{sec:graphical-rep-chain-complex}

\begin{figure}
    \centering
    \includegraphics[width=0.2\textwidth]{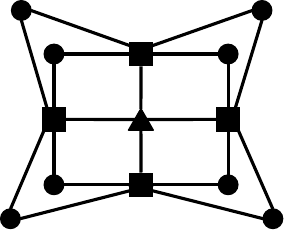}
    \caption{ \label{fig:chain-complex-example}
        Example of graphical representation of a chain complex. Circles, squares and triangles represent basis elements of $C_2$, $C_1$ and $C_0$ respectively. The connections between circles and squares are specified by the boundary operator $\partial_2$, seen as a biadjacency matrix, while the connections between triangles and squares are specified by $\partial_1$. The chain complex condition $\partial_1 \circ \partial_2 = 0$ can be verified by checking that the neighborhood of any circle intersects with the neighborhood of any triangle on an even number of squares.
    }
\end{figure}

Let us consider a length-2 chain complex
\begin{align}
    C_2 \xrightarrow{\partial_2} C_1 \xrightarrow{\partial_1} C_0,
\end{align}
where $C_2$, $C_1$ and $C_0$ have respective dimensions $n_2$, $n_1$ and $n_0$. We also equip each of the three spaces with a basis. The boundary operators $\partial_1$ and $\partial_2$ can then be represented as matrices of dimension $n_1 \times n_0$ and $n_2 \times n_1$, respectively. Interpreting $\partial_1$ as the biadjacency matrix of a graph with $n_1$ nodes of type $C_1$ (one for each basis element of $C_1$) and $n_0$ nodes of type $C_0$, and $\partial_2$ as the biadjacency matrix of a graph with $n_2$ nodes of type $C_2$ and $n_1$ nodes of type $C_1$, we can represent the chain   complex as a graph with three types of nodes. The edges between nodes of type $C_2$ and $C_1$ are specified by the matrix $\partial_2$, the edges between nodes of type $C_1$ and $C_0$ are specified by the matrix $\partial_1$, etc. We will always represent nodes of type $C_2$ by circles, nodes of type $C_1$ by squares, and nodes of type $C_0$ by triangles. An example of such chain complex representation is shown in \cref{fig:chain-complex-example}.

In the case that our chain complex represents a subsystem code
\begin{align}
    \mathcal{G} \xrightarrow{H_G^T} \mathcal{E} \xrightarrow{H_S \Omega} \mathcal{S}
\end{align}
we can moreover split the nodes of type $\mathcal{E}$ into $X$ and $Z$ type, each node representing a single-qubit $X$ or $Z$ error at a certain location. We call the nodes of type $\mathcal{G}$ \textit{gauge nodes}, the nodes of type $\mathcal{E}$ \textit{error nodes}, and the nodes of type $\mathcal{S}$ \textit{detector nodes}. Note that due to the $\Omega$ matrix in the definition of $\partial_1$, detector nodes are connected to opposite Pauli error nodes in the graph than in their definition: a stabilizer $S=X_1Z_2$ is represented by a triangular node connected to square nodes $Z_1$ and $X_2$.
An example of gauge complex, corresponding to the Bacon-Shor code \cite{bacon2006operator}, is shown in \cref{fig:bacon-shor-equivalence-b}.

\begin{figure}
    \centering
    \tikzsetnextfilename{sec2-bacon-shor-equivalence-a}
    \subfloat[]{ \label{fig:bacon-shor-equivalence-a}
        \scalebox{.95}{\begin{tikzpicture}[scale=0.92]
    \foreach \x in {0,...,3}
        \foreach \y in {0,...,3}
            \node[dot=5pt, fill=black] (\x\y) at (\x,\y) {};

    \foreach \x in {0,...,3}
        \foreach \y [count=\yi] in {0,...,2}
            \draw[black,thick] (\x\y) -- (\x\yi) (\y\x) -- (\yi\x);

    \node[] () at (-0.3,1) {\footnotesize X};
    \node[] () at (-0.3,2) {\footnotesize X};

    \node[] () at (0,-0.3) {\footnotesize Z};
    \node[] () at (1,-0.3) {\footnotesize Z};
\end{tikzpicture}}
    }
    \tikzsetnextfilename{sec2-bacon-shor-equivalence-b}
    \subfloat[]{ \label{fig:bacon-shor-equivalence-b}
        \scalebox{.95}{\begin{tikzpicture}[scale=0.92, triangle/.style = {fill=black, regular polygon, regular polygon sides=3}]   
    \foreach \y in {0,...,2}
        \node[triangle,fill=black, inner sep=0pt, minimum size=8pt] (SX\y) at (1.5,\y+0.5) {};

    \foreach \y in {0,...,2}
        \foreach \x in {0,...,3}
            \draw[black,thick] (SX\y) -- (\x,\y) (SX\y) -- (\x,\y+1);

    \foreach \x in {0,...,2}
        \foreach \y in {0,...,3}
            \node[circle, fill=black, inner sep=0pt, minimum size=5pt] (GZ\x\y) at (\x+0.5,\y) {};

    \foreach \y in {0,...,3}
        \foreach \x in {0,...,2}
            \draw[black,thick] (\x+0.5,\y) -- (\x,\y) (\x+0.5,\y) -- (\x+1,\y);

    \foreach \x in {0,...,3}
        \foreach \y in {0,...,3}
            \node[fill=niceblue, inner sep=0, minimum size=12pt] (\x\y) at (\x,\y) {\color{white} \tiny Z};

    \foreach \x in {0,...,2}
        \node[triangle,fill=black, inner sep=0, minimum size=8pt] (SZ\x) at (4+\x+0.5, 1.5) {};

    \foreach \x in {0,...,2}
        \foreach \y in {0,...,3}
            \draw[black,thick] (SZ\x) -- (4+\x,\y) (SZ\x) -- (4+\x+1,\y);

    \foreach \x in {0,...,3}
        \foreach \y in {0,...,2}
            \node[circle, fill=black, inner sep=0pt, minimum size=5pt] (GX\x\y) at (4+\x,\y+0.5) {};

    \foreach \y in {0,...,2}
        \foreach \x in {0,...,3}
            \draw[black,thick] (4+\x,\y+0.5) -- (4+\x,\y) (4+\x,\y+0.5) -- (4+\x,\y+1);

    \foreach \x in {0,...,3}
        \foreach \y in {0,...,3}
            \node[fill=nicered, inner sep=0, minimum size=12pt] (\x\y) at (4+\x,\y) {\color{white} \tiny X};
            
    \node[] () at (0,-0.3) {\phantom{\footnotesize Z}};
    \node[] () at (1,-0.3) {\phantom{\footnotesize Z}};
\end{tikzpicture}}
    }
    \tikzsetnextfilename{sec2-bacon-shor-equivalence-c}
    \subfloat[]{ \label{fig:bacon-shor-equivalence-c}
        \scalebox{.95}{\begin{tikzpicture}[scale=0.92, triangle/.style = {fill=black, regular polygon, regular polygon sides=3}]
    \foreach \y in {0,...,2}
        \node[triangle,fill=black, inner sep=0, minimum size=8pt] (SX\y) at (0,\y+0.5) {};

    \foreach \y in {0,...,2}
        \draw[black,thick] (SX\y) -- (0,\y) (SX\y) -- (0,\y+1);

    \foreach \y in {0,...,3}
        \node[fill=niceblue, inner sep=0, minimum size=12pt] (0\y) at (0,\y) {\color{white} \tiny Z};

    \foreach \x in {0,...,2}
        \node[triangle,fill=black, inner sep=0, minimum size=8pt] (SZ\x) at (1+\x+0.5,1.5) {};

    \foreach \x in {0,...,2}
        \draw[black,thick] (SZ\x) -- (1+\x,1.5) (SZ\x) -- (1+\x+1,1.5);

    \foreach \x in {0,...,3}
        \node[fill=nicered, inner sep=0, minimum size=12pt] (\x0) at (1+\x,1.5) {\color{white} \tiny X};

    \node[] () at (0,-0.3) {\phantom{\footnotesize Z}};
    \node[] () at (1,-0.3) {\phantom{\footnotesize Z}};
\end{tikzpicture}}
    }
    \caption{Gauge complex and reduced gauge complex of the Bacon-Shor code \cite{bacon2006operator}. \textbf{(a)} Bacon-Shor code on a $4 \times 4$ grid with qubits on vertices. Vertical and horizontal edges represent $X$ and $Z$ gauge operators, respectively. The $X$ stabilizers and $Z$ stabilizers can be generated by pairs of rows and columns, respectively. \textbf{(b)} Gauge complex associated to those gauge and stabilizer generators. \textbf{(c)} The reduced gauge complex: using rule A, we can merge all the $Z$ nodes in the same row, and all the $X$ nodes in the same column, showing the equivalence of the Bacon-Shor code decoding problem to that of two repetition codes.}
    \label{fig:bacon-shor-equivalence}
\end{figure}
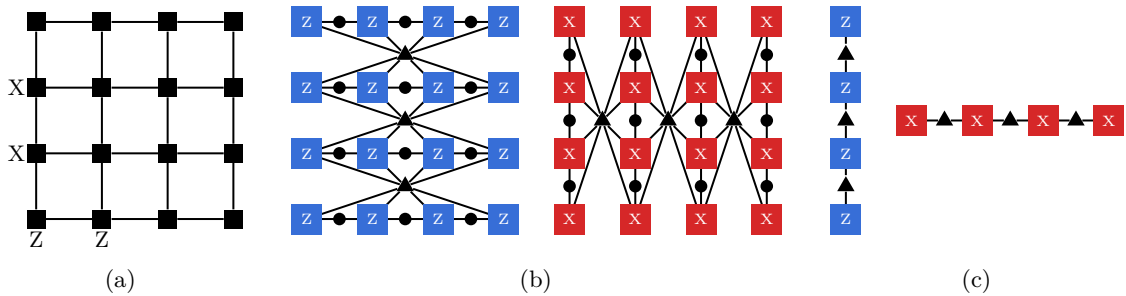

\subsection{Fault-tolerant maps}
\label{sec:fault-tolerant-maps}

We would now like to develop a notion of equivalence between chain complexes, such that two chain complexes are equivalent if they have the same distance and number of logical cosets, and if the minimum-weight decoding function of one chain complex can be mapped to the minimum-weight decoding function of any other chain complex in the equivalence class.

A natural way to transform a chain complex while preserving its structure is to use a chain map.
A chain map $f:C_{\bullet} \rightarrow C'_{\bullet}$ between two length-2 chain complexes
\begin{align*}
    C_2 \xrightarrow{\partial_2} C_1 \xrightarrow{\partial_1} C_0
\end{align*}
and
\begin{align*}
    C'_2 \xrightarrow{\partial'_2} C'_1 \xrightarrow{\partial'_1} C'_0
\end{align*}
is a triple of linear maps $(f_0,f_1,f_2)$ with $f_i:\mathcal{C}_i \rightarrow \mathcal{C}'_i$ for $i \in \{0,1,2\}$, such that the diagram
\begin{equation}
    \begin{tikzcd}[column sep=20pt, row sep=20pt, every cell/.append style={inner sep=4pt}]
        {C_2}  & {C_1} & {C_0} \\
        {C'_2} & {C'_1} & {C'_0}
        \arrow["{\partial_2}", from=1-1, to=1-2]
        \arrow["{f_2}"', from=1-1, to=2-1]
        \arrow["{\partial_1}", from=1-2, to=1-3]
        \arrow["{f_1}"', from=1-2, to=2-2]
        \arrow["{f_0}"', from=1-3, to=2-3]
        \arrow["{\partial'_2}", from=2-1, to=2-2]
        \arrow["{\partial'_1}", from=2-2, to=2-3]
    \end{tikzcd}
\end{equation}
commutes, that is, $f_1 \circ \partial_2 = \partial_2' \circ f_2$ and $f_0 \circ \partial_1=\partial_1' \circ f_1$.
It can be shown that chain maps preserve cycles (elements of $\ker \partial_1$) and boundaries (elements of $\Im \partial_0$). Moreover, any chain map $f:C_{\bullet} \rightarrow C'_{\bullet}$ induces a well-defined map $f^\star:H_1(C_{\bullet}) \rightarrow H_1(C'_{\bullet})$ on the homology groups.
We note that chain maps preserving the properties of CSS codes have been studied previously in the context of decoding topological codes~\cite{delfosse2014decoding,kubica2023efficient,bauer2025planar}.

However, the commutation requirement of chain maps is too strong for our purpose: we will soon present some examples of maps that preserve the desired properties of our complex without making the diagram above commute.
We introduce here a weaker notion of chain maps, which still induces a well-defined map on the homology groups and preserves enough structure to make distance and decoding preservation possible.

\begin{definition}[Weak chain map]
        A weak chain map $f:C_{\bullet} \rightarrow C'_{\bullet}$ between two length-2 chain complexes
    \begin{align*}
        C_2 \xrightarrow{\partial_2} C_1 \xrightarrow{\partial_1} C_0
    \end{align*}
    and
    \begin{align*}
        C'_2 \xrightarrow{\partial'_2} C'_1 \xrightarrow{\partial'_1} C'_0
    \end{align*}
    is a triple of linear maps $(f_0,f_1,f_2)$ with $f_i:\mathcal{C}_i \rightarrow \mathcal{C}'_i$ for $i \in \{0,1,2\}$, such that $\Im(f_1 \circ \partial_2) \subset \Im(\partial_2')$ and $f_0 \circ \partial_1=\partial_1' \circ f_1$.
\end{definition}
We note that chain maps are particular instances of weak chain maps, since if $f$ is a chain map, we have $\Im(f_1 \circ \partial_2) = \Im(\partial_2' \circ f_2) \subset \Im(\partial_2')$. We now show that weak chain maps also induce well-defined maps at the homology level.

\begin{theorem}
    \label{theorem:weak-chain-map-induces-homology-map}
    A weak chain map $f:C_{\bullet} \rightarrow C'_{\bullet}$ induces a map $f^*:H_1(C_{\bullet}) \rightarrow H_1(C'_{\bullet})$ defined as $f^*([x])=[f_1(x)]$ for any $x \in \ker(\partial_1)$. This map is well-defined, that is, $f_1(x) \in \ker(\partial_1')$ for all $x \in \ker(\partial_1)$, and $[f_1(x_1)]=[f_1(x_2)]$ for all $x_1, x_2 \in \ker(\partial_1)$ such that $[x_1]=[x_2]$.
\end{theorem}
\begin{proof}
    Let us start by showing that $f_1$ maps cycles to cycles. Let $x \in \ker(\partial_1)$. Then $\partial_1'(f_1(x))=f_0(\partial_1(x))=0$. Therefore, $f_1(x) \in \ker(\partial_1')$.

    Let us now show that the coset of $f_1(x)$ only depends on the coset of $x$. Let $x_1,x_2 \in \ker(\partial_1)$ such that $[x_1]=[x_2]$. Then $f_1(x_1)+f_1(x_2)=f_1(x_1+x_2)$. Since $[x_1+x_2]=0$, we must have $x_1+x_2 \in \Im(\partial_2)$. Therefore, there exists $y \in C_2$ such that $x_1+x_2=\partial_2 (y)$. Hence, $f_1(x_1+x_2)=f_1(\partial_2(y)) \in \Im(\partial_2')$ by the definition of a weak chain map. Therefore, $[f_1(x_1) + f_1(x_2)]=0$, that is, $[f_1(x_1)]=[f_1(x_2)]$.
\end{proof}

We are particularly interested in the case where the induced map $f^\star$ is an isomorphism, that is, a one-to-one map from logical errors of one chain complex to logical errors of the other chain complex. A chain map that is not necessarily an isomorphism, but whose induced map is, is called a \textit{quasi-isomorphism}. We generalize this notion for weak chain maps.
\begin{definition}[Weak quasi-isomorphism]
    A weak quasi-isomorphism is a weak chain map $f$ such that the induced map $f^\star$ is a group isomorphism.
\end{definition}
For the purpose of this work, we would also like our chain maps to preserve the distance and decoding function of the chain complexes.

\begin{definition}[Distance-preserving map]
    A weak chain map $f: C_\bullet \rightarrow C'_\bullet$ is said to be \textit{distance-preserving} if $d(C_\bullet)=d(C'_\bullet)$.
\end{definition}

\begin{definition}[Decoding-preserving map]
    A weak chain map $f: C_\bullet \rightarrow C'_\bullet$ is said to be \textit{decoding-preserving} if the minimum-weight decoding functions $x^\star$ and $x'^\star$ of $C_\bullet$ and $C'_\bullet$ are related by $f_1 \circ x^\star = x'^\star \circ f_0$.
\end{definition}

We can finally define our notion of fault-tolerant map:
\begin{definition}[Fault-tolerant map]
    A weak chain map is called \textit{fault-tolerant} if it is a distance-preserving and decoding-preserving weak quasi-isomorphism.
\end{definition}
This notion allows us to define an equivalence class of chain complex as follows.
\begin{definition}[Equivalence of chain complexes]
    Two chain complexes $C_{\bullet}$ and $C'_{\bullet}$ are equivalent if there exists a fault-tolerant map $f:C_{\bullet} \rightarrow C'_{\bullet}$ and a fault-tolerant map $g:C'_{\bullet} \rightarrow C_{\bullet}$
\end{definition}

\begin{figure}
    \centering
    \tikzsetnextfilename{sec2-rule-a}
    \subfloat[]{ \label{fig:reduction-rule-a}
        \begin{tikzpicture}[scale=0.92, triangle/.style = {fill=black, regular polygon, regular polygon sides=3}]
    \node[fill=black, inner sep=0, minimum size=6pt] (Q1) at (0,0) {};
    \node[circle, fill=black, inner sep=0pt, minimum size=5pt] (G1) at (1,0) {};
    \node[fill=black, inner sep=0, minimum size=6pt] (Q2) at (2,0) {};

    \node[triangle,fill=black, inner sep=0, minimum size=10pt] (S1) at (1,1) {};

    \draw[black,thick] (Q1) -- (Q2);
    \draw[black,thick] (Q1) -- (S1);
    \draw[black,thick] (Q2) -- (S1);

    \node[circle,fill=black, inner sep=0, minimum size=1pt] () at (-0.5,0.1) {};
    \node[circle,fill=black, inner sep=0, minimum size=1pt] () at (-0.5,0) {};
    \node[circle,fill=black, inner sep=0, minimum size=1pt] () at (-0.5,-0.1) {};

    \node[circle,fill=black, inner sep=0, minimum size=1pt] () at (2.5,0.1) {};
    \node[circle,fill=black, inner sep=0, minimum size=1pt] () at (2.5,0) {};
    \node[circle,fill=black, inner sep=0, minimum size=1pt] () at (2.5,-0.1) {};

    \node[circle,fill=black, inner sep=0, minimum size=1pt] () at (0.9,1.5) {};
    \node[circle,fill=black, inner sep=0, minimum size=1pt] () at (1,1.5) {};
    \node[circle,fill=black, inner sep=0, minimum size=1pt] () at (1.1,1.5) {};

    \draw[black,thick] (Q1) -- (-0.6,0.3);
    \draw[black,thick] (Q1) -- (-0.6,-0.3);

    \draw[black,thick] (Q2) -- (2.6,0.3);
    \draw[black,thick] (Q2) -- (2.6,-0.3);

    \draw[black,thick] (S1) -- (0.7,1.6);
    \draw[black,thick] (S1) -- (1.3,1.6);

    \draw[<->, thick] (3,0.5) -- (4,0.5);

    \node[fill=black, inner sep=0, minimum size=6pt] (Q) at (4+1,0) {};

    \node[triangle,fill=black, inner sep=0, minimum size=10pt] (S) at (4+1,1) {};

    \draw[black,thick] (Q) -- (S);

    \node[circle,fill=black, inner sep=0, minimum size=1pt] () at (5-0.5,0.1) {};
    \node[circle,fill=black, inner sep=0, minimum size=1pt] () at (5-0.5,0) {};
    \node[circle,fill=black, inner sep=0, minimum size=1pt] () at (5-0.5,-0.1) {};

    \node[circle,fill=black, inner sep=0, minimum size=1pt] () at (5+0.5,0.1) {};
    \node[circle,fill=black, inner sep=0, minimum size=1pt] () at (5+0.5,0) {};
    \node[circle,fill=black, inner sep=0, minimum size=1pt] () at (5+0.5,-0.1) {};

    \node[circle,fill=black, inner sep=0, minimum size=1pt] () at (4+0.9,1.5) {};
    \node[circle,fill=black, inner sep=0, minimum size=1pt] () at (4+1,1.5) {};
    \node[circle,fill=black, inner sep=0, minimum size=1pt] () at (4+1.1,1.5) {};

    \draw[black,thick] (Q) -- (5-0.6,0.3);
    \draw[black,thick] (Q) -- (5-0.6,-0.3);

    \draw[black,thick] (Q) -- (5+0.6,0.3);
    \draw[black,thick] (Q) -- (5+0.6,-0.3);

    \draw[black,thick] (S) -- (4+0.7,1.6);
    \draw[black,thick] (S) -- (4+1.3,1.6);

\end{tikzpicture}
    }
    \tikzsetnextfilename{sec2-rule-b}
    \subfloat[]{ \label{fig:reduction-rule-b}
        \begin{tikzpicture}[scale=0.92, triangle/.style = {fill=black, regular polygon, regular polygon sides=3}]
    \begin{scope}[xshift=-3.5cm]
        \node[circle, fill=black, inner sep=0pt, minimum size=5pt] (G1) at (0,0) {};
        \node[fill=black, inner sep=0, minimum size=6pt] (Q1) at (1,0) {};
        \node[circle, fill=black, inner sep=0pt, minimum size=5pt] (G2) at (2,0.5) {};
        \node[circle, fill=black, inner sep=0pt, minimum size=5pt] (G3) at (2,-0.5) {};

        \draw[black,thick] (G1) -- (Q1);
        \draw[black,thick] (Q1) -- (G2);
        \draw[black,thick] (Q1) -- (G3);

        \node[circle,fill=black, inner sep=0, minimum size=1pt] () at (2,0.15) {};
        \node[circle,fill=black, inner sep=0, minimum size=1pt] () at (2,0) {};
        \node[circle,fill=black, inner sep=0, minimum size=1pt] () at (2,-0.15) {};

        \node[circle,fill=black, inner sep=0, minimum size=1pt] () at (2.4,0.6) {};
        \node[circle,fill=black, inner sep=0, minimum size=1pt] () at (2.4,0.5) {};
        \node[circle,fill=black, inner sep=0, minimum size=1pt] () at (2.4,0.4) {};

        \node[circle,fill=black, inner sep=0, minimum size=1pt] () at (2.4,-0.4) {};
        \node[circle,fill=black, inner sep=0, minimum size=1pt] () at (2.4,-0.5) {};
        \node[circle,fill=black, inner sep=0, minimum size=1pt] () at (2.4,-0.6) {};

        \draw[black,thick] (G2) -- (2.5,0.8);
        \draw[black,thick] (G2) -- (2.5,0.2);

        \draw[black,thick] (G3) -- (2.5,-0.8);
        \draw[black,thick] (G3) -- (2.5,-0.2);
    \end{scope}

    \begin{scope}
        \draw[<->, thick] (-0.5,0) -- (0.5,0);
    \end{scope}

    \begin{scope}[xshift=-1cm]
        \node[circle, fill=black, inner sep=0pt, minimum size=5pt] (G2) at (2,0.5) {};
        \node[circle, fill=black, inner sep=0pt, minimum size=5pt] (G3) at (2,-0.5) {};

        \node[circle,fill=black, inner sep=0, minimum size=1pt] () at (2,0.15) {};
        \node[circle,fill=black, inner sep=0, minimum size=1pt] () at (2,0) {};
        \node[circle,fill=black, inner sep=0, minimum size=1pt] () at (2,-0.15) {};

        \node[circle,fill=black, inner sep=0, minimum size=1pt] () at (2.4,0.6) {};
        \node[circle,fill=black, inner sep=0, minimum size=1pt] () at (2.4,0.5) {};
        \node[circle,fill=black, inner sep=0, minimum size=1pt] () at (2.4,0.4) {};

        \node[circle,fill=black, inner sep=0, minimum size=1pt] () at (2.4,-0.4) {};
        \node[circle,fill=black, inner sep=0, minimum size=1pt] () at (2.4,-0.5) {};
        \node[circle,fill=black, inner sep=0, minimum size=1pt] () at (2.4,-0.6) {};

        \draw[black,thick] (G2) -- (2.5,0.8);
        \draw[black,thick] (G2) -- (2.5,0.2);

        \draw[black,thick] (G3) -- (2.5,-0.8);
        \draw[black,thick] (G3) -- (2.5,-0.2);
    \end{scope}

\end{tikzpicture}
    }
    \caption{Reduction rules A and B. Interpreting the represented chain complex as a CSS code, we can see that rule A corresponds to merging the two qubits in the support of a weight-two $X$ stabilizer, while rule B corresponds to removing a qubit in the support of a weight-one $X$ stabilizer.}
    \label{fig:reduction-rules}
\end{figure}
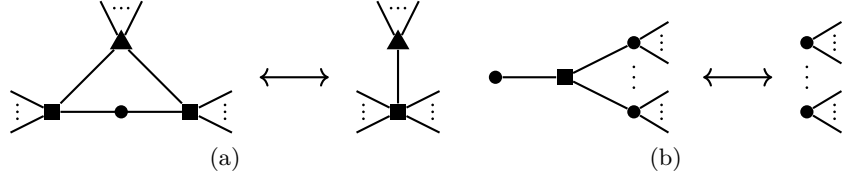
We now introduce two specific fault-tolerant maps, called reduction rules A and B.

\subsection{Rule A} \label{sec:rule-a}

Rule A is represented graphically in \cref{fig:reduction-rule-a}.
To describe it, let us consider a chain complex
\begin{align*}
    C_2 \xrightarrow{\partial_2} C_1 \xrightarrow{\partial_1} C_0
\end{align*}
along with bases $s_1,\ldots,s_{n_0}$ for $C_0$, $e_1,\ldots,e_{n_1}$ for $C_1$ and $g_1,\ldots,g_{n_2}$ for $C_2$.
To apply the rule, we assume the existence of a node of type $C_2$ only connected to two nodes of type $C_1$.
Without loss of generality, we assume that the node of type $C_2$ is $g_{n_2}$ and the two nodes of type $C_1$ are $e_{n_1-1}$ and $e_{n_1}$, that is, $\partial_2(g_{n_2})=e_{n_1-1}+e_{n_1}$.
Rule A then consists of removing $g_{n_2}$ and merging $e_{n_1-1}$ and $e_{n_1}$.
More precisely, it is a pair of linear maps $(f_A,h_A)$ with $f_A:C_{\bullet}:C'_{\bullet}$ and $h_A:C_{\bullet}:C'_{\bullet}$, where $C'_{\bullet}$ is defined as follows:
\begin{itemize}
    \item $C'_2$ has dimension $n_2-1$ and is generated by a basis $g'_1,\ldots,g'_{n_2-1}$,
    \item $C'_1$ has dimension $n_1-1$ and is generated by a basis $e'_1,\ldots,e'_{n_1-1}$,
    \item $C'_0=C_0$,
    \item The boundary operators transform the following way:
    \begin{align}
        \partial_2 &=
        \begin{pNiceArray}{ccc|c}[first-row,first-col]
            & g_1 & \cdots & g_{n_2-1} & g_{n_2} \\
            e_1 & \Block{3-3}{\tilde{\partial}_2}  &   &   & 0 \\
            \vdots & & &   & \vdots \\
            e_{n_1-2} &   &   &   & 0 \\
            \cline{1-4}
            e_{n_1-1} & \Block{1-3}{a} & &   & 1 \\
            e_{n_1} & \Block{1-3}{b} & &   & 1
        \end{pNiceArray}
        & \mapsto\;\;\; &
        \partial'_2=
        \begin{pNiceArray}{ccc}[first-row,first-col]
            & g'_1 & \cdots & g'_{n_2-1} \\
            e'_1 & \Block{3-3}{\tilde{\partial}_2}  &   & \\
            \vdots & & & \\
            e'_{n_1-2} &   &   & \\
            \cline{1-3}
            e'_{n_1-1} & \Block{1-3}{a+b} & &
        \end{pNiceArray},
        \\[2em]
        \partial_1 & =
        \begin{pNiceArray}{ccc|cc}[first-row, first-col]
            & e_1 & \cdots & e_{n_1-2} & e_{n_1-1} & e_{n_1} \\
            s_1 & \Block{3-3}{\tilde{\partial}_1} & & & \Block{3-1}{c} & \Block{3-1}{c} \\
            \vdots & & & & & \\
            s_{n_0} & & & & &
        \end{pNiceArray}
        & \mapsto\;\;\; &
        \partial'_1=
        \begin{pNiceArray}{ccc|c}[first-row, first-col]
            & e'_1 & \cdots & e'_{n_1-2} & e'_{n_1-1} \\
            s'_1 & \Block{3-3}{\tilde{\partial}_1} & & & \Block{3-1}{c} \\
            \vdots & & & & \\
            s'_{n_0} & & & &
        \end{pNiceArray}.
    \end{align}
    Note that $\partial_1$ can be written this way because if a node $s_i$ is connected to $e_{n_1-1}$, it must also be connected to $e_{n_1}$ for the chain complex condition to be valid (otherwise, $s_i$ and $g_{n_2}$ would have a neighborhood overlapping on an odd number of elements).
\end{itemize}
The map $f_A=(f_0,f_1,f_2)$ is defined as follows:
\begin{itemize}
    \item $f_2(g_{n_2})=0$ and $f_2(g_i)=g'_i$ for all $i \leq n_2-1$,
    In matrix form:
    \begin{align}
        f_2=
        \begin{pNiceArray}{ccc|c}[first-row, first-col]
            & g_1 & \cdots & g_{n_2-1} & g_{n_2} \\
            g'_1 & \Block{3-3}{I_{n_2-1}}  &   &   & 0 \\
            \vdots & & &   & \vdots \\
            g'_{n_2-1} &   &   &   & 0
        \end{pNiceArray},
    \end{align}
    \item $f_1(e_{n_1})=e'_{n_1-1}$ and $f_1(e_i)=e'_i$ for all $i \leq n_1-1$,
    In matrix form:
    \begin{align}
        f_1=
        \begin{pNiceArray}{ccc|c}[first-row, first-col]
            & e_1 & \cdots & e_{n_1-1} & e_{n_1} \\
            e'_1 & \Block{4-3}{I_{n_1-1}}  &   &   & 0 \\
            \vdots & & &   & \vdots \\
            e'_{n_1-2} &   &   &   & 0 \\
            e'_{n_1-1} &   &   &   & 1
        \end{pNiceArray},
    \end{align}
    \item $f_0(s_i)=s'_i$ for all $i \leq n_0$,
\end{itemize}
and the map $h_A=(h_0,h_1,h_2)$ is defined as follows:
\begin{itemize}
    \item $h_2(g'_i)=g_i$ for all $i \leq n_2-1$. In matrix form:
    \begin{align}
        h_2=
        \begin{pNiceArray}{ccc}[first-row,first-col]
            & g'_1 & \cdots & g'_{n_2-1} \\
            g_1 & \Block{3-3}{I_{n_2-1}}  &   & \\
            \vdots & & & \\
            g_{n_2-1} &   &   & \\
            \cline{1-3}
            g_{n_2} & \Block{1-3}{0} & &
        \end{pNiceArray},
    \end{align}
    \item $h_1(e'_i)=e_i$ for all $i \leq n_1-1$. In matrix form:
    \begin{align}
        h_1=
        \begin{pNiceArray}{ccc}[first-row,first-col]
            & e'_1 & \cdots & e'_{n_1-1} \\
            e_1 & \Block{3-3}{I_{n_1-1}}  &   & \\
            \vdots & & & \\
            e_{n_1-1} &   &   & \\
            \cline{1-3}
            e_{n_1} & \Block{1-3}{0} & &
        \end{pNiceArray},
    \end{align}
    \item $h_0(s'_i)=s_i$ for all $i \leq n_0$.
\end{itemize}

To prove that $f_A$ is a fault-tolerant map, we will need the following two lemmas:
\begin{lemma}
    \label{lemma:condition-quasi-isomorphism}
    Let $f:C_{\bullet} \rightarrow C'_{\bullet}$ and $h:C'_{\bullet} \rightarrow C_{\bullet}$ be two weak chain maps between length-2 chain complexes $C_{\bullet}$ and $C'_{\bullet}$ such that $[f_1(h_1(x))]=[x]$ for all $x \in C'_1$, and $[h_1(f_1(x))]=[x]$ for all $x \in C_1$. Then $f$ and $h$ are both weak quasi-isomorphisms, and we have $f^\star \circ h^\star = \mathrm{id}_{H_1(C'_\bullet)}$ and $h^\star \circ f^\star = \mathrm{id}_{H_1(C_\bullet)}$.
\end{lemma}
\begin{proof}
    Let $x \in C'_1$. We have $f^\star(h^\star([x]))=[f_1(h_1(x))]=[x]$. Therefore, $f^\star \circ h^\star$ is the identity on $H_1(C'_\bullet)$. Similarly, for $x \in C_1$, we have $h^\star(f^\star([x]))=[h_1(f_1(x))]=[x]$, which shows that $h^\star \circ f^\star$ is the identity on $H_1(C_\bullet)$. Thus, $f^\star$ and $h^\star$ are both isomorphisms, and therefore $f$ and $h$ are weak quasi-isomorphisms.
\end{proof}
\begin{lemma}
    \label{lemma:distance-preserving-and-decoding-preserving}
    Let $f:C_{\bullet} \rightarrow C'_{\bullet}$ and $h:C'_{\bullet} \rightarrow C_{\bullet}$ be two weak chain maps between length-2 chain complexes $C_{\bullet}$ and $C'_{\bullet}$ such that
    \begin{itemize}
        \item$f_0 \circ h_0 = \mathrm{id}_{C'_1}$ and $h_0 \circ f_0 = \mathrm{id}_{C_1}$,
        \item $f_1$ preserves the weight of any minimal-weight logical error,
        \item $f_1$ and $h_1$ are weight-nonincreasing, i.e. $|f_1(x)| \leq |x|$ for every $x \in C_1$ and $|h_1(x)| \leq |x|$ for every $x \in C'_1$.
    \end{itemize}
    Then $f$ and $h$ are both distance-preserving and decoding-preserving.
\end{lemma}

\begin{proof}
    We start by showing that $f$ is decoding-preserving.
    Let $s \in \Im(\partial_1)$.
    We need to show that $f_1(x^\star(s))=x'^\star(f_0(s))$, or in other words, that $f_1(x^\star(s))$ is the minimum-weight element $x' \in C'_1$ such that $\partial'_1(x')=f_0(s)$. Let us start by showing that the constraint $\partial'_1(x')=f_0(s)$ is fulfilled:
    \begin{align}
        \partial'_1(f_1(x^\star(s))) &= f_0(\partial_1(x^\star(s))) = f_0(s).
    \end{align}
    We then need to show that $f_1(x^\star(s))$ is the minimum-weight element of $C'_1$ satisfying this constraint. We prove this by contradiction. Let $y \in C'_1$ satisfying $\partial'_1(y)=f_0(s)$ and $|y|<|f_1(x^\star(s))|$. Using the weight-nonincreasing properties of $f_1$ and $h_1$, we get the series of inequalities:
    \begin{align}
        |h_1(y)| \leq |y| < |f_1(x^\star(s))| \leq |x^\star(s)|.
    \end{align}
    Moreover, $\partial_1(h_1(y))=h_0(\partial'_1(y))=h_0(f_0(s))=s$.
    Therefore, $h_1(y)$ is an element of $C_1$ satisfying the syndrome constraint with a lower weight than $|x^\star(s)|$, which contradicts the definition of $x^\star(s)$.
    Thus, $f_1(x^\star(s))$ is the minimum-weight element of $C'_1$ satisfying the constraint, and $f$ is decoding-preserving.
    The same proof with $f$ and $h$ switched allows us to show that $h$ is decoding-preserving as well.

    We now show that $f$ and $h$ are distance-preserving, meaning that $C_\bullet$ and $C_\bullet'$ have the same distance.
    Let $x \in C_1$ be a minimal-weight non-trivial logical error.
    By applying the proof above to $s=0$, we can see that $f_1(x)$ is a minimum-weight logical error in $C'_\bullet$.
    Moreover, since $f$ is a weak chain map, \cref{theorem:weak-chain-map-induces-homology-map} shows that it maps non-trivial logical errors to non-trivial logicals errors. Therefore, $f_1(x)$ is a minimum-weight non-trivial logical error of $C'_\bullet$. Since by assumption, $|f_1(x)|=|x|$, we have $d(C'_\bullet) = d(C_\bullet)$.
    Therefore, both $f$ and $h$ are distance-preserving.
\end{proof}

\begin{theorem} \label{theorem:rule-a}
    $f_A$ and $h_A$ are fault-tolerant maps
\end{theorem}
\begin{proof}
    We start by proving that $f_A$ and $h_A$ are weak chain maps. We first show that the diagram for $f_A$,
    \[\begin{tikzcd}[column sep=20pt, row sep=20pt, every cell/.append style={inner sep=4pt}]
        {C_2} & {C_1} & {C_0} \\
        {C'_2} & {C'_1} & {C'_0}
        \arrow["{\partial_2}", from=1-1, to=1-2]
        \arrow["{f_2}"', from=1-1, to=2-1]
        \arrow["{\partial_1}", from=1-2, to=1-3]
        \arrow["{f_1}"', from=1-2, to=2-2]
        \arrow["{f_0}"', from=1-3, to=2-3]
        \arrow["{\partial'_2}", from=2-1, to=2-2]
        \arrow["{\partial'_1}", from=2-2, to=2-3]
    \end{tikzcd}\]
    commutes, making $f_A$ a chain map, and therefore a weak chain map. We can prove this by multiplying the matrices of the different operators involved:
    \begin{align}
        f_1 \circ \partial_2 &= \partial'_2 \circ f_2 =
        \begin{pNiceArray}{c|c}
            \Block{2-1}{\tilde{\partial}_2}  & 0 \\
            & \vdots \\
            a+b &  0
        \end{pNiceArray}, \\
        f_0 \circ \partial_1 &= \partial'_1 \circ f_1 =
        \begin{pNiceArray}{ccc}
            \tilde{\partial}_1 & c & c
        \end{pNiceArray}.
    \end{align}
    Similarly, we show that the right square of the diagram for $h_A$,
    \[\begin{tikzcd}[column sep=20pt, row sep=20pt, every cell/.append style={inner sep=4pt}]
        {C'_2} & {C'_1} & {C'_0} \\
        {C_2} & {C_1} & {C_0}
        \arrow["{\partial'_2}", from=1-1, to=1-2]
        \arrow["{h_2}"', from=1-1, to=2-1]
        \arrow["{\partial'_1}", from=1-2, to=1-3]
        \arrow["{h_1}"', from=1-2, to=2-2]
        \arrow["{h_0}"', from=1-3, to=2-3]
        \arrow["{\partial_2}", from=2-1, to=2-2]
        \arrow["{\partial_1}", from=2-2, to=2-3]
    \end{tikzcd}\]
    commutes, using matrix multiplication:
    \begin{align}
        h_0 \circ \partial_1' = \partial_1 \circ h_1 =
        \begin{pNiceArray}{cc}
            \tilde{\partial}_1 & c
        \end{pNiceArray}.
    \end{align}
    Moreover, for the left square, we have
    \begin{align}
        h_1 \circ \partial_2' =
        \begin{pNiceArray}{c}
            \tilde{\partial}_2 \\
            a + b \\
            0
        \end{pNiceArray}.
    \end{align}
    Applying it to any $x \in C_2'$, we obtain
    \begin{align}
        h_1(\partial_2'(x)) =
        \begin{pNiceArray}{c}
            \tilde{\partial}_2  x\\
            (a + b)^T x \\
            0
        \end{pNiceArray}.
    \end{align}
    Every element in the image of $\partial_2$ can be written as
    \begin{align}
        \partial_2
        \begin{pNiceArray}{c}
            x\\
            y
        \end{pNiceArray}
        =
        \begin{pNiceArray}{c}
            \tilde{\partial}_2  x\\
            a^T x + y \\
            b^T x + y
        \end{pNiceArray},
    \end{align}
    where $x \in \mathbb{Z}_2^{n_2-1}$ and $y \in \mathbb{Z}_2$.
    Choosing $y=b^T x$, we obtain
    \begin{align}
        \partial_2
        \begin{pNiceArray}{c}
            x\\
            y
        \end{pNiceArray}
        =
        \begin{pNiceArray}{c}
            \tilde{\partial}_2  x\\
            a^T x + b^T x \\
            b^T x + b^T x
        \end{pNiceArray}
        =
        \begin{pNiceArray}{c}
            \tilde{\partial}_2  x\\
            (a+b)^T x \\
            0
        \end{pNiceArray}.
    \end{align}
    Therefore, $h_1(\partial_2'(x)) \in \Im(\partial_2)$ for all $x \in C_2'$, and $h$ is a weak chain map.

    We now show using \cref{lemma:condition-quasi-isomorphism} that $f$ and $h$ are both weak quasi-isomorphisms. Let us calculate $f_1 \circ h_1$ and $h_1 \circ f_1$ in matrix form:
    \begin{align}
        f_1 \circ h_1 &= I_{n_1 - 1}, \\
        h_1 \circ f_1 &=
        \begin{pNiceArray}{c|c}
            \Block{4-1}{I_{n_1 - 1}}  & 0 \\
                                      & \vdots \\
                                      & 0 \\
                                      & 1 \\
            \cline{1-2}
            \Block{1-2}{0} &
        \end{pNiceArray}
        =
        I_{n_1} +
        \begin{pNiceArray}{cccc|c}
            \Block{5-4}{0}  & & & & 0 \\
                            & & & & \vdots \\
                            & & & & 0 \\
                            & & & & 1 \\
                            & & & & 1
        \end{pNiceArray}.
    \end{align}
    Applying $h_1 \circ f_1$ to any $x \in C_1$, we obtain
    \begin{align}
        h_1(f_1(x)) &= x +
        \begin{pmatrix}
            0 \\
            \vdots \\
            0 \\
            x_{n_1} \\
            x_{n_1}
        \end{pmatrix}
        = x + x_{n_1} (e_{n_1-1} + e_{n_1})
        = x + {\color{teal}x_{n_1}}\partial_2(g_{n_2}),
    \end{align}
    where $x_{n_1} = x \cdot e_{n_1}$.
    Therefore, $h_1(f_1(x)) \in x + \Im(\partial_2)$, and $[h_1(f_1(x))]=[x]$.
    Moreover, since $f_1(h_1(x')) = x'$ for all $x' \in C_1'$, we also have $[f_1(h_1(x'))]=[x']$.
    By \cref{lemma:condition-quasi-isomorphism}, we can conclude that $f_A$ and $h_A$ are both weak quasi-isomorphisms.

    It remains to show that $f_A$ and $h_A$ are distance- and decoding-preserving. For this, we show that the three conditions of \cref{lemma:distance-preserving-and-decoding-preserving} are satisfied.
    Let us start with the first one.
    For every basis element $s_i$ of $C_0$ and $s'_j$ of $C'_0$, we have
    \begin{align}
        h_0(f_0(s_i)) = h_0(s'_i)=s_i, \\
        f_0(h_0(s'_j)) = f_0(s_j)=s_j.
    \end{align}
    Therefore, $h_0$ and $f_0$ are inverses of each other.

    To see that $f_1(x)$ has the same weight as $x$ if $x \in C_1$ is a minimum-weight logical error, let us consider four cases: $x$ is supported on (1) both $e_{n_1}$ and $e_{n_1-1}$, (2) $e_{n_1}$ but not $e_{n_1-1}$ (3) $e_{n_1-1}$ but not $e_{n_1}$, (4) neither $e_{n_1}$ nor $e_{n_1-1}$.
    We can see that the first case is not possible since $x+e_{n_1}+e_{n_1-1}=x+\partial_2(g_{n_2})$ would then be another non-trivial logical error with a weight reduced by two. In the second case, since $f_1(e_{n_1})=e_{n_1-1}$ while all the other nodes remain unchanged, $f_1$ preserves the weight of $x$.
    In the third and fourth case, $f_1$ acts as an identity on all the nodes in the support of $x$, and it therefore preserves its weight as well.

    Finally, we show that $f_1$ and $h_1$ are weight-nonincreasing.
    Let $x' \in C'_1$ such that $x'=\sum_i x'_i e'_i$.
    We have
    \begin{align}
        h_1(x')=\sum_{i=1}^{n_1-1} x'_i h_1(e'_i) = \sum_{i=1}^{n_1-1} x'_i e_i.
    \end{align}
    Therefore,
    \begin{align}
        |h_1(x')|=\sum_{i=1}^{n_1-1} x'_i = |x'|,
    \end{align}
    where in this equation the $x'_i$ are interpreted as integers.
    Thus, $h_1$ is weight-preserving, and therefore weight-nonincreasing.

    Let $x \in C_1$ such that $x=\sum_i x_i e_i$.
    We have:
    \begin{align}
        f_1(x)=\sum_{i=1}^{n_1} x_i f_1(e_i) = (x_{n_1 - 1} + x_{n_1}) e'_{n_1} + \sum_{i=1}^{n_1-2} x_i e'_i,
    \end{align}
    Therefore,
    \begin{align}
        |f_1(x)|=\sum_{i=1}^{n_1-2} x'_i + (x_{n_1-1} + x_{n_1} \mod 2).
    \end{align}
    Since
    \begin{align}
        (x_{n_1-1} + x_{n_1} \mod 2) \leq x_{n_1-1} + x_{n_1},
    \end{align}
    we have that $|f_1(x)| \leq |x|$ and $f_1$ is weight-nonincreasing.

    We finally apply \cref{lemma:distance-preserving-and-decoding-preserving} to conclude that $f_A$ and $h_A$ are distance- and decoding-preserving.
\end{proof}

As an application of rule A, we show in \cref{fig:bacon-shor-equivalence-c} that the gauge complex of the Bacon-Shor code can be reduced to two connected components, where each component can be seen as the Tanner graph of a classical repetition code (one for $X$ errors and one for $Z$ errors).
This is a new proof of the well-known fact that decoding the Bacon-Shor code is equivalent to decoding two repetition codes \cite{aliferis2007subsystem}.

\subsection{Rule B} \label{sec:rule-b}

We now move on to the description of rule B.
We assume the existence of a node of type $C_2$ connected to only one node of type $C_1$.
Without loss of generality, we assume that the node of type $C_2$ is $g_{n_2}$ and the node of type $C_1$ is $e_{n_1}$, that is, $\partial_2(g_{n_2})=e_{n_1}$.
Rule B then consists of removing $g_{n_2}$ and $e_{n_1}$.
More precisely, it is described by a pair of linear maps $(f_B,h_B)$ between the original chain complex $C_{\bullet}$ and a new chain complex $C'_{\bullet}$ defined in the following way:
\begin{itemize}
    \item $C'_2$ has dimension $n_2-1$ and is generated by a basis $g'_1,\ldots,g'_{n_2-1}$
    \item $C'_1$ has dimension $n_1-1$ and is generated by a basis $e'_1,\ldots,e'_{n_1-1}$
    \item $C'_0=C_0$
    \item The boundary operators $\partial'_1$ and $\partial'_2$ are defined by removing all the edges connected to $g_{n_2}$ and $e_{n_1}$ in the original boundary operators $\partial_1$ and $\partial_2$. More precisely, writing down the transformation in matrix form, we have:
    \begin{align}
        \partial_2 &=
        \begin{pNiceArray}{ccc|c}[first-row,first-col]
            & g_1 & \cdots & g_{n_2-1} & g_{n_2} \\
            e_1 & \Block{3-3}{\tilde{\partial}_2}  &   &   & 0 \\
            \vdots & & &   & \vdots \\
            e_{n_1-1} &  &   &   & 0 \\
            \cline{1-3}
            e_{n_1} & \Block{1-3}{a} &   &   & 1
        \end{pNiceArray}
        & \mapsto\;\;\;
        \partial'_2=\tilde{\partial}_2,
        \\[2em]
        \partial_1 & =
        \begin{pNiceArray}{ccc|c}[first-row, first-col]
            & e_1 & \cdots & e_{n_1-1} & e_{n_1} \\
            s_1 & \Block{3-3}{\tilde{\partial}_1} & & & 0 \\
            \vdots & & & & \vdots \\
            s_{n_0} & & & & 0
        \end{pNiceArray}
        & \mapsto\;\;\;
        \partial'_1=\tilde{\partial}_1.
    \end{align}
    Note that $\partial_1$ must have this form for the chain complex condition to be valid.
    Indeed, suppose that there existed a node $s_i \in C_0$ such that $\partial_1^T s_i =e_{n_1}+\ldots$.
    Then, the overlap with the neighbourhood of $g_{n_2}$, $\partial_2(g_{n_2})=e_{n_1}$, would be exactly $e_{n_1}$ and hence $\partial_2^T \partial_1^T s_i \neq 0$. This is why no triangle is connected to the square in \cref{fig:reduction-rule-b}.
\end{itemize}
The map $f_B=(f_0,f_1,f_2)$ is defined as follows:
\begin{itemize}
    \item $f_2(g_{n_2})=0$ and $f_2(g_i)=g'_i$ for all $i \leq n_2-1$. Or in matrix form:
    \begin{align}
        f_2=
        \begin{pNiceArray}{ccc|c}[first-row, first-col]
            & g_1 & \cdots & g_{n_2-1} & g_{n_2} \\
            g'_1 & \Block{3-3}{I_{n_2-1}}  &   &   & 0 \\
            \vdots & & &   & \vdots \\
            g'_{n_2-1} &   &   &   & 0
        \end{pNiceArray},
    \end{align}
    \item $f_1(e_{n_1})=0$ and $f_1(e_i)=e'_i$ for all $i \leq n_1-1$. Or in matrix form:
    \begin{align}
        f_1=
        \begin{pNiceArray}{ccc|c}[first-row, first-col]
            & e_1 & \cdots & e_{n_1-1} & e_{n_1} \\
            e'_1 & \Block{3-3}{I_{n_1-1}}  &   &   & 0 \\
            \vdots & & &   & \vdots \\
            e'_{n_1-1} &   &   &   & 0
        \end{pNiceArray},
    \end{align}
    \item $f_0$ is the identity map on $C_0$
\end{itemize}
and the map $h_B=(h_0,h_1,h_2)$ is defined as follows:
\begin{itemize}
    \item $h_2(g'_i)=g_i$ for all $i \leq n_2-1$. In matrix form:
    \begin{align}
        h_2=
        \begin{pNiceArray}{ccc}[first-row,first-col]
            & g'_1 & \cdots & g'_{n_2-1} \\
            g_1 & \Block{3-3}{I_{n_2-1}}  &   & \\
            \vdots & & & \\
            g_{n_2-1} &   &   & \\
            \cline{1-3}
            g_{n_2} & \Block{1-3}{0} & &
        \end{pNiceArray},
    \end{align}
    \item $h_1(e'_i)=e_i$ for all $i \leq n_1-1$. In matrix form:
    \begin{align}
        h_1=
        \begin{pNiceArray}{ccc}[first-row,first-col]
            & e'_1 & \cdots & e'_{n_1-1} \\
            e_1 & \Block{3-3}{I_{n_1-1}}  &   & \\
            \vdots & & & \\
            e_{n_1-1} &   &   & \\
            \cline{1-3}
            e_{n_1} & \Block{1-3}{0} & &
        \end{pNiceArray},
    \end{align}
    \item $h_0(s'_i)=s_i$ for all $i \leq n_0$.
\end{itemize}
Let us now show that $f_B$ and $h_B$ are fault-tolerant maps.

\begin{theorem} \label{theorem:rule-b}
    $f_B$ and $h_B$ are fault-tolerant maps
\end{theorem}
\begin{proof}
    We start by proving that $f_B$ and $h_B$ are weak chain maps. We first show that the diagram for $f_B$,
    \[\begin{tikzcd}[column sep=20pt, row sep=20pt, every cell/.append style={inner sep=4pt}]
        {C_2} & {C_1} & {C_0} \\
        {C'_2} & {C'_1} & {C'_0}
        \arrow["{\partial_2}", from=1-1, to=1-2]
        \arrow["{f_2}"', from=1-1, to=2-1]
        \arrow["{\partial_1}", from=1-2, to=1-3]
        \arrow["{f_1}"', from=1-2, to=2-2]
        \arrow["{f_0}"', from=1-3, to=2-3]
        \arrow["{\partial'_2}", from=2-1, to=2-2]
        \arrow["{\partial'_1}", from=2-2, to=2-3]
    \end{tikzcd}\]
    commutes.
    We can prove this by multiplying the matrices of the different operators involved:
    \begin{align}
        f_1 \circ \partial_2 &= \partial'_2 \circ f_2 =
        \begin{pNiceArray}{c|c}
            \tilde{\partial}_2  & 0
        \end{pNiceArray}, \\
        f_0 \circ \partial_1 &= \partial'_1 \circ f_1 = \partial_1.
    \end{align}
    Similarly, we show that right square of the diagram for $h_B$,
    \[\begin{tikzcd}[column sep=20pt, row sep=20pt, every cell/.append style={inner sep=4pt}]
        {C'_2} & {C'_1} & {C'_0} \\
        {C_2} & {C_1} & {C_0}
        \arrow["{\partial'_2}", from=1-1, to=1-2]
        \arrow["{h_2}"', from=1-1, to=2-1]
        \arrow["{\partial'_1}", from=1-2, to=1-3]
        \arrow["{h_1}"', from=1-2, to=2-2]
        \arrow["{h_0}"', from=1-3, to=2-3]
        \arrow["{\partial_2}", from=2-1, to=2-2]
        \arrow["{\partial_1}", from=2-2, to=2-3]
    \end{tikzcd}\]
    commutes, using matrix multiplication:
    \begin{align}
        h_0 \circ \partial_1' &= \partial_1 \circ h_1 = \tilde{\partial}_1.
    \end{align}
    Moreover, we have on the left square:
    \begin{align}
        h_1 \circ \partial_2' =
        \begin{pNiceArray}{c}
            \tilde{\partial}_2 \\
            0
        \end{pNiceArray}.
    \end{align}
    Applying it to any $x \in C_2'$, we obtain:
    \begin{align}
        h_1(\partial_2'(x)) =
        \begin{pNiceArray}{c}
            \tilde{\partial}_2  x\\
            0
        \end{pNiceArray}.
    \end{align}
    Every element in the image of $\partial_2$ can be written as:
    \begin{align}
        \partial_2
        \begin{pNiceArray}{c}
            x\\
            y
        \end{pNiceArray}
        =
        \begin{pNiceArray}{c}
            \tilde{\partial}_2  x\\
            a^T x + y \\
        \end{pNiceArray}.
    \end{align}
    where $x \in \mathbb{Z}_2^{n_2-1}$ and $y \in \mathbb{Z}_2$.
    Choosing $y=a^T x$, we obtain:
    \begin{align}
        \partial_2
        \begin{pNiceArray}{c}
            x\\
            y
        \end{pNiceArray}
        =
        \begin{pNiceArray}{c}
            \tilde{\partial}_2  x\\
            a^T x + a^T x \\
        \end{pNiceArray}
        =
        \begin{pNiceArray}{c}
            \tilde{\partial}_2  x\\
            0
        \end{pNiceArray}.
    \end{align}
    Therefore, $h_1(\partial_2'(x)) \in \Im(\partial_2)$ for all $x \in C_2'$, and $h$ is a weak chain map.

    We now show using \cref{lemma:condition-quasi-isomorphism} that $f$ and $h$ are both weak quasi-isomorphisms. We first calculate $f_1 \circ h_1$ and $h_1 \circ f_1$ in matrix form:
    \begin{align}
        f_1 \circ h_1 &= I_{n_1 - 1}, \\
        h_1 \circ f_1 &=
        \begin{pNiceArray}{c|c}
            \Block{4-1}{I_{n_1 - 1}}  & 0 \\
                                      & \vdots \\
                                      & 0 \\
            \cline{1-2}
            \Block{1-2}{0} &
        \end{pNiceArray}
        =
        I_{n_1} +
        \begin{pNiceArray}{cccc|c}
            \Block{4-4}{0}  & & & & 0 \\
                            & & & & \vdots \\
                            & & & & 0 \\
            \cline{1-4}
            \Block{1-4}{0}  & & & & 1
        \end{pNiceArray}.
    \end{align}
    Applying $h_1 \circ f_1$ to any $x \in C_1$, we obtain:
    \begin{align}
        h_1(f_1(x)) &= x +
        \begin{pmatrix}
            0 \\
            \vdots \\
            0 \\
            x_{n_1}
        \end{pmatrix}
        = x + x_{n_1} e_{n_1}
        = x + {\color{teal}x_{n_1}}\partial_2(g_{n_2})
    \end{align}
    where $x_{n_1} = x \cdot e_{n_1}$.
    Therefore, $h_1(f_1(x)) \in x + \Im(\partial_2)$, and $[h_1(f_1(x))]=[x]$.
    Moreover, since $f_1(h_1(x')) = x'$ for all $x' \in C_1'$, we also have $[f_1(h_1(x'))]=[x']$.
    By \cref{lemma:condition-quasi-isomorphism}, we can conclude that $f_B$ and $h_B$ are both weak quasi-isomorphisms.

    It remains to show that $f_B$ and $h_B$ are distance- and decoding-preserving. For this, we show that the three conditions of \cref{lemma:distance-preserving-and-decoding-preserving} are satisfied.
    Let us start with the first one.
    For every basis element $s_i$ of $C_0$ and $s'_j$ of $C'_0$, we have:
    \begin{align}
        h_0(f_0(s_i)) = h_0(s'_i)=s_i, \\
        f_0(h_0(s'_j)) = f_0(s_j)=s_j.
    \end{align}
    Therefore, $h_0$ and $f_0$ are inverses of each other.

    To see that $f_1(x)$ has the same weight as $x$ if $x$ is a minimum-weight logical error, we notice that $x$ cannot be supported on $e_{n_1}$, since $x+e_{n_1}=x+\partial_2(g_{n_2})$ would then be another non-trivial logical error with a weight reduced by one.
    Since all the other nodes remain unchanged through the action of $f_1$, $x$ and $f_1(x)$ must have the same weight.

    It remains to show that $f_1$ and $h_1$ are weight-nonincreasing.
    Let $x' \in C'_1$ such that $x'=\sum_i x'_i e'_i$.
    We have:
    \begin{align}
        h_1(x')=\sum_{i=1}^{n_1-1} x'_i h_1(e'_i) = \sum_{i=1}^{n_1-1} x'_i e_i.
    \end{align}
    Therefore,
    \begin{align}
        |h_1(x')|=\sum_{i=1}^{n_1-1} x'_i = |x'|,
    \end{align}
    where in this equation the $x'_i$ are interpreted as integers.
    Thus, $h_1$ is weight-preserving, and therefore weight-nonincreasing.

    Let $x \in C_1$ such that $x=\sum_i x_i e_i$.
    We have:
    \begin{align}
        f_1(x)=\sum_{i=1}^{n_1} x_i f_1(e_i) = \sum_{i=1}^{n_1-1} x_i e'_i.
    \end{align}
    Therefore,
    \begin{align}
        |f_1(x)|=\sum_{i=1}^{n_1-1} x_i \leq \sum_{i=1}^{n_1} x_i = |x|.
    \end{align}
    We finally apply \cref{lemma:distance-preserving-and-decoding-preserving} to conclude that $f_A$ and $h_A$ are distance- and decoding-preserving.
\end{proof}

\section{Spacetime codes}
\label{sec:spacetime-codes}
In this section, we review the spacetime codes formalism.
Following Refs~\cite{bacon2015sparse,gottesman2022opportunities} we define the spacetime code associated to a Clifford circuit as a subsystem code, with gauge operators for each gate, measurement, and input state stabilizer.
The stabilizers of this subsystem code then either correspond to redundancies between measurements---often referred to as detectors~\cite{gidney2021stim,higgott2025sparse,derks2024designing}---or measurements that could be made redundant by the addition of measurements later in the circuit (which we call incomplete detectors).
Our approach is slightly different to that of Delfosse and Paetznick~\cite{delfosse2023spacetime}, where spacetime codes are stabilizer codes, with only complete detectors considered as stabilizers.
We do however make use of results and definitions from their work.

\subsection{Outcome code}

Let us consider a circuit consisting Clifford gates and (multi-qubit) Pauli measurements.
The input to the circuit is defined by a stabilizer group $\mathcal S$, such that the input state is in the codespace of the corresponding stabilizer code.
We note that this, for example, allows us to consider auxiliary qubits that are prepared in Pauli eigenstates.

The core idea behind spacetime codes is the observation that, in a circuit designed to do error-correction, there are redundancies amongst measurement outcomes\footnote{Such redundancies between measurement outcomes can also appear in circuits that are not associated to any particular code; see \cref{fig:outcome-code-example} for an example.}.
For instance, in the circuit consisting of multiple rounds of stabilizer measurements of a stabilizer code, each pair $(m_i,m_{i+1})$ of successive measurement outcomes of a given stabilizer should give the same result, that is $m_i+m_{i+1}=0$, where $m_i$ and $m_{i+1}$ are binary variables encoding the measurement outcomes.

Suppose that the input stabilizer $\mathcal S$ contains the stabilizers being measured in the circuit, then the first set of stabilizer measurement outcomes will be deterministic, i.e., $m_1=0$.
While the constraint $m_i+m_{i+1}=0$ only depends on the circuit itself, the constraint $m_1=0$ depends on the circuit and the input stabilizer group.
To put these two types of constraints on equal footing, we formally consider input stabilizers as measurements at the beginning of the circuit, whose output is post-selected to be 0.
Denoting the input measurements by $s_1,\ldots,s_k$ (one for each generator of the input stabilizer group), we always have $s_1=\ldots=s_k=0$ after post-selection.
However, when enumerating the redundancies between the measurements in a circuit, we will always suppose that those initial stabilizer measurements can take any value, as a trick to retain the stabilizers involved in a given redundancy.
For instance, if a measurement $m_1$ somewhere in the circuit measures a stabilizer coming from the input space (that we replace by an initial measurement $s_1$), we will have $m_1+s_1=0 $ instead of $m_1=0$.
The fact that $s_1=0$ is only used when decoding a particular instance of the circuit. The example in \cref{fig:outcome-code-example} illustrates this idea.

It was shown in Ref.~\cite{delfosse2023spacetime} that, for any Clifford circuit with multi-qubit measurements, the set of measurement outcomes $\{m_1,\ldots,m_k\}$ forms an affine space. By considering input stabilizers as measurements themselves, as discussed above, this observation can be generalized to include stabilizers as well, that is, $\{m_1,\ldots,m_k,s_1,\ldots,s_{k}\}$ forms an affine space. From now on, unless stated otherwise, we will stop making the distinction between measurements arising from input stabilizers or from actual measurements.
By changing the sign of some measurement outcomes in post-processing, one can transform this affine space into a linear space. In other words, the set of measurement outcomes forms a (classical) linear code. Following the terminology of Ref.~\cite{delfosse2023spacetime}, we call this classical code the outcome code. The outcome code can be obtained efficiently from a given circuit, by following the procedure described in Ref.~\cite[Algorithm 1]{delfosse2023spacetime}.

If errors occur during the circuit, some checks of the outcome code might be unsatisfied. The objective of spacetime code constructions is to map unsatisfied outcome code checks back to errors happening at different locations of the circuit---in other words, to rigorously formulate the decoding problem in the context of a circuit-level error model.

\begin{figure}
    \centering
    \tikzsetnextfilename{sec3-outcome-code-example}
    \begin{quantikz}
                                                              & \labeledwire{niceblue}{(1\char044 1)} &                                     & \labeledwire{niceblue}{(1\char044 2)} & \targ{}   & \labeledwire{niceblue}{(1\char044 3)} & \gate[2]{\labeledmeas{m_2}{M_{XX}}} & \labeledwire{niceblue}{(1\char044 4)} & & \labeledwire{niceblue}{(1\char044 5)} \\
                                                              & \labeledwire{niceblue}{(2\char044 1)} & \gate[2]{\labeledmeas{m_1}{M_{XZ}}} & \labeledwire{niceblue}{(2\char044 2)} & \ctrl{-1} & \labeledwire{niceblue}{(2\char044 3)} &                                     & \labeledwire{niceblue}{(2\char044 4)} & & \labeledwire{niceblue}{(2\char044 5)} \\
                                                              & \labeledwire{niceblue}{(3\char044 1)} &                                     & \labeledwire{niceblue}{(3\char044 2)} & \ctrl{1}  & \labeledwire{niceblue}{(3\char044 3)} & \gate{\labeledmeas{m_3}{M_Z}}       & \labeledwire{niceblue}{(3\char044 4)} & & \labeledwire{niceblue}{(3\char044 5)} \\
        \lstick{$\overset{{\color{niceblue} s_1}}{\ket{0}}$} & \labeledwire{niceblue}{(4\char044 1)} &                                     & \labeledwire{niceblue}{(4\char044 2)} & \targ{}   & \labeledwire{niceblue}{(4\char044 3)} & \gate{\labeledmeas{m_4}{M_Z}}       & \labeledwire{niceblue}{(4\char044 4)} & & \labeledwire{purple}{(4\char044 5)}
    \end{quantikz}
    \caption{ \label{fig:outcome-code-example}
        Example of a circuit with redundant measurements. Red labels in measurement boxes represent the measurement outcomes, and blue labels on top of input states represent input stabilizers. The outcome code of the circuit has two checks, which we can write as $\textcolor{nicered}{m_1}+\textcolor{nicered}{m_2}+\textcolor{nicered}{m_3}=0$ and $\textcolor{niceblue}{s_1}+\textcolor{nicered}{m_3}+\textcolor{nicered}{m_4}=0$.
        We also indicate in blue the coordinates $(i,t)$ of each spacetime location of the circuit. The last column of the circuit is added to fulfill the requirement that the number of timesteps must be odd.
    }
\end{figure}
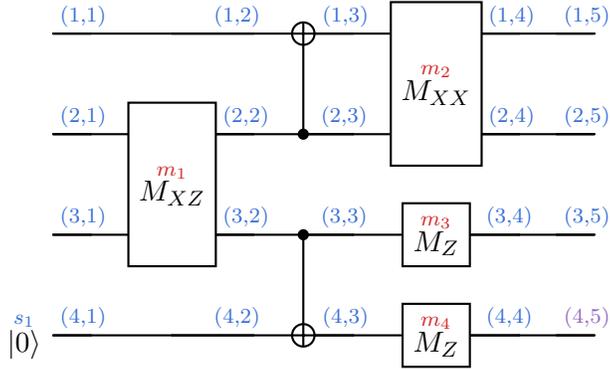

\subsection{Detectors of the spacetime code}
\label{sec:detectors-spacetime-code}

Once we know the outcome code of a circuit, the next step is to associate outcome code checks to fault locations and error types ($X$, $Y$ or $Z$) in the circuit.
For this, we first need to introduce a coordinate system representing possible fault locations in the circuit.
Let us assume that every qubit is associated to a horizontal wire encompassing the whole circuit, and divide the circuit into $T$ time steps.
We call $(i,t)$ the coordinates corresponding to row (or qubit) $i$ and column (or time step) $t$ of the circuit, where $i \in \{1,\ldots,n\}$ and $t \in \{1,\ldots,T\}$, with $t=1$ representing the column of the circuit located before all the gates, and $t=T$ representing the column located after all the gates.
The formalism requires that $T$ is odd, but this is not really a restriction as any circuit with even $T$ can be extended by one time step without changing the properties of the spacetime code.
An example of circuit with a coordinate system as above is shown in \cref{fig:outcome-code-example}

We define a \textit{spacetime operator} as an assignment of a Pauli operator $P \in \{I,X,Y,Z\}$ to each coordinate of the circuit.
We denote $P_{i,t}$ the spacetime operator comprising $P$ at coordinates $(i,t)$ and identity everywhere else.
More generally, for an $n$-qubit Pauli operator $A$, we denote $A_t$ the corresponding spacetime operator located at time $t$.

We now consider an error model where Pauli errors can occur at every location of the circuit with a certain probability, and call the corresponding spacetime operator a \textit{fault}.
The goal is to associate to each check of the outcome code a spacetime operator, called a \textit{detector}, such that a fault anticommutes with the detector if and only if it is detectable by the check.

We construct detectors in following way.
Let us define the propagation of a Pauli operator $P$ from time $t$ to time $t' \geq t$ as $\Pi_{t \rightarrow t'}(P)=UPU^\dagger$ where $U$ is the unitary defined by all the unitary gates between $t$ and $t'$ ($U=I$ if $t=t'$).
Any measurement between $t$ and $t'$ is considered to act as the identity, independent of whether $P$ commutes or anticommutes with it.
Following the terminology of Refs. \cite{bacon2015sparse, gottesman2022opportunities}, we then define the spackle of $P$ at $t$ as
\begin{align}
    \spackle_t(P)=\prod_{s=t}^T \left[ \Pi_{t\rightarrow s}(P) \right]_s.
\end{align}
By extension, the spackle of a measurement $M$ occurring between steps $t$ and $t+1$ is the spackle at $t+1$ of the Pauli operator $P$ being measured, i.e. $\spackle(M)=\spackle_{t+1}(P)$. Finally, the spackle of multiple measurements is the product of the individual spackles, i.e. $\spackle(M_1,\ldots,M_k)=\spackle(M_1) \ldots \spackle(M_k)$. Intuitively, the spackle of $P$ is its propagation to all the following time steps, leaving a trace of the propagation at every time step. Examples of spackles are shown in \cref{fig:spackles}.

\begin{figure}[t]
    \centering
    \subfloat[Spackle of $m_1$]{ \label{fig:spackle-m1}
        \tikzsetnextfilename{sec3-spackle-m1}
        \begin{quantikz}
                                                                  &                                     &                         & \targ{}   & \labeledwire{purple}{X} & \gate[2]{\labeledmeas{m_2}{M_{XX}}} & \labeledwire{purple}{X} & \labeledwire{purple}{X}   \\
                                                                  & \gate[2]{\labeledmeas{m_1}{M_{XZ}}} & \labeledwire{purple}{X} & \ctrl{-1} & \labeledwire{purple}{X} & \ghost{M_{XX}}                      & \labeledwire{purple}{X} & \labeledwire{purple}{X}   \\
                                                                  & \ghost{M_{XZ}}                      & \labeledwire{purple}{Z} & \ctrl{1}  & \labeledwire{purple}{Z} & \gate{\labeledmeas{m_3}{M_Z}}       & \labeledwire{purple}{Z} & \labeledwire{purple}{Z}   \\
            \lstick{$\overset{{\color{niceblue} s_1}} {\ket{0}}$} &                                     &                         & \targ{}   &                         & \gate{\labeledmeas{m_4}{M_Z}}       &                         &
        \end{quantikz}
     }
    \hspace{20pt}
    \subfloat[Spackle of $m_2$]{ \label{fig:spackle-m2}
        \tikzsetnextfilename{sec3-spackle-m2}
        \begin{quantikz}
                                                                  &                                     & \targ{}   & \gate[2]{\labeledmeas{m_2}{M_{XX}}} & \labeledwire{purple}{X} & \labeledwire{purple}{X}  \\
                                                                  & \gate[2]{\labeledmeas{m_1}{M_{XZ}}} & \ctrl{-1} & \ghost{M_{XX}}                      & \labeledwire{purple}{X} & \labeledwire{purple}{X}  \\
                                                                  & \ghost{M_{XZ}}                      & \ctrl{1}  & \gate{\labeledmeas{m_3}{M_Z}}       &                         &                          \\
            \lstick{$\overset{{\color{niceblue} s_1}} {\ket{0}}$} &                                     & \targ{}   & \gate{\labeledmeas{m_4}{M_Z}}       &                         &
        \end{quantikz}
    }
    \\
    \subfloat[Spackle of $m_3$]{ \label{fig:spackle-m3}
        \tikzsetnextfilename{sec3-spackle-m3}
        \begin{quantikz}
                                                                  &                                     & \targ{}   & \gate[2]{\labeledmeas{m_2}{M_{XX}}} &                         &                          \\
                                                                  & \gate[2]{\labeledmeas{m_1}{M_{XZ}}} & \ctrl{-1} & \ghost{M_{XX}}                      &                         &                          \\
                                                                  & \ghost{M_{XZ}}                      & \ctrl{1}  & \gate{\labeledmeas{m_3}{M_Z}}       & \labeledwire{purple}{Z} & \labeledwire{purple}{Z}  \\
            \lstick{$\overset{{\color{niceblue} s_1}} {\ket{0}}$} &                                     & \targ{}   & \gate{\labeledmeas{m_4}{M_Z}}       &                         &
        \end{quantikz}
    }
    \hspace{20pt}
    \subfloat[Spackle of $s_1$]{ \label{fig:spackle-s1}
        \tikzsetnextfilename{sec3-spackle-s1}
        \begin{quantikz}
                                                                  &                         &                                     &                         & \targ{}   &                         & \gate[2]{\labeledmeas{m_2}{M_{XX}}} &                         &                          \\
                                                                  &                         & \gate[2]{\labeledmeas{m_1}{M_{XZ}}} &                         & \ctrl{-1} &                         & \ghost{M_{XX}}                      &                         &                          \\
                                                                  &                         & \ghost{M_{XZ}}                      &                         & \ctrl{1}  & \labeledwire{purple}{Z} & \gate{\labeledmeas{m_3}{M_Z}}       & \labeledwire{purple}{Z} & \labeledwire{purple}{Z}  \\
            \lstick{$\overset{{\color{niceblue} s_1}} {\ket{0}}$} & \labeledwire{purple}{Z} &                                     & \labeledwire{purple}{Z} & \targ{}   & \labeledwire{purple}{Z} & \gate{\labeledmeas{m_4}{M_Z}}       & \labeledwire{purple}{Z} & \labeledwire{purple}{Z}
        \end{quantikz}
    }
    \caption{Spackles of measurements corresponding to \textbf{(a)} $m_1$, \textbf{(b)} $m_2$, \textbf{(c)} $m_3$ and \textbf{(d)} $s_1$.}
    \label{fig:spackles}
\end{figure}
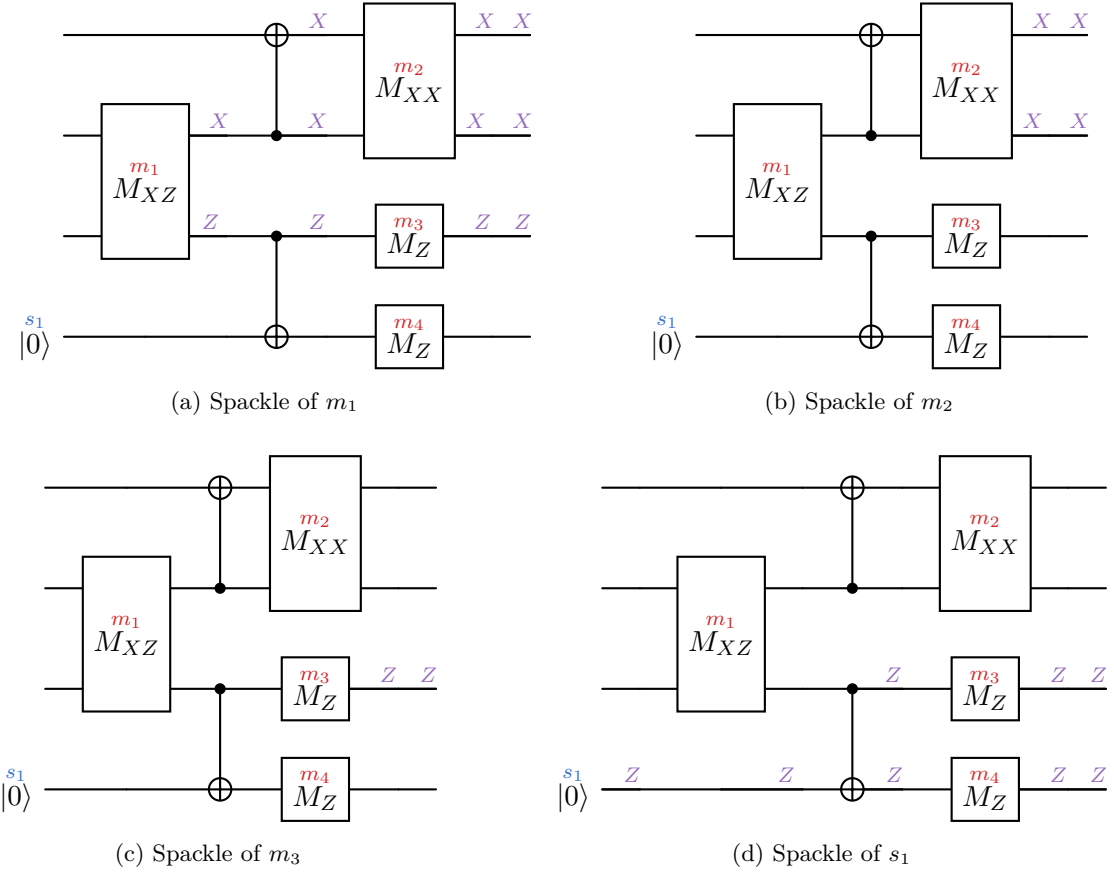

We define the back-propagation of a Pauli operator from $t'$ to $t \leq t'$ as $\overleftarrow \Pi_{t \leftarrow t'}(P) = U^\dagger P U$, where $U$ is the unitary defined by all the unitary gates between $t$ and $t'$. Then we define the backwards spackle (or \emph{backle}) of a Pauli operator $P$ to be
\begin{align}
    \backle_t(P) = \prod_{s=t}^1 [\overleftarrow \Pi_{s \leftarrow t}(P)]_s.
\end{align}
And we define $\backle(M)=\backle_{t}(P)$ where $M$ is a measurement of the Pauli operator $P$ between times $t$ and $t+1$.

We are now ready to define the notion of detector. The \textit{detector associated to a check of the outcome code} involving the measurements $M_1,\ldots,M_k$, is $\spackle(M_1,\ldots,M_k)$, i.e., the spackle of all the measurements involved in the check.
Examples of detectors are shown in \cref{fig:spacetime-stab}. It was shown in Ref. \cite{delfosse2023spacetime} that a fault anticommutes with a detector if and only if it anticommutes with an odd number of measurements of the check. Detectors are therefore in one-to-one correspondence with the checks of the outcome code, but with additional information on the possible locations and error types that can violate those checks.
We note that detectors are sometimes defined to directly be the checks of the outcome code, e.g., in the \textsc{Stim} software package~\cite{gidney2021stim}.

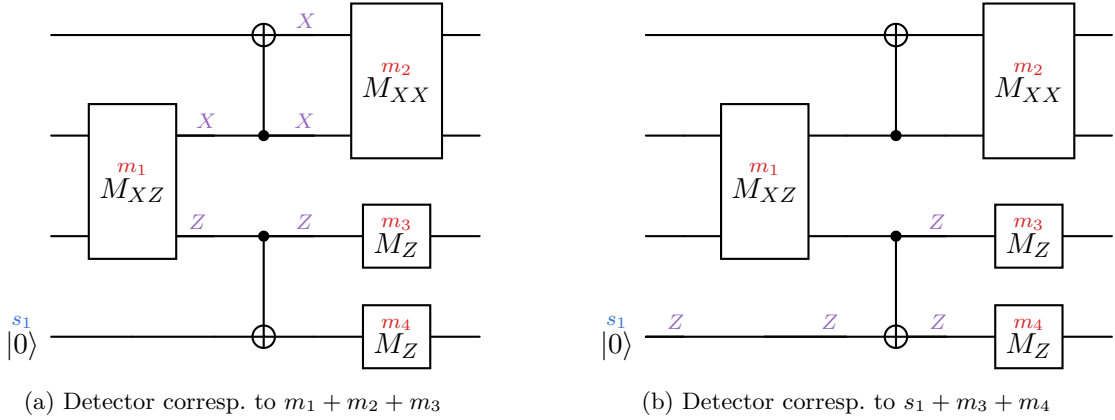
\begin{figure}[t]
    \centering
    \subfloat[Detector corresp.\ to $m_1+m_2+m_3$]{ \label{fig:spacetime-stab-a}
        \tikzsetnextfilename{sec3-spacetime-stab-a}
        \begin{quantikz}
                                                                  &                                     &                         & \targ{}   & \labeledwire{purple}{X} & \gate[2]{\labeledmeas{m_2}{M_{XX}}} &  \\
                                                                  & \gate[2]{\labeledmeas{m_1}{M_{XZ}}} & \labeledwire{purple}{X} & \ctrl{-1} & \labeledwire{purple}{X} & \ghost{M_{XX}}                      &  \\
                                                                  & \ghost{M_{XZ}}                      & \labeledwire{purple}{Z} & \ctrl{1}  & \labeledwire{purple}{Z} & \gate{\labeledmeas{m_3}{M_Z}}       &  \\
            \lstick{$\overset{{\color{niceblue} s_1}} {\ket{0}}$} &                                     &                         & \targ{}   &                         & \gate{\labeledmeas{m_4}{M_Z}}       &
        \end{quantikz}
    }
    \hspace{20pt}
    \subfloat[Detector corresp.\ to $s_1+m_3+m_4$]{ \label{fig:spacetime-stab-b}
        \tikzsetnextfilename{sec3-spacetime-stab-b}
        \begin{quantikz}
                                                                  &                         &                                     &                         & \targ{}   &                         & \gate[2]{\labeledmeas{m_2}{M_{XX}}} &  \\
                                                                  &                         & \gate[2]{\labeledmeas{m_1}{M_{XZ}}} &                         & \ctrl{-1} &                         & \ghost{M_{XX}}                      &  \\
                                                                  &                         & \ghost{M_{XZ}}                      &                         & \ctrl{1}  & \labeledwire{purple}{Z} & \gate{\labeledmeas{m_3}{M_Z}}       &  \\
            \lstick{$\overset{{\color{niceblue} s_1}} {\ket{0}}$} & \labeledwire{purple}{Z} &                                     & \labeledwire{purple}{Z} & \targ{}   & \labeledwire{purple}{Z} & \gate{\labeledmeas{m_4}{M_Z}}       &
        \end{quantikz}
    }
    \caption{Detectors corresponding to the two outcome code checks of the circuit in \cref{fig:outcome-code-example}.}
    \label{fig:spacetime-stab}
\end{figure}

\subsection{Spacetime subsystem code}
\label{sec:spacetime-subsystem-code}

To define a spacetime code with the right decoding properties, we need to include the detectors defined above in the stabilizer group.
Notice however that there are many faults that are not detectors but which commute with the detectors and have a trivial effect on the output of the circuit.
This can happen when a fault consists of a certain Pauli operator before a gate together with its propagation after the gate (e.g. $X$ before a Hadamard gate and $Z$ after the gate).
Another example is a fault consisting of a Pauli operator right after the measurement of this operator.
Intuitively, these faults correspond to some degrees of freedom we have when decoding a spacetime code. Indeed, a decoding solution containing $X_i$ ($Z_i$) before a gate is equivalent to the same solution with $UX_iU^\dag$ ($UZ_iU^\dag$) after the gate.
We do not want these faults to be logical operators of the spacetime code, so we consider them as gauge operators in a spacetime subsystem code.

We now specify the generators of the gauge group of the spacetime subsystem code, which we refer to as \textit{elementary propagation operators}.
For each gate of the circuit $U$, we have generators $P_{i,t}(UP_iU^\dag)_{t+1}$ for all $i \in \supp U$ and for $P \in \{X, Z\}$.
We call these generators \textit{gate propagation operators}; see \cref{fig:propagation-operations} for some examples.
For each input (stabilizer) measurement $S$ we have the generator $S_{1}$.
For each measurement of an $m$-qubit Pauli operator $M$ at time $t$, we have the generator $Q_{t+1}$.
We also have the generators $R_t R_{t+1}$ for $R$ in the centralizer of $M$ in the $m$-qubit Pauli group\footnote{We note that one can derive these generators by considering the elementary propagation operators of the standard ``Hadamard test" circuit for measuring a Pauli operator.}.
For example, if $M = X_i X_j$, then the corresponding generators are $X_{j,t}X_{j,t+1}$, $X_{i,{t+1}}X_{j,{t+1}}$, $X_{i,t}X_{i,t+1}$, and $Z_{i,t}Z_{j,t}Z_{i,t+1}Z_{j,t+1}$.
We use the term \textit{measurement propagation operators} to refer to generators associated with measurements.

Since the spacetime code is now seen as a subsystem code, we can associate a gauge complex to it, which we call the \textit{spacetime complex} of the circuit.
In order to define the distance of a circuit, we need to equip the space of spacetime operators (middle space $C_1$ of the spacetime complex) with a basis. In the rest of this paper, we choose single-qubit $X$ and $Z$ operators at every spacetime location of the circuit as our basis. Writing a vector in this basis is equivalent to writing it in the binary symplectic format. This choice is well-suited for noise models in which $X$ and $Z$ errors are considered to be independent, since single-qubit $Y$ operators are for instance weight-two in this basis.
We believe that our formalism could be extended to the case of correlated errors by defining a notion of distance/decoding function for an overcomplete generating set of $C_1$, but we leave this extension to future work.

Equipped with our basis for $C_1$, we can now define the \textit{distance of a circuit} as the distance of the corresponding spacetime complex. We also say that two circuits are \textit{equivalent} if their spacetime complexes are equivalent. An example spacetime complex, and its reduction using rule A, is shown in \cref{fig:example-spacetime-code-gauge-complex}

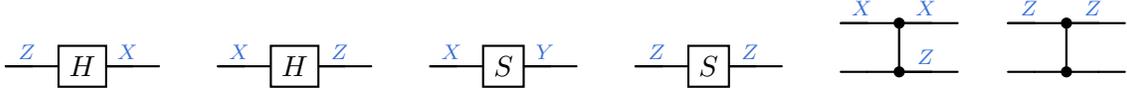
\begin{figure}[t]
    \centering
    \begin{align*}
        \tikzsetnextfilename{sec3-propagation-operators-1}
        \begin{quantikz}[column sep=10pt]
            & \labeledwire{niceblue}{Z}  & \gate{H} & \labeledwire{niceblue}{X}  &
        \end{quantikz}
        \hspace{15pt}
        \tikzsetnextfilename{sec3-propagation-operators-2}
        \begin{quantikz}[column sep=10pt]
            & \labeledwire{niceblue}{X}  & \gate{H} & \labeledwire{niceblue}{Z}  &
        \end{quantikz}
        \hspace{15pt}
        \tikzsetnextfilename{sec3-propagation-operators-3}
        \begin{quantikz}[column sep=10pt]
            & \labeledwire{niceblue}{X}  & \gate{S} & \labeledwire{niceblue}{Y}  &
        \end{quantikz}
        \hspace{15pt}
        \tikzsetnextfilename{sec3-propagation-operators-4}
        \begin{quantikz}[column sep=10pt]
            & \labeledwire{niceblue}{Z}  & \gate{S} & \labeledwire{niceblue}{Z}  &
        \end{quantikz}
        \hspace{15pt}
        \tikzsetnextfilename{sec3-propagation-operators-5}
        \begin{quantikz}[column sep=10pt]
            & \labeledwire{niceblue}{X}  & \ctrl{1} & \labeledwire{niceblue}{X}  &  \\
            &  & \ctrl{-1} & \labeledwire{niceblue}{Z}  &  \\
        \end{quantikz}
        \hspace{15pt}
        \tikzsetnextfilename{sec3-propagation-operators-6}
        \begin{quantikz}[column sep=10pt]
            & \labeledwire{niceblue}{Z}  & \ctrl{1} & \labeledwire{niceblue}{Z}  &  \\
            &  & \ctrl{-1} &  &  \\
        \end{quantikz}
    \end{align*}
    \caption{Examples of elementary propagation operators (gauge operators in a spacetime subsystem code).}
    \label{fig:propagation-operations}
\end{figure}

\subsection{Stabilizers of the spacetime subsystem code}
\label{sec:stabilizers-spacetime-code}

Here we characterize the stabilizers of the spacetime code, i.e. the gauge operators that commute with all gauge operators, in terms of spackles and backles. We show that detectors, in particular, are spacetime code stabilizers.
We begin with two lemmas: one about the properties of elementary propagation operators, and one about spackles and backles.
\begin{lemma} \label{lemma:elementary-prop-completeness}
    Every $m$-qubit gate or measurement has $2m$ associated independent and commuting elementary propagation operators. Furthermore, any Pauli operator commuting with all of these elementary propagation operators must be in their multiplicative span (up to a phase).
\end{lemma}
\begin{proof}
    First consider an $m$-qubit unitary gate $U$, which has $2m$ associated elementary propagation operators $P_{i,t}(UP_iU^\dag)_{t+1}$ for $i \in \supp U$ and $P \in \{X,Z\}$.
    Consider two generators $P_{i,t}(UP_iU^\dag)_{t+1}$ and $Q_{j,t}(UQ_jU^\dag)_{t+1}$.
    Clifford unitaries preserve commutation relations so $[P,Q] = [UPU^\dag,UQU^\dag]$ and therefore $P_{i,t}(UP_iU^\dag)_{t+1}$ and $Q_{j,t}(UQ_jU^\dag)_{t+1}$ commute for all $P$ and for all $i$.
    Independence follows from the independence of the operators $\{ P_{i,t} : P\in\{X,Z\}, i\in\supp U \}$.
    Now consider an $m$-qubit measurement $M$, with associated elementary propagation operators $M_{t+1}$ and $R_t R_{t+1}$ for $R$ in the centralizer of $M$.
    The centralizer has $2m-1$ commuting generators so in total we have $2m$ independent and commuting elementary propagation operators.

    Let $S$ denote the set of $2m$ independent and commuting elementary propagation operators associated with a gate or measurement.
    The support of $S$ spans $2m$ qubits, so $S$ defines a stabilizer state (a stabilizer code with no encoded qubits), whose centralizer within the Pauli group is exactly $S$.
    Therefore, any operator commuting with all elements of $S$ must be equal (up to a phase) to a product of operators in $S$.
\end{proof}

\begin{lemma} \label{lemma:spackles-backles-gauge-ops}
Spackles and backles are gauge operators of the spacetime code
\end{lemma}
\begin{proof}
    For any $n$-qubit Pauli operator $P$ and time steps $t,t' \in \{1,\ldots,T\}$, $t \leq t'$, the spacetime operator $P_t [\Pi_{t\rightarrow {t'}}(P)]_{t'}$ is a product of gate propagation operators, and is therefore a gauge operator.
    Furthermore, for any measurement of a Pauli $M$ between times $t-1$ and $t$, $[\Pi_{t\rightarrow {t'}}(M)]_{t'}=M_t M_t [\Pi_{t\rightarrow {t'}}(M)]_{t'}$ is a gauge operator as a product of gate propagation operators and a measurement propagation operator.
    Consequently, $\spackle_t(M)=\prod_{s=t}^T \left[ \Pi_{t\rightarrow s}(M) \right]_s$ is a gauge operator itself, and so is any spacetime operator of the form $\spackle(M_1,\ldots,M_\ell)$.
    Similarly, for any input stabilizer $S$, $[\Pi_{1\rightarrow {t'}}(S)]_{t'}=S_1 S_1 [\Pi_{1\rightarrow {t'}}(S)]_{t'}$ is a gauge operator as a product of gate propagation operators and input stabilizer (gauge) operators.
    Finally, to show that backles are gauge operators, we repeat the argument above, replacing $\Pi_{s \rightarrow t}(M)$ by $\overleftarrow \Pi_{s \leftarrow t}(M)$.
\end{proof}

\begin{proposition} \label{prop:stabilizers_backle_spackle}
    The stabilizers of the spacetime subsystem code are either the spackle of a set of measurements and input stabilizers or the backle of a set of measurements.
\end{proposition}
\begin{proof}
    Since spackles and backles are gauge operators by \cref{lemma:spackles-backles-gauge-ops}, those that commute with all gauge operators are stabilizers of the spacetime code.

    In the other direction, let $S$ be a stabilizer of the spacetime code, and consider the support of $S$ at times $t$ and $t+1$, which we write as $P_t$ and $Q_{t+1}$, respectively.
    Since $S$ is a stabilizer, $P_t Q_{t+1}$ must commute with all the elementary propagation operators associated with the circuit elements acting between times $t$ and $t+1$.
    Consider some $m$-qubit circuit element $U$ whose elementary propagation operators share support with $P_t$ and $Q_{t+1}$.
    The operator $P_t Q_{t+1}$, restricted to the support of $U$, must commute with all of its elementary propagation operators, and therefore, by \cref{lemma:elementary-prop-completeness}, must be in their multiplicative span\footnote{We do not need to worry about the phases because the stabilizer group of a subsystem code is defined as the centre of the gauge group modulo phases.}.
    This is true for all circuit elements whose elementary propagation operators share support with $P_t Q_{t+1}$, and therefore $P_t Q_{t+1}$ must be equal to a product of elementary propagation operators.

    Let $t_s > 1$ be the time step such that $G$ has no support on qubit $i$ at earlier time steps.
    And let $P_{i,t_s}$ denote the support of $G$ on qubit $i$ at time $t_s$.
    Suppose that the circuit element acting on qubit $i$ between time $t_{s}-1$ and $t_s$ is a unitary gate.
    By \cref{lemma:elementary-prop-completeness}, there are $2m$ independent gate propagation operators associated with an $m$-qubit gate, and so there will always exist a gate propagation operator that anticommutes with $P_{i,t_s}$.
    Now suppose that the circuit element acting on qubit $i$ between times $t_s-1$ and $t_s$ is an $m$-qubit measurement $M$.
    By \cref{lemma:elementary-prop-completeness}, there are $2m$ independent measurement propagation operators associated with $M$, and for $S$ to commute with all of these operators it must be contained in their multiplicative span.
    Furthermore, the only measurement propagation operator wholly supported on qubits at time $t_s$ is $M_{t_s}$ itself, and therefore $P_{i,t_s} = M_{i,t_s}$.
    Now consider the case for $t_s=1$.
    Observe that $P_{i,1}$ will commute with all the gauge operators associated with the input stabilizers if the support of $S$ at the first time step is equal to an operator that commutes with the input stabilizers.

    To summarize, since $S$ is a stabilizer it must be either preceded by the beginning of the circuit (and the $t=1$ support of $S$ must commute with the input stabilizers) or by a set of measurements such that the support of $S$ immediately following each measurement is equal to the operator being measured.

    By an analogous argument, $S$ must either be followed by the end of the circuit or by a set of measurements such that the support of $S$ immediately preceding each measurement is equal to the operator being measured.
    Recall that $S$ must also be equal to a product of elementary propagation operators for every two consecutive timesteps.
    When $S$ is not preceded by the beginning of the circuit or followed by the end of the circuit, these conditions imply that $S$ is a spackle or backle of a set of measurements.
    Now suppose that $S$ is instead preceded by the beginning of the circuit and followed by the end of the circuit.
    Because $T$ is odd, the support of $S$ at the beginning of the circuit must be equal to an input stabilizer, as otherwise it would not be a product of elementary propagation operators.
    Therefore, $S$ is equal to the spackle of an input stabilizer.
\end{proof}

\begin{corollary} \label{corr:detectors_eq_stabilizers}
    Detectors are stabilizers of the spacetime subsystem code.
\end{corollary}
\begin{proof}
    Let $M_1,\ldots,M_k$ be a set of measurements that form a check of the outcome code.
    We show that $\spackle(M_1,\ldots,M_k)$ commutes with all the gauge operators, which implies (by \cref{prop:stabilizers_backle_spackle}) that it is a stabilizer of the spacetime subsystem code.
    First note that all the elementary propagation operators for a single gate commute with each other, by \cref{lemma:elementary-prop-completeness}.
    For every $M_i$, $\spackle_t(M_i)$ is constructed from elementary gate propagation operators, and therefore commutes with all elementary gate propagation operators occurring between times $t$ and $T$.
    And, in fact, this operator will commute with all elementary gate propagation operators, as any such operator occurring before time $t$ does not share support with $\spackle_t(M_i)$.
    The remaining elementary propagation operators that we have not considered thus far are those that are associated with measurements.
    These operators may anticommute with spackles of measurements in general, but they necessarily commute with those operators when the measurements constitute a check of the outcome code.
    Without loss of generality, suppose that the circuit contains a measurement $N$ that occurs after measurements $M_1,\ldots,M_l$ but before measurements $M_{l+1},\ldots,M_k$.
    The elementary propagation operator associated with $N$ commutes with $\spackle(M_{l+1},\ldots,M_k)$ by definition.
    Now suppose that this elementary propagation operator anticommutes with $\spackle(M_1,\ldots,M_l)$.
    Then $M_1,\ldots,M_k$ cannot possibly define a check of the outcome code as (informally) the measurement $N$ will destroy the correlation between $M_1,\ldots,M_l$ and $M_{l+1},\ldots,M_k$\footnote{We do not present a formal argument here as it would necessitate the introduction of substantial additional notation and the result can be straightforwardly derived from \cite[Algorithm 1]{delfosse2023spacetime}.}.
\end{proof}

While we have shown that detectors are stabilizers of the spacetime code, they are not the only ones.
As a simplified example, consider a measurement-only circuit where the operator $M$ is measured at time $t$.
Now suppose that all subsequent measurements commute with $M$ and that $M$ cannot be decomposed into a product of any subset of the subsequent measurements.
Observe that $\spackle(M)$ is a stabilizer of the spacetime code, but it does not have an associated check of the outcome code.
We call such operators \emph{incomplete detectors}, as they could be made into detectors by appending further measurements to the circuit.

\subsection{Logical operators of the spacetime subsystem code}
\label{sec:logical-ops-spacetime-code}

We will characterize the logical operators of a spacetime subsystem code step by step, starting from a trivial example.
Consider the trivial circuit on $n$ wires with $T$ time steps, i.e., the circuit consisting exclusively of identity gates.
Without loss of generality, suppose that the input state to this circuit is in the codespace of some $[\![n,k,d]\!]$ stabilizer code defined by the stabilizer group $\mathcal S$.
Let $\mathcal G$, $\mathcal D$, and $\mathcal L$ denote the gauge, stabilizer, and logical groups of the spacetime subsystem code associated with the circuit.

We have $2n(T-1)$ gauge generators coming from the identity gates and $(n-k)$ gauge generators coming from the input stabilizers, all of which are independent.
For each input stabilizer generator $S$, we have a stabilizer generator of the spacetime code $\spackle_1(S)$, which---in the language of the previous section---is an incomplete detector.
Now consider the spacetime operator $\spackle_1(L)$, where $L$ is a logical operator of the input stabilizer code.
This operator commutes with all the gauge operators (the argument is identical as for $\spackle_1(S)$) but is not itself contained in the gauge group and is therefore a bare logical operator of the spacetime code.
To see this, note that $\spackle_1(L) = L_1 L_2 \ldots L_T$, where $T$ is odd.
The gauge group contains $S_1$ for all input stabilizers $S$, as well as $X_t X_{t+1}$ and $Z_t Z_{t+1}$ for all time steps $t$.
Any odd-weight operator constructed from these components must act as a stabilizer at time step $t=1$ and therefore cannot be equal to $\spackle_1(L)$.
There are $2k$ independent logical operators of the input stabilizer code, and therefore we have $2k$ generators of $\mathcal L$.
Let $\mathsf{Mat}(\mathcal G)$ denote the symplectic representation of the generators of $\mathcal G$, and similarly for $\mathsf{Mat}(\mathcal D)$ and $\mathsf{Mat}(\mathcal L)$.
We have
\begin{equation} \label{eq:GSL_counting}
\begin{aligned}
    \RANK \mathsf{Mat}(\mathcal G) + \RANK (\ker \mathsf{Mat}(\mathcal G)) &= \RANK \mathsf{Mat}(\mathcal G) + \RANK \mathsf{Mat}(\mathcal D) + \RANK \mathsf{Mat}(\mathcal L), \\
    &= [2n(T-1) + (n-k)] + (n-k) + 2k = 2nT,
\end{aligned}
\end{equation}
as expected from the rank-nullity theorem (applied to $\mathsf{Mat}(\mathcal G)$).

Now suppose that we replace the identity gates in our circuit with unitary gates.
Consider, for example, replacing an identity gate with the Hadamard gate.
This transforms the gate propagation operators as follows:
\begin{equation}
    X_{i,t} X_{i,t+1} \mapsto X_{i,t} Z_{i,t+1}, \quad
    Z_{i,t} Z_{i,t+1} \mapsto Z_{i,t} X_{i,t+1},
\end{equation}
i.e., we apply the Hadamard to the qubit with index $(i,t+1)$.
The transformation of the stabilizers and logical operators is analogous.
Therefore, adding this gate to the circuit is equivalent to applying a Clifford unitary to a qubit of the spacetime code, which will not change the number of stabilizers or encoded qubits.
Furthermore, the stabilizer generators will still be of the form $\spackle_1(S)$ and the logical operators will still be of the form $\spackle_1(L)$.
Iterating this argument allows us to conclude that replacing identity gates with unitary gates does not change the number of stabilizers and encoded qubits of the spacetime code, or the form of the stabilizers and logical operators.

We now consider replacing identity gates with measurements.
Recall that the input to the circuit is defined by the stabilizer group $\mathcal S$; we call the corresponding stabilizer code the input code.
Any $n$-qubit Pauli operator can be written in the form $L S D$, where $L$ is a logical operator of the input code, $S$ is a stabilizer of the input code, and $D$ is a destabilizer (sometimes called pure error) of the input code (see, e.g.,~\cite{iyer2015hardnessDecoding}).
We now consider the effect of measuring each of these types of operators in turn.
If we measure a stabilizer $S$ between times $t$ and $t+1$ we gain an additional stabilizer generator $\spackle_1(S) \spackle_t(S)$.
But we also gain a relation between gauge operators as $\spackle_1(S) \spackle_t(S)$ is equivalent to $\backle_{t-1}(S)$.
The operators $\spackle_1(L)$ are still logical operators, as $S$ and $L$ commute for all $L$.
Now suppose that we measure a logical operator of the input code between times $t$ and $t+1$.
We denote the measured logical operator as $L$ and its conjugate partner as $M$.
In this case, $\spackle_1(L)$ becomes a stabilizer, we gain an additional stabilizer $\spackle_1(L) \spackle_t(L)$, and the operator $\spackle_1(M)$ is no longer a logical operator of the spacetime code as it anticommutes with $L_t$.
Finally, suppose that we measure a destabilizer $D$ between times $t$ and $t+1$, and that $D$ anticommutes exclusively with the stabilizer generator $S$.
The operator $\spackle_1(S)$ is no longer a stabilizer, but is replaced by the operator $\spackle_t(D)$.

Now we are ready to consider a generic circuit containing unitary gates and measurements.
We start from a circuit consisting of unitary gates only, and consider replacing some unitary gates with measurements.
For example, suppose we replace a gate between times $t$ and $t+1$ with a measurement of a Pauli operator $P$.
To understand the effect of this operator we need to consider the propagations of the input stabilizers and logical operators, i.e., $\Pi_{1\rightarrow t}(S)$ and $\Pi_{1\rightarrow t}(L)$ for all input stabilizer generators $S$ and input logical operators $L$.
If the measurement is equal to $\Pi_{1\rightarrow t}(S)$\footnote{The measurement can also be in the span of $\Pi_{1\rightarrow t}(S)$ and the effect would be analogous.}, then we gain a stabilizer $\spackle_1(S)\spackle_t(P)$ and one relation between gauge generators.
Similarly, if the measurement is equal to $\Pi_{1\rightarrow t}(L)$ for some logical operator then we lose two logical operators ($\spackle_1(L)$ and its conjugate partner) and gain two stabilizers ($\spackle_1(L)$ and $\spackle_1(L) \spackle_t(P)$).
And finally, if the operator anticommutes with exclusively $\Pi_{1\rightarrow t}(S)$, then we lose $\spackle_1(S)$ from the stabilizer group but gain $\spackle_t(P)$.

The above argument gives us a method for generating the stabilizers and logical operators of a spacetime subsystem code.
Concretely, we start from a spacetime code and replace all the measurements with identity gates.
The stabilizers and logical operators of this spacetime code are simply the spackles of the input stabilizers and logical operators.
We then add the original measurements back into the circuit one by one (in the order in which they appear in the circuit), and update the gauge, stabilizer, and logical operator groups accordingly.
Examples of logical operators in the general case are shown in \cref{fig:example-spacetime-code-logicals-1,fig:example-spacetime-code-logicals-2}.

\begin{figure}
    \centering
    \subfloat[Stabilizers]{ \label{fig:example-spacetime-code-stabilizers}
        \tikzsetnextfilename{sec3-example-spacetime-code-stabilizers}
        \begin{quantikz}
            & \gate[2]{S_{XX}} & \labeledwire{purple}{X}  & \gate{H} & \labeledwire{purple}{Z}  & \gate{H} & \labeledwire{purple}{X}  & \gate[2]{M_{XX}} & \labeledwire{orange}{X}  & \gate{H} & \labeledwire{orange}{Z}  & \gate{H} & \labeledwire{orange}{X}  & \gate[2]{M_{XX}} & \labeledwire{nicegreen}{X}  &  \\
            &                  & \labeledwire{purple}{X}  & \gate{H} & \labeledwire{purple}{Z}  & \gate{H} & \labeledwire{purple}{X}  & & \labeledwire{orange}{X}  & \gate{H} & \labeledwire{orange}{Z}  & \gate{H} & \labeledwire{orange}{X}  & & \labeledwire{nicegreen}{X}  &  \\
        \end{quantikz}
    }
    \\
    \subfloat[First pair of logical operators]{ \label{fig:example-spacetime-code-logicals-1}
        \tikzsetnextfilename{sec3-example-spacetime-code-logicals-1}
        \begin{quantikz}
            & \gate[2]{S_{XX}} & \labeledwire{orange}{Z}  & \labeledwire{purple}{X}  & \gate{H} & \labeledwire{purple}{Z}  & \gate{H} & \labeledwire{purple}{X}  & \gate[2]{M_{XX}} & \labeledwire{purple}{X}  & \gate{H} & \labeledwire{purple}{Z}  & \gate{H} & \labeledwire{purple}{X}  & \gate[2]{M_{XX}} & \labeledwire{purple}{X}  &  \\
            &                  & \labeledwire{orange}{Z}  &                          & \gate{H} &                          & \gate{H} &                          &                  &                          & \gate{H} &                          & \gate{H} &                          &                  &                          &  \\
        \end{quantikz}
    }
    \\
    \subfloat[Second pair of logical operators]{ \label{fig:example-spacetime-code-logicals-2}
    \tikzsetnextfilename{sec3-example-spacetime-code-logicals-2}
        \begin{quantikz}
            & \gate[2]{S_{XX}} & \labeledwire{orange}{X} & \labeledwire{purple}{Z} & \gate{H}                & \labeledwire{purple}{X}            & \gate{H} & \labeledwire{purple}{Z}  & \gate[2]{M_{XX}} & \labeledwire{purple}{Z}  & \gate{H} & \labeledwire{purple}{X}  & \gate{H} & \labeledwire{purple}{Z}  & \gate[2]{M_{XX}} & \labeledwire{purple}{Z} &  \\
            &                                            &                         & \labeledwire{purple}{Z} & \gate{H} & \labeledwire{purple}{X} & \gate{H} & \labeledwire{purple}{Z}  &                  & \labeledwire{purple}{Z}  & \gate{H} & \labeledwire{purple}{X}  & \gate{H} & \labeledwire{purple}{Z}  &                  & \labeledwire{purple}{Z} &  \\
        \end{quantikz}
    }
    \\
    \subfloat[Spacetime complex]{ \label{fig:example-spacetime-code-gauge-complex}
        \includegraphics[width=0.66\textwidth]{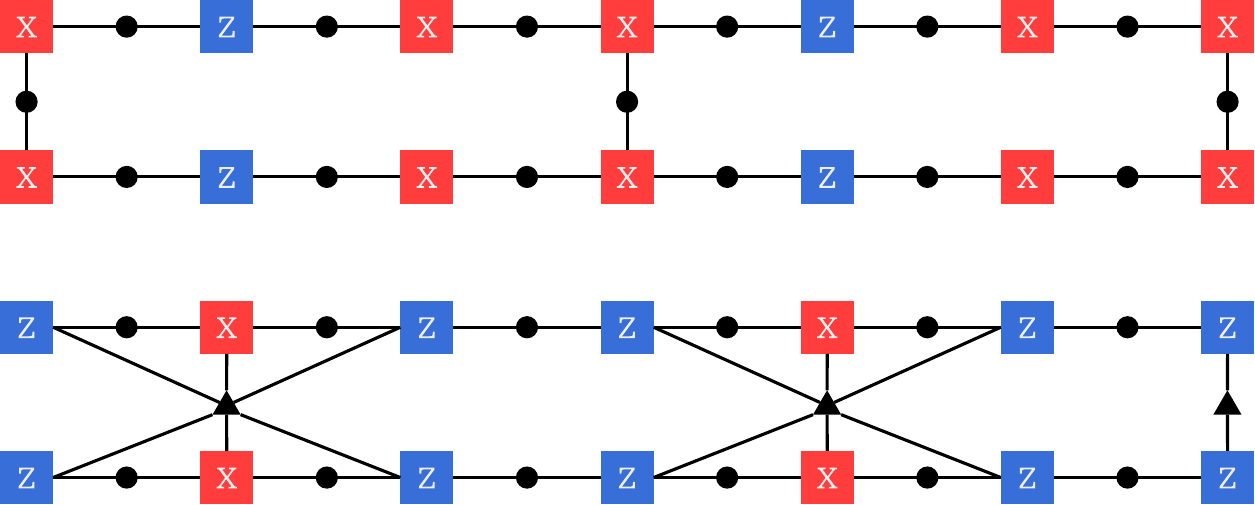}
    }
    \hspace{15pt}
    \subfloat[Reduced spacetime complex]{ \label{fig:example-spacetime-code-gauge-complex-reduced}
        \includegraphics[width=0.24\textwidth]{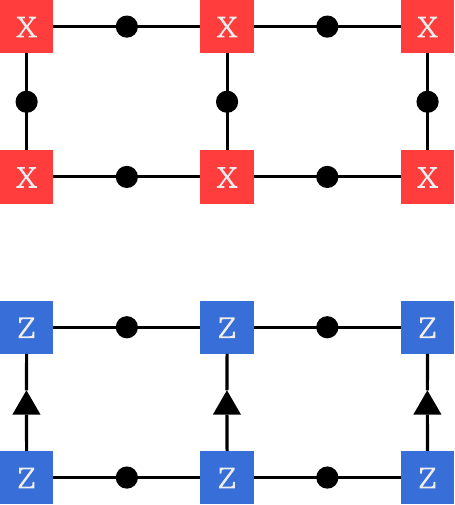}
    }
    \caption{Example of a spacetime code, with its stabilizers, logical operators, and gauge complexes. \textbf{(a)}
    A Clifford circuit annotated with three generators of the stabilizer group of its associated spacetime subsystem code.
    The notation $S_{XX}$ means that the input state stabilized by $XX$. The purple and orange stabilizers are detectors, while the green operator is an incomplete detector. \textbf{(b-c)} First and second pairs of logical operators. In orange are the logical errors (dressed logical operators of the subsystem code) and in purple are the logical correlations (bare logical operators). \textbf{(d)} Spacetime complex of the circuit. Circles are gauge nodes, squares are qubit nodes, and triangles are detector nodes. Note that due to the symplectic matrix $\Omega$, an $X \leftrightarrow Z$ permutation is applied to the stabilizer generators in order to obtain the detector nodes. \textbf{(e)} Reduced gauge complex after applying Rule A on each gauge node associated to a Hadamard gate.
    \label{fig:example-spacetime-code}
    }
\end{figure}

\subsection{Relationship to other frameworks}

To end this section, we would like to discuss the relationship between the spacetime code formalism and two other popular frameworks used to analyze error-correcting circuits: \textsc{Stim}'s detector error model \cite{gidney2021stim,higgott2025sparse, derks2024designing} and ZX-calculus with Pauli webs \cite{mcewen2023relaxing,bombin2024unifying,rodatz2024floquetifying,delafuente2024dynamical,rodatz2025fault}.

In \textsc{Stim}'s detector error model, three types of objects are constructed from a given circuit: the detectors, the logical observables and the error mechanisms.
A \textit{detector} is a set of measurements that have a deterministic parity in the absence of errors, i.e. a check of the outcome code.
A \textit{logical observable} is a set of measurements whose value correspond to measuring a logical operator.
An \textit{error mechanism} is a set of errors that flip the same detectors and logical observables.
From those three types of objects, one can draw a tripartite graph with error mechanism nodes connected to a detector or logical observable node when the error mechanism flips that detector or logical observable in a circuit execution.

We see three main differences between the spacetime code and the detector error model formalisms:
\begin{enumerate}
    \item In spacetime codes, the groups of spacetime stabilizers, gauge operators, and logical correlations are derived directly from the circuit and the input stabilizer space. This stands in contrast to the detector error model, where detectors and logical observables are specified manually. The distinction, however, is largely practical rather than fundamental: defining detectors and observables by hand makes it possible to choose a low-weight basis that simplifies decoding, whereas generating them automatically facilitates the proof of fault-tolerant properties of circuits.
    \item Spacetime codes are constructed with gauge operators first, while detector error models are constructed with logical observables first. In a spacetime code, the counterparts of the logical observables---the logical correlations---can be derived from the chain complex (by computing the cohomology). In a detector error model, the counterparts of the gauge operators can be derived by looking at all the combinations of error mechanisms that commute with all detectors and logical observables.
    \item In our formalism, we take as our basis of errors all the single-qubit $X$ and $Z$ errors occurring at spacetime locations immediately before and after each gate. As noted at the end of \cref{sec:spacetime-subsystem-code}, we expect that the framework can be extended to incorporate correlated errors, although this is not yet achieved in its present form. On the other hand, the detector error model has a flexible definition of error mechanisms, which can include correlated multi-qubit errors.
\end{enumerate}

In the ZX-calculus formalism, a given circuit is turned into a ZX diagram, and a set of \textit{Pauli webs}---colored edges representing both detectors and logical correlations---is drawn following some specific rules \cite{bombin2024unifying}. If the circuit has some input and output qubits, the corresponding edges are called \textit{outer ports} of the ZX diagram. We can then partition the Pauli webs as follows:
\begin{itemize}
    \item Pauli webs with no support in an outer port: those are equivalent to our full detectors.
    \item Pauli webs with support in outer ports that are either input or output, but not both: those are equivalent to our incomplete detectors.
    \item Pauli webs with support on both input and output outer ports: those are the logical correlations.
\end{itemize}
The set of trivial errors (the counterparts of our gauge operators) can also be read off from the ZX diagrams, as the set of constraints at each node that allow one to define the Pauli webs.
Moreover, the recent development of distance-preserving \cite{rodatz2024floquetifying,delafuente2024dynamical} and fault-equivalent \cite{rodatz2025fault} rewrite rules for ZX-diagrams makes the two formalisms comparable.
The precise relationship between ZX diagrams and spacetime complexes---and therefore fault-tolerant maps and distance-preserving rewrite rules---is left to future work.

Finally, the spacetime code formalism can be compared with the dynamical code formalism of Fu and Gottesman~\cite{fu2024error}, where subsystem codes are likewise constructed from circuits and a corresponding notion of distance is defined.
A more detailed comparison will be provided in upcoming work by the authors~\cite{fu2025subsystem}.

\section{Fault-tolerant measurement-based quantum computing}
\label{sec:fault-tolerant-cluster-states}
Quantum computing architectures based on the photonic platform are distinguished by an inherent scalability, allowing them to clear a major hurdle en route to the construction of devices capable of solving classically intractable problems~\cite{bourassa2021blueprintscalable}. However, photonic approaches are only viable with low-depth optical circuits to preclude noise---especially loss---from accruing to levels incompatible with fault tolerance. This restriction invites reformulating the standard circuit model of quantum computing into the measurement-based paradigm~\cite{jozsa2005introduction,gottesman1999demonstrating,nielsen2003quantum,leung2004quantum,raussendorf2001oneway}, where flying qubits need only traverse a handful of physical gates before being measured out. In measurement-based quantum computing (MBQC), the entanglement is generated offline in a resource state called a \emph{graph state} or \emph{cluster state}, while the computation is embedded into a sequence of adaptive local measurements and classical feedforward. The ability to construct graph states and perform measurements in arbitrary bases---either directly or through auxiliary magic states---unlocks universality.

The construction of fault-tolerant graph states from quantum error correcting (QEC) codes dates back to early surface code work~\cite{raussendorf2003measurement,nielsen2005fault,raussendorf2005longrange,raussendorf2006fault,raussendorf2007topological,browne2016one}.
Since then, this work has been generalized to CSS codes in a process called~\emph{foliation}~\cite{bolt2016foliated,bolt2018decoding,BrownUniversalFTMBQC2020,hillmann2024single}.
In this procedure, the two disjoint Tanner graphs of the code (containing $X$-type and $Z$-type stabilizers, respectively) are interpreted as alternating layers of the foliated graph state. Data qubits between adjacent layers are entangled with each other, resulting in teleportation chains propagating quantum information from layer to layer with byproduct Hadamard gates.
The $X$ and $Z$ checks are replaced by \emph{detectors}, in this case combinations of $X$ measurement outcomes in consecutive layers that add to zero in the absence of errors.
Foliation has also been generalized to non-CSS codes and subsystem codes~\cite{BrownUniversalFTMBQC2020} and to certain Floquet codes~\cite{paesani2023highthreshold}.

In this section, we start by introducing the reader to measurement-based quantum computing, defining important terminology for the rest of the paper.
We then study the fault-tolerant properties of MBQC circuits through the lens of the associated spacetime code.
Finally, we show how the associated spacetime complex can be simplified to another chain complex, the cluster state complex, generalizing the chain complex formalism of Ref.~\cite{newman2020generating}.

\subsection{Measurement-based quantum computing}
\label{sec:mbqc}

The fundamental observation underlying the MBQC formalism is that any single-qubit diagonal unitary gate
\begin{align}
    U_z(\phi)=e^{-\frac{\phi}{2} Z}
\end{align}
followed by a Hadamard can be implemented through the following teleportation circuit \cite{browne2016one}:

\begin{align} \label{eq:teleportation}
    \tikzsetnextfilename{sec4-teleportation}
    \begin{quantikz}
        \lstick{$\ket{\psi}$} &  & \ctrl{1}  & & \meterD{X_{\phi}} \\
        \lstick{$\ket{+}$}    &  & \ctrl{-1} & & \gate{X} \wire[u][1]{c} & \rstick{$HU_z(\phi) \ket{\psi}$}
    \end{quantikz}
\end{align}
where $X_{\phi}$ is the measurement basis consisting of the two projectors $\ket{0}\bra{0} \pm e^{i\phi} \ket{1}\bra{1}$. In particular, the Hadamard gate is implemented by taking $\phi=0$. Since any single-qubit gate can be decomposed into Hadamard and diagonal unitaries, we can implement any arbitrary single-qubit operation via successive teleportation circuits of the form of \cref{eq:teleportation}. Moreover, adding simple $\texttt{CZ}$ gates (with no measurement) to our set of operations, we can see that any multi-qubit unitary operation can be implemented with only $\texttt{CZ}$ gates, $\ket{+}$ ancilla state preparation, single-qubit measurements, and Pauli corrections.

In circuits of the form above, the Pauli correction can always be moved to the very end, as long as we perform the measurements sequentially and adapt the measurement bases depending on the previous measurements outcomes. For instance, two successive teleportation circuits with measurement bases $X_{\phi_1}$ and $X_{\phi_2}$, and without any Pauli correction, implements the following unitary:
\begin{align}
    X^{m_2} H U_z(\phi_2) X^{m_1} H U_z(\phi_1) = X^{m_2} Z^{m_1} H U_z(-\phi_2) H U_z(\phi_1)
\end{align}
where $m_1$ and $m_2$ are the outcomes of the first and second measurements, respectively. The Pauli correction has been propagated to the end, with the effect of transforming the second unitary from $U(\phi_2)$ to $U(-\phi_2)$. By changing the second measurement basis from $X_{\phi_2}$ to $X_{-\phi_2}$, we therefore recover the correct unitary. The same reasoning applies when adding more teleportation circuits and $\texttt{CZ}$ gates.

In other words, we can implement any unitary operation with a circuit of $\texttt{CZ}$ gates (between the input state and some ancilla qubits initialized in the $\ket{+}$ states), followed by a sequence of measurements (whose bases depend on the operation we are implementing and the outcome of the previous measurements), and finally some Clifford corrections (which can be done as a classical post-processing step).
Any such circuit can therefore be constructed from the following data, often referred to as an \textit{MBQC pattern} \cite{broadbent2009parallelizing}:
\begin{enumerate}
    \item A graph $G$ with adjacency matrix $A$, where each node represents a qubit and each edge a $\texttt{CZ}$ gate between the corresponding qubits.
    \item Two sets of nodes, the \textit{input nodes} $\mathcal{I}$ and \textit{output nodes} $\mathcal{O}$, such that input nodes represent the input state (as opposed to the ancilla qubits initialized in the $\ket{+}$ state), and the output nodes represent the qubits that are not measured at the end of the circuit.
    \item A list $\mathcal{B}$ of angles assigned to each non-output node, representing the basis in which it is measured (not taking into account the change of basis due to Pauli propagations).
    \item A list of integers $\mathcal{T}$ assigned to each non-output node, representing the time step at which it is measured.
\end{enumerate}
We write such an MBQC pattern as $\mbqc(A,\mathcal{I},\mathcal{O},\mathcal{B},\mathcal{T})$.
The corresponding circuit made of \texttt{CZ} gates and measurements is called the \textit{circuit realization of the MBQC pattern}, or \textit{MBQC circuit} for short.
We refer to the state obtained after applying all the \texttt{CZ} gates to an input state as a \textit{cluster state} or \textit{graph state}.
We refer to the graph $G$ as the \textit{cluster state graph} of the MBQC pattern.

As an example, the following circuit
\begin{align}
    \tikzsetnextfilename{sec4-example-circuit}
    \begin{quantikz}
        &  & \gate{H} & \ctrl{1} & \ctrl{1} & \gate{T} & \gate{H} &  \\
        &  &          & \targ{} & \ctrl{-1} &          &          &
    \end{quantikz}
\end{align}
can be represented by the following MBQC pattern:
\begin{center}
\includegraphics[width=0.22\linewidth]{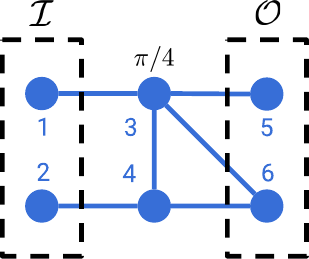}
\end{center}
where the blue nodes (apart from the output modes) are measured in the order of the blue indices, i.e.\ $\mathcal{T}=[1,2,3,4]$. We only label the measurement basis of non-output nodes for non-zero angles, that is, $\mathcal{B}=[0,0,\pi/4,0]$. This means, that to run the circuit in a measurement-based way, we first need to initialize the non-input nodes $\{3,4,5,6\}$ in the $\ket{+}$ state, perform a $\texttt{CZ}$ for each edge, and measure the nodes $\{1,2,3,4\}$ in increasing order in their respective basis, up to a sign which can be computed through the propagation of Pauli correction operators. We call this circuit the \textit{circuit associated to the MBQC pattern}.

It should be noted that the mapping from a circuit to an MBQC pattern is far from unique in general: both the order of measurements and the graph itself can be modified without changing the corresponding channel, and a great deal of research has been done on how to optimize MBQC patterns; see for example \cite{danos2005parsimonious, broadbent2009parallelizing, backens2021there, mcelvanney2023flow}.

In the context of quantum error correction and spacetime codes, we are mainly interested in mapping \textit{Clifford} circuits into MBQC patterns. In this case, it becomes unnecessary to update the measurement bases dynamically and the order of measurements becomes arbitrary. Indeed, if $U_z(\phi)$ is a Clifford gate, any Pauli correction $P$ propagates into a new Pauli correction $P'$, i.e.\ $U_z(\phi)P=P'U_z(\phi)$, by the definition of a Clifford gate, without a need to update the angle $\phi$.
The measured angle on each qubit is therefore exactly the basis specified in $\mathcal{B}$. Moreover, since each component of $\mathcal{B}$ is either $0$ or $\pi/2$ in the Clifford case, we can replace it by a vector $b$ such that $b_i=0$ if $\mathcal{B}_i=0$, and $b_i=1$ otherwise.
We call any such MBQC pattern built from a Clifford circuit a \textit{Clifford MBQC pattern} and denote it $\mbqc(A,\mathcal{I},\mathcal{O},b)$.

\subsection{Spacetime code of an MBQC circuit} \label{sec:spacetime-code-cluster-state}

As discussed above, the circuit realization of the pattern $\mbqc(A,\mathcal{I},\mathcal{O},b)$ can be written as a sequence of \texttt{CZ} gates, acting on some qubits divided into two categories: the input qubits $\mathcal{I}$, characterized by a stabilizer group $S_{\mathcal{I}}$, and the non-input qubits initialized in the $\ket{+}$ state. A subset of the qubits (the non-output qubits) is then measured in the $X$ or $Y$ basis.
We associate to this circuit a noise model where errors can either happen before or after all the \texttt{CZ} gates.
When the circuit is built by mapping a stabilizer or Floquet code experiment to MBQC, this noise model maps directly to the phenomenological noise model for the code (as we will prove in \cref{sec:from-stabilizer-codes-to-cluster-states} and \cref{sec:from-floquet-codes-to-cluster-states}).
To define the spacetime code corresponding to this noise model, we consider that all the \texttt{CZ} gates are applied at the same time step of the circuit and merged into a single unitary, that we call a \textit{\texttt{CZ} network}.
Since all the \texttt{CZ} gates commute, their order does not matter and this unitary is entirely characterized by the adjacency matrix $A$ of the cluster state graph, defined by $A_{i,j}=1$ if and only if there is a \texttt{CZ} gate between qubits $i$ and $j$.
To draw \texttt{CZ} networks within a circuit, we introduce a new circuit notation: a blue box surrounding a set of gates is treated as a single gate when constructing the spacetime code. In other words, no error can happen within such a box.
Here is an example of MBQC circuit using this notation:

\begin{align}
    \tikzsetnextfilename{sec4-example-mbqc-circuit}
    \begin{quantikz}
                           & \gatebox{3}{2} \qw & \ctrl{2}  & \meterD{X} \\
                           & \ctrl{1}           &           &            \\
        \lstick{$\ket{+}$} & \ctrl{-1}          & \ctrl{-2} & \meterD{Y}
    \end{quantikz}
\end{align}

We therefore currently have a circuit with two timesteps: $t=1$ before the \texttt{CZ} network and $t=2$ after.
However, as discussed in \cref{sec:spacetime-subsystem-code}, we need an odd number of steps to define a spacetime code.
We therefore add a layer of identity gates after the \texttt{CZ} network, adding a third timestep $t=3$ in the circuit.
For instance, the circuit above becomes:

\begin{align}
    \tikzsetnextfilename{sec4-adding-identities}
    \begin{quantikz}
                           & \gatebox{3}{2} \qw & \ctrl{2}  & \gate{I} &\meterD{X} \\
                           & \ctrl{1}           &           & \gate{I} &           \\
        \lstick{$\ket{+}$} & \ctrl{-1}          & \ctrl{-2} & \gate{I} &\meterD{Y}
    \end{quantikz}
\end{align}

We are now ready to describe the spacetime code. Every qubit $i$ initialized in the $\ket{+}$ state is associated with a gauge operator $X_{i,1}$, every input stabilizer $S$ to a gauge operator $S_1$ (i.e.\ the spacetime operator corresponding to $S$ at $t=1$), and every measurement on qubit $j$ to a gauge operator $X_{j,3}$ or $Y_{j,3}$ depending on the measurement basis.

As discussed in \cref{sec:stabilizers-spacetime-code}, every spacetime stabilizer is either the spackle of some input gauge operators (coming from the input stabilizers) or the backle of some output gauge operators (coming from the final measurements).
We give the following characterization of those spacetime stabilizers:

\begin{proposition} \label{prop:cs-spacetime-stabilizers}
    Let $\mbqc(A,\mathcal{I},\mathcal{O},b)$ be a Clifford MBQC pattern and $\mathcal{S}$ its input stabilizer group.
    The stabilizers of the corresponding spacetime code are the spacetime operators $S$ that commute with all the gauge operators and have either:
    \begin{itemize}
        \item The restriction of $S$ to the input qubits at $t=1$ is equal to an input stabilizer:
        \begin{align} \label{eq:mbqc-spacetime-stabilizer-cond-1}
            S_{\mathcal{I},1} \in \mathcal{S}.
        \end{align}
        \item The restriction of $S$ to the output qubits at $t=2$ and $t=3$ is the identity:
        \begin{align} \label{eq:mbqc-spacetime-stabilizer-cond-2}
            S_{\mathcal{O},2}=S_{\mathcal{O},3}=I.
        \end{align}
    \end{itemize}
\end{proposition}
\begin{proof}
    Let $S$ be a stabilizer of the spacetime code. By \cref{prop:stabilizers_backle_spackle}, $S$ is either the spackle of some input stabilizers of the MBQC circuits---including both the single-qubit $X$ stabilizers characterizing the $\ket{+}$ states, and elements of $\mathcal{S}$ on the input qubits---or the backle of some $X/Y$ measurements at the end of the circuit.
    In the first case, the restriction of $S$ to the input qubits at $t=1$ must be an element of $\mathcal{S}$, which includes the identity in the case that $S$ is the spackle of single-qubit $X$ stabilizers at non-input qubits only.
    In the second case, the restriction of $S$ to the input qubits at $t=1$ can be any operator that commutes with the input stabilizers.
    However, as the backle of some final measurements, it cannot have any support on the output qubits. Therefore, $S_{\mathcal{O},2}=S_{\mathcal{O},3}=I$.

    In the other direction, let $S$ be a spacetime operator that commutes with all gauge operators of the spacetime code and obeys either \cref{eq:mbqc-spacetime-stabilizer-cond-1} or \cref{eq:mbqc-spacetime-stabilizer-cond-2}.
    Since $S$ commutes with all gauge operators, it commutes in particular with the single-qubit $X$ operators at time $t=1$ associated to the $\ket{+}$ states on non-input qubits, and with the single-qubit $X/Y$ measurement propagation operators at $t=3$ on non-output qubits.
    This implies that $S$ must have either $X$ or $I$ on its non-input qubits at $t=1$, either $X$ or $I$ on its non-output qubits measured in the $X$ basis at $t=3$, and either $Y$ or $I$ on its non-output qubits measured in the $Y$ basis at $t=3$.

    If $S$ obeys \cref{eq:mbqc-spacetime-stabilizer-cond-1}, we have $S_{\mathcal{I},1} \in \mathcal{S}$ and $S$ restricted to $t=1$ can be written as product of measurement propagation operators.
    Moreover, since $S$ commutes with all gauge operators, it must be a gate propagation operator of the \texttt{CZ} network when restricted to $t=1$ and $t=2$, by \cref{lemma:elementary-prop-completeness}.
    By the same argument, its restrictions to $t=2$ and $t=3$ must be identical, as there are only identity gates in-between those two timesteps. Therefore, $S$ must be the spackle of some initial stabilizers (including the single-qubit $X$ stabilizers on non-input qubits).

    If $S$ obeys \cref{eq:mbqc-spacetime-stabilizer-cond-2}, we have $S_{\mathcal{O},3}=I$ and $S$ restricted to $t=3$ can be written as a product of measurement propagation operators. By a similar argument as above, it must be identical at $t=2$ and $t=3$, and a \texttt{CZ} network gate propagation operator between $t=1$ and $t=2$, meaning that $S$ is the backle of some final measurements.

    In all cases, $S$ is either a spackle or a backle and must therefore be a gauge operator by \cref{lemma:spackles-backles-gauge-ops}. Since $S$ commutes with all gauge operators and is a gauge operator, it is a stabilizer of the spacetime code.
\end{proof}

\subsection{Cluster state complex} \label{sec:cluster-state-complex}

It has recently been shown that many types of foliated and non-foliated cluster states can be understood using the language of chain complexes
\cite{nickerson2018measurement,newman2020generating, hillmann2024single}. In particular, those works associate to a cluster state graph a chain complex of the form
\begin{align}
    C_3 \xrightarrow{\partial_3} C_2 \xrightarrow{\partial_2} C_1 \xrightarrow{\partial_1} C_0
\end{align}
where it is assumed that the nodes of the cluster state graph can be bipartitioned into primal and dual nodes, such that no edge connects a primal node to a dual node.
In this case, detectors can also be bipartitioned into primal and dual detectors depending on the type of qubits in their support.
The spaces $C_2$ and $C_1$ then represent, respectively, the primal and dual nodes, while $C_3$ and $C_0$ represent, respectively, the primal and dual detectors.
This representation works for the foliation of any CSS code and also allows for the discovery of many non-foliated cluster states~\cite{newman2020generating}.
However, this representation breaks down when considering the foliation of non-CSS codes, due to the potential presence of $Y$ measurements in the MBQC pattern and length-3 cycles in the cluster state graph, preventing the existence of a bipartition of the nodes \cite{BrownUniversalFTMBQC2020}.
Moreover, the exact relationship between this chain complex and the spacetime complex of the corresponding spacetime code has not been previously elucidated.

In this section, we generalize this chain complex formalism to any generic (Clifford) MBQC pattern.
Using the reduction rules of \cref{sec:fault-tolerant-maps}, we then show that this chain complex is equivalent to the spacetime complex of the associated circuit.
\begin{definition}[Cluster state complex]
\label{def:cluster-state-complex}
Consider a pattern $\mbqc(A,\mathcal{I},\mathcal{O},b)$ and an input stabilizer group $\mathcal{S}$.
We denote by $n_Q$ the dimension of $A$, $n_I$ the number of elements of $\mathcal I$, $n_O$ the number of elements of $\mathcal O$, $n_S$ the rank of $\mathcal S$, and $B=\text{diag}(b)$ the matrix with diagonal $b$ and zero elsewhere.
The corresponding $\emph{cluster state complex}$ is a length-2 chain complex
\begin{align} \label{eq:cluster-state-chain-complex}
    C_0 \xrightarrow{G^T} C_1 \xrightarrow{H} C_2
\end{align}
where $C_0$, $C_1$ and $C_2$ are finite-dimensional vector spaces over $\mathbb{F}_2$, and:
\begin{itemize}
    \item $C_0$ has dimension $n_Q+n_S$, and we call the basis elements of this space \emph{gauge nodes}.
    We have one gauge node per vertex of the cluster state graph, and one gauge node per input stabilizer.
    We partition those nodes into five sets:
    \begin{itemize}
    \item $\mathcal{G}_S$ (nodes corresponding to input stabilizers),
    \item $\mathcal{G}_{I\overline{O}}$ (nodes corresponding to input non-output vertices),
    \item $\mathcal{G}_{\overline{I}O}$ (nodes corresponding to output non-input vertices),
    \item $\mathcal{G}_{IO}$ (nodes corresponding to input output vertices),
    \item $\mathcal{G}_{\overline{I}\overline{O}}$ (nodes corresponding to non-input non-output vertices).
    \end{itemize}
    \item $C_1$ has dimension $n_Q+n_I+n_O$, and we call the basis elements of this space \emph{error nodes}.
    We have one error node per vertex of the cluster state graph, called a \emph{$Z$-type node}, interpreted as a $Z$ error applied to the corresponding qubit (at any timestep, as all the $Z$ errors on a given wire are equivalent).
    We also have one additional error node per input and output vertex of the cluster state graph, called an \emph{$X$-type node}, interpreted as an $X$ error on this qubit applied either before the \texttt{CZ} network (for input vertices) or after (for output vertices).
    Since these nodes do not correspond to any vertex of the original cluster state graph, we call them \emph{virtual nodes}, and we call the other basis elements \emph{physical nodes}.
    Given a vector $x \in C_1$, we denote $x_{\mathcal{I}}$ ($x_{\mathcal{O}}$) its restriction to the physical input (output) nodes, and  $x_{\tilde{\mathcal{I}}}$ ($x_{\tilde{\mathcal{O}}}$) its restriction onto the virtual ones.
    We emphasize that if a vertex is both input and output, it is associated to two distinct virtual nodes in $C_1$.
    We partition the nodes into six sets:
    \begin{itemize}
    \item $\mathcal{X}_{I\overline{O}}$ (virtual nodes corresponding to input non-output vertices),
    \item $\mathcal{X}_{\overline{I}O}$ (virtual nodes corresponding to output non-input vertices),
    \item $\mathcal{X}_{\overline{I}\overline{O}}$ (physical nodes corresponding to non-input non-output vertices),
    \item $\mathcal{X}_{IO}$ (first set of virtual nodes corresponding to input output vertices),
    \item $\mathcal{X}'_{IO}$ (second set of virtual nodes corresponding to input output vertices),
    \item $\mathcal{Z}_{I}$ (physical nodes corresponding to input vertices) and
    $\mathcal{Z}_{\overline{I}}$ (physical nodes corresponding to non-input vertices).
    \end{itemize}
    \item $G$ is the $(n_Q+n_S) \times (n_Q+n_I+n_O)$ matrix:
    \begin{equation} \label{eq:G-matrix}
        \def\arraystretch{1.4}
        G = \begin{blockarray}{ccccccc}
                                                     & \mathcal{X}_{I\overline{O}}       & \mathcal{X}_{IO}       & \mathcal{X}_{\overline{I}O} & \mathcal{X}'_{IO} & \mathcal{Z}_{I}  & \mathcal{Z}_{\overline{I}} \\
            \begin{block}{c(cccc|cc)}
              \mathcal{G}_S                          & S^X_{\mathcal{X}_{I\overline{O}}} & S^X_{\mathcal{X}_{IO}} & 0                           & 0                 & S^Z              & 0 \\[5pt] \cline{2-7}
              \mathcal{G}_{I\overline{O}}            & I                                 & 0                      & 0                           & 0                 & \multicolumn{2}{c}{\multirow{4}{*}{$A+B$}} \\
              \mathcal{G}_{\overline{I}O}            & 0                                 & 0                      & I                           & 0                 & \\
              \mathcal{G}_{IO}                       & 0                                 & I                      & 0                           & I                 & \\
              \mathcal{G}_{\overline{I}\overline{O}} & 0                                 & 0                      & 0                           & 0                 & \\
            \end{block}
        \end{blockarray}
    \end{equation}
    where $S = (S^X | S^Z)$ is the parity-check matrix of the input stabilizer group, and $S^X_{E}$ for a set $E$ denotes $S^X$ restricted to the column indices in $E$.
    Moreover, the rows and columns of $A$ and $B$ are reordered to correspond to the indicated row and column labels.
    We note that the matrix $G$, interpreted as the biadjacency matrix between basis elements of $C_0$ and $C_1$, has a simple graphical construction:
    \begin{enumerate}
        \item Insert a $Z$-type node for every vertex of the cluster state graph.
        \item Insert an $X$-type node for every input and output vertex of the cluster state graph.
        \item Insert a gauge node for every $Z$-type vertex $z$ of the graph, and connect it to all the $Z$-type nodes whose corresponding vertices in the cluster state graph were connected to the vertex corresponding to $z$. Moreover, if $z$ corresponds to an input or output vertex of the cluster state graph, connect $z$ to all the $X$-type nodes corresponding to this vertex. Finally, if $z$ corresponds to a vertex of the cluster state graph measured in the $Y$ basis, connect the gauge node to $z$.
        \item For every input stabilizer, insert a gauge node connected to all the input $X$ and $Z$-type nodes that make up the stabilizer.
    \end{enumerate}
    \item $C_2$ has dimension $m$, and represents the space of cluster state detectors.
    A detector is a vector $x \in \ker G$ such that either $(x_{\mathcal{I}} | x_{\tilde{\mathcal{I}}}) \in \mathcal{S}$ or $(x_{\mathcal{O}} | x_{\tilde{\mathcal{O}}})=0$.
    $H$ is an $m \times (n_Q+n_I+n_O)$ matrix whose rowspace is the space of detectors.
\end{itemize}
The chain complex condition $HG^T=0$ (or equivalently $GH^T=0$) is fulfilled, since every row of $H$ belongs to $\ker{G}$ by construction.
\end{definition}

Examples of cluster state complexes are shown in \cref{fig:zz-meas-cs-complex}, \cref{fig:i-and-o-cs-complex} and \cref{fig:y-meas-cs-complex}. Before proving that this chain complex is equivalent to the gauge complex of the corresponding spacetime code, let us start by building some intuition.
We begin by providing a graphical representation of the cluster state complex. As described in \cref{sec:graphical-rep-chain-complex}, the graphical representation of the chain complex in \cref{eq:cluster-state-chain-complex}, contains three types of nodes: gauge nodes ($C_0$), error nodes ($C_1$) and detector nodes ($C_2$). However, it is possible to compress this graph and have a single node representing both a gauge operator and an error.

\begin{figure}
    \centering
    \begin{minipage}{.4\linewidth}%
    \subfloat[]{
        \tikzsetnextfilename{sec4-zz-meas-cluster-state}
        \begin{quantikz}
            \lstick{$\ket{+}$} & \gatebox{7}{6} &           & \ctrl{2}  &           &           & \ctrl{6}  &            \\
            \lstick{$\ket{+}$} &                &           &           & \ctrl{2}  & \ctrl{5}  &           &            \\
            \lstick{$\ket{+}$} & \ctrl{2}       &           & \ctrl{-2} &           &           &           & \meterD{X} \\
            \lstick{$\ket{+}$} &                & \ctrl{2}  &           & \ctrl{-2} &           &           & \meterD{X} \\
                               & \ctrl{-2}      &           &           &           &           &           & \meterD{X} \\
                               &                & \ctrl{-2} &           &           &           &           & \meterD{X} \\
            \lstick{$\ket{+}$} &                &           &           &           & \ctrl{-5} & \ctrl{-6} & \meterD{X}
        \end{quantikz}
    }
    \\
    \makebox[\linewidth]{\subfloat[]{
        \includegraphics[width=0.5\textwidth]{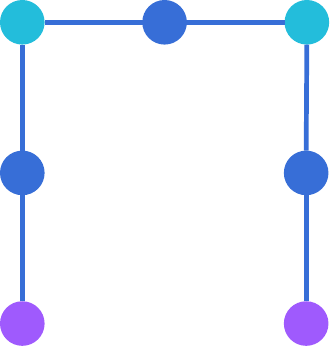}
    }}
    \end{minipage}
    \begin{minipage}{.3\linewidth}%
    \subfloat[]{ \label{fig:zz-meas-cs-complex}
        \includegraphics[width=0.86\textwidth]{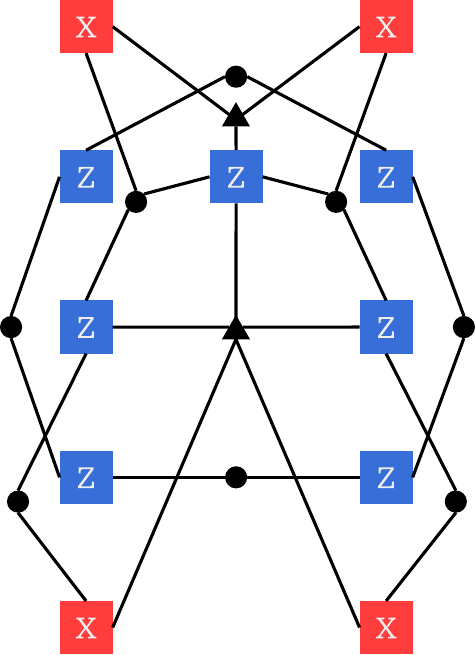}
    }
    \\
    \subfloat[]{ \label{fig:zz-meas-augmented-cs-complex}
        \includegraphics[width=0.86\textwidth]{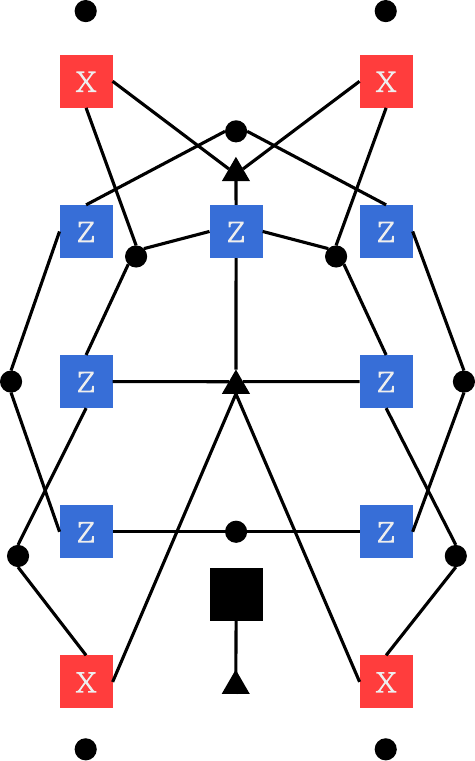}
    }
    \end{minipage}
    \begin{minipage}{.2\linewidth}
    \subfloat[]{ \label{fig:zz-meas-ccs}
        \includegraphics[width=1\textwidth]{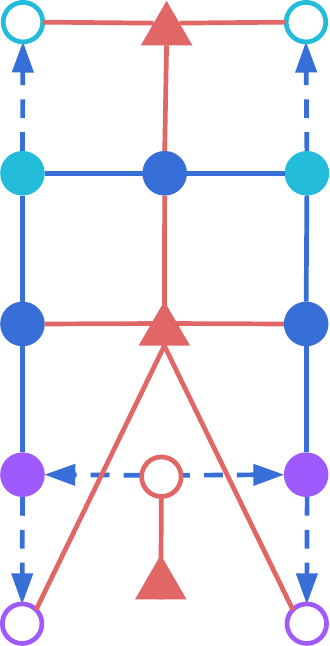}
    }
    \\
    \subfloat[]{ \label{fig:zz-meas-co-ccs}
        \includegraphics[width=1\textwidth]{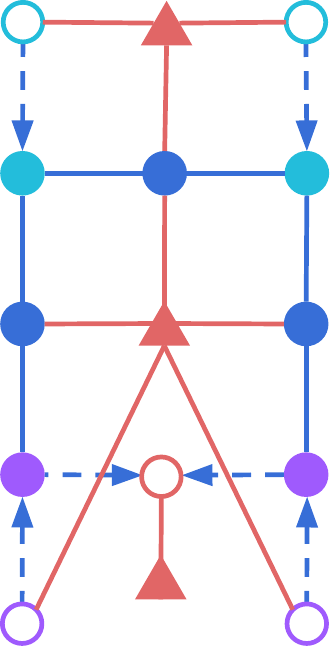}
    }
    \end{minipage}
    \caption{ \label{fig:zz-meas-cluster-state}
        The different representations of a cluster state. \textbf{(a)} Circuit consisting of two successive teleportations of a two-qubit state, followed by a $ZZ$ measurement. We assume that the input state is stabilized by $\mathcal{S}=\{I, ZZ\}$. \textbf{(b)} Cluster state graph of the corresponding MBQC pattern. Purple and light blue nodes indicate input and output qubits, respectively. \textbf{(c)} The corresponding cluster state complex, where circles indicate gauge nodes, squares error nodes ($X$-type and $Z$-type), and triangles detector nodes. \textbf{(d)} Augmented cluster state complex, where we have added a new gauge node for each input and output $X$-type node, and one new error node and detector node (black square and triangle connected to it) for the gauge node associated to the input $ZZ$ stabilizer. In this complex we have the same number of error and gauge nodes, allowing to interpret the biadjacency matrix between circles and squares as an adjacency matrix. \textbf{(e)} Compressed representation of the cluster state complex. Each circle now represents both an error operator ($Z$ for filled circles and $X$ for empty circles) and a gauge operator (defined by its neighborhood). Red triangles represent the detectors. \textbf{(f)} Compressed co-representation of the cluster state complex, where the arrows of the previous graph have been inverted. Error nodes now have the opposite interpretation: the filled circles are $X$ and the empty ones are $Z$.
        This graph is useful to derive the detectors, as the set of nodes with an empty boundary whose restriction to the input qubits is an input stabilizer or whose restriction on the output qubits is empty.
    }
\end{figure}

\begin{figure}
    \centering
    \begin{minipage}{.4\linewidth}%
    \subfloat[]{
        \tikzsetnextfilename{sec4-mbqc-circuit-both-input-and-output}
        \begin{quantikz}
            \lstick{$\ket{+}$} & \gatebox{5}{3} & \ctrl{1}  & \ctrl{4}  &            \\
            \lstick{$\ket{+}$} & \ctrl{2}       & \ctrl{-1} &           & \meterD{X} \\
                               &                & \ctrl{2}  &           &            \\
                               & \ctrl{-2}      &           &           & \meterD{X} \\
            \lstick{$\ket{+}$} &                & \ctrl{-2} & \ctrl{-4} & \meterD{X}
        \end{quantikz}
    }
    \\
    \makebox[\linewidth]{\subfloat[]{
        \includegraphics[width=0.5\textwidth]{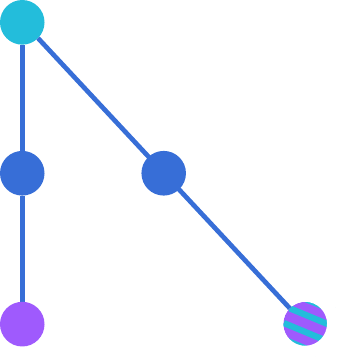}
    }}
    \end{minipage}
    \begin{minipage}{.3\linewidth}%
    \subfloat[]{ \label{fig:i-and-o-cs-complex}
        \includegraphics[width=0.8\textwidth]{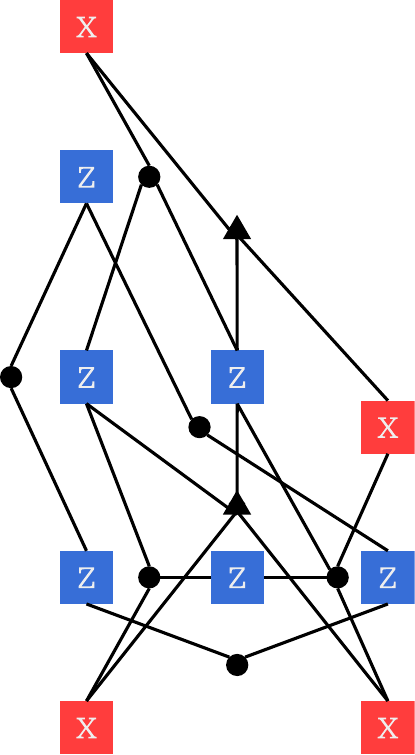}
    }
    \\
    \subfloat[]{ \label{fig:i-and-o-augmented-cs-complex}
        \includegraphics[width=0.8\textwidth]{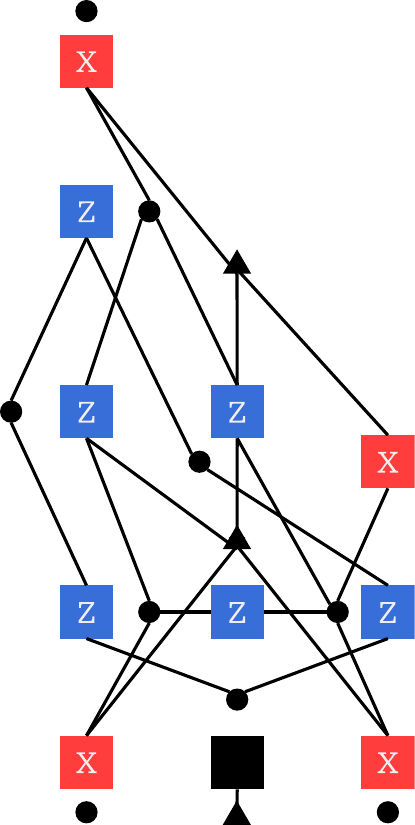}
    }
    \end{minipage}
    \begin{minipage}{.2\linewidth}
    \subfloat[]{ \label{fig:i-and-o-ccs}
        \includegraphics[width=1\textwidth]{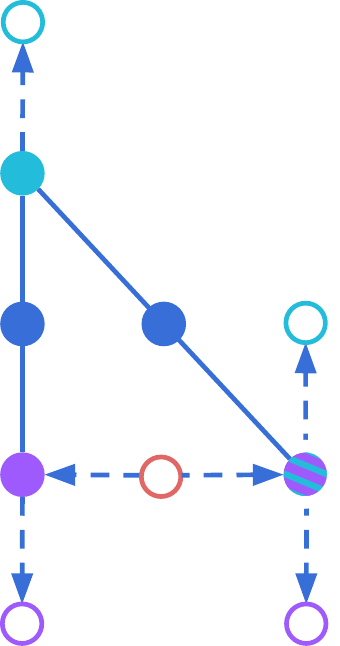}
    }
    \\
    \subfloat[]{ \label{fig:i-and-o-co-ccs}
        \includegraphics[width=1\textwidth]{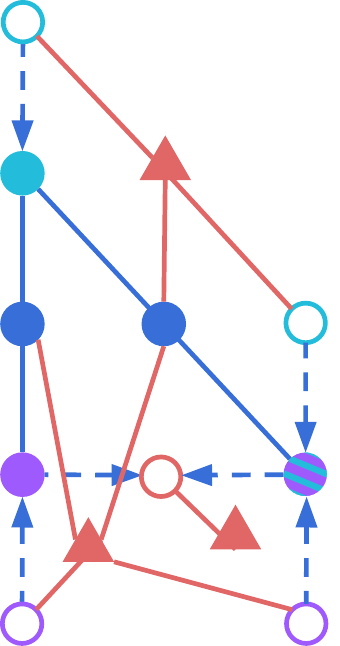}
    }
    \end{minipage}
    \caption{The different representations of a cluster state, where one node is both an input and output. See \cref{fig:zz-meas-cluster-state} for a more detailed description of the drawing conventions. \textbf{(a)} Circuit consisting of two successive teleportations of one of the two input qubits, followed by a $ZZ$ measurement. We assume that the input state is stabilized by $\mathcal{S}=\{I, ZZ\}$. \textbf{(b)} Cluster state graph of the corresponding MBQC pattern. \textbf{(c)} Associated cluster state complex. The input-output node of the cluster state graph leads to a $Z$-type node in the complex, as well as two $X$-type nodes, representing the two different $X$ errors on that qubit, before and after the \texttt{CZ} network. \textbf{(d)} Augmented cluster state complex. \textbf{(e)} Compressed representation of the cluster state complex, without detector nodes for clarity \textbf{(f)} Compressed co-representation of the cluster state complex, with detector nodes.
    }
\end{figure}

\begin{figure}
    \centering
    \subfloat[]{
            \raisebox{2.2em}{%
            \tikzsetnextfilename{sec4-example-mbqc-y-measurement}
            \begin{quantikz}
                \lstick{$\ket{+}$} & \labeledwire{nicered}{X} & \ctrl{1}  & \labeledwire{nicered}{Y} & \meterD{Y} \\
                \lstick{$\ket{+}$} & \labeledwire{nicered}{X} & \ctrl{-1} & \labeledwire{nicered}{Y} & \meterD{Y}
            \end{quantikz}
        }
    }
    \subfloat[]{ \centering \label{fig:y-meas-cs}
        \includegraphics[width=0.05\textwidth]{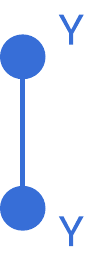}
    }
    \hspace{10pt}
    \subfloat[]{ \centering \label{fig:y-meas-cs-complex}
        \includegraphics[width=0.1\textwidth]{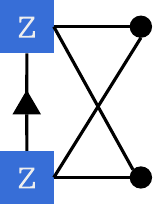}
    }
    \hspace{10pt}
    \subfloat[]{ \centering \label{fig:y-meas-ccs}
        \includegraphics[width=0.08\textwidth]{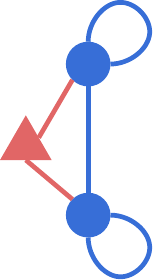}
    }
    \caption{
        Different representation of a cluster state with $Y$ measurements. See \cref{fig:zz-meas-cluster-state} for a more detailed description of the drawing conventions. \textbf{(a)} Two-qubit circuit with one \texttt{CZ} gate, $\ket{+}$ initialization and $Y$ measurement on all qubits. It has one spacetime stabilizer, labeled in red. \textbf{(b)} Associated cluster state graph. \textbf{(c)} Corresponding cluster state complex. The $Y$ measurements are reflected through the gauge nodes connected to their own associated error node. \textbf{(d)} Compressed representation of the cluster state complex. The self-loops are due to the $Y$ measurements of the corresponding qubits.
    }
    \label{fig:y-meas-cluster-state-example}
\end{figure}

\begin{definition}[Compressed representation of a cluster state complex] \label{def:compressed-rep-cluster-state}
    Given a cluster state complex
    \begin{align}
        C_0 \xrightarrow{G^T} C_1 \xrightarrow{H} C_2
    \end{align}
    we can construct an equivalent complex, called the \textit{augmented cluster state complex}
    \begin{align}
        C'_0 \xrightarrow{\tilde{G}^T} C'_1 \xrightarrow{\tilde{H}} C'_2
    \end{align}
    such that $C'_0$ and $C'_1$ have the same dimension. For this, we add $n_I+n_O$ basis elements to $C_0$, i.e.\ $C'_0=C_0 \oplus \mathbb{Z}_2^{n_I+n_O}$, and $n_S$ basis elements to $C_1$, i.e.\ $C'_1=C_1 \oplus \mathbb{Z}_2^{n_S}$.
    We build the square matrix $\tilde{G}$, of size $n_Q+n_I+n_O+n_S$, by adding $n_I+n_O$ zero rows and $n_S$ zero columns to $G$. Finally, we add $n_S$ basis elements to $C_2$ and define $\tilde{H}$ as
    \begin{align}
        \tilde{H} = \begin{pmatrix}
            H & 0 \\
            0 & I_{n_S}
        \end{pmatrix}
    \end{align}
    The directed graph defined by the adjacency matrix
    \begin{align}
        B=\begin{pmatrix}
            \tilde{G}^T & \tilde{H} \\
            \tilde{H} & 0
        \end{pmatrix}
    \end{align}
    is called the \textit{compressed representation of the cluster state complex}. The first $n_Q+n_I+n_O+n_S$ rows and columns of $B$ correspond to the error/gauge nodes, while the last $m$ rows and columns correspond to the detector nodes.

    This graph can be constructed directly by starting from the cluster state graph and adding one new vertex per input qubit, output qubit and input stabilizer.
    Since we now have two vertices for each input and output qubit, we label the one present in the original graph as a $Z$-type vertex and the new one as an $X$-type vertex. We then add one directed edge from every $Z$-type vertex to its corresponding $X$-type vertex.
    For every input stabilizer $(h_X|h_Z)$, we add a directed edge between its corresponding vertex in the graph, and every $X$-type ($Z$-type) input vertex in the support of $h_X$ ($h_Z$).
    We also add a self-loop to every vertex $i$ whose corresponding qubit is measured in the $Y$ basis.

    Finally, we add detectors to the graph.
    They can be found directly from our current graph, by first inverting all the arrows, i.e.\ drawing the graph with adjacency matrix $\tilde{G}$ instead of $\tilde{G}^T$.
    In this new graph, the $X$-type vertices are now interpreted as $Z$ errors and the $Z$-type vertices as $X$ errors (due to the symplectic matrix $\Omega$ present in the second boundary map of the chain complex).
    The detectors can be identified as independent sets of vertices with an empty boundary (this corresponds to finding a basis of $\ker(\tilde{G})$) whose restriction to the input qubits is an input stabilizer or whose restriction to the output qubits is empty. Note that if a vertex $v$ is connected to another vertex $u$ by a directed edge, then $u$ is in the boundary of $v$ only if the directed edge goes from $v$ to $u$. The graph with arrows inverted and adjacency matrix
    \begin{align}
        B^T=\begin{pmatrix}
            \tilde{G} & \tilde{H} \\
            \tilde{H} & 0
        \end{pmatrix}
    \end{align}
    is called the \textit{compressed co-representation of the cluster state complex}.
\end{definition}

An example of the representation and co-representation of a cluster state complex is shown in \cref{fig:zz-meas-ccs} and \cref{fig:zz-meas-co-ccs}.
We are now ready to prove the main theorem of this section:

\begin{figure}
    \centering
    \subfloat[Non-input, $X$ measurement]{ \centering
        \includegraphics[width=0.45\textwidth]{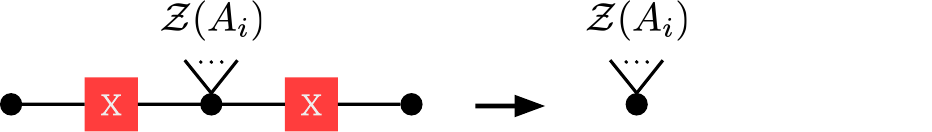}
    }
    \hspace{10pt}
    \subfloat[Input, $X$ measurement]{ \centering
        \includegraphics[width=0.45\textwidth]{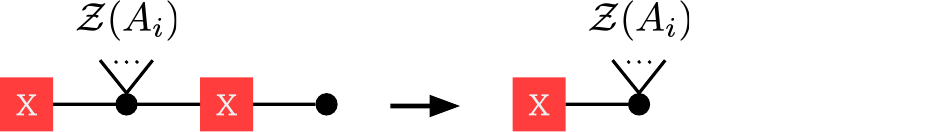}
    }
    \\
    \subfloat[Non-input, $Y$ measurement]{ \centering
        \includegraphics[width=0.45\textwidth]{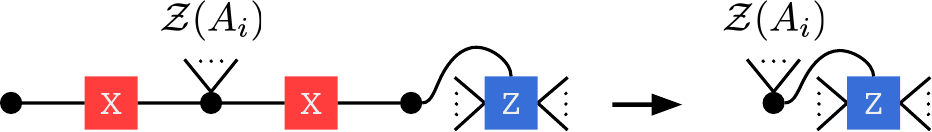}
    }
    \hspace{10pt}
    \subfloat[Input, $Y$ measurement]{ \centering
        \includegraphics[width=0.45\textwidth]{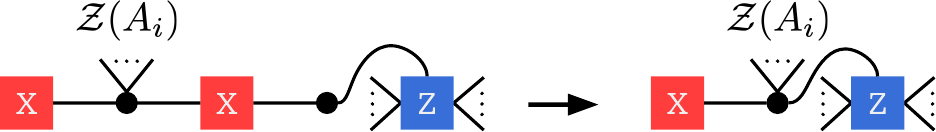}
    }
    \\
    \subfloat[Non-input, output]{ \centering
        \includegraphics[width=0.45\textwidth]{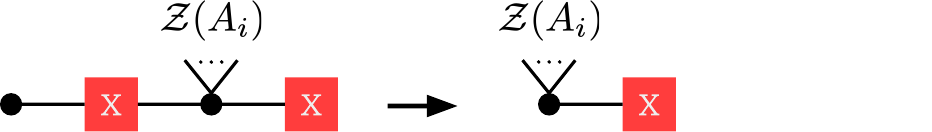}
    }
    \hspace{10pt}
    \subfloat[Input, output]{ \centering
        \includegraphics[width=0.45\textwidth]{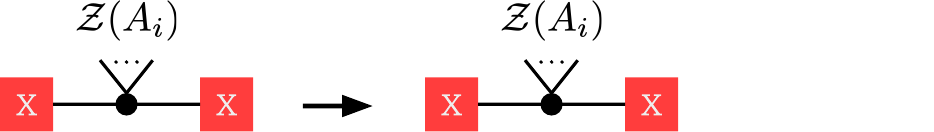}
    }
    \caption{ \label{fig:six-cases}
        From the spacetime complex of the MBQC circuit to the cluster state complex.
    }
\end{figure}

\begin{theorem} \label{theorem:mbqc-circuit-equivalent-to-cluster-state-complex}
    Given a pattern $\mbqc(A,\mathcal{I},\mathcal{O},b)$ and an input stabilizer group $\mathcal{S}$, the spacetime complex of its circuit implementation and its cluster state complex are equivalent.
\end{theorem}
\begin{proof}
    Let's begin by setting up some notations. Let $e_i \in \mathbb{Z}_2^n$ the vector containing a $1$ at position $i$ and $0$ everywhere else, for any $i \in \{1,\ldots,n\}$.
    For any matrix $M$, we write $M_i := Me_i$ for the $i$th column.
    Given the spacetime complex of a circuit, a column $j$ of that circuit, and a binary vector $v$, we write $\mathcal{X}_j(v)$ ($\mathcal{Z}_j(v)$) for the set of all $X$ ($Z$) nodes of the chain complex on column $j$ and row $i$ for all $i$ such that $v_i=1$.

    We now consider the chain complex $C_\bullet$ associated to the spacetime code of the MBQC circuit
    \begin{align}
        C_2 \xrightarrow{G^T} C_1 \xrightarrow{H \Omega} C_0
    \end{align}
    We want to show that there exists a fault-tolerant map to the cluster state complex $C'_\bullet$:
    \begin{align}
        C'_0 \xrightarrow{G'^T} C'_1 \xrightarrow{H'} C'_2
    \end{align}
    with spaces and boundary maps as defined in \cref{def:cluster-state-complex}. We will use a combination of rules A and B to construct this fault-tolerant map. We will start by describing the succession of reductions we are applying and how they affect the error and gauge nodes. We will then see how they affect the detector nodes.

    For every row $i$ of the circuit, the \texttt{CZ} network induces the following two gauge operators:
    \begin{align}
        g^X_i &= \mathcal{X}_1(e_i)+\mathcal{X}_2(e_i)+\mathcal{Z}(A_i) \\
        g^Z_i &= \mathcal{Z}_1(e_i)+\mathcal{Z}_2(e_i)
    \end{align}
    Since the second gauge operator has weight 2, we can use rule A to remove it from the chain complex and replace the nodes $\mathcal{Z}_1(e_i)$ and $\mathcal{Z}_2(e_i)$ by a single node, that we denote $\mathcal{Z}(e_i)$.
    We can often simplify the first gauge operator as well, in a way that depends on whether the wire corresponds to an input qubit or not, and whether it is measured in the $X$ basis, $Y$ basis or is an output qubit. Those six cases are illustrated in \cref{fig:six-cases}. Let us study them separately:
    \begin{enumerate}
        \item Non-input, $X$ basis measurement: we have two weight-1 gauge operators connected to $\mathcal{X}_1(e_i)$ and $\mathcal{X}_2(e_i)$ respectively. We use rule B to remove those nodes from the chain complex, giving us the simplified gauge operator $g^X_i = \mathcal{Z}(A_i)$.
        \item Non-input, $Y$ basis measurement: we have one weight-1 gauge operators connected to $\mathcal{X}_1(e_i)$ and one weight-2 gauge operator connecting $\mathcal{X}_2(e_i)$ with $\mathcal{Z}(e_i)$. We use rules A and B to remove $\mathcal{X}_1(e_i)$ from the chain complex and merge $\mathcal{X}_2(e_i)$ with $\mathcal{Z}(e_i)$, giving us the new gauge operator $g^X_i = \mathcal{Z}(A_i+e_i)$.
        \item Non-input, output: we have one weight-1 gauge operators connected to $\mathcal{X}_1(e_i)$, which we remove using rule B, giving us the new gauge operator $g^X_i = \mathcal{Z}(A_i)+\mathcal{X}_2(e_i)$.
        \item Input, $X$ basis measurement: we have one weight-1 gauge operators connected to $\mathcal{X}_2(e_i)$, which we remove using rule B, giving us the new gauge operator $g^X_i = \mathcal{Z}(A_i)+\mathcal{X}_1(e_i)$.
        \item Input, $Y$ basis measurement: we have one weight-2 gauge operator connecting $\mathcal{X}_2(e_i)$ with $\mathcal{Z}(e_i)$. We use rules A to merge $\mathcal{X}_2(e_i)$ with $\mathcal{Z}(e_i)$, giving us the new gauge operator $g^X_i = \mathcal{Z}(A_i+e_i)+\mathcal{X}_1(e_i)$.
        \item Input, output: we cannot simplify such a row of the chain complex, and we therefore keep the gauge operator $g^X_i = \mathcal{Z}(A_i)+\mathcal{X}_1(e_i)$ + $\mathcal{X}_2(e_i)$.
    \end{enumerate}
    The reduced chain complex now has exactly $n_Q+n_I+n_O$ nodes: one $\mathcal{Z}(e_i)$ for every wire $i$, plus $\mathcal{X}_1(e_i)$ if it is an input wire and $\mathcal{X}_2(e_i)$ if it is an output wire.
    Moreover, it has one gauge operator per wire, connecting to $\mathcal{X}_1(e_i)$ ($\mathcal{X}_2(e_i)$) for input (output) qubits, as well as all the nodes in $\mathcal{Z}(A_i+b_ie_i)=\mathcal{Z}\left((A+B)_i\right)$.
    Finally, every input stabilizer $j$ corresponds to a gauge operator $g^S_j=\mathcal{X}_1(S^X_j)+\mathcal{Z}(S^Z_j)$. Putting those gauge operators in a matrix form gives exactly the $G$ matrix of \cref{eq:G-matrix}.

    So far, we have shown the existence of a fault-tolerant map
    \[\begin{tikzcd}[column sep=20pt, row sep=20pt, every cell/.append style={inner sep=4pt}]
        {C_2} & {C_1} & {C_0} \\
        {C'_2} & {C'_1} & {C'_0}
        \arrow["{G^T}", from=1-1, to=1-2]
        \arrow["{f_2}"', from=1-1, to=2-1]
        \arrow["{H\Omega}", from=1-2, to=1-3]
        \arrow["{f_1}"', from=1-2, to=2-2]
        \arrow["{f_0}"', from=1-3, to=2-3]
        \arrow["{G'^T}", from=2-1, to=2-2]
        \arrow["{\partial'_1}", from=2-2, to=2-3]
    \end{tikzcd}\]
    It remains to show that $H'=\partial'_1$, i.e., the detectors of the spacetime code transform into the cluster state detectors, which we defined as the elements $x \in \ker(G)$ such that either $(x_{\mathcal{I}} | x_{\tilde{\mathcal{I}}}) \in \mathcal{S}$ or $(x_{\mathcal{O}} | x_{\tilde{\mathcal{O}}})=0$. We will first show that every spacetime code stabilizer maps into a cluster state detector under our fault-tolerant map. We will then show that every cluster state detector comes from a spacetime code stabilizer.

    Let us start by showing that every spacetime code stabilizer maps into a cluster state detector.
    Let $s \in C_0$ be a spacetime code stabilizer and $S=H^T s$ its binary symplectic representation. Let $x=(H\Omega)^T s$ be its representation in $C_1$ (i.e. $S$ with $X$ and $Z$ exchanged), let $s'=f_0(s) \in C'_0$, and let $x'={\partial'}^T_1 s' \in C'_1$.
    We know that $x \in \ker\left(G'^T\right)$ by the chain complex condition of $C'_\bullet$.
    We therefore just need to show that either $(x'_{\mathcal{I}} | x'_{\tilde{\mathcal{I}}}) \in \mathcal{S}$ or $(x'_{\mathcal{O}} | x'_{\tilde{\mathcal{O}}})=0$.
    We first observe that $x$ and $x'$ are supported on the same rows of the chain complex, since we only applied rule $A$ on gauge operators connecting nodes of the same row.
    In particular, $x$ and $x'$ have support on the same input and output rows.
    By \cref{prop:cs-spacetime-stabilizers}, we know that either $S_{\mathcal{I},1} \in \mathcal{S}$ or $S_{\mathcal{O},2}=S_{\mathcal{O},3}=0$. In the first case, since for every input row $i$, $\mathcal{Z}_1(e_i)$ maps to $\mathcal{Z}(e_i)$ (physical input node) and $\mathcal{X}_1(e_i)$ maps to $\mathcal{X}_1(e_i)$ (virtual input node), we have $(x'_{\mathcal{I}} | x'_{\tilde{\mathcal{I}}}) \in \mathcal{S}$ (since $S$ and $x$ have $X$ and $Z$ exchanged).
    In the second case, since $f$ cannot map a stabilizer not supported on output nodes to a stabilizer supported on them by the previous observation, we must have $x_{\tilde{\mathcal{O}}}=x_{\mathcal{O}}=0$.

    We now show that every cluster state detector comes from a spacetime code stabilizer.
    Let $x' \in \ker(G')$ such that either $(x'_{\mathcal{I}} | x'_{\tilde{\mathcal{I}}} ) \in \mathcal{S}$ or $(x'_{\mathcal{O}} | x'_{\tilde{\mathcal{O}}})=0$. We want to show that there exists a spacetime stabilizer $s \in C_0$ such that $x'=f_1 \left((H\Omega)^T s \right)$.
    Let us map back $x'$ to a vector $x \in C_1$ the following way: every $Z$-type node in the support of $x'$ is mapped to the corresponding three nodes $Z_{i,1}$, $Z_{i,2}$ and $Z_{i,3}$ of $C_\bullet$. Every $X$-type input node in the support of $x'$ is mapped to the input node $X_{i,1}$.
    Every $X$-type output node is mapped to the two output nodes $X_{i,2}$ and $X_{i,3}$. It is easy to see that applying $f_1$ to this new vector $x$ gives us back $x'$.
    Moreover, since input nodes map to input nodes and output nodes to output nodes, it is clear that $x$ is a spacetime stabilizer by \cref{prop:cs-spacetime-stabilizers}.
    Therefore, there exists $s$ such that $x'=f_1(x)=\left((H\Omega)^T s \right)$, showing that any cluster state detector is in the image of a spacetime stabilizer.
\end{proof}

\section{Cluster state complex from a Clifford circuit}
\label{sec:from-spacetime-codes-to-cluster-states}
\begin{figure}
    \begin{center}
        \begin{tabular}{>{\centering\arraybackslash}m{3cm} >{\centering\arraybackslash}m{7cm}}
            \tikzsetnextfilename{sec5-swap}
            \begin{quantikz}
                & \gate[2]{\text{\texttt{SWAP}}} &  \\
                & \ghost{\text{\texttt{SWAP}}} &
             \end{quantikz}
             &
             \includegraphics[width=0.5\linewidth]{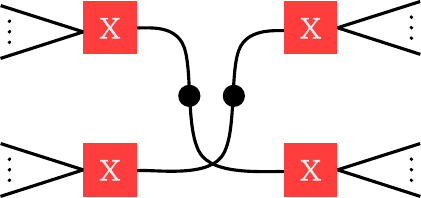}
            \\
            \hline
            \tikzsetnextfilename{sec5-swap-compilation}
            \begin{quantikz}
                & \ctrl{1} & \targ{}& \ctrl{1} &  \\
                & \targ{}& \ctrl{-1} & \targ{}&
            \end{quantikz}
            &
            \addstackgap[15pt]{
                \includegraphics[width=\linewidth]{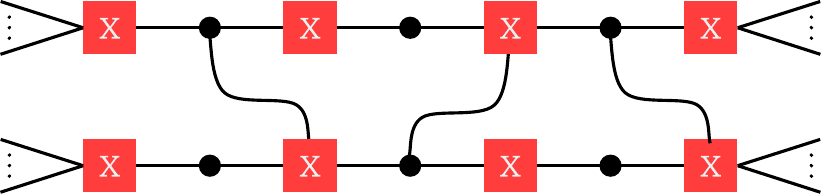}
            }
            \\
        \end{tabular}
    \end{center}
    \caption{ \label{fig:swap-example}
        \texttt{SWAP} gate vs \texttt{CNOT} compilation. We draw the circuits on the left, and the $X$ connected component of the gauge part of the chain complexes on the right. While the two circuits represent the same channel, they do not have equivalent chain complexes: the weight-3 gauge nodes cannot be reduced to weight-2 gauge nodes through fault-tolerant maps. This formalizes the observation that weight-1 $X$ errors happening on the top wire before the last \texttt{CNOT} propagate into weight-2 errors, thereby potentially reducing the distance of a spacetime code containing this component.
    }
\end{figure}

The goal of this section is to give a constructive proof of the following theorem:
\begin{theorem}
    Let $\mathcal{C}$ a Clifford circuit made of single-qubit Clifford gates, control-Pauli gates (where the Pauli operator can be multi-qubit) and single-qubit Pauli measurements at the end of the circuit. Then there exists an MBQC circuit equivalent to $\mathcal{C}$.
\end{theorem}
Note that different compilations of a given unitary with different sets of gates can give rise to spacetime codes with different distances. For instance, while a \texttt{SWAP} gate preserves the weight of any input Pauli error, its compilation into three \texttt{CNOT}s leads to some error-spreading that can decrease the distance of a code, as shown in \cref{fig:swap-example}. However, we will show that for the specific case of the compilation of our initial Clifford circuit into its MBQC version, the chain complexes are equivalent.

The proof works by progressively transforming circuit elements of the original Clifford circuit until a fully measurement-based circuit is obtained, that is, a circuit made only of ancilla qubits initialized in the $\ket{+}$ state, a $\texttt{CZ}$ network, and final Pauli measurements. It is made of two main steps:
\begin{enumerate}
    \item Compiling the single-qubit gates into \texttt{H} and \texttt{HS} gates, and the controlled-Pauli gates into \texttt{CZ} networks. This involves compiling some \texttt{H} gates into their MBQC version when $Y$ measurements are present. This is the content of \cref{sec:pre-compilation}.
    \item Compiling all the \texttt{H} and \texttt{HS} gates into their MBQC version and merging all the \texttt{CZ} networks. This can involve adding additional pairs of \texttt{H} gates in-between certain \texttt{CZ} networks. This is the content of \cref{sec:merge-cz-networks}.
\end{enumerate}
\cref{sec:general-lemmas} and \cref{sec:mbqc-compilation-h-hs} provide lemmas and propositions that are useful for the rest of the proof.

We often consider partial chain complexes associated to particular circuit elements, and we define input (output) nodes of such a partial chain complex to be those nodes corresponding to the $X$ and $Z$ spacetime operators at the beginning (end) of the circuit element on the wires that are not initialized (measured) on that circuit element.
We will always represent them by wires with ellipses, on the left for input nodes and on the right for output nodes.
Furthermore, we require that input and output nodes are preserved throughout any transformation.
To transform a circuit element while preserving the fault-tolerant properties of the full circuit that contains it, we use rules A and B, introduced in \cref{sec:fault-tolerant-maps}, on the partial chain complex corresponding to that circuit element.

An important observation that makes the diagrammatic proofs below work, is that applying rules A and B to a given partial chain complex does not change the support of the detectors on the input and output nodes. Therefore, any element of the kernel of the gauge group that came from a spackle or a backle in the original circuit, will still come from the same spackle or backle after the transformation.
Consequently, while the detectors of a circuit element are not fully defined without considering the full circuit, showing that the gauge part of the partial chain complex transforms correctly is enough to show that the detectors transform into the detectors as well.
Thus, we will not need to draw detector nodes when deriving the equivalence of two partial chain complexes coming from circuit elements.

\subsection{General lemmas}
\label{sec:general-lemmas}

We start by proving a few preliminary lemmas that will be useful throughout this section and the next.

\begin{lemma} \label{lemma:hadamard-square}
    Inserting two successive Hadamard gates in a circuit preserves its equivalence class.
\end{lemma}
\begin{proof}
Using the reduction rule A shows that the chain complexes of the two circuits, drawn below, are equivalent:
\begin{center}
    \begin{tabular}{>{\centering\arraybackslash}m{3cm} >{\centering\arraybackslash}m{6cm}}
        \tikzsetnextfilename{sec5-double-h}
        \begin{quantikz}
            & \gate{H} & \gate{H} &
         \end{quantikz}
         &
         \includegraphics[width=\linewidth]{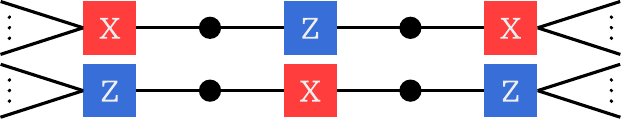}
        \\
        \hline
        \\
        \tikzsetnextfilename{sec5-identity-wire}
        \begin{quantikz}
            &  &  &
        \end{quantikz}
        &
        \addstackgap[15pt]{
            \includegraphics[width=0.35\linewidth]{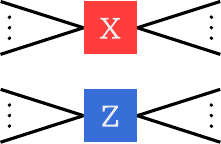}
        }
        \\
    \end{tabular}
\end{center}
\end{proof}

\begin{lemma} \label{lemma:pushing-h-away}
    Let $U$ be a multi-qubit unitary gate that can be written as $U=U'H_i$, where $H_i$ is a Hadamard gate applied on qubit $i$ and $U'$ another unitary gate.
    Then any circuit with the gates $H_i$ and $U'$ applied at successive time steps is equivalent to the circuit with the gate $U$.
    This can be represented by the following circuit equivalence:
    \begin{align}
        \tikzsetnextfilename{sec5-putting-h-away-1}
        \begin{quantikz}
            & \mygate{1}{H} \gatebox[U]{4}{2} & \mygate{4}{U'}      &          \\
            &                                               &                                      &          \\
            & \setwiretype{n} \ \vdots \                    &                                      & \vdots \ \\
            &                                               &                                      &
        \end{quantikz}
        \sim
        \tikzsetnextfilename{sec5-putting-h-away-2}
        \begin{quantikz}
            & \gate{H}                   & \mygate{4}{U'} \gatebox{4}{1} &          \\
            &                            &                                                  &          \\
            & \setwiretype{n} \ \vdots \ &                                                  & \vdots \ \\
            &                            &                                                  &
        \end{quantikz}
    \end{align}
    The same result applies if $H_i$ is applied after $U'$.
\end{lemma}
\begin{proof}
    Starting from the circuit element on the right-hand side, we apply rule A between the first and second timesteps of every wire. This gives us the circuit element of the left-hand side.
\end{proof}

\begin{lemma} \label{lemma:h-cz-sandwich}
    Any gate made of Hadamard gates sandwiched between two networks of \texttt{CZ}s can be compiled into a single \texttt{CZ} network as follows:
    \begin{align}
        \tikzsetnextfilename{sec5-h-cz-sandwich-1}
        \begin{quantikz}[row sep=6pt]
            & \gatebox{6}{5} \qw     & \mygate{6}{CZ_A} & \mygate{1}{H} & \mygate{6}{CZ_B} &        & \\
            & \setwiretype{n} \vdots &                  & \vdots        &                  & \vdots & \\
            &                        &                  & \mygate{1}{H} &                  &        & \\
            &                        &                  &               &                  &        & \\
            & \setwiretype{n} \vdots &                  & \vdots        &                  & \vdots & \\
            &                        &                  &               &                  &        &
        \end{quantikz}
        \sim
        \tikzsetnextfilename{sec5-h-cz-sandwich-2}
        \begin{quantikz}[row sep=5pt]
            \lstick{$\ket{+}$} & \gatebox{9}{6} \qw     &                  & \ctrl{3}  & \ \ldots \ &           & \mygate{3}{CZ_B} &            \\
                               & \setwiretype{n} \vdots &                  &           & \ \ldots \ &           &                  & \vdots     \\
            \lstick{$\ket{+}$} &                        &                  &           & \ \ldots \ & \ctrl{3}  &                  &            \\
                               &                        & \mygate{6}{CZ_A} & \ctrl{-3} & \ \ldots \ &           &                  & \meterD{X} \\
                               & \setwiretype{n} \vdots &                  & \vdots    & \ \ldots \ &           &                  & \vdots     \\
                               &                        &                  &           & \ \ldots \ & \ctrl{-3} &                  & \meterD{X} \\
                               &                        &                  &           & \ \ldots \ &           & \mygate{3}{CZ_B} &            \\
                               & \setwiretype{n} \vdots &                  & \vdots    & \ \ldots \ &           &                  & \vdots     \\
                               &                        &                  &           & \ \ldots \ &           &                  &
        \end{quantikz}
    \end{align}
\end{lemma}
\begin{proof}
    Let us begin by setting up some notation. Let $n$ be the number of wires of the circuit on the left, and $a$ the number of wires containing a Hadamard.
    Without loss of generality, we assume that the Hadamard gates lie on the first $a$ rows of the circuit, as in the circuit above.
    Let $A$ ($B$) the adjacency matrix associated to the first (second) \texttt{CZ} network.
    Let $e_i \in \mathbb{Z}_2^n$ be the vector containing a $1$ at position $i$ and $0$ everywhere else, for any $i \in \{1,\ldots,n\}$.
    For any matrix $M$, we write $M_i := Me_i$. We use $\odot$ for the element-wise multiplication of vectors.
    Given the chain complex of a circuit, a column $j$ of that circuit, and a binary vector $v$, we write $\mathcal{X}_j(v)$ ($\mathcal{Z}_j(v)$) for the set of all $X$ ($Z$) nodes of the chain complex on column $j$ and row $i$ for all $i$ such that $v_i=1$.

    We can now characterize the chain complex associated to the circuit on the left.
    Let $\mathcal{H}=\{1,\ldots,a\}$ the set of rows containing a Hadamard, and $h$ the indicator vector of the set, i.e.\ $h_i=1$ if and only if $i \in \mathcal{H}$. We use the notation $\overline{\mathcal{H}}=\{a,\ldots,n\}$ and $\overline{h}$ for the rows that do not contain a Hadamard.

    For every row $i \in \{1,\ldots,n\}$, we have two gauge operators, one for the propagation of $X$ and one for the propagation of $Z$, that we denote $g^X_i$ and $g^Z_i$ respectively. If $i \in \mathcal{H}$, the $X$ operator on row $i$ propagates into $Z$'s after the gate through different pathways: through the Hadamard on that wire, through the \texttt{CZ} gates in $CZ_A$ connected to the wire $i$ and ending on a wire $j \in \overline{\mathcal{H}}$, and through those ending on a wire $j \in \mathcal{H}$, propagating to an $X$ with the Hadamard, and onto $Z$s through the gates in $CZ_B$.
    Through this last process, the $X$ operator on row $i$ also propagates into $X$'s after the gate on the wires $j \in \mathcal H$ connected to $i$ through $CZ_A$.
    Using vector notation, the resulting gauge operator can therefore be written as
    \begin{align}
        g^X_i = \mathcal{X}_1(e_i) + \mathcal{X}_2(A_i \odot h) + \mathcal{Z}_2 \left(e_i + A_i \odot \overline{h} + B(A_i \odot h)\right).
    \end{align}
    Similarly, the $Z$ operator on row $i$ propagates onto an $X$ on that wire, and onto $Z$s through the second \texttt{CZ} network, resulting in the following gauge operator:
    \begin{align}
        g^Z_i = \mathcal{Z}_1(e_i) + \mathcal{X}_2(e_i) + \mathcal{Z}_2(B_i).
    \end{align}
    A similar analysis for the wires $i \in \overline{\mathcal{H}}$ gives us the following gauge operators:
    \begin{align}
        g^X_i &= \mathcal{X}_1(e_i) + \mathcal{X}_2(e_i + A_i \odot h) + \mathcal{Z}_2 \left(A_i \odot \overline{h} + B_i + B(A_i \odot h) \right), \\
        g^Z_i &= \mathcal{Z}_1(e_i) + \mathcal{Z}_2(e_i).
    \end{align}

    We now move to the circuit on the right.
    We write $\widetilde{\mathcal{H}}=\{1,\ldots,a\}$ for the set of rows that start with a $\ket{+}$, $\widetilde{\mathcal{H}}'=\{a,\ldots,2a\}$ for the set of rows with a final measurement, and $\widetilde{\overline{\mathcal{H}}}=\{2a,a+n\}$ for the rest of the rows.
    As before, we denote by $\widetilde{h}$, $\widetilde{h}'$ and $\widetilde{\overline{h}}$ their corresponding indicator vectors.
    We denote by $\tilde{g}^X_i$ and $\tilde{g}^Z_i$ the new gauge operators corresponding to the propagation of $X$ and $Z$ on every wire $i$.
    We use $\widetilde{A}$ and $\widetilde{B}$ to denote the adjacency matrices of the \texttt{CZ} networks now extended to all the new wires, and $\{\tilde{e}_i\}_i$ the canonical basis of $\mathbb{Z}_2^{n+a}$.
    We also define a new adjacency matrix $C$ for the new \texttt{CZ} network in the middle.
    To construct the chain complex, we first note that each single-qubit $Z$ operator before the gate propagates into a single $Z$ operator on that same wire, forming a weight-2 gauge operator, i.e.\ for every $i \in \{1,\ldots,n+a\}$ we have
    \begin{align}
        \tilde{g}^Z_i = \mathcal{Z}_1(\tilde{e}_i) + \mathcal{Z}_2(\tilde{e}_i).
    \end{align}
    We also have single-qubit gauge operators on each wire $i \in \widetilde{\mathcal{H}} \cup \widetilde{\overline{\mathcal{H}}}$, corresponding either to the initial state or to the final measurement, and connected to an $X$-type node.

    The rest of the gauge operators correspond to the propagation of $X$ on each wire.
    For every wire $i$, we have
    \begin{align} \label{eq:g-tilde-x-1}
        \tilde{g}^X_i &=
        \left\{
            \begin {aligned}
                 & \mathcal{X}_1(\tilde{e}_i) + \mathcal{X}_2(\tilde{e}_i) + \mathcal{Z}_2 \left(C\tilde{e}_i + \widetilde{B}_i)\right) \quad & i \in \widetilde{\mathcal{H}}, \\
                 & \mathcal{X}_1(\tilde{e}_i) + \mathcal{X}_2(\tilde{e}_i) + \mathcal{Z}_2 \left(C\tilde{e}_i + \widetilde{A}_i\right) & i \in \widetilde{\mathcal{H}}', \\
                 & \mathcal{X}_1(\tilde{e}_i) + X_2(\tilde{e}_i) + \mathcal{Z}_2\left(\widetilde{A}_i + \widetilde{B}_i\right) & i \in \widetilde{\overline{\mathcal{H}}}.
            \end{aligned}
        \right.
    \end{align}
    We now transform this chain complex into the chain complex of the original circuit. We start by contracting the weight-1 gauge operators using rule B, eliminating the nodes $X_1(\tilde{e}_i)$ for $i \in \mathcal{H}$ and $X_2(\tilde{e}_i)$ for $i \in \mathcal{H}'$. We also contract the weight-2 gauge operators $\tilde{g}^Z_i$ for all $i$, thereby replacing each $\mathcal{Z}_2(\tilde{e}_i)$ by $\mathcal{Z}_1(\tilde{e}_i)$. The gauge operators now become
    \begin{align} \label{eq:g-tilde-x-2}
        \tilde{g}^X_i &=
        \left\{
            \begin {aligned}
                 & \mathcal{X}_2(\tilde{e}_i) + \mathcal{Z}_1 \left(C\tilde{e}_i + \widetilde{B}_i \right) \quad & i \in \widetilde{\mathcal{H}}, \\
                 & \mathcal{X}_1(\tilde{e}_i) + \mathcal{Z}_1 \left(C\tilde{e}_i + \widetilde{A}_i \right) & i \in \widetilde{\mathcal{H}}', \\
                 & \mathcal{X}_1(\tilde{e}_i) + X_2(\tilde{e}_i) + \mathcal{Z}_1\left(\widetilde{A}_i + \widetilde{B}_i\right) & i \in \widetilde{\overline{\mathcal{H}}}.
            \end{aligned}
        \right.
    \end{align}
    We then relabel $Z$ nodes to $\mathcal{Z}_1$ if they are in $\widetilde{\mathcal{H}}'$ and $\mathcal{Z}_2$ if they are in $\widetilde{\mathcal{H}}$ or $\widetilde{\overline{\mathcal{H}}}$.
    With this new labeling, \cref{eq:g-tilde-x-2} becomes
    \begin{align} \label{eq:g-tilde-x-3}
        \tilde{g}^X_i &=
        \left\{
            \begin {aligned}
                 & \mathcal{X}_2(\tilde{e}_i) + \mathcal{Z}_1(C\tilde{e}_i) + \mathcal{Z}_2\left(\widetilde{B}_i \right) \quad & i \in \widetilde{\mathcal{H}}, \\
                 & \mathcal{X}_1(\tilde{e}_i) + \mathcal{Z}_1 \left(\widetilde{A}_i \odot \tilde{h}' \right) + \mathcal{Z}_2\left(C\tilde{e}_i + \widetilde{A}_i \odot \widetilde{\overline{h}}\right) & i \in \widetilde{\mathcal{H}}', \\
                 & \mathcal{X}_1(\tilde{e}_i) + X_2(\tilde{e}_i) + \mathcal{Z}_1\left(\widetilde{A}_i \odot \tilde{h}'\right) + \mathcal{Z}_2\left(\widetilde{A}_i \odot \widetilde{\overline{h}} + \widetilde{B}_i\right) & i \in \widetilde{\overline{\mathcal{H}}}.
            \end{aligned}
        \right.
    \end{align}
    The objective is to now remove the connections between $\mathcal{X}_1$ and $\mathcal{Z}_1$ nodes in last two lines of \cref{eq:g-tilde-x-3}.
    For this, we perform the following change-of-basis operation for all $i \in \widetilde{\mathcal{H}}' \cup \widetilde{\overline{\mathcal{H}}}$:
    \begin{align}
        \tilde{g}_i^X \longleftarrow \tilde{g}_i^X + \sum_{j \in \supp{C \left( \widetilde{A}_i \odot \tilde{h}' \right)}}{\tilde{g}_j^X},
    \end{align}
    which cancels all the $\mathcal{Z}_1$ nodes and gives us the following new gauge operators $\tilde{g}'^X_i$:
    \begin{align}
        \tilde{g}'^X_i &=
        \left\{
            \begin {aligned}
                 & \mathcal{X}_2(\tilde{e}_i) + \mathcal{Z}_1(C\tilde{e}_i) + \mathcal{Z}_2\left(\widetilde{B}_i \right) \quad & i \in \widetilde{\mathcal{H}}, \\
                 & \mathcal{X}_1(\tilde{e}_i) + \mathcal{X}_2\left(C \left( \widetilde{A}_i \odot \tilde{h}' \right) \right) + \mathcal{Z}_2\left(C\tilde{e}_i + \widetilde{A}_i \odot \overline{h} + \widetilde{B} \left(C \left( \widetilde{A}_i \odot \tilde{h}' \right) \right) \right) & i \in \widetilde{\mathcal{H}}', \\
                 & \mathcal{X}_1(\tilde{e}_i) + \mathcal{X}_2(\tilde{e}_i) + \mathcal{X}_2\left(C \left( \widetilde{A}_i \odot \tilde{h}' \right) \right) + \mathcal{Z}_2\left(\widetilde{A}_i \odot \overline{h} + \widetilde{B}_i + \widetilde{B} \left(C \left( \widetilde{A}_i \odot \tilde{h}' \right) \right) \right) & i \in \widetilde{\overline{\mathcal{H}}}.
            \end{aligned}
        \right.
    \end{align}
    We notice that $\mathcal{X}_1(\tilde{e}_i)$ and $\mathcal{Z}_1(\tilde{e}_i)$ are only present in a gauge operator for $i \in \widetilde{\mathcal{H}}' \cup \widetilde{\overline{\mathcal{H}}}$, while $\mathcal{X}_2(\tilde{e}_i)$ and $\mathcal{Z}_2(\tilde{e}_i)$ are only present for $i \in \widetilde{\mathcal{H}}' \cup \widetilde{\overline{\mathcal{H}}}$. We can therefore relabel the error nodes as follows:
    \begin{align}
        \mathcal X_1(\tilde e_i) &\mapsto X_1(e_{i-a}) \quad i \in \widetilde H' \cup \widetilde{\overline{\mathcal H}}, \\
        \mathcal Z_1(\tilde e_i) &\mapsto Z_1(e_{i-a}) \quad i \in \widetilde H' \cup \widetilde{\overline{\mathcal H}}, \\
        \mathcal X_2(\tilde e_i) &\mapsto
        \left\{
            \begin {aligned}
                & \mathcal X_2(e_i) \quad & i \in \widetilde{\mathcal H}, \\
                & \mathcal X_2(e_{i-a}) \quad & i \in \widetilde{\overline{\mathcal H}},
            \end{aligned}
        \right. \\
        \mathcal Z_2(\tilde e_i) &\mapsto
        \left\{
            \begin {aligned}
                & \mathcal Z_2(e_i) \quad & i \in \widetilde{\mathcal H}, \\
                & \mathcal Z_2(e_{i-a}) \quad & i \in \widetilde{\overline{\mathcal H}}.
            \end{aligned}
        \right.
    \end{align}
    This map is an isomorphism in the middle space of the chain complex and therefore trivially induces a fault-tolerant map. Under this transformation, we observe that $\tilde{g}'^X_i$ maps to $g_i^Z$ for $i \in \{1,\ldots,a\}$, and to $g_{i-a}^X$ for $i \in \{a,\ldots,n+a\}$. By introducing new error nodes $\mathcal{Z}_1(e_i)$ for $i \in \{a,\ldots,n+a\}$, and weight-2 gauge operators that map them to $\mathcal{Z}_2(e_i)$, we recover all the gauge operators of the original chain complex.
\end{proof}

\subsection{Compilation of the circuit into \texttt{H}, \texttt{HS} and \texttt{CZ} networks}
\label{sec:pre-compilation}

In this first step of the mapping, we show that all the gates of the circuit can be compiled into \texttt{H}, \texttt{HS} and \texttt{CZ} networks, while remaining in the same equivalence class of chain complexes. We begin by compiling all the single-qubit gates.

\begin{proposition}
    All the single-qubit gates of the circuit can be compiled into $\texttt{H}$ and $\texttt{HS}$ without changing the equivalence class of the circuit
\end{proposition}
\begin{proof}
    We first note that any Pauli gates can be removed from the circuit without changing the corresponding chain complex.
    Indeed, since the circuit is Clifford, Pauli operators commute through it, resulting in either sign changes in the measurements or output Pauli operators.
    In both cases, the detectors and gauge operators only change up to signs, which are modded out in our construction.

    As a consequence, any $\texttt{S}^\dag$ gate can be turned into an $\texttt{S}$ gate, since $S^\dag=SZ$.
    We therefore just need to show that any $\texttt{S}$ gate can be compiled into $\texttt{HS}$ gates.
    For every $\texttt{S}$ gate of the circuit we distinguish between two cases. In the first case, it is immediately followed by a Hadamard. In this case, we merge $\texttt{S}$ and $\texttt{H}$ into an $\texttt{HS}$ gate using rule A:
    \begin{center}
        \begin{tabular}{>{\centering\arraybackslash}m{3cm} >{\centering\arraybackslash}m{6cm}}
            \tikzsetnextfilename{sec5-s-and-h}
            \begin{quantikz}
                & \gate{S} & \gate{H} &
            \end{quantikz}
             &
             \includegraphics[width=\linewidth]{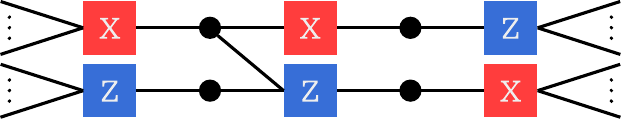}
            \\
            \hline
            \tikzsetnextfilename{sec5-hs}
            \begin{quantikz}
                & \gate{HS} &
            \end{quantikz}
            &
            \addstackgap[15pt]{
                \includegraphics[width=0.7\linewidth]{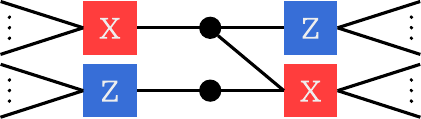}
            }
            \\
        \end{tabular}
    \end{center}
    In the second case, we use \cref{lemma:hadamard-square} to append two Hadamard gates after the \texttt{H} and use the first case to merge one of the Hadamard gates with the \texttt{S}, as illustrated by the following chain of equivalences:
    \begin{align}
        \tikzsetnextfilename{sec5-s-gate}
        \begin{quantikz}
            & \gate{S} &
        \end{quantikz}
        \sim
        \tikzsetnextfilename{sec5-s-with-double-h}
        \begin{quantikz}
            & \gate{S} & \gate{H} & \gate{H} &
        \end{quantikz}
        \sim
        \tikzsetnextfilename{sec5-hs-with-h}
        \begin{quantikz}
            & \gate{HS} & \gate{H} &
        \end{quantikz}
    \end{align}
\end{proof}

The next step is to show that multi-qubit gates can be compiled into \texttt{CZ} networks and Hadamard gates.

\begin{proposition}
    Any controlled Pauli gate can be compiled into a \texttt{CZ} network followed and preceded by some Hadamard gates.
\end{proposition}
\begin{proof}
We first note that any controlled Pauli gate can be written as
\begin{align}
    \tikzsetnextfilename{sec5-ctrl-pauli-decomposition-1}
    \begin{quantikz}[row sep=9pt]
        &  & \gate[9, disable auto height]{
            \begin{array}{ccccccccc}
                Y\\ \vdots \\ Y \\ X \\ \vdots \\ X \\ Z \\ \vdots \\ Z
            \end{array}
        } &  &  \\
        & \setwiretype{n} \vdots     &           & \vdots     &  \\
        &                            &           &            &  \\
        &                            &           &            &  \\
        & \setwiretype{n} \vdots     &           & \vdots     & \\
        &                            &           &            &  \\
        &                            &           &            &  \\
        & \setwiretype{n} \vdots     &           & \vdots     & \\
        &                            &           &            &  \\
        &                            & \ctrl{-1} &            &
    \end{quantikz}
    \sim
    \tikzsetnextfilename{sec5-ctrl-pauli-decomposition-2}
    \begin{quantikz}
        & \gatebox{7}{9} \qw   & \ \ldots \ &           &           & \ \ldots \ &           &                             & \ \ldots \ & \mygate{1}{Y}  &  \\
        &                      & \ \ldots \ &           &           & \ \ldots \ &           & \mygate{1}{Y} & \ \ldots \  &            &                &  \\
        &                      & \ \ldots \ &           &           & \ \ldots \ & \targ{}   &                             & \ \ldots \ &                &  \\
        &                      & \ \ldots \ &           & \targ{}   & \ \ldots \ &           &                             & \ \ldots \ &                &  \\
        &                      & \ \ldots \ & \ctrl{2}  &           & \ \ldots \ &           &                             & \ \ldots \ &                &  \\
        & \ctrl{1}             & \ \ldots \ &           &           & \ \ldots \ &           &                             & \ \ldots \ &                &  \\
        & \ctrl{-1}            & \ \ldots \ & \ctrl{-2} & \ctrl{-3} & \ \ldots \ & \ctrl{-4} & \ctrl{-5}                   & \ \ldots \ & \ctrl{-6}      &
    \end{quantikz}
\end{align}
and both \texttt{CNOT} and \texttt{CY} can be compiled using \texttt{CZ} and \texttt{H} as:
\begin{align}
    \tikzsetnextfilename{sec5-cnot-compilation-1}
    \begin{quantikz}
        & \targ{}   & \\
        & \ctrl{-1} &
    \end{quantikz}
    &=
    \tikzsetnextfilename{sec5-cnot-compilation-2}
    \begin{quantikz}
        & \gate{H} & \ctrl{1}  & \gate{H} & \\
        &          & \ctrl{-1} &          &
    \end{quantikz}
    \\
    \tikzsetnextfilename{sec5-cy-compilation-1}
    \begin{quantikz}
        & \gate{Y}  &  \\
        & \ctrl{-1} &
    \end{quantikz}
    &=
    \tikzsetnextfilename{sec5-cy-compilation-2}
    \begin{quantikz}
        & \targ{}   & \ctrl{1}  & \\
        & \ctrl{-1} & \ctrl{-1} &
    \end{quantikz}
    =
    \tikzsetnextfilename{sec5-cy-compilation-3}
    \begin{quantikz}
        & \gate{H} & \ctrl{1}  & \gate{H} & \ctrl{1}  & \\
        &          & \ctrl{-1} &          & \ctrl{-1} &
    \end{quantikz}
\end{align}
Using the compilations of \texttt{CX} and \texttt{CY} into Hadamard gates and \texttt{CZ}s, as well as  \cref{lemma:pushing-h-away} to push the \texttt{H} gates away from the blue box, we compile the circuit above to the following equivalent circuit:
\begin{center}
    \tikzsetnextfilename{sec5-pushing-h-away-of-ctrl-pauli}
    \begin{quantikz}
        & \gate{H} & \gatebox{7}{13} \qw & \ \ldots \ &           &          & \ \ldots \ &           &                & \ \ldots \ & \ctrl{6}  & \mygate{1}{H} &           & \ \ldots \ & \ctrl{6}   &          & \\
        & \gate{H} &                     & \ \ldots \ &           &           & \ \ldots \ &           & \ctrl{5}       & \ \ldots \ &           & \mygate{1}{H} & \ctrl{5}  & \ \ldots \ &            &          &  \\
        & \gate{H} &                     & \ \ldots \ &           &           & \ \ldots \ & \ctrl{4}  &                & \ \ldots \ &           &               &           & \ \ldots \ &            & \gate{H} &  \\
        & \gate{H} &                     & \ \ldots \ &           & \ctrl{3}  & \ \ldots \ &           &                & \ \ldots \ &           &               &           & \ \ldots \ &            & \gate{H} &  \\
        &          &                     & \ \ldots \ & \ctrl{2}  &           & \ \ldots \ &           &                & \ \ldots \ &           &               &           & \ \ldots \ &            &          &  \\
        &          & \ctrl{1}            & \ \ldots \ &           &           & \ \ldots \ &           &                & \ \ldots \ &           &               &           & \ \ldots \ &            &          &  \\
        &          & \ctrl{-1}           & \ \ldots \ & \ctrl{-2} & \ctrl{-3} & \ \ldots \ & \ctrl{-4} & \ctrl{-5}      & \ \ldots \ & \ctrl{-6} &               & \ctrl{-5} & \ \ldots \ & \ctrl{-6}  &          &
    \end{quantikz}
\end{center}
To get a single \texttt{CZ} network in the middle, we compile the \texttt{H}s into their measurement-based versions, using \cref{lemma:h-cz-sandwich}.
\end{proof}

\subsection{Measurement-based \texttt{H} and \texttt{HS} gates}
\label{sec:mbqc-compilation-h-hs}

The next step is to show that \texttt{H} and $\texttt{HS}$ gates are equivalent to their measurement-based version.

\begin{proposition}
    The two following circuits, corresponding to a Hadamard gate and its measurement-based version, are equivalent:
    \begin{center}
        \tikzsetnextfilename{sec5-h-gate-mbqc-1}
        \begin{quantikz}
            & \gate{H} &
        \end{quantikz}
        $\sim$
        \tikzsetnextfilename{sec5-h-gate-mbqc-2}
        \begin{quantikz}
                               & \ctrl{1}  \gatebox{2}{1}  & \meterD{X} \\
            \lstick{$\ket{+}$} & \ctrl{-1}                 &
        \end{quantikz}
    \end{center}
\end{proposition}

\begin{proof}
    While this theorem could be seen as a special case of \cref{lemma:h-cz-sandwich} where the two \texttt{CZ} networks are empty, we choose to do a more direct proof for clarity here. Here are the circuits and corresponding chain complexes whose equivalence we want to show:
    \begin{center}
        \begin{tabular}{>{\centering\arraybackslash}m{3cm} >{\centering\arraybackslash}m{4cm}}
            \tikzsetnextfilename{sec5-h-gate-mbqc-proof-1}
            \begin{quantikz}
                & \gate{H} &
             \end{quantikz}
             &
             \includegraphics[width=\linewidth]{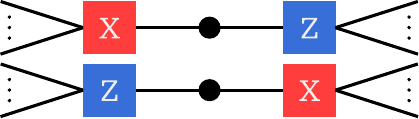}
            \\
            \hline
            \tikzsetnextfilename{sec5-h-gate-mbqc-proof-2}
            \begin{quantikz}
                & \ctrl{1} & \meterD{X} \\
                \lstick{$\ket{+}$} & \ctrl{-1} &
            \end{quantikz}
            &
            \addstackgap[15pt]{
                \includegraphics[width=\linewidth]{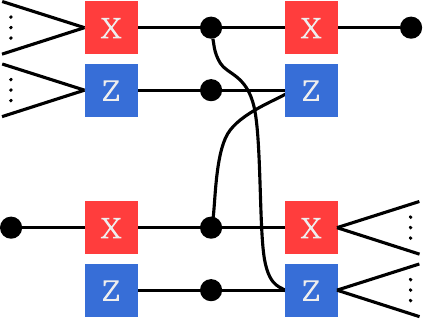}
            }
            \\
        \end{tabular}
    \end{center}
    Starting from the chain complex of the measurement-based Hadamard, we apply rule A on the two weight-2 gauge operators linking the $Z$ nodes, and rule B on the two weight-1 gauge operators. This gives us the following new chain complex:
    \begin{center}
        \includegraphics[width=0.3\linewidth]{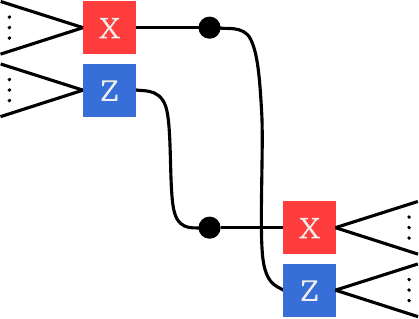}
    \end{center}
    which is the same as the chain complex of the original Hadamard.
\end{proof}

We now move on to the $\texttt{HS}$ gate.

\begin{proposition}
    The two following circuits, corresponding to an \texttt{HS} gate and its measurement-based version, are equivalent:
    \begin{align} \label{eq:hs-mbqc-compilation}
        \tikzsetnextfilename{sec5-hs-gate-mbqc-1}
        \begin{quantikz}
            & \gate{HS} &
        \end{quantikz}
        \sim
        \tikzsetnextfilename{sec5-hs-gate-mbqc-2}
        \begin{quantikz}
                               & \ctrl{1}  & \meterD{Y} \\
            \lstick{$\ket{+}$} & \ctrl{-1} &
        \end{quantikz}
    \end{align}
\end{proposition}
\begin{proof}
    We can compare the chain complexes for the direct and MBQC implementations
    \begin{center}
        \begin{tabular}{>{\centering\arraybackslash}m{3cm} >{\centering\arraybackslash}m{4cm}}
            \tikzsetnextfilename{sec5-hs-gate-mbqc-proof-1}
            \begin{quantikz}
                & \gate{HS} &
            \end{quantikz}
            &
            \includegraphics[width=\linewidth]{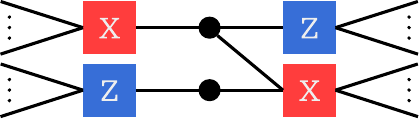}
            \\
            \hline
            \tikzsetnextfilename{sec5-hs-gate-mbqc-proof-2}
            \begin{quantikz}
                & \ctrl{1} & \meterD{Y} \\
                \lstick{$\ket{+}$} & \ctrl{-1} &
            \end{quantikz}
            &
            \addstackgap[15pt]{
                \includegraphics[width=\linewidth]{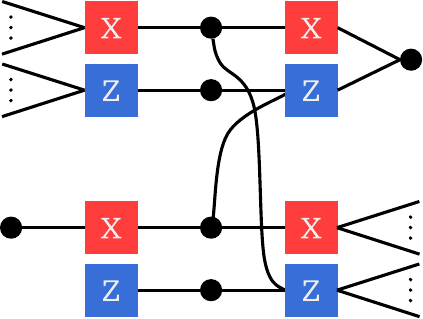}
            }
            \\
        \end{tabular}
    \end{center}
    and use the following sequence of fault-tolerant maps to prove their equivalence:
    \begin{center}
        \begin{tabular}{>{\centering\arraybackslash}m{4cm} >{\centering\arraybackslash}m{0.7cm} >{\centering\arraybackslash}m{4cm} >{\centering\arraybackslash}m{0.7cm} >{\centering\arraybackslash}m{4cm}}
            \includegraphics[width=\linewidth]{figures/chain-complex-mbqc-s-1.pdf}
            &
            $$\xrightarrow{\;\;\;\;}$$
            &
            \includegraphics[width=\linewidth]{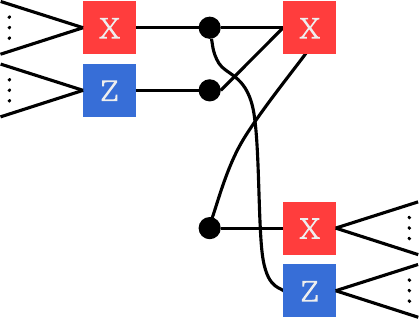}
            &
            $$\xrightarrow{\;\;\;\;}$$
            &
            \includegraphics[width=\linewidth]{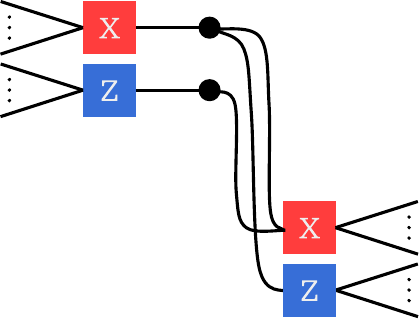}
            \\
        \end{tabular}
    \end{center}
\end{proof}

\subsection{Merge of \texttt{CZ} networks}
\label{sec:merge-cz-networks}

Our circuit is now composed of \texttt{CZ} networks, \texttt{H} gates and \texttt{HS} gates. The objective of this section is to show that we can turn the \texttt{H} and \texttt{HS} gates into their measurement-based versions, yielding a circuit only composed of \texttt{CZ} networks, and that all those \texttt{CZ} networks can be merged.

We begin by noting that, in general, \texttt{CZ} networks cannot be merged. For instance, the following two circuits are not equivalent:
\begin{center}
    \tikzsetnextfilename{sec5-double-cz-no-blue-box}
    \begin{quantikz}
        & \ctrl{1}  & \ctrl{1}  & \\
        & \ctrl{-1} & \ctrl{-1} &
    \end{quantikz}
    $\nsim$
    \tikzsetnextfilename{sec5-double-cz-in-blue-box}
    \begin{quantikz}
        & \ctrl{1} \gatebox{2}{2} & \ctrl{1}  & \\
        & \ctrl{-1}               & \ctrl{-1} &
    \end{quantikz}
\end{center}
since the circuit on the left propagates single-qubit $X$ errors between the two \texttt{CZ} into weight-2 errors, while the circuit on the right is equivalent to the identity gate and does not propagate any weight-1 error into a weight-2 error. However, as we will see, we can merge \texttt{CZ} networks after turning all the single-qubit gates into their measurement-based versions, as long as all the original \texttt{CZ} networks are Hadamard-separated.

\begin{definition}[Hadamard separation]
    We say that a pair of \texttt{CZ} networks is Hadamard-separated if, for each qubit of their intersection, the first \texttt{CZ} network is immediately followed by an \texttt{H} gate on that qubit, before any other potential gates. We say that a circuit, compiled into \texttt{CZ} networks, \texttt{H} and \texttt{HS} gates, is Hadamard-separated if every pair of \texttt{CZ} networks is Hadamard-separated.
\end{definition}

We start by turning our circuit into an equivalent Hadamard-separated circuit.
For this, we identify all the pairs of \texttt{CZ} networks that are not separated by any gate on some of their common wires, and insert a pair of Hadamard gates on those wires in between the two networks. We also identify all the \texttt{HS} gates immediately following a \texttt{CZ} network, and insert a pair of Hadamard gates before each such \texttt{HS} gate.
By \cref{lemma:hadamard-square}, the new circuit is equivalent to the original one. As an example, our method turns two successive \texttt{CZ} gates into a single \texttt{CZ} network the following way:
\begin{center}
    \tikzsetnextfilename{sec5-double-cz}
    \begin{quantikz}
        & \ctrl{1}  & \ctrl{1}  & \\
        & \ctrl{-1} & \ctrl{-1} &
    \end{quantikz}
    $\sim$
    \tikzsetnextfilename{sec5-double-cz-double-h}
    \begin{quantikz}
        & \ctrl{1}  & \gate{H} & \gate{H} & \ctrl{1}  & \\
        & \ctrl{-1} & \gate{H} & \gate{H} & \ctrl{-1} &
    \end{quantikz}
    $\sim$
    \tikzsetnextfilename{sec5-double-cz-double-h-mbqc-compilation}
    \begin{quantikz}
        \lstick{$\ket{+}$} & \gatebox{6}{4} \qw  &           & \ctrl{1}  & \ctrl{5}  & \\
        \lstick{$\ket{+}$} &                     & \ctrl{1}  & \ctrl{-1} &           & \meterD{X} \\
                           & \ctrl{1}            & \ctrl{-1} &           &           & \meterD{X} \\
                           & \ctrl{-1}           & \ctrl{1}  &           &           & \meterD{X} \\
        \lstick{$\ket{+}$} &                     & \ctrl{-1} & \ctrl{1}  &           & \meterD{X} \\
        \lstick{$\ket{+}$} &                     &           & \ctrl{-1} & \ctrl{-5} &
    \end{quantikz}
\end{center}

Our proof that we can merge Hadamard-separated \texttt{CZ} networks splits into the following steps:
\begin{enumerate}
    \item Showing that a \texttt{CZ} network can be merged with all the measurement-based \texttt{H} gates on its right.
    \item Showing that any \texttt{HS} gate can be merged with a network resulting from the merge of a network with an \texttt{H} gate on the left of the \texttt{HS} gate.
    \item Showing that the resulting \texttt{CZ} network can be merged with all the \texttt{CZ} networks that immediately follow it.
    \item Showing that a \texttt{CZ} network can be merged with all the measurement-based \texttt{H} and \texttt{HS} gates on its left. This last part is useful in the case of a circuit starting with some \texttt{H} and \texttt{HS} gates, without any \texttt{CZ} network on their left to be merged with.
\end{enumerate}

\begin{lemma} \label{lemma:h-and-hs-merge-right}
    Any \texttt{H} gate can be merged into a \texttt{CZ} network on its right by compiling it into its measurement-based version. Moreover, any \texttt{HS} gate with an \texttt{H} gate and a \texttt{CZ} network on its right can be merged into a single \texttt{CZ} network.
    \begin{align}
        \label{eq:h-merge-right}
        \tikzsetnextfilename{sec5-h-merge-right-1}
        \begin{quantikz}[row sep=6pt]
            &                        & \mygate{4}{CZ} & \mygate{1}{H} &  \\
            &                        &                &               & \\
            & \setwiretype{n} \vdots &                & \vdots        & \\
            &                        &                &               &
        \end{quantikz}
        & \sim
        \tikzsetnextfilename{sec5-h-merge-right-2}
        \begin{quantikz}[row sep=6pt]
            \lstick{$\ket{+}$} & \gatebox{5}{3} \qw     &                & \ctrl{1}  &            \\
                               &                        & \mygate{4}{CZ} & \ctrl{-1} & \meterD{X} \\
                               &                        &                &           &            \\
                               & \setwiretype{n} \vdots &                & \vdots    &            \\
                               &                        &                &           &
        \end{quantikz}
        \\
        \label{eq:hs-merge-right}
        \tikzsetnextfilename{sec5-hs-merge-right-1}
        \begin{quantikz}[row sep=6pt]
            &                        & \mygate{4}{CZ} & \gate{H} & \gate{HS} & \\
            &                        &                &          &           & \\
            & \setwiretype{n} \vdots &                & \vdots   &           & \\
            &                        &                &          &           &
        \end{quantikz}
        & \sim
        \tikzsetnextfilename{sec5-hs-merge-right-2}
        \begin{quantikz}[row sep=6pt]
            \lstick{$\ket{+}$} & \gatebox{6}{4} \qw     &                &           & \ctrl{1}  &            \\
            \lstick{$\ket{+}$} &                        &                & \ctrl{1}  & \ctrl{-1} & \meterD{Y} \\
                               &                        & \mygate{4}{CZ} & \ctrl{-1} &           & \meterD{X} \\
                               &                        &                &           &           &            \\
                               & \setwiretype{n} \vdots &                & \vdots    &           &            \\
                               &                        &                &           &           &
        \end{quantikz}
    \end{align}
\end{lemma}
\begin{proof}
    \cref{eq:h-merge-right} is a direct consequence of \cref{lemma:pushing-h-away,lemma:h-cz-sandwich}:
    \tikzsetnextfilename{sec5-h-merge-right-proof-1}
    \begin{align}
        \tikzsetnextfilename{sec5-h-merge-right-proof-1}
        \begin{quantikz}[row sep=6pt]
            &                        & \mygate{4}{CZ} & \mygate{1}{H} &  \\
            &                        &                &               & \\
            & \setwiretype{n} \vdots &                & \vdots        & \\
            &                        &                &               &
        \end{quantikz}
        \sim
        \tikzsetnextfilename{sec5-h-merge-right-proof-2}
        \begin{quantikz}[row sep=6pt]
            & \gatebox{4}{3} \qw     & \mygate{4}{CZ} & \mygate{1}{H} &  \\
            &                        &                &               & \\
            & \setwiretype{n} \vdots &                & \vdots        & \\
            &                        &                &               &
        \end{quantikz}
        \sim
        \tikzsetnextfilename{sec5-h-merge-right-proof-3}
        \begin{quantikz}[row sep=6pt]
            \lstick{$\ket{+}$} & \gatebox{5}{3} \qw     &                & \ctrl{1}  &            \\
                               &                        & \mygate{4}{CZ} & \ctrl{-1} & \meterD{X} \\
                               &                        &                &           &            \\
                               & \setwiretype{n} \vdots &                & \vdots    &            \\
                               &                        &                &           &
        \end{quantikz}
    \end{align}
    Using this equivalence, along with \cref{eq:hs-mbqc-compilation} to compile the \texttt{HS} gate, showing \cref{eq:hs-merge-right} comes down to showing
    \begin{align} \label{eq:hs-merge-right-mbqc}
        \tikzsetnextfilename{sec5-hs-merge-right-proof-1}
        \begin{quantikz}[row sep=6pt]
            \lstick{$\ket{+}$} & \gatebox{6}{3} \qw     &                &           & \ctrl{1}  &            \\
            \lstick{$\ket{+}$} &                        &                & \ctrl{1}  & \ctrl{-1} & \meterD{Y} \\
                               &                        & \mygate{4}{CZ} & \ctrl{-1} &           & \meterD{X} \\
                               &                        &                &           &           &            \\
                               & \setwiretype{n} \vdots &                & \vdots    &           &            \\
                               &                        &                &           &           &
        \end{quantikz}
        \sim
        \tikzsetnextfilename{sec5-hs-merge-right-proof-2}
        \begin{quantikz}[row sep=6pt]
            \lstick{$\ket{+}$} & \gatebox{6}{4} \qw     &                &           & \ctrl{1}  &            \\
            \lstick{$\ket{+}$} &                        &                & \ctrl{1}  & \ctrl{-1} & \meterD{Y} \\
                               &                        & \mygate{4}{CZ} & \ctrl{-1} &           & \meterD{X} \\
                               &                        &                &           &           &            \\
                               & \setwiretype{n} \vdots &                & \vdots    &           &            \\
                               &                        &                &           &           &
        \end{quantikz}
    \end{align}
    The left circuit is associated to the following chain complex:
    \begin{center}
        \includegraphics[width=0.45\linewidth]{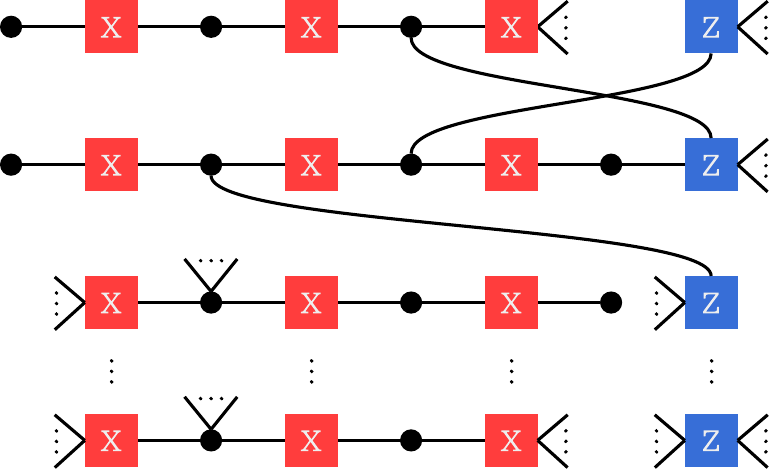}
    \end{center}
    where for every row, we have contracted the three $Z$ nodes of the row using rule A, and have placed the resulting node on the right of all the $X$ nodes. Moreover, we have put ellipses on the left-most gauge operator of each row of the \texttt{CZ} network, indicating that those gauge operators are connected to the $Z$ nodes at rows depending on the adjacency matrix of the \texttt{CZ} network.

    A first application of rules A and B gives the following equivalent diagram:
    \begin{center}
        \includegraphics[width=0.45\linewidth]{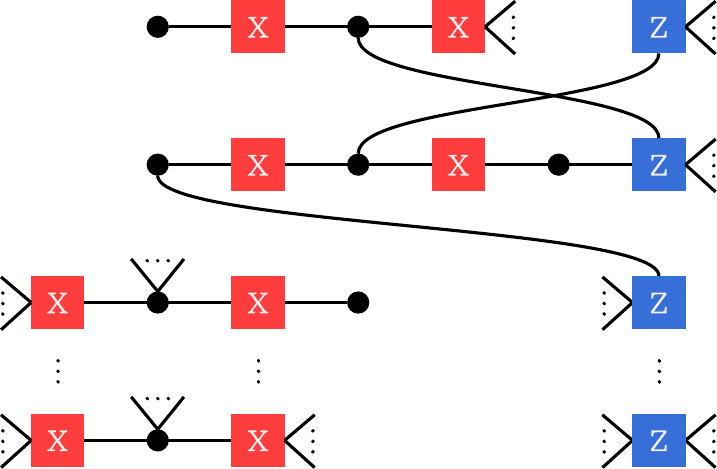}
    \end{center}
    Using rule A on the first gauge operator on the second row then gives:
    \begin{center}
        \includegraphics[width=0.35\linewidth]{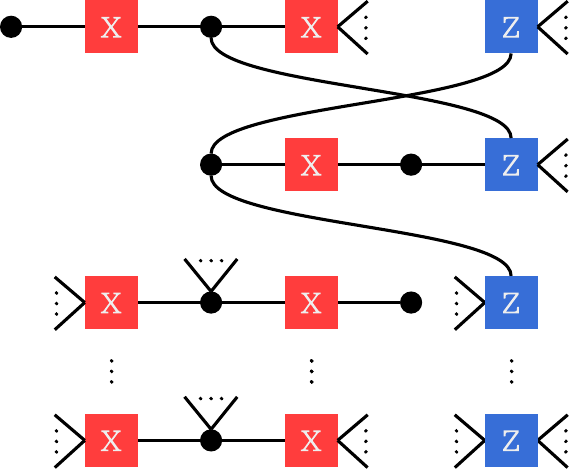}
    \end{center}
    Finally, using rule B to add a new $X$ node, we recover exactly the chain complex of the circuit on the right of \cref{eq:hs-merge-right-mbqc}:
    \begin{center}
        \includegraphics[width=0.35\linewidth]{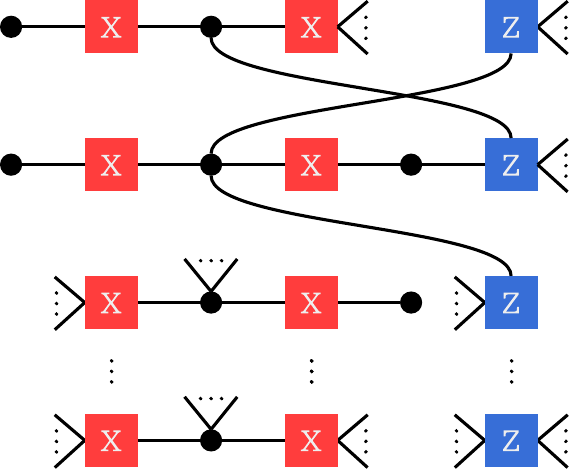}
    \end{center}

\end{proof}

\begin{lemma} \label{lemma:cz-cz-merge}
    We have the equivalence:
    \begin{align} \label{eq:cz-cz-merge}
        \tikzsetnextfilename{sec5-cz-cz-merge-1}
        \begin{quantikz}[row sep=8pt, column sep=10pt]
                               & \gatebox{12}{5} \qw     &                  &           &           &           & \mygate{6}{CZ_B} & \\
                               & \setwiretype{n} \vdots  &                  &           &           &           &                  & \\
                               &                         &                  &           &           &           &                  & \\
            \lstick{$\ket{+}$} &                         &                  & \ctrl{3}  &           &           &                  & \\
                               & \setwiretype{n} \vdots  &                  &           &           &           &                  & \\
            \lstick{$\ket{+}$} &                         &                  &           & \ctrl{3}  &           &                  & \\
                               &                         & \mygate{6}{CZ_A} & \ctrl{-3} &           &           &                  & \\
                               & \setwiretype{n} \vdots  &                  &           &           &           &                  & \\
                               &                         &                  &           & \ctrl{-3} &           &                  & \\
                               &                         &                  &           &           &           &                  & \\
                               & \setwiretype{n} \vdots  &                  &           &           &           &                  & \\
                               &                         &                  &           &           &           &                  &
        \end{quantikz}
        \sim
        \tikzsetnextfilename{sec5-cz-cz-merge-2}
        \begin{quantikz}[row sep=8pt, column sep=10pt]
                               & \gatebox{12}{6} \qw     &                  &           &           &           & \mygate{6}{CZ_B} & \\
                               & \setwiretype{n} \vdots  &                  &           &           &           &                  & \\
                               &                         &                  &           &           &           &                  & \\
            \lstick{$\ket{+}$} &                         &                  & \ctrl{3}  &           &           &                  & \\
                               & \setwiretype{n} \vdots  &                  &           &           &           &                  & \\
            \lstick{$\ket{+}$} &                         &                  &           & \ctrl{3}  &           &                  & \\
                               &                         & \mygate{6}{CZ_A} & \ctrl{-3} &           &           &                  & \\
                               & \setwiretype{n} \vdots  &                  &           &           &           &                  & \\
                               &                         &                  &           & \ctrl{-3} &           &                  & \\
                               &                         &                  &           &           &           &                  & \\
                               & \setwiretype{n} \vdots  &                  &           &           &           &                  & \\
                               &                         &                  &           &           &           &                  &
        \end{quantikz}
    \end{align}
\end{lemma}
\begin{proof}
    The left-hand side of \cref{eq:cz-cz-merge} is associated to the following chain complex:
    \begin{center}
        \includegraphics[width=0.5\textwidth]{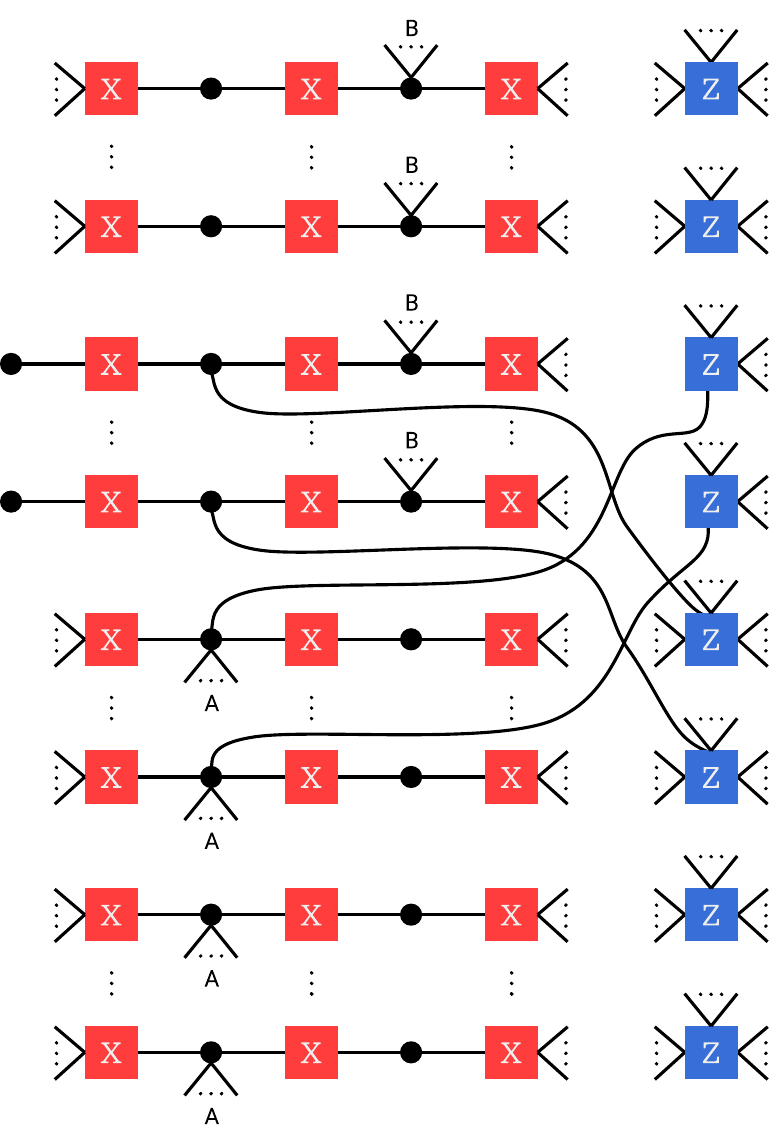}
    \end{center}
    where we have labeled the ellipses A and B depending on whether a gauge operator corresponds to $\texttt{CZ}_A$ or $\texttt{CZ}_B$. A first application of rules A and B gives the following equivalent diagram:
    \begin{center}
        \includegraphics[width=0.5\textwidth]{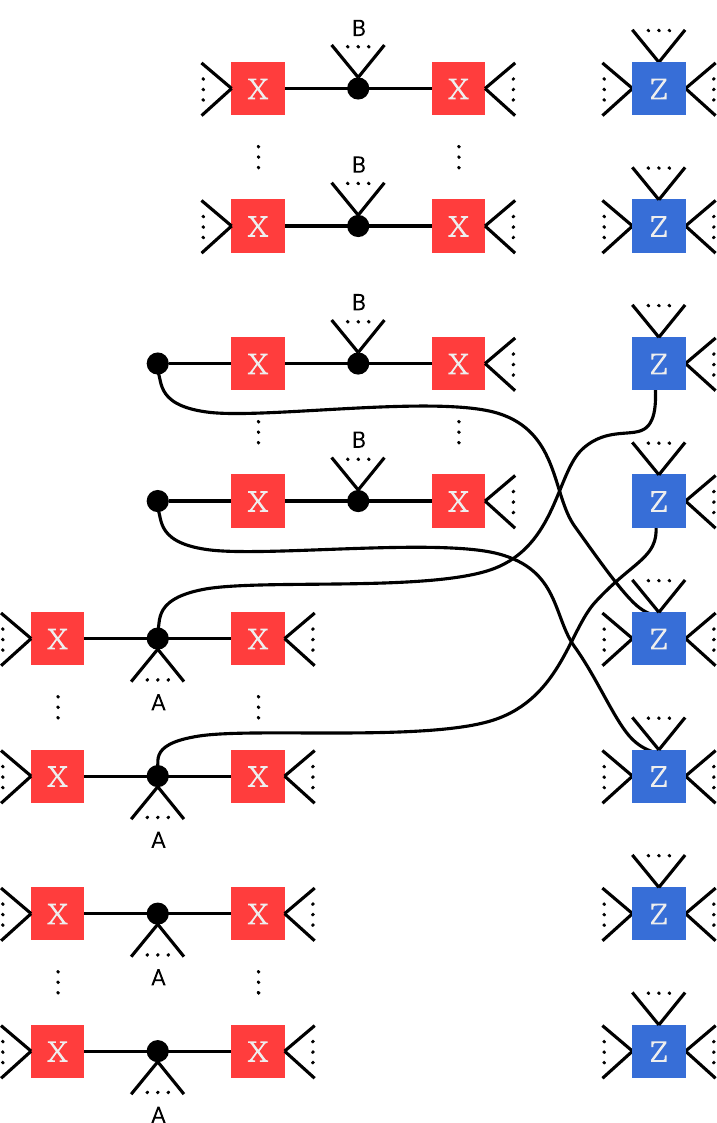}
    \end{center}
    We then apply rule B on the third and fourth rows to eliminate the weight-1 gauge operators, rule A to merge the new weight-2 gauge operators of those rows, and finally rule B to add a new X node:
    \begin{center}
        \begin{tabular}{
            >{\centering\arraybackslash}m{6cm}
            >{\centering\arraybackslash}m{0.7cm}
            >{\centering\arraybackslash}m{6cm}}
            \includegraphics[width=\linewidth]{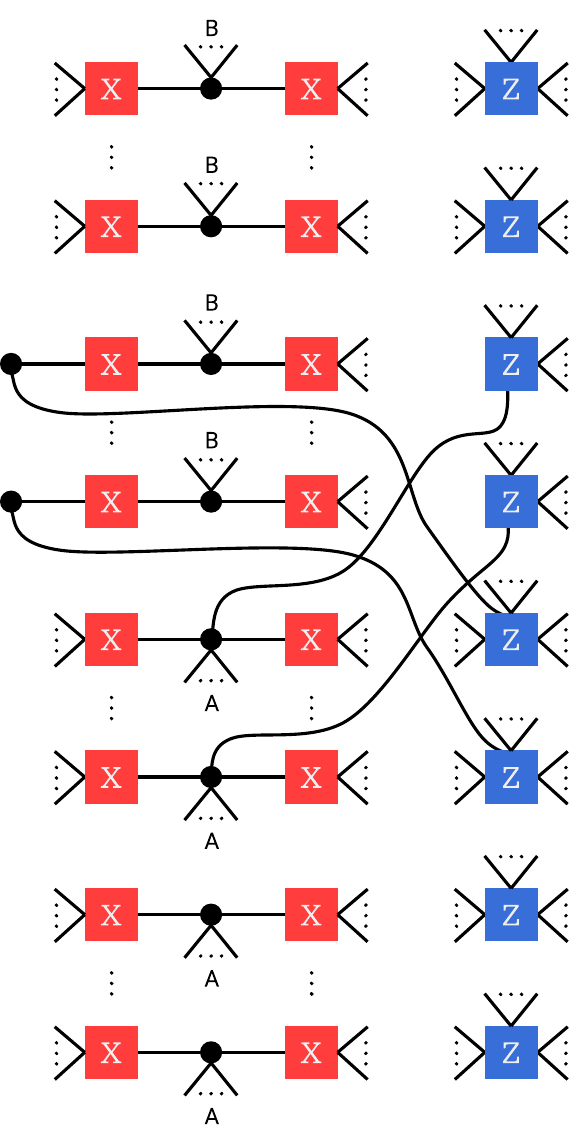}
            &
            $$\xrightarrow{\;\;\;\;}$$
            &
            \includegraphics[width=\linewidth]{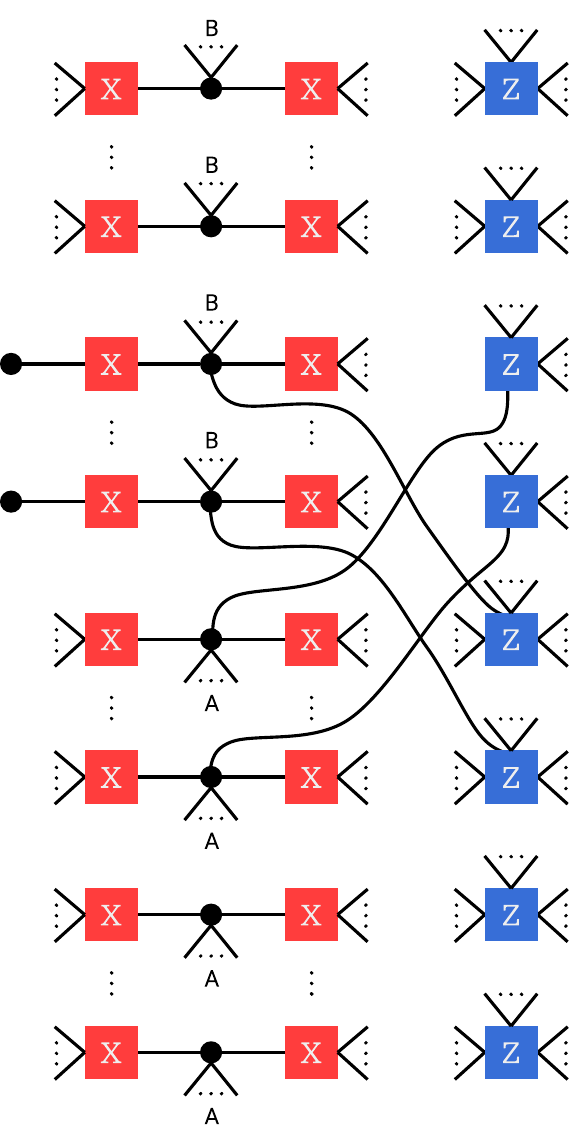}
        \end{tabular}
    \end{center}
    This last diagram represents the chain complex at the right-hand side of \cref{eq:hs-merge-right}, showing the desired circuit equivalence.
\end{proof}

The final step of the MBQC mapping is to merge any \texttt{H} or \texttt{HS} gate at the beginning of the circuit into a CZ network on its right. This can be done using the following lemma:
\begin{lemma} \label{lemma:h-and-hs-merge-left}
    We have the following circuit equivalences:
    \begin{align}
        \label{eq:h-merge-left}
        \tikzsetnextfilename{sec5-h-merge-left-1}
        \begin{quantikz}[row sep=6pt]
            & \mygate{1}{H}          & \mygate{4}{CZ} &  \\
            &                        &                & \\
            & \setwiretype{n} \vdots &                & \\
            &                        &                &
        \end{quantikz}
        & \sim
        \tikzsetnextfilename{sec5-h-merge-left-2}
        \begin{quantikz}[row sep=6pt]
                               & \gatebox{5}{3} \qw     & \ctrl{1}  &                & \meterD{X} \\
            \lstick{$\ket{+}$} &                        & \ctrl{-1} & \mygate{4}{CZ} &            \\
                               &                        &           &                &            \\
                               & \setwiretype{n}        & \vdots    &                &            \\
                               &                        &           &                &
        \end{quantikz}
        \\
        \label{eq:hs-merge-left}
        \tikzsetnextfilename{sec5-h-merge-left-3}
        \begin{quantikz}[row sep=6pt]
            & \mygate{1}{HS}         & \mygate{4}{CZ} &  \\
            &                        &                & \\
            & \setwiretype{n} \vdots &                & \\
            &                        &                &
        \end{quantikz}
        & \sim
        \tikzsetnextfilename{sec5-h-merge-left-4}
        \begin{quantikz}[row sep=6pt]
                               & \gatebox{5}{3} \qw     & \ctrl{1}  &                & \meterD{Y} \\
            \lstick{$\ket{+}$} &                        & \ctrl{-1} & \mygate{4}{CZ} &            \\
                               &                        &           &                &            \\
                               & \setwiretype{n}        & \vdots    &                &            \\
                               &                        &           &                &
        \end{quantikz}
    \end{align}
\end{lemma}
\begin{proof}
    The equivalence in \cref{eq:h-merge-left} is direct consequence of \cref{lemma:pushing-h-away} and \cref{lemma:h-cz-sandwich}. To prove \cref{eq:hs-merge-left}, we start by compiling the \texttt{HS} gate using \cref{eq:hs-mbqc-compilation}. It remains to show that:
    \begin{align}
        \tikzsetnextfilename{sec5-h-merge-left-proof-1}
        \begin{quantikz}[row sep=6pt]
                               &                        & \ctrl{1}  &                & \meterD{Y} \\
            \lstick{$\ket{+}$} &                        & \ctrl{-1} & \mygate{4}{CZ} &            \\
                               &                        &           &                &            \\
                               & \setwiretype{n}        & \vdots    &                &            \\
                               &                        &           &                &
        \end{quantikz}
        \sim
        \tikzsetnextfilename{sec5-h-merge-left-proof-2}
        \begin{quantikz}[row sep=6pt]
                               & \gatebox{5}{3} \qw     & \ctrl{1}  &                & \meterD{Y} \\
            \lstick{$\ket{+}$} &                        & \ctrl{-1} & \mygate{4}{CZ} &            \\
                               &                        &           &                &            \\
                               & \setwiretype{n}        & \vdots    &                &            \\
                               &                        &           &                &
        \end{quantikz}
    \end{align}
    The left-hand side of this equation is associated to the following chain complex, where we have used the same contraction of $Z$ nodes as in the proof of \cref{lemma:h-and-hs-merge-right}:
    \begin{center}
        \includegraphics[width=0.55\textwidth]{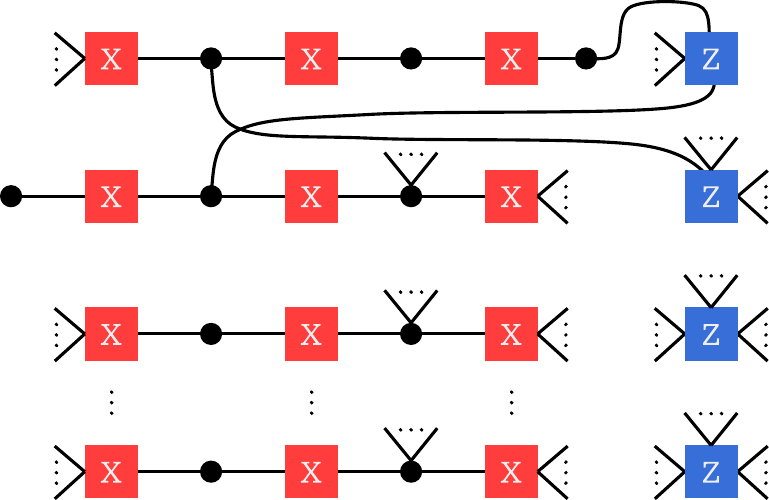}
    \end{center}
    A first application of rules A and B gives the following equivalent diagram:
    \begin{center}
        \includegraphics[width=0.5\textwidth]{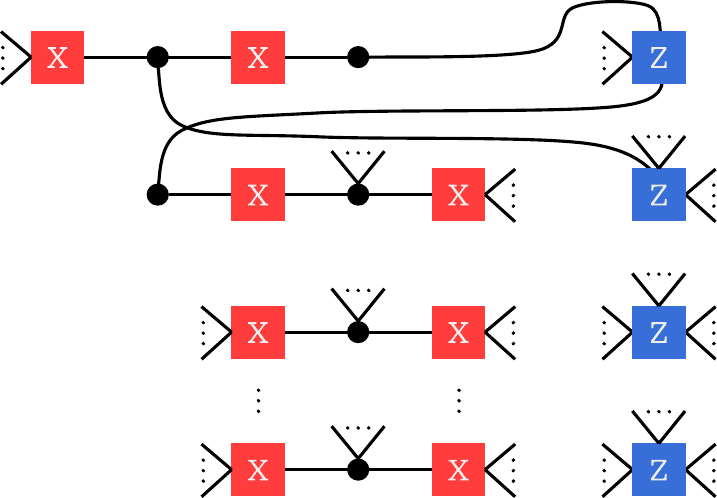}
    \end{center}
    We now use rule A on the second row, followed by rule B to add a new $X$ node:
    \begin{center}
        \begin{tabular}{
            >{\centering\arraybackslash}m{6cm}
            >{\centering\arraybackslash}m{0.7cm}
            >{\centering\arraybackslash}m{6.5cm}}
            \includegraphics[width=\linewidth]{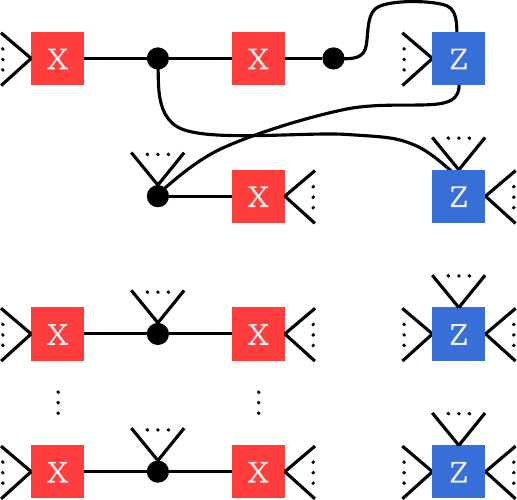}
            &
            $$\xrightarrow{\;\;\;\;}$$
            &
            \includegraphics[width=\linewidth]{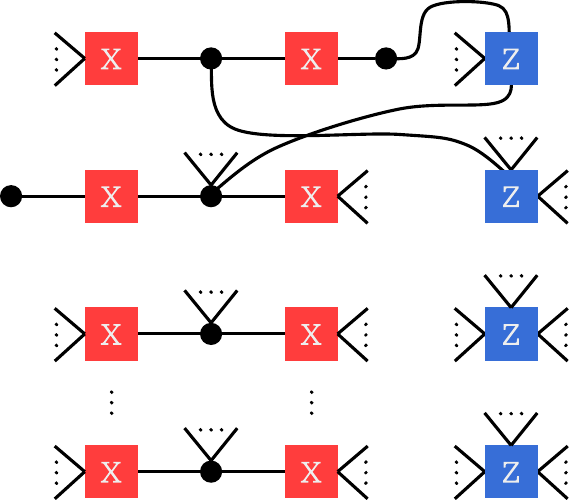}
        \end{tabular}
    \end{center}
    This last diagram corresponds to the chain complex of the circuit at the right-hand side of \cref{eq:hs-merge-left}.
\end{proof}

\subsection{Example: transversal logical gate}

As an application of our circuit transformation technique, we derive here the cluster state complex corresponding to the application of a transversal $S$ gate on the Bell state stabilized by $X^{\otimes 2}$ and $Z^{\otimes 2}$, with measurement of those stabilizers before and after the application of the gate.
This protocol is a simplified version of the transversal $S$ gate applied to the $[[4,2,2]]$ code or more general 2D color codes~\cite{bombin2006topologicalquantum,kubica2015universaltransversal,kubica2015unfolding}.
Since applying $S^{\otimes 2}$ to the Bell state only changes its stabilizers up to a sign, the circuit should have at least two spacetime stabilizers.
Here is the circuit corresponding to the protocol, with the two detectors labeled in blue and red:
\begin{align*}
    \tikzsetnextfilename{sec5-transversal-s-1}
    \begin{quantikz}[row sep=10pt, column sep=8pt]
                           &                           &                           & \gate[2, disable auto height]{ \begin{array}{cc} Z\\Z \end{array} } & \labeledwire{nicered}{Z} & \gate[2, disable auto height]{ \begin{array}{cc} X\\X \end{array} } & \labeledwire{nicered}{Z} & \labeledwire{niceblue}{X} & \gate{S} & \labeledwire{nicered}{Z} & \labeledwire{niceblue}{Y} & \gate[2, disable auto height]{ \begin{array}{cc} X\\X \end{array} } & \labeledwire{nicered}{Z} & \labeledwire{niceblue}{Z} & \gate[2, disable auto height]{ \begin{array}{cc} Z\\Z \end{array} }  &                           &                           &            \\
                           &                           &                           &                                                                     & \labeledwire{nicered}{Z} &                                                                     & \labeledwire{nicered}{Z} & \labeledwire{niceblue}{X} & \gate{S} & \labeledwire{nicered}{Z} & \labeledwire{niceblue}{Y} &                                                                     & \labeledwire{nicered}{Z} & \labeledwire{niceblue}{Z} &                                                                      &                           &                           &            \\
        \lstick{$\ket{+}$} & \labeledwire{nicered}{X}  &                           & \ctrl{-2}                                                           &                          &                                                                     &                          &                           &          &                          &                           &                                                                     &                          &                           &                                                                      & \labeledwire{nicered}{X}  &                           & \meterD{X} \\
        \lstick{$\ket{+}$} & \labeledwire{niceblue}{X} &                           &                                                                     &                          & \ctrl{-3}                                                           &                          &                           &          &                          &                           &                                                                     &                          &                           &                                                                      & \labeledwire{niceblue}{X} &                           & \meterD{X} \\
        \lstick{$\ket{+}$} & \labeledwire{niceblue}{X} &                           &                                                                     &                          &                                                                     &                          &                           &          &                          &                           & \ctrl{-4}                                                           &                          &                           &                                                                      & \labeledwire{niceblue}{X} &                           & \meterD{X} \\
        \lstick{$\ket{+}$} & \labeledwire{nicered}{X}  & \labeledwire{niceblue}{X} &                                                                     &                          &                                                                     &                          &                           &          &                          &                           &                                                                     &                          &                           & \ctrl{-5}                                                            & \labeledwire{nicered}{X}  & \labeledwire{niceblue}{X} & \meterD{X}
    \end{quantikz}
\end{align*}
Since $X^{\otimes 2}$ turns into $Y^{\otimes 2}$ through the application of $S^{\otimes 2}$, the associated detector involves the measurement of both $X^{\otimes 2}$ and $Z^{\otimes 2}$ after the gate.
The reason we change the measurement order before and after the gate is to obtain a more compact MBQC circuit, as we will soon understand.

To turn this circuit into an MBQC circuit, the first step is to rewrite the controlled-Pauli gates using \texttt{CZ} gates, and compile the \texttt{S} gate into \texttt{H} and \texttt{HS} gates:
\begin{align*}
    \tikzsetnextfilename{sec5-transversal-s-2}
    \begin{quantikz}[row sep=10pt, column sep=7pt]
                           &                           &                           & \ctrl{2} \gatebox{3}{2} &           & & \labeledwire{nicered}{Z} & & \mygate{1}{H} \gatebox{4}{4} & \ctrl{3}  &           & \mygate{1}{H} & & \labeledwire{nicered}{Z} & \labeledwire{niceblue}{X} & \gate{HS} & \labeledwire{nicered}{X} &  \labeledwire{niceblue}{Y} & \gate{H} & \labeledwire{nicered}{Z} & \labeledwire{niceblue}{Y} & & \mygate{1}{H} \gatebox{5}{4} & \ctrl{4}  &           & \mygate{1}{H} & & \labeledwire{nicered}{Z} & \labeledwire{niceblue}{Z} & & \ctrl{5} \gatebox{6}{2} &           & &                           &                           &            \\
                           &                           &                           &                         & \ctrl{1}  & & \labeledwire{nicered}{Z} & & \mygate{1}{H}                &           & \ctrl{2}  & \mygate{1}{H} & & \labeledwire{nicered}{Z} & \labeledwire{niceblue}{X} & \gate{HS} & \labeledwire{nicered}{X} &  \labeledwire{niceblue}{Y} & \gate{H} & \labeledwire{nicered}{Z} & \labeledwire{niceblue}{Y} & & \mygate{1}{H}                &           & \ctrl{3}  & \mygate{1}{H} & & \labeledwire{nicered}{Z} & \labeledwire{niceblue}{Z} & &                         & \ctrl{4}  & &                           &                           &            \\
        \lstick{$\ket{+}$} & \labeledwire{nicered}{X}  &                           & \ctrl{-2}               & \ctrl{-1} & &                          & &                              &           &           &               & &                          &                           &           &                          &                            &          &                          &                           & &                              &           &           &               & &                          &                           & &                         &           & & \labeledwire{nicered}{X}  &                           & \meterD{X} \\
        \lstick{$\ket{+}$} & \labeledwire{niceblue}{X} &                           &                         &           & &                          & &                              & \ctrl{-3} & \ctrl{-2} &               & &                          &                           &           &                          &                            &          &                          &                           & &                              &           &           &               & &                          &                           & &                         &           & & \labeledwire{niceblue}{X} &                           & \meterD{X} \\
        \lstick{$\ket{+}$} & \labeledwire{niceblue}{X} &                           &                         &           & &                          & &                              &           &           &               & &                          &                           &           &                          &                            &          &                          &                           & &                              & \ctrl{-4} & \ctrl{-3} &               & &                          &                           & &                         &           & & \labeledwire{niceblue}{X} &                           & \meterD{X} \\
        \lstick{$\ket{+}$} & \labeledwire{nicered}{X}  & \labeledwire{niceblue}{X} &                         &           & &                          & &                              &           &           &               & &                          &                           &           &                          &                            &          &                          &                           & &                              &           &           &               & &                          &                           & & \ctrl{-5}               & \ctrl{-4} & & \labeledwire{nicered}{X}  & \labeledwire{niceblue}{X} & \meterD{X}
    \end{quantikz}
\end{align*}
We then move the Hadamard gates out of the blue boxes using \cref{lemma:pushing-h-away}, and compile the two successive \texttt{H} gates coming after the \texttt{HS} gate into the identity. This is the trick that will give us a more compressed MBQC circuit and the reason we chose to measure the $X^{\otimes 2}$ operator first after the transversal gate.
\begin{align*}
    \tikzsetnextfilename{sec5-transversal-s-3}
    \begin{quantikz}[row sep=10pt, column sep=7pt]
                           &                           &                           & \ctrl{2} \gatebox{3}{2} &           & & \labeledwire{nicered}{Z} & \mygate{1}{H} & \labeledwire{nicered}{X} & & \ctrl{3} \gatebox{4}{2} &           & & \labeledwire{nicered}{X} & \labeledwire{niceblue}{Z} & \mygate{1}{H} & \labeledwire{nicered}{Z} & \labeledwire{niceblue}{X} & \gate{HS} & \labeledwire{nicered}{X} &  \labeledwire{niceblue}{Y} & & \ctrl{4} \gatebox{5}{2} &           & &  \labeledwire{nicered}{X} &  \labeledwire{niceblue}{X} & \mygate{1}{H} & \labeledwire{nicered}{Z} & \labeledwire{niceblue}{Z} & & \ctrl{5} \gatebox{6}{2} &           & &                           &                           &            \\
                           &                           &                           &                         & \ctrl{1}  & & \labeledwire{nicered}{Z} & \mygate{1}{H} & \labeledwire{nicered}{X} & &                         & \ctrl{2}  & & \labeledwire{nicered}{X} & \labeledwire{niceblue}{Z} & \mygate{1}{H} & \labeledwire{nicered}{Z} & \labeledwire{niceblue}{X} & \gate{HS} & \labeledwire{nicered}{X} &  \labeledwire{niceblue}{Y} & &                         & \ctrl{3}  & &  \labeledwire{nicered}{X} &  \labeledwire{niceblue}{X} & \mygate{1}{H} & \labeledwire{nicered}{Z} & \labeledwire{niceblue}{Z} & &                         & \ctrl{4}  & &                           &                           &            \\
        \lstick{$\ket{+}$} & \labeledwire{nicered}{X}  &                           & \ctrl{-2}               & \ctrl{-1} & &                          &               &                          & &                         &           & &                          &                          &               &                          &                           &           &                          &                            & &                         &           & &                           &                            &               &                          &                            & &                         &           & & \labeledwire{nicered}{X}  &                           & \meterD{X} \\
        \lstick{$\ket{+}$} & \labeledwire{niceblue}{X} &                           &                         &           & &                          &               &                          & & \ctrl{-3}               & \ctrl{-2} & &                          &                          &               &                          &                           &           &                          &                            & &                         &           & &                           &                            &               &                          &                            & &                         &           & & \labeledwire{niceblue}{X} &                           & \meterD{X} \\
        \lstick{$\ket{+}$} & \labeledwire{niceblue}{X} &                           &                         &           & &                          &               &                          & &                         &           & &                          &                          &               &                          &                           &           &                          &                            & & \ctrl{-4}               & \ctrl{-3} & &                           &                            &               &                          &                            & &                         &           & & \labeledwire{niceblue}{X} &                           & \meterD{X} \\
        \lstick{$\ket{+}$} & \labeledwire{nicered}{X}  & \labeledwire{niceblue}{X} &                         &           & &                          &               &                          & &                         &           & &                          &                          &               &                          &                           &           &                          &                            & &                         &           & &                           &                            &               &                          &                            & & \ctrl{-5}               & \ctrl{-4} & & \labeledwire{nicered}{X}  & \labeledwire{niceblue}{X} & \meterD{X}
    \end{quantikz}
\end{align*}
We now iteratively apply \cref{lemma:h-and-hs-merge-left} and \cref{lemma:cz-cz-merge} to compile the \texttt{H} and \texttt{HS} gates into their MBQC versions and merge the resulting \texttt{CZ} networks.
Starting with the Hadamard gates on the left, we obtain the following circuit:
\begin{align*}
    \tikzsetnextfilename{sec5-transversal-s-4}
    \begin{quantikz}[row sep=10pt, column sep=7pt]
        \lstick{$\ket{+}$} & \labeledwire{nicered}{X}  &                           & \gatebox{6}{6} &           &           & \ctrl{2}  & \ctrl{5}  &           & & \labeledwire{nicered}{X}  & \labeledwire{niceblue}{Z} & \mygate{1}{H} & \labeledwire{nicered}{Z}  & \labeledwire{niceblue}{X} & \gate{HS} & \labeledwire{nicered}{X}  &  \labeledwire{niceblue}{Y} & & \ctrl{6} \gatebox{7}{2} &           & &  \labeledwire{nicered}{X} &  \labeledwire{niceblue}{X} & \mygate{1}{H} & \labeledwire{nicered}{Z}  & \labeledwire{niceblue}{Z} & & \ctrl{5} \gatebox{8}{2} &           & &                           &                           &            \\
        \lstick{$\ket{+}$} & \labeledwire{nicered}{X}  &                           &                &           & \ctrl{2}  &           &           & \ctrl{4}  & & \labeledwire{nicered}{X}  & \labeledwire{niceblue}{Z} & \mygate{1}{H} & \labeledwire{nicered}{Z}  & \labeledwire{niceblue}{X} & \gate{HS} & \labeledwire{nicered}{X}  &  \labeledwire{niceblue}{Y} & &                         & \ctrl{5}  & &  \labeledwire{nicered}{X} &  \labeledwire{niceblue}{X} & \mygate{1}{H} & \labeledwire{nicered}{Z}  & \labeledwire{niceblue}{Z} & &                         & \ctrl{6}  & &                           &                           &            \\
                           &                           &                           & \ctrl{2}       &           &           & \ctrl{-2} &           &           & &                           &                           &               &                           &                           &           &                           &                            & &                         &           & &                           &                            &               &                           &                           & &                         &           & &                           &                           & \meterD{X} \\
                           &                           &                           &                & \ctrl{1}  & \ctrl{-2} &           &           &           & &                           &                           &               &                           &                           &           &                           &                            & &                         &           & &                           &                            &               &                           &                           & &                         &           & &                           &                           & \meterD{X} \\
        \lstick{$\ket{+}$} & \labeledwire{nicered}{X}  &                           & \ctrl{-2}      & \ctrl{-1} &           &           &           &           & & \labeledwire{nicered}{X}  &                           &               & \labeledwire{nicered}{X}  &                           &           & \labeledwire{nicered}{X}  &                            & &                         &           & & \labeledwire{nicered}{X}  &                            &               & \labeledwire{nicered}{X}  &                           & &                         &           & & \labeledwire{nicered}{X}  &                           & \meterD{X} \\
        \lstick{$\ket{+}$} & \labeledwire{niceblue}{X} &                           &                &           &           &           & \ctrl{-5} & \ctrl{-4} & & \labeledwire{niceblue}{X} &                           &               & \labeledwire{niceblue}{X} &                           &           & \labeledwire{niceblue}{X} &                            & &                         &           & & \labeledwire{niceblue}{X} &                            &               & \labeledwire{niceblue}{X} &                           & &                         &           & & \labeledwire{niceblue}{X} &                           & \meterD{X} \\
        \lstick{$\ket{+}$} & \labeledwire{niceblue}{X} &                           &                &           &           &           &           &           & & \labeledwire{niceblue}{X} &                           &               & \labeledwire{niceblue}{X} &                           &           & \labeledwire{niceblue}{X} &                            & & \ctrl{-6}               & \ctrl{-5} & & \labeledwire{niceblue}{X} &                            &               & \labeledwire{niceblue}{X} &                           & &                         &           & & \labeledwire{niceblue}{X} &                           & \meterD{X} \\
        \lstick{$\ket{+}$} & \labeledwire{nicered}{X}  & \labeledwire{niceblue}{X} &                &           &           &           &           &           & & \labeledwire{nicered}{X}  & \labeledwire{niceblue}{X} &               & \labeledwire{nicered}{X}  & \labeledwire{niceblue}{X} &           & \labeledwire{nicered}{X}  & \labeledwire{niceblue}{X}  & &                         &           & & \labeledwire{nicered}{X}  & \labeledwire{niceblue}{X}  &               & \labeledwire{nicered}{X}  & \labeledwire{niceblue}{X} & & \ctrl{-5}               & \ctrl{-6} & & \labeledwire{nicered}{X}  & \labeledwire{niceblue}{X} & \meterD{X}
    \end{quantikz}
\end{align*}
Next, we compile the \texttt{H} and \texttt{HS} gates, giving the following circuit:
\begin{align*}
    \tikzsetnextfilename{sec5-transversal-s-5}
    \begin{quantikz}[row sep=10pt, column sep=7pt]
        \lstick{$\ket{+}$} & \labeledwire{nicered}{X}  & \labeledwire{niceblue}{X} & & \gatebox{11}{12} &           &           &           &           &           &           &           &           & \ctrl{2}  & \ctrl{10}  &            & & \labeledwire{nicered}{X}  & \labeledwire{niceblue}{X}  & \mygate{1}{H} & \labeledwire{nicered}{Z}  & \labeledwire{niceblue}{Z} & & \ctrl{10} \gatebox{12}{2} &           & &                           &                           &            \\
        \lstick{$\ket{+}$} & \labeledwire{nicered}{X}  & \labeledwire{niceblue}{X} & &                  &           &           &           &           &           &           &           & \ctrl{2}  &           &            & \ctrl{9}   & & \labeledwire{nicered}{X}  & \labeledwire{niceblue}{X}  & \mygate{1}{H} & \labeledwire{nicered}{Z}  & \labeledwire{niceblue}{Z} & &                           & \ctrl{9}  & &                           &                           &            \\
        \lstick{$\ket{+}$} & \labeledwire{niceblue}{X} &                           & &                  &           &           &           &           &           &           & \ctrl{2}  &           & \ctrl{-2} &            &            & & \labeledwire{niceblue}{Y} &                            &               & \labeledwire{niceblue}{Y} &                           & &                           &           & & \labeledwire{niceblue}{Y} &                           & \meterD{Y} \\
        \lstick{$\ket{+}$} & \labeledwire{niceblue}{X} &                           & &                  &           &           &           &           &           & \ctrl{2}  &           & \ctrl{-2} &           &            &            & & \labeledwire{niceblue}{Y} &                            &               & \labeledwire{niceblue}{Y} &                           & &                           &           & & \labeledwire{niceblue}{Y} &                           & \meterD{Y} \\
        \lstick{$\ket{+}$} & \labeledwire{nicered}{X}  &                           & &                  &           &           & \ctrl{2}  & \ctrl{5}  &           &           & \ctrl{-2} &           &           &            &            & & \labeledwire{nicered}{X}  &                            &               & \labeledwire{nicered}{X}  &                           & &                           &           & & \labeledwire{nicered}{X}  &                           & \meterD{X} \\
        \lstick{$\ket{+}$} & \labeledwire{nicered}{X}  &                           & &                  &           & \ctrl{2}  &           &           & \ctrl{4}  & \ctrl{-2} &           &           &           &            &            & & \labeledwire{nicered}{X}  &                            &               & \labeledwire{nicered}{X}  &                           & &                           &           & & \labeledwire{nicered}{X}  &                           & \meterD{X} \\
                           &                           &                           & & \ctrl{2}         &           &           & \ctrl{-2} &           &           &           &           &           &           &            &            & &                           &                            &               &                           &                           & &                           &           & &                           &                           & \meterD{X} \\
                           &                           &                           & &                  & \ctrl{1}  & \ctrl{-2} &           &           &           &           &           &           &           &            &            & &                           &                            &               &                           &                           & &                           &           & &                           &                           & \meterD{X} \\
        \lstick{$\ket{+}$} & \labeledwire{nicered}{X}  &                           & & \ctrl{-2}        & \ctrl{-1} &           &           &           &           &           &           &           &           &            &            & & \labeledwire{nicered}{X}  &                            &               & \labeledwire{nicered}{X}  &                           & &                           &           & & \labeledwire{nicered}{X}  &                           & \meterD{X} \\
        \lstick{$\ket{+}$} & \labeledwire{niceblue}{X} &                           & &                  &           &           &           & \ctrl{-5} & \ctrl{-4} &           &           &           &           &            &            & & \labeledwire{niceblue}{X} &                            &               & \labeledwire{niceblue}{X} &                           & &                           &           & & \labeledwire{niceblue}{X} &                           & \meterD{X} \\
        \lstick{$\ket{+}$} & \labeledwire{niceblue}{X} &                           & &                  &           &           &           &           &           &           &           &           &           & \ctrl{-10} & \ctrl{-9}  & & \labeledwire{niceblue}{X} &                            &               & \labeledwire{niceblue}{X} &                           & &                           &           & & \labeledwire{niceblue}{X} &                           & \meterD{X} \\
        \lstick{$\ket{+}$} & \labeledwire{nicered}{X}  & \labeledwire{niceblue}{X} & &                  &           &           &           &           &           &           &           &           &           &            &            & & \labeledwire{nicered}{X}  & \labeledwire{niceblue}{X}  &               & \labeledwire{nicered}{X}  & \labeledwire{niceblue}{X} & & \ctrl{-10}                & \ctrl{-9} & & \labeledwire{nicered}{X}  & \labeledwire{niceblue}{X} & \meterD{X}
    \end{quantikz}
\end{align*}
Finally, compiling the last \texttt{H} gates gives us:
\begin{align*}
    \tikzsetnextfilename{sec5-transversal-s-6}
    \begin{quantikz}[row sep=10pt, column sep=7pt]
        \lstick{$\ket{+}$} &                           &                           & & \gatebox{14}{16} &           &           &           &           &           &           &           &           &           &            &           &           & \ctrl{2}  & \ctrl{12}  &           & &                           &                           &            \\
        \lstick{$\ket{+}$} &                           &                           & &                  &           &           &           &           &           &           &           &           &           &            &           & \ctrl{2}  &           &            & \ctrl{11} & &                           &                           &            \\
        \lstick{$\ket{+}$} & \labeledwire{nicered}{X}  & \labeledwire{niceblue}{X} & &                  &           &           &           &           &           &           &           &           & \ctrl{2}  & \ctrl{10}  &           &           & \ctrl{-2} &            &           & & \labeledwire{nicered}{X}  & \labeledwire{niceblue}{X} & \meterD{X} \\
        \lstick{$\ket{+}$} & \labeledwire{nicered}{X}  & \labeledwire{niceblue}{X} & &                  &           &           &           &           &           &           &           & \ctrl{2}  &           &            & \ctrl{9}  & \ctrl{-2} &           &            &           & & \labeledwire{nicered}{X}  & \labeledwire{niceblue}{X} & \meterD{X} \\
        \lstick{$\ket{+}$} & \labeledwire{niceblue}{X} &                           & &                  &           &           &           &           &           &           & \ctrl{2}  &           & \ctrl{-2} &            &           &           &           &            &           & & \labeledwire{niceblue}{Y} &                           & \meterD{Y} \\
        \lstick{$\ket{+}$} & \labeledwire{niceblue}{X} &                           & &                  &           &           &           &           &           & \ctrl{2}  &           & \ctrl{-2} &           &            &           &           &           &            &           & & \labeledwire{niceblue}{Y} &                           & \meterD{Y} \\
        \lstick{$\ket{+}$} & \labeledwire{nicered}{X}  &                           & &                  &           &           & \ctrl{2}  & \ctrl{5}  &           &           & \ctrl{-2} &           &           &            &           &           &           &            &           & & \labeledwire{nicered}{X}  &                           & \meterD{X} \\
        \lstick{$\ket{+}$} & \labeledwire{nicered}{X}  &                           & &                  &           & \ctrl{2}  &           &           & \ctrl{4}  & \ctrl{-2} &           &           &           &            &           &           &           &            &           & & \labeledwire{nicered}{X}  &                           & \meterD{X} \\
                           &                           &                           & & \ctrl{2}         &           &           & \ctrl{-2} &           &           &           &           &           &           &            &           &           &           &            &           & &                           &                           & \meterD{X} \\
                           &                           &                           & &                  & \ctrl{1}  & \ctrl{-2} &           &           &           &           &           &           &           &            &           &           &           &            &           & &                           &                           & \meterD{X} \\
        \lstick{$\ket{+}$} & \labeledwire{nicered}{X}  &                           & & \ctrl{-2}        & \ctrl{-1} &           &           &           &           &           &           &           &           &            &           &           &           &            &           & & \labeledwire{nicered}{X}  &                           & \meterD{X} \\
        \lstick{$\ket{+}$} & \labeledwire{niceblue}{X} &                           & &                  &           &           &           & \ctrl{-5} & \ctrl{-4} &           &           &           &           &            &           &           &           &            &           & & \labeledwire{niceblue}{X} &                           & \meterD{X} \\
        \lstick{$\ket{+}$} & \labeledwire{niceblue}{X} &                           & &                  &           &           &           &           &           &           &           &           &           & \ctrl{-10} & \ctrl{-9} &           &           &            &           & & \labeledwire{niceblue}{X} &                           & \meterD{X} \\
        \lstick{$\ket{+}$} & \labeledwire{nicered}{X}  & \labeledwire{niceblue}{X} & &                  &           &           &           &           &           &           &           &           &           &            &           &           &           & \ctrl{-12} & \ctrl{-11} & & \labeledwire{nicered}{X}  & \labeledwire{niceblue}{X} & \meterD{X}
    \end{quantikz}
\end{align*}

\clearpage
We can now draw the corresponding cluster state complex, in its compressed representation.
To make the figure more readable, we omit the detector nodes, but circle in green the support of the two detectors derived above (the red-labeled one on the left figure and the blue-labeled one on the right figure):
\begin{center}
    \includegraphics[width=0.15\textwidth]{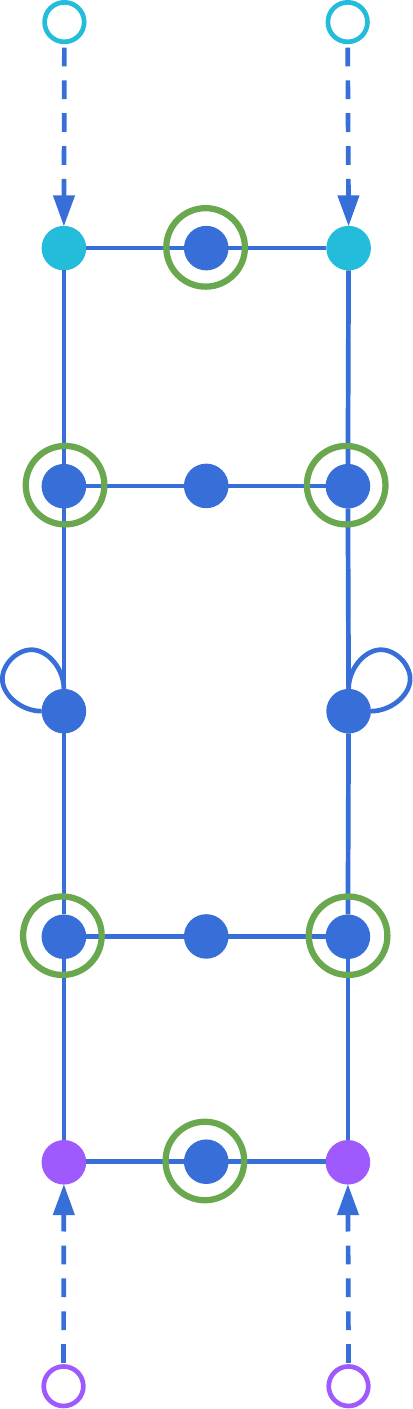}
    \hspace{4em}
    \includegraphics[width=0.15\textwidth]{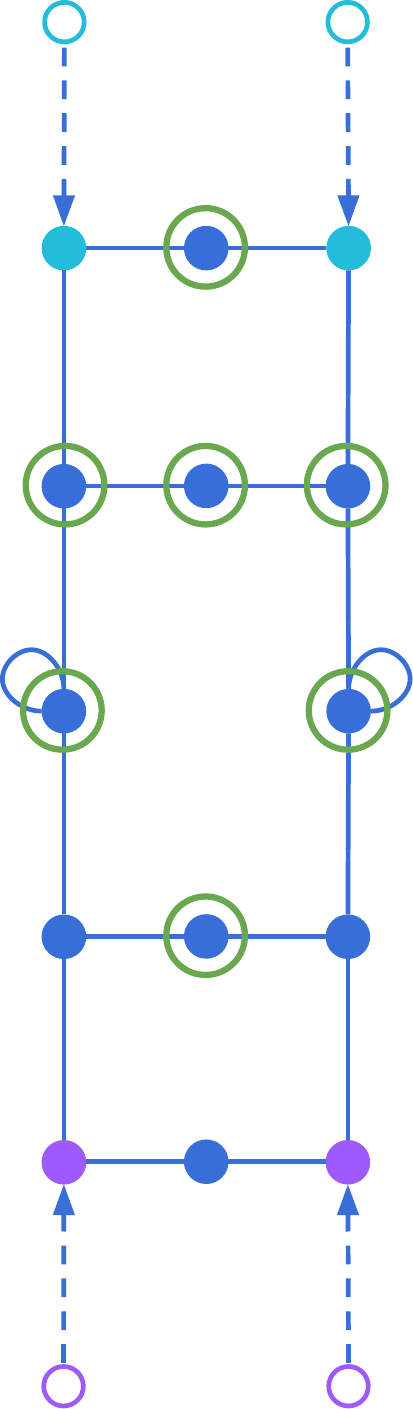}
\end{center}


\section{Cluster state complex from a stabilizer code}
\label{sec:from-stabilizer-codes-to-cluster-states}
\subsection{CSS stabilizer codes}

\begin{figure}[ht]
    \centering
    \subfloat[]{
        \includegraphics[width=0.15\linewidth]{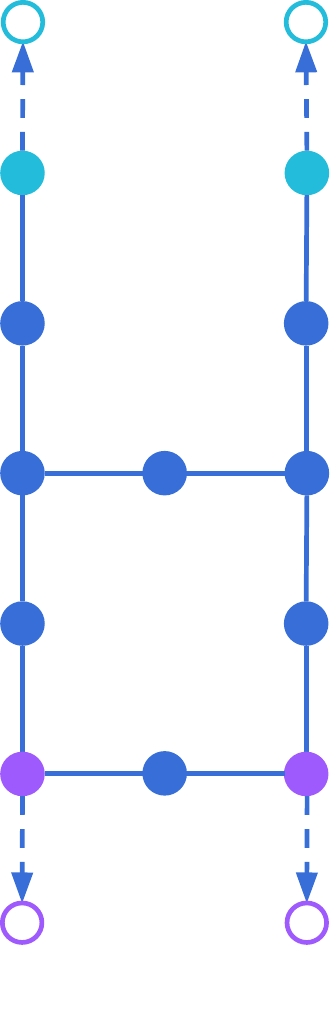}
    }
    \hspace{2em}
    \subfloat[]{
        \includegraphics[width=0.15\linewidth]{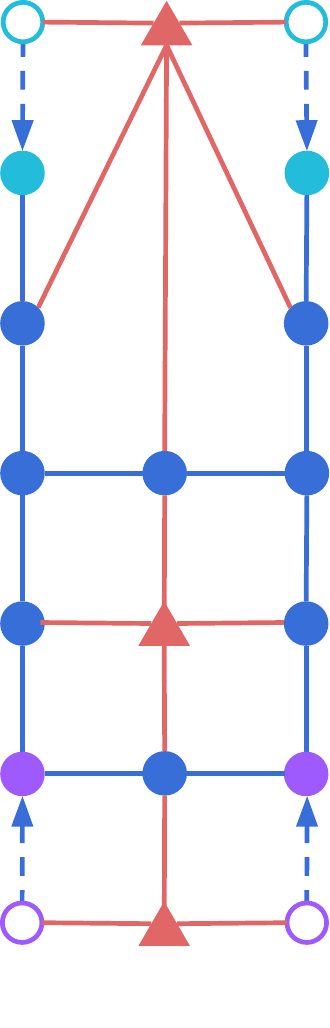}
    }
    \hspace{4em}
    \subfloat[]{
        \includegraphics[width=0.15\linewidth]{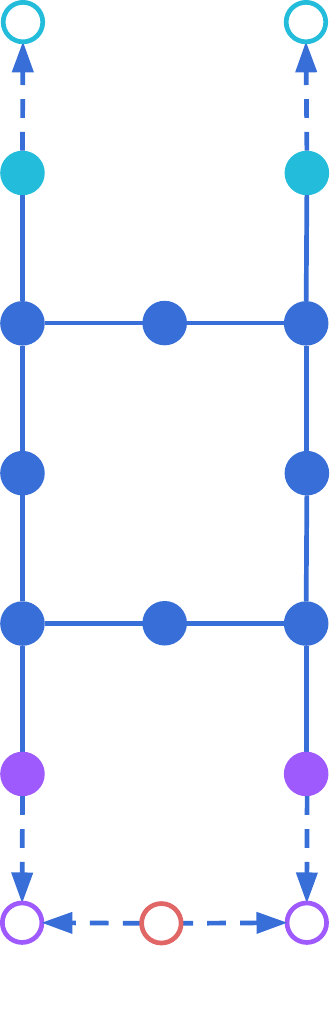}
    }
    \hspace{2em}
    \subfloat[]{
        \includegraphics[width=0.15\linewidth]{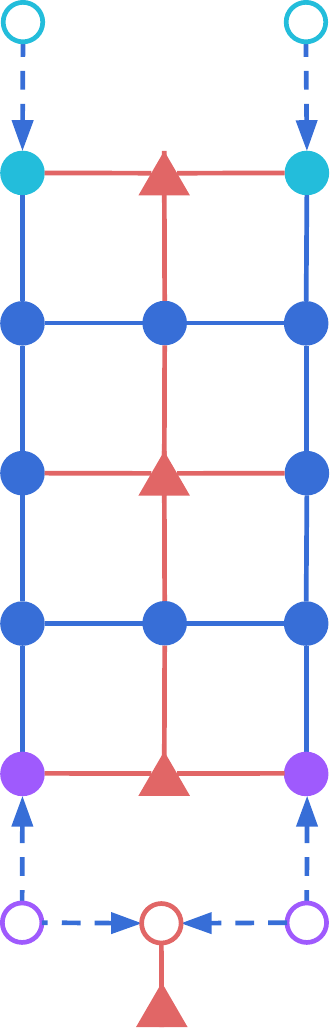}
    }
    \caption{ \label{fig:css-cluster-states}
        Cluster state complexes of two foliated CSS codes, with stabilizer groups $\langle ZZ \rangle$ (a and b) and $\langle XX \rangle$ (c and d). \textbf{(a)} and \textbf{(c)} Compressed representation of the cluster state complex, without detector nodes for clarity. \textbf{(b)} and \textbf{(d)} Compressed co-representation of the cluster state complex, with detector nodes.
    }
\end{figure}

We show here that we recover the usual foliated cluster state when the input circuit corresponds to a phenomenological memory experiment with a CSS code \cite{bolt2016foliated}. Two examples of cluster state complexes coming from a CSS code are shown in \cref{fig:css-cluster-states}.

\begin{theorem}
    Consider the circuit made of $T$ rounds of stabilizer measurements of an $[[n,k,d]]$ CSS stabilizer code, where $m_X$ $X$ stabilizers and $m_Z$ $Z$ stabilizers are measured each round using some ancilla qubits, and errors are assumed to only happen in-between successive rounds. Let $H^X$ and $H^Z$ be the parity-check matrices associated to the $X$ and $Z$ stabilizers, respectively. This circuit is equivalent to a cluster state complex described by the following graph (in the compressed representation):
    \begin{itemize}
        \item It contains $2T+3$ layers $\{0,\ldots,2T+2\}$ of error nodes, where each odd layer $\ell = 2i+1$ (even layer $\ell=2i$) contains $n$ \textnormal{data nodes} $q^{(\ell)}_1,\ldots,q^{(\ell)}_n$ and, if $\ell \leq 2T$, $m_Z$ ($m_X$) \textnormal{ancilla nodes} $a^{(\ell)}_1,\ldots,a^{(\ell)}_{m_Z}$ ($a^{(\ell)}_1,\ldots,a^{(\ell)}_{m_X}$). The nodes on layers $0$ and $2T+2$ are virtual, while all the other nodes are physical.
        \item Within a given layer $\ell \leq 2T$, data nodes and ancilla nodes (both physical and virtual) are connected to each other according to the matrices $H^Z$ (for odd layers) and $H^X$ (for even layers), seen as biadjacency matrices, that is, $a^{(\ell)}_i$ and $q^{(\ell)}_j$ are connected if and only if $H^Z_{ij}=1$ (for odd layers) or $H^X_{ij}=1$ (for even layers).
        \item For $1 \leq \ell \leq 2T$, every data node $q_i^{(\ell)}$ is connected to the data node $q_i^{(\ell+1)}$ in the next layer. For $\ell = 0$, there is directed edge from $q_i^{(\ell+1)}$ to $q_i^{(\ell)}$. For $\ell = 2T+1$, there is directed edge from $q_i^{(\ell)}$ and $q_i^{(\ell+1)}$.
        \item For every ancilla node $a_i^{(\ell)}$ ($1 \leq \ell \leq 2T$) connected to a set of data nodes $q^{(\ell)}_{j_1},\ldots,q^{(\ell)}_{j_k}$, we have a detector node connected to $a_i^{(\ell)}$, $q^{(\ell+1)}_{j_1},\ldots,q^{(\ell+1)}_{j_k}$, $a_i^{(\ell+2)}$ if $\ell \leq 2T-2$, and $q^{(\ell+3)}_{j_1},\ldots,q^{(\ell+3)}_{j_k}$ if $\ell=2T-1$.
        These detectors correspond to spackles in the spacetime code of the circuit.
        Corresponding to the backles, we also have one detector node per ancilla node $a_i^{(1)}$ on layer one, connected to $a_i^{(1)}$ and $q^{(0)}_{j_1},\ldots,q^{(0)}_{j_k}$, and one per ancilla node $a_i^{(2)}$ on layer two, connected to $a_i^{(2)}$ and $q^{(1)}_{j_1},\ldots,q^{(1)}_{j_k}$.
        Finally, we have one detector node per virtual ancilla node, connected only to it.
    \end{itemize}
\end{theorem}
\begin{proof}
The circuit consists of several rounds of simultaneous measurements of all the stabilizers, where errors are assumed to happen in-between measurement rounds, on both the data qubits and ancilla qubits used to perform the measurements. For instance, two rounds of stabilizer measurements correspond to the following circuit:
\begin{align*}
    \tikzsetnextfilename{sec6-css-proof-1}
    \begin{quantikz}[row sep=10pt, column sep=6pt]
                           & \phantomgate{H}          & \gatebox{15}{7} & \mygate{3}{S_1^Z} &        & \mygate{3}{S_{m_Z}^Z} & \mygate{3}{S_1^X} &        & \mygate{3}{S_{m_X}^X} & & \labeledwire{nicered}{Z} & \gatebox{15}{7}    & \mygate{3}{S_1^Z} &        & \mygate{3}{S_{m_Z}^Z} & \mygate{3}{S_1^X} &        & \mygate{3}{S_{m_X}^X} & &                          &            \\[-1.5em]
                           & \setwiretype{n}          & \vdots          &                   & \cdots &                       &                   & \cdots &                       & &                          &                    &                   & \cdots &                       &                   & \cdots &                       & &                          &            \\[-1.1em]
                           &                          &                 &                   &        &                       &                   &        &                       & & \labeledwire{nicered}{Z} &                    &                   &        &                       &                   &        &                       & &                          &            \\
        \lstick{$\ket{+}$} & \labeledwire{nicered}{X} &                 & \ctrl{-1}         &        &                       &                   &        &                       & & \phantomgate{H}          &                    &                   &        &                       &                   &        &                       & & \labeledwire{nicered}{X} & \meterD{X} \\[-1.2em]
                           & \setwiretype{n}          & \vdots          &                   & \ddots &                       &                   &        &                       & &                          &                    &                   &        &                       &                   &        &                       & &                          &            \\[-0.7em]
        \lstick{$\ket{+}$} &                          &                 &                   &        & \ctrl{-3}             &                   &        &                       & & \phantomgate{H}          &                    &                   &        &                       &                   &        &                       & &                          & \meterD{X} \\
        \lstick{$\ket{+}$} &                          &                 &                   &        &                       & \ctrl{-4}         &        &                       & & \phantomgate{H}          &                    &                   &        &                       &                   &        &                       & &                          & \meterD{X} \\[-1.2em]
                           & \setwiretype{n}          & \vdots          &                   &        &                       &                   & \ddots &                       & &                          &                    &                   &        &                       &                   &        &                       & &                          &            \\[-0.7em]
        \lstick{$\ket{+}$} &                          &                 &                   &        &                       &                   &        & \ctrl{-6}             & & \phantomgate{H}          &                    &                   &        &                       &                   &        &                       & &                          & \meterD{X} \\
        \lstick{$\ket{+}$} & \labeledwire{nicered}{X} &                 &                   &        &                       &                   &        &                       & & \phantomgate{H}          &                    & \ctrl{-7}         &        &                       &                   &        &                       & & \labeledwire{nicered}{X} & \meterD{X} \\[-1.2em]
                           & \setwiretype{n}          & \vdots          &                   &        &                       &                   &        &                       & &                          &                    &                   & \ddots &                       &                   &        &                       & &                          &            \\[-0.7em]
        \lstick{$\ket{+}$} &                          &                 &                   &        &                       &                   &        &                       & & \phantomgate{H}          &                    &                   &        & \ctrl{-9}             &                   &        &                       & &                          & \meterD{X} \\
        \lstick{$\ket{+}$} &                          &                 &                   &        &                       &                   &        &                       & & \phantomgate{H}          &                    &                   &        &                       & \ctrl{-10}        &        &                       & &                          & \meterD{X} \\[-1.2em]
                           & \setwiretype{n}          & \vdots          &                   &        &                       &                   &        &                       & &                          &                    &                   &        &                       &                   & \ddots &                       & &                          &            \\[-0.7em]
        \lstick{$\ket{+}$} &                          &                 &                   &        &                       &                   &        &                       & & \phantomgate{H}          &                    &                   &        &                       &                   &        & \ctrl{-12}            & &                          & \meterD{X}
    \end{quantikz}
\end{align*}
where the sets $\{S_i^X\}$ and $\{S_i^Z\}$ are generators of the $X$ and $Z$ stabilizer groups respectively. We also provide (in red) an example of detector, corresponding the repeated measurement of $S_1^Z$ in the two rounds. The $Z$-part of this detector is in the support of $S_1^Z$.

We can transform every $X$ stabilizer measurement into $Z$ measurements by inserting Hadamard gates before and after the \texttt{CX} gates.
Denoting by $\bar{S}_i^X$ the operator coming from $S_i^X$ where every $X$ has been turned into a $Z$,
and using \cref{lemma:pushing-h-away}, we get the following equivalent circuit:

\begin{align*}
    \tikzsetnextfilename{sec6-css-proof-2}
    \begin{quantikz}[row sep=10pt, column sep=5pt]
                           & \phantomgate{H}         & \gatebox{15}{8}     & \mygate{3}{S_1^Z} &        & \mygate{3}{S_{m_Z}^Z} & \mygate{1}{H} & \mygate{3}{\bar{S}_1^X}  &        & \mygate{3}{\bar{S}_{m_X}^X}  & & \labeledwire{nicered}{X} & \gate{H} & \labeledwire{nicered}{Z} & &  \gatebox{15}{8} \qw & \mygate{3}{S_1^Z} &        & \mygate{3}{S_{m_Z}^Z} & \mygate{1}{H} & \mygate{3}{\bar{S}_1^X}  &        & \mygate{3}{\bar{S}_{m_X}^X}  & & \phantomgate{H}          & \gate{H} &            \\[-1.5em]
                           & \setwiretype{n}         & \vdots              &                   & \cdots &                       &               &                          & \cdots &                              & &                         &          &                           & &                      &                   & \cdots &                       &               &                          & \cdots &                              & &                          &          &            \\[-1.1em]
                           &                         &                     &                   &        &                       & \mygate{1}{H} &                          &        &                              & & \labeledwire{nicered}{X} & \gate{H} & \labeledwire{nicered}{Z} & &                      &                   &        &                       & \mygate{1}{H} &                          &        &                              & &                          & \gate{H} &            \\
        \lstick{$\ket{+}$} & \labeledwire{nicered}{X} &                    & \ctrl{-1}         &        &                       &               &                          &        &                              & &                         &          &                           & &                      &                   &        &                       &               &                          &        &                              & & \labeledwire{nicered}{X} & \meterD{X} \\[-1.2em]
                           & \setwiretype{n}         & \vdots              &                   & \ddots &                       &               &                          &        &                              & &                         &          &                           & &                      &                   &        &                       &               &                          &        &                              & &                          &            \\[-0.7em]
        \lstick{$\ket{+}$} &                         &                     &                   &        & \ctrl{-3}             &               &                          &        &                              & &                         &          &                           & &                      &                   &        &                       &               &                          &        &                              & &                          & \meterD{X} \\
        \lstick{$\ket{+}$} &                         &                     &                   &        &                       &               & \ctrl{-4}                &        &                              & &                         &          &                           & &                      &                   &        &                       &               &                          &        &                              & &                          & \meterD{X} \\[-1.2em]
                           & \setwiretype{n}         & \vdots              &                   &        &                       &               &                          & \ddots &                              & &                         &          &                           & &                      &                   &        &                       &               &                          &        &                              & &                          &            \\[-0.7em]
        \lstick{$\ket{+}$} &                         &                     &                   &        &                       &               &                          &        & \ctrl{-6}                    & &                         &          &                           & &                      &                   &        &                       &               &                          &        &                              & &                          & \meterD{X} \\
        \lstick{$\ket{+}$} & \labeledwire{nicered}{X} &                    &                   &        &                       &               &                          &        &                              & &                         &          &                           & &                      & \ctrl{-7}         &        &                       &               &                          &        &                              & & \labeledwire{nicered}{X} & \meterD{X} \\[-1.2em]
                           & \setwiretype{n}         & \vdots              &                   &        &                       &               &                          &        &                              & &                         &          &                           & &                      &                   & \ddots &                       &               &                          &        &                              & &                          &            \\[-0.7em]
        \lstick{$\ket{+}$} &                         &                     &                   &        &                       &               &                          &        &                              & &                         &          &                           & &                      &                   &        & \ctrl{-9}             &               &                          &        &                              & &                          & \meterD{X} \\
        \lstick{$\ket{+}$} &                         &                     &                   &        &                       &               &                          &        &                              & &                         &          &                           & &                      &                   &        &                       &               & \ctrl{-10}               &        &                              & &                          & \meterD{X} \\[-1.2em]
                           & \setwiretype{n}         & \vdots              &                   &        &                       &               &                          &        &                              & &                         &          &                           & &                      &                   &        &                       &               &                          & \ddots &                              & &                          &            \\[-0.7em]
        \lstick{$\ket{+}$} &                         &                     &                   &        &                       &               &                          &        &                              & &                         &          &                           & &                      &                   &        &                       &               &                          &        & \ctrl{-12}                   & &                          & \meterD{X}
    \end{quantikz}
\end{align*}

We finally use \cref{lemma:h-cz-sandwich} to compile the Hadamard gates into their measurement-based versions within each blue box, and \cref{lemma:h-and-hs-merge-left,lemma:cz-cz-merge} to merge the \texttt{CZ} networks at different rounds, giving us the following circuit for our two rounds of measurements:

\begin{align*}
    \tikzsetnextfilename{sec6-css-proof-3}
    \begin{quantikz}[row sep=10pt, column sep=2pt]
        \lstick{$\ket{+}$} & \phantomgate{H}           & & \gatebox{27}{28} &                   &        &                       &           &                       &           &                         &        &                             & &           &                       &           & &                     &                   &        &                       &           &                       &           &                         &        &                             & &           &                       & \ctrl{3}  & &                          & &            \\[-1.2em]
                           & \setwiretype{n}           & & \vdots           &                   &        &                       &           &                       &           &                         &        &                             & &           &                       &           & &                     &                   &        &                       &           &                       &           &                         &        &                             & &           & \reflectbox{$\ddots$} &           & &                          & &            \\[-0.7em]
        \lstick{$\ket{+}$} &                           & &                  &                   &        &                       &           &                       &           &                         &        &                             & &           &                       &           & &                     &                   &        &                       &           &                       &           &                         &        &                             & & \ctrl{3}  &                       &           & &                          & &            \\
        \lstick{$\ket{+}$} &                           & &                  &                   &        &                       &           &                       &           &                         &        &                             & &           &                       &           & &                     &                   &        &                       &           &                       & \ctrl{3}  & \mygate{3}{\bar{S}_1^X} &        & \mygate{3}{\bar{S}_{m_X}^X} & &           &                       & \ctrl{-3} & &                          & & \meterD{X} \\[-1.5em]
                           & \setwiretype{n}           & & \vdots           &                   &        &                       &           &                       &           &                         &        &                             & &           &                       &           & &                     &                   &        &                       &           & \reflectbox{$\ddots$} &           &                         & \cdots &                             & &           & \reflectbox{$\ddots$} &           & &                          & &            \\[-1.1em]
        \lstick{$\ket{+}$} &                           & &                  &                   &        &                       &           &                       &           &                         &        &                             & &           &                       &           & &                     &                   &        &                       & \ctrl{3}  &                       &           &                         &        &                             & & \ctrl{-3} &                       &           & &                          & & \meterD{X} \\
        \lstick{$\ket{+}$} &                           & &                  &                   &        &                       &           &                       &           &                         &        &                             & &           &                       & \ctrl{3}  & &                     & \mygate{3}{S_1^Z} &        & \mygate{3}{S_{m_Z}^Z} &           &                       & \ctrl{-3} &                         &        &                             & &           &                       &           & &                          & & \meterD{X} \\[-1.5em]
                           & \setwiretype{n}           & & \vdots           &                   &        &                       &           &                       &           &                         &        &                             & &           & \reflectbox{$\ddots$} &           & &                     &                   & \cdots &                       &           & \reflectbox{$\ddots$} &           &                         &        &                             & &           &                       &           & &                          & &            \\[-1.1em]
        \lstick{$\ket{+}$} &                           & &                  &                   &        &                       &           &                       &           &                         &        &                             & & \ctrl{3}  &                       &           & &                     &                   &        &                       & \ctrl{-3} &                       &           &                         &        &                             & &           &                       &           & &                          & & \meterD{X} \\
        \lstick{$\ket{+}$} & \labeledwire{nicered}{X}  & &                  &                   &        &                       &           &                       & \ctrl{3}  & \mygate{3}{\bar{S}_1^X} &        & \mygate{3}{\bar{S}_{m_X}^X} & &           &                       & \ctrl{-3} & &                     &                   &        &                       &           &                       &           &                         &        &                             & &           &                       &           & & \labeledwire{nicered}{X} & & \meterD{X} \\[-1.5em]
                           & \setwiretype{n}           & & \vdots           &                   &        &                       &           & \reflectbox{$\ddots$} &           &                         & \cdots &                             & &           & \reflectbox{$\ddots$} &           & &                     &                   &        &                       &           &                       &           &                         &        &                             & &           &                       &           & &                          & &            \\[-1.1em]
        \lstick{$\ket{+}$} & \labeledwire{nicered}{X}  & &                  &                   &        &                       & \ctrl{3}  &                       &           &                         &        &                             & & \ctrl{-3} &                       &           & &                     &                   &        &                       &           &                       &           &                         &        &                             & &           &                       &           & & \labeledwire{nicered}{X} & & \meterD{X} \\
                           &                           & &                  & \mygate{3}{S_1^Z} &        & \mygate{3}{S_{m_Z}^Z} &           &                       & \ctrl{-3} &                         &        &                             & &           &                       &           & &                     &                   &        &                       &           &                       &           &                         &        &                             & &           &                       &           & &                          & & \meterD{X} \\[-1.5em]
                           & \setwiretype{n}           & & \vdots           &                   & \cdots &                       &           & \reflectbox{$\ddots$} &           &                         &        &                             & &           &                       &           & &                     &                   &        &                       &           &                       &           &                         &        &                             & &           &                       &           & &                          & &            \\[-1.1em]
                           &                           & &                  &                   &        &                       & \ctrl{-3} &                       &           &                         &        &                             & &           &                       &           & &                     &                   &        &                       &           &                       &           &                         &        &                             & &           &                       &           & &                          & & \meterD{X} \\
        \lstick{$\ket{+}$} & \labeledwire{nicered}{X}  & &                  & \ctrl{-1}         &        &                       &           &                       &           &                         &        &                             & &           &                       &           & &                     &                   &        &                       &           &                       &           &                         &        &                             & &           &                       &           & & \labeledwire{nicered}{X} & & \meterD{X} \\[-1.2em]
                           & \setwiretype{n}           & & \vdots           &                   & \ddots &                       &           &                       &           &                         &        &                             & &           &                       &           & &                     &                   &        &                       &           &                       &           &                         &        &                             & &           &                       &           & &                          & &            \\[-0.7em]
        \lstick{$\ket{+}$} &                           & &                  &                   &        & \ctrl{-3}             &           &                       &           &                         &        &                             & &           &                       &           & &                     &                   &        &                       &           &                       &           &                         &        &                             & &           &                       &           & &                          & & \meterD{X} \\
        \lstick{$\ket{+}$} &                           & &                  &                   &        &                       &           &                       &           & \ctrl{-7}               &        &                             & &           &                       &           & &                     &                   &        &                       &           &                       &           &                         &        &                             & &           &                       &           & &                          & & \meterD{X} \\[-1.2em]
                           & \setwiretype{n}           & & \vdots           &                   &        &                       &           &                       &           &                         & \ddots &                             & &           &                       &           & &                     &                   &        &                       &           &                       &           &                         &        &                             & &           &                       &           & &                          & &            \\[-0.7em]
        \lstick{$\ket{+}$} &                           & &                  &                   &        &                       &           &                       &           &                         &        & \ctrl{-9}                   & &           &                       &           & &                     &                   &        &                       &           &                       &           &                         &        &                             & &           &                       &           & &                          & & \meterD{X} \\
        \lstick{$\ket{+}$} & \labeledwire{nicered}{X}  & &                  &                   &        &                       &           &                       &           &                         &        &                             & &           &                       &           & &                     & \ctrl{-13}        &        &                       &           &                       &           &                         &        &                             & &           &                       &           & & \labeledwire{nicered}{X} & & \meterD{X} \\[-1.2em]
                           & \setwiretype{n}           & & \vdots           &                   &        &                       &           &                       &           &                         &        &                             & &           &                       &           & &                     &                   & \ddots &                       &           &                       &           &                         &        &                             & &           &                       &           & &                          & &            \\[-0.7em]
        \lstick{$\ket{+}$} &                           & &                  &                   &        &                       &           &                       &           &                         &        &                             & &           &                       &           & &                     &                   &        & \ctrl{-15}            &           &                       &           &                         &        &                             & &           &                       &           & &                          & & \meterD{X} \\
        \lstick{$\ket{+}$} &                           & &                  &                   &        &                       &           &                       &           &                         &        &                             & &           &                       &           & &                     &                   &        &                       &           &                       &           & \ctrl{-19}              &        &                             & &           &                       &           & &                          & & \meterD{X} \\[-1.2em]
                           & \setwiretype{n}           & & \vdots           &                   &        &                       &           &                       &           &                         &        &                             & &           &                       &           & &                     &                   &        &                       &           &                       &           &                         & \ddots &                             & &           &                       &           & &                          & &            \\[-0.7em]
        \lstick{$\ket{+}$} &                           & &                  &                   &        &                       &           &                       &           &                         &        &                             & &           &                       &           & &                     &                   &        &                       &           &                       &           &                         &        & \ctrl{-21}                  & &           &                       &           & &                          & & \meterD{X}
    \end{quantikz}
\end{align*}
Since this circuit is made only of a \texttt{CZ} network with qubits either initialized in the $\ket{+}$ state or uninitialized, and measurements in the $X$ basis, it defines an MBQC pattern $\mbqc(A,\mathcal{I},\mathcal{O},b=\bm{0})$. Using \cref{def:compressed-rep-cluster-state}, we get the compressed representation by adding $2n+m_X+m_Z$ virtual nodes to represent $X$ errors at the input and output wires of the circuit, as well as gauge operators corresponding to the $X$ and $Z$ stabilizers of the input qubits. However, each one of the $m_Z$ new gauge operators associated to an input $Z$ stabilizer is the same as the gauge operator associated to the first measurement of this $Z$ stabilizer, and can therefore be removed from the graph. The new graph corresponds exactly to the one described in the theorem statement.

The pure $X$-type detector drawn on the circuit above (corresponding to successive measurements of $S_1^Z$) is supported on the ancilla qubits used to measure $S_1^Z$ in layers one and three, and on the data qubits in round two.
This readily generalizes to the successive measurement of any other stabilizer $S_i^Z$ or $S_i^X$.

It remains to derive the input and output detectors described in the theorem statement.
Since the part of the original circuit representing the first round of $Z$ measurements has not been changed during the transformation, the following detector, corresponding to the backle of the first $S_1^Z$ measurement, remains a detector after the transformation:
\begin{align*}
    \tikzsetnextfilename{sec6-css-proof-4}
    \begin{quantikz}[row sep=10pt, column sep=6pt]
                           & \phantomgate{H} \labeledwire{nicered}{Z} & \gatebox{9}{7} & \mygate{3}{S_1^Z} &        & \mygate{3}{S_{m_Z}^Z} & \mygate{3}{S_1^X} &        & \mygate{3}{S_{m_X}^X} & \phantomgate{H} &                          &             \\[-1.5em]
                           & \setwiretype{n}                          & \vdots         &                   & \cdots &                       &                   & \cdots &                       &                 &                          &             \\[-1.1em]
                           & \phantomgate{H} \labeledwire{nicered}{Z} &                &                   &        &                       &                   &        &                       &                 &                          &             \\
        \lstick{$\ket{+}$} & \labeledwire{nicered}{X}                 &                & \ctrl{-1}         &        &                       &                   &        &                       &                 & \labeledwire{nicered}{X} &  \meterD{X} \\[-1.2em]
                           & \setwiretype{n}                          & \vdots         &                   & \ddots &                       &                   &        &                       &                 &                          &             \\[-0.7em]
        \lstick{$\ket{+}$} &                                          &                &                   &        & \ctrl{-3}             &                   &        &                       &                 &                          &  \meterD{X} \\
        \lstick{$\ket{+}$} &                                          &                &                   &        &                       & \ctrl{-4}         &        &                       &                 &                          &  \meterD{X} \\[-1.2em]
                           & \setwiretype{n}                          & \vdots         &                   &        &                       &                   & \ddots &                       &                 &                          &             \\[-0.7em]
        \lstick{$\ket{+}$} &                                          &                &                   &        &                       &                   &        & \ctrl{-6}             &                 &                          &  \meterD{X}
    \end{quantikz}
\end{align*}
where the $Z$ labels on the data qubits represent the support of $S_1^Z$.
This spacetime stabilizer becomes the detector connecting $a_1^{(1)}$ to $q^{(0)}_{j_1},\ldots,q^{(0)}_{j_k}$ in the compressed representation of the cluster state complex.
We can also derive the detectors corresponding to the backle of the first round of $X$-type stabilizer measurements.
For instance, for $S^1_X$ we draw the backle of the $X$ measurements of both the ancilla qubit used to measure $S^1_X$ and of the data qubits in the support of $S^1_X$ (due to the replacement of the Hadamard gates by their MBQC version):
\begin{align*}
    \tikzsetnextfilename{sec6-css-proof-5}
    \begin{quantikz}[row sep=10pt, column sep=2pt]
        \lstick{$\ket{+}$} & \phantomgate{H}           & & \gatebox{12}{10} &                   &        &                       &           &                       & \ctrl{3}  & \mygate{3}{\bar{S}_1^X} &        & \mygate{3}{\bar{S}_{m_X}^X} & \phantomgate{H} &                          & & \meterD{X} \\[-1.5em]
                           & \setwiretype{n}           & & \vdots           &                   &        &                       &           & \reflectbox{$\ddots$} &           &                         & \cdots &                             &                 &                          & &            \\[-1.1em]
        \lstick{$\ket{+}$} &                           & &                  &                   &        &                       & \ctrl{3}  &                       &           &                         &        &                             &                 &                          & & \meterD{X} \\
                           & \labeledwire{nicered}{X}  & &                  & \mygate{3}{S_1^Z} &        & \mygate{3}{S_{m_Z}^Z} &           &                       & \ctrl{-3} &                         &        &                             &                 & \labeledwire{nicered}{X} & & \meterD{X} \\[-1.5em]
                           & \setwiretype{n}           & & \vdots           &                   & \cdots &                       &           & \reflectbox{$\ddots$} &           &                         &        &                             &                 &                          & &            \\[-1.1em]
                           & \labeledwire{nicered}{X}  & &                  &                   &        &                       & \ctrl{-3} &                       &           &                         &        &                             &                 & \labeledwire{nicered}{X} & & \meterD{X} \\
        \lstick{$\ket{+}$} &                           & &                  & \ctrl{-1}         &        &                       &           &                       &           &                         &        &                             &                 &                          & & \meterD{X} \\[-1.2em]
                           & \setwiretype{n}           & & \vdots           &                   & \ddots &                       &           &                       &           &                         &        &                             &                 &                          & &            \\[-0.7em]
        \lstick{$\ket{+}$} &                           & &                  &                   &        & \ctrl{-3}             &           &                       &           &                         &        &                             &                 &                          & & \meterD{X} \\
        \lstick{$\ket{+}$} & \labeledwire{nicered}{X}  & &                  &                   &        &                       &           &                       &           & \ctrl{-7}               &        &                             &                 & \labeledwire{nicered}{X} & & \meterD{X} \\[-1.2em]
                           & \setwiretype{n}           & & \vdots           &                   &        &                       &           &                       &           &                         & \ddots &                             &                 &                          & &            \\[-0.7em]
        \lstick{$\ket{+}$} &                           & &                  &                   &        &                       &           &                       &           &                         &        & \ctrl{-9}                   &                 &                          & & \meterD{X}
    \end{quantikz}
\end{align*}
where the $X$ labels on the data qubits represent the support of $\overline{S_1^X}$.
They allow the $Z$s to cancel on the top wires.
This gives us the detector connecting $a^{(2)}_1$ to $q^{(1)}_{j_1},\ldots,q^{(1)}_{j_k}$ in the compressed representation of the cluster state complex.

Let us now derive the detectors corresponding to the spackle of $X$ and $Z$ measurements at the last round.
The spacetime stabilizer corresponding to the last measurement of $S_1^Z$ is as follows:
\begin{align*}
    \tikzsetnextfilename{sec6-css-proof-6}
    \begin{quantikz}[row sep=10pt, column sep=2pt]
        \lstick{$\ket{+}$} &                                          & & \gatebox{15}{13} \push{\hspace{0.5em}\cdots\hspace{0.5em}} &                   &        &                       &           &                       &           &                         &        &                             &           &                       & \ctrl{3}  & \phantomgate{H} & \labeledwire{nicered}{Z} & &            \\[-1.2em]
                           & \setwiretype{n}                          & & \vdots                                                     &                   &        &                       &           &                       &           &                         &        &                             &           & \reflectbox{$\ddots$} &           &                 &                          & &            \\[-0.4em]
        \lstick{$\ket{+}$} &                                          & & \push{\hspace{0.5em}\cdots\hspace{0.5em}}                  &                   &        &                       &           &                       &           &                         &        &                             & \ctrl{3}  &                       &           &                 & \labeledwire{nicered}{Z} & &            \\
        \lstick{$\ket{+}$} & \phantomgate{H} \labeledwire{nicered}{X} & & \push{\hspace{0.5em}\cdots\hspace{0.5em}}                  &                   &        &                       &           &                       & \ctrl{3}  & \mygate{3}{\bar{S}_1^X} &        & \mygate{3}{\bar{S}_{m_X}^X} &           &                       & \ctrl{-3} &                 & \labeledwire{nicered}{X} & & \meterD{X} \\[-1.5em]
                           & \setwiretype{n}                          & & \vdots                                                     &                   &        &                       &           & \reflectbox{$\ddots$} &           &                         & \cdots &                             &           & \reflectbox{$\ddots$} &           &                 &                          & &            \\[-1.1em]
        \lstick{$\ket{+}$} & \phantomgate{H} \labeledwire{nicered}{X} & & \push{\hspace{0.5em}\cdots\hspace{0.5em}}                  &                   &        &                       & \ctrl{3}  &                       &           &                         &        &                             & \ctrl{-3} &                       &           &                 & \labeledwire{nicered}{X} & & \meterD{X} \\
        \lstick{$\ket{+}$} &                                          & & \push{\hspace{0.5em}\cdots\hspace{0.5em}}                  & \mygate{3}{S_1^Z} &        & \mygate{3}{S_{m_Z}^Z} &           &                       & \ctrl{-3} &                         &        &                             &           &                       &           &                 &                          & & \meterD{X} \\[-1.5em]
                           & \setwiretype{n}                          & & \vdots                                                     &                   & \cdots &                       &           & \reflectbox{$\ddots$} &           &                         &        &                             &           &                       &           &                 &                          & &            \\[-1.1em]
        \lstick{$\ket{+}$} &                                          & & \push{\hspace{0.5em}\cdots\hspace{0.5em}}                  &                   &        &                       & \ctrl{-3} &                       &           &                         &        &                             &           &                       &           &                 &                          & & \meterD{X} \\
        \lstick{$\ket{+}$} & \phantomgate{H} \labeledwire{nicered}{X} & &                                                            & \ctrl{-1}         &        &                       &           &                       &           &                         &        &                             &           &                       &           &                 & \labeledwire{nicered}{X} & & \meterD{X} \\[-1.2em]
                           & \setwiretype{n}                          & & \vdots                                                     &                   & \ddots &                       &           &                       &           &                         &        &                             &           &                       &           &                 &                          & &            \\[-0.7em]
        \lstick{$\ket{+}$} &                                          & &                                                            &                   &        & \ctrl{-3}             &           &                       &           &                         &        &                             &           &                       &           &                 &                          & & \meterD{X} \\
        \lstick{$\ket{+}$} &                                          & &                                                            &                   &        &                       &           &                       &           & \ctrl{-7}               &        &                             &           &                       &           &                 &                          & & \meterD{X} \\[-1.2em]
                           & \setwiretype{n}                          & & \vdots                                                     &                   &        &                       &           &                       &           &                         & \ddots &                             &           &                       &           &                 &                          & &            \\[-0.7em]
        \lstick{$\ket{+}$} &                                          & &                                                            &                   &        &                       &           &                       &           &                         &        & \ctrl{-9}                   &           &                       &           &                 &                          & & \meterD{X}
    \end{quantikz}
\end{align*}
where the $X$ and $Z$ labels on the data qubits both represent the support of $S_1^Z$.
This spacetime stabilizer becomes a detector connecting $a^{(2T-1)}_1$ with $q^{(2T)}_{j_1},\ldots,q^{(2T)}_{j_k}$ and $q^{(2T+2)}_{j_1},\ldots,q^{(2T+2)}_{j_k}$ (due to the $Z$'s on the last layer of data qubits) in the compressed representation of the cluster state complex.
Similarly, the spacetime stabilizer corresponding to the final measurement of $S_1^X$ is as follows:
\begin{align*}
    \tikzsetnextfilename{sec6-css-proof-7}
    \begin{quantikz}[row sep=10pt, column sep=2pt]
        \lstick{$\ket{+}$} & \phantomgate{H} \labeledwire{nicered}{X} & & \gatebox{15}{13} \push{\hspace{0.5em}\cdots\hspace{0.5em}} &                   &        &                       &           &                       &           &                         &        &                             &           &                       & \ctrl{3}  & \phantomgate{H} & \labeledwire{nicered}{X} & &            \\[-1.2em]
                           & \setwiretype{n}                          & & \vdots                                                     &                   &        &                       &           &                       &           &                         &        &                             &           & \reflectbox{$\ddots$} &           &                 &                          & &            \\[-0.4em]
        \lstick{$\ket{+}$} & \phantomgate{H} \labeledwire{nicered}{X} & & \push{\hspace{0.5em}\cdots\hspace{0.5em}}                  &                   &        &                       &           &                       &           &                         &        &                             & \ctrl{3}  &                       &           &                 & \labeledwire{nicered}{X} & &            \\
        \lstick{$\ket{+}$} &                                          & & \push{\hspace{0.5em}\cdots\hspace{0.5em}}                  &                   &        &                       &           &                       & \ctrl{3}  & \mygate{3}{\bar{S}_1^X} &        & \mygate{3}{\bar{S}_{m_X}^X} &           &                       & \ctrl{-3} &                 &                          & & \meterD{X} \\[-1.5em]
                           & \setwiretype{n}                          & & \vdots                                                     &                   &        &                       &           & \reflectbox{$\ddots$} &           &                         & \cdots &                             &           & \reflectbox{$\ddots$} &           &                 &                          & &            \\[-1.1em]
        \lstick{$\ket{+}$} &                                          & & \push{\hspace{0.5em}\cdots\hspace{0.5em}}                  &                   &        &                       & \ctrl{3}  &                       &           &                         &        &                             & \ctrl{-3} &                       &           &                 &                          & & \meterD{X} \\
        \lstick{$\ket{+}$} &                                          & & \push{\hspace{0.5em}\cdots\hspace{0.5em}}                  & \mygate{3}{S_1^Z} &        & \mygate{3}{S_{m_Z}^Z} &           &                       & \ctrl{-3} &                         &        &                             &           &                       &           &                 &                          & & \meterD{X} \\[-1.5em]
                           & \setwiretype{n}                          & & \vdots                                                     &                   & \cdots &                       &           & \reflectbox{$\ddots$} &           &                         &        &                             &           &                       &           &                 &                          & &            \\[-1.1em]
        \lstick{$\ket{+}$} &                                          & & \push{\hspace{0.5em}\cdots\hspace{0.5em}}                  &                   &        &                       & \ctrl{-3} &                       &           &                         &        &                             &           &                       &           &                 &                          & & \meterD{X} \\
        \lstick{$\ket{+}$} &                                          & &                                                            & \ctrl{-1}         &        &                       &           &                       &           &                         &        &                             &           &                       &           &                 &                          & & \meterD{X} \\[-1.2em]
                           & \setwiretype{n}                          & & \vdots                                                     &                   & \ddots &                       &           &                       &           &                         &        &                             &           &                       &           &                 &                          & &            \\[-0.7em]
        \lstick{$\ket{+}$} &                                          & &                                                            &                   &        & \ctrl{-3}             &           &                       &           &                         &        &                             &           &                       &           &                 &                          & & \meterD{X} \\
        \lstick{$\ket{+}$} & \phantomgate{H} \labeledwire{nicered}{X} & &                                                            &                   &        &                       &           &                       &           & \ctrl{-7}               &        &                             &           &                       &           &                 & \labeledwire{nicered}{X} & & \meterD{X} \\[-1.2em]
                           & \setwiretype{n}                          & & \vdots                                                     &                   &        &                       &           &                       &           &                         & \ddots &                             &           &                       &           &                 &                          & &            \\[-0.7em]
        \lstick{$\ket{+}$} &                                          & &                                                            &                   &        &                       &           &                       &           &                         &        & \ctrl{-9}                   &           &                       &           &                 &                          & & \meterD{X}
    \end{quantikz}
\end{align*}
and turns into a detector connecting $a^{(2T)}_1$ with $q^{(2T+1)}_{j_1},\ldots,q^{(2T+1)}_{j_k}$ in the compressed representation of the cluster state complex.

Finally, the detectors connected to the virtual ancilla nodes come from the construction of the augmented cluster state complex (the identity part of $\tilde{H}$) in \cref{def:compressed-rep-cluster-state}.
\end{proof}

\subsection{General stabilizer codes}

\begin{figure}
    \centering
    \subfloat[]{
        \includegraphics[width=0.17\textwidth]{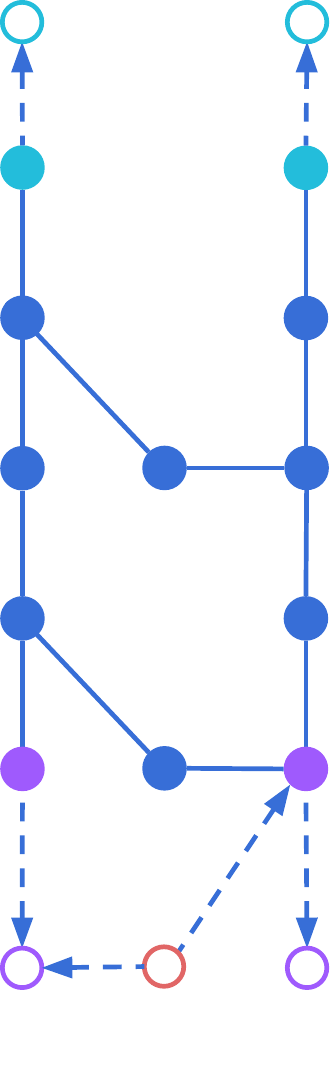}
    }
    \hspace{2em}
    \subfloat[]{
        \includegraphics[width=0.17\textwidth]{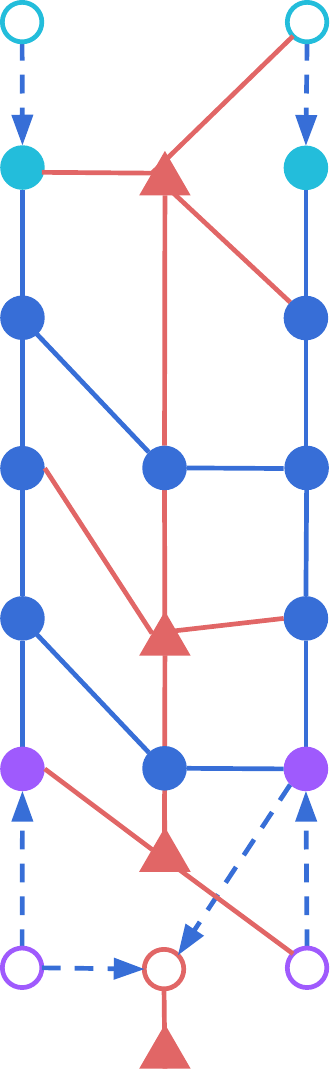}
    }
    \hspace{2em}
    \subfloat[]{
        \includegraphics[width=0.17\textwidth]{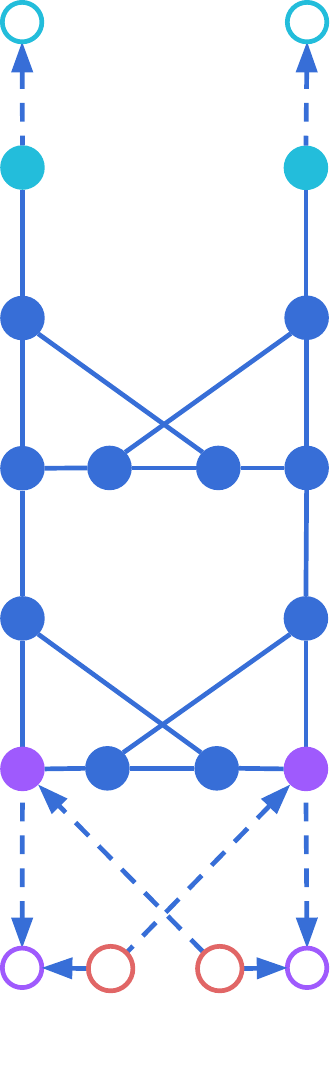}
    }
    \hspace{2em}
    \subfloat[]{
        \includegraphics[width=0.17\textwidth]{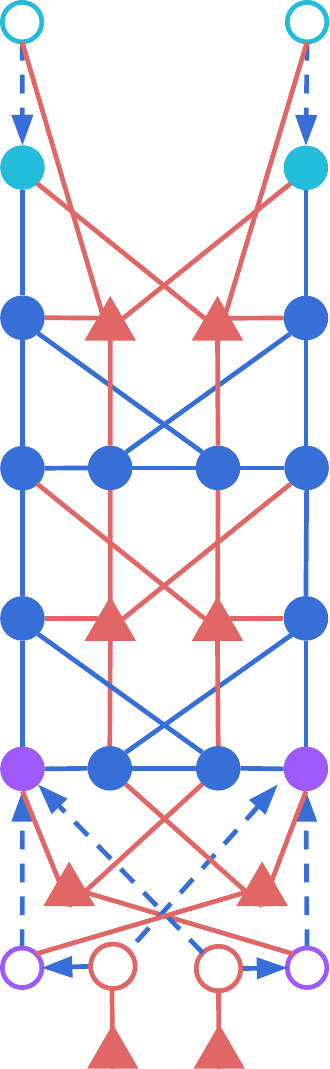}
    }
    \caption{ \label{fig:non-css-cluster-states}
        Cluster state complexes of two foliated non-CSS codes, with stabilizer groups $\langle XZ \rangle$ (a and b) and $\langle XZ, ZX \rangle$ (c and d). \textbf{(a)} and \textbf{(c)} Compressed representation of the cluster state complex, without detector nodes for clarity. \textbf{(b)} and \textbf{(d)} Compressed co-representation of the cluster state complex, with detector nodes. In the second complex, the two ancilla nodes at each layer are connected due to the odd intersection between the $X$-part of $XZ$ (i.e. $X_1$) and the $Z$-part of $ZX$ (i.e. $Z_1$)
    }
\end{figure}

We derive here the cluster state complex for a non-CSS code phenomenological memory experiment. The corresponding MBQC circuits are the same as the one described in \cite{BrownUniversalFTMBQC2020}. Examples of such cluster state complexes are given in \cref{fig:non-css-cluster-states}.

\begin{theorem}
Consider the circuit made of $T$ rounds of stabilizer measurements of a general $[[n,k,d]]$ stabilizer code, where $m$ stabilizer generators $\{S_1,\ldots,S_m\}$ are measured each round using some ancilla qubits, and errors are assumed to only happen in-between successive rounds.
Let $H$ be the parity-check matrices associated to $S_1,\ldots, S_m$, written in the binary symplectic format. We write $H=(H^X|H^Z)$ to distinguish the $X$ part and the $Z$ part of the parity-check matrix, and $S_i=(S_i^X | S_j^Z)$ to distinguish the $X$ and $Z$ parts of each stabilizer $S_i$.
The circuit is equivalent to a cluster state complex described by the following graph (in its compressed representation):
\begin{itemize}
    \item It contains $2T+3$ layers $\{0,\ldots,2T+2 \}$, where each layer $\ell$ contains $n$ data nodes $q_1^{(\ell)},\ldots,q_n^{(\ell)}$, and each odd layer $\ell \leq 2T-1$ contains $m$ ancilla nodes $a_1^{(\ell)},\ldots,a_m^{(\ell)}$. The nodes on layers $0$ and $2T+2$ are virtual, while all the other nodes are physical.
    \item For any odd layer $\ell \leq 2T-1$, the ancilla node $a_i^{(\ell)}$ is connected to the data nodes at the same level according to $H_Z$ and at a level above according to $H_X$, that is, $a_i^{(\ell)}$ is connected to $q_j^{(\ell)}$ if and only if $H^Z_{ij}=1$ and to $q_j^{(\ell+1)}$ if and only if $H^X_{ij}=1$.
    \item The ancilla nodes $a_i^{(\ell)}$ and $a_j^{(\ell)}$ are connected if and only if $S_i^Z$ and $S_j^X$ have an intersection of odd weight, i.e. $|S_i^Z \cdot S_j^X|=1$ modulo two. Note that for $i=j$, it means there is a self-loop on the node $a_i^{(\ell)}$ if and only if $S_i$ has an odd number of $Y$.
    For $1 \leq \ell \leq 2T$, every data node $q_i^{(\ell)}$ is connected to the data node $q_i^{(\ell+1)}$ in the next layer. For $\ell = 0$, there is directed edge from $q_i^{(\ell+1)}$ to $q_i^{(\ell)}$. For $\ell = 2T+1$, there is directed edge from $q_i^{(\ell)}$ and $q_i^{(\ell+1)}$.
    \item For every ancilla node $a_i^{(\ell)}$ ($1 \leq \ell \leq 2T$, $\ell$ odd) connected to a set of data node $q^{(\ell_1)}_{j_1},\ldots,q^{(\ell_k)}_{j_k}$ (with $\ell_1,\ldots,\ell_k \in \{\ell,\ell+1\}$ and $j_1,\ldots,j_k \in \{1,\ldots,n\}$), we have a detector node connected to $a_i^{(\ell)}$, $q^{(\ell_1+1)}_{j_1},\ldots,q^{(\ell_k+1)}_{j_k}$, $a_i^{(\ell+2)}$ if $\ell \leq 2T-2$, and $q^{(2T+2)}_{j_i}$ for all $i$ such that $\ell_i=2T$.
    These detectors correspond to spackles in the spacetime code of the circuit.
    Corresponding to the backles, we also have one detector node per ancilla node $a_i^{(1)}$ on the first layer, connected to $a_i^{(1)}$ and $q^{(\ell_1-1)}_{j_1},\ldots,q^{(\ell_k-1)}_{j_k}$.
    Finally, we have one detector node per virtual ancilla node, connected only to it.
\end{itemize}
\end{theorem}

\begin{proof}
One round of stabilizer measurements for a non-CSS code can be represented by the following circuit:
\begin{align*}
    \tikzsetnextfilename{sec6-general-proof-1}
    \begin{quantikz}[row sep=10pt, column sep=6pt]
                           & \phantomgate{H} & \gatebox{6}{4}  & \mygate{3}{S_1}   &        & \mygate{3}{S_{m}} &  &            \\[-1.5em]
                           & \setwiretype{n} & \vdots          &                   & \cdots &                   &  &            \\[-1.1em]
                           &                 &                 &                   &        &                   &  &            \\
        \lstick{$\ket{+}$} &                 &                 & \ctrl{-1}         &        &                   &  & \meterD{M_1} \\[-1.2em]
                           & \setwiretype{n} & \vdots          &                   & \ddots &                   &  &            \\[-0.7em]
        \lstick{$\ket{+}$} &                 &                 &                   &        & \ctrl{-3}         &  & \meterD{M_m}
    \end{quantikz}
\end{align*}
where $M_i=X$ if there is an even number of $Y$ in $S_i$, and $M_i=Y$ otherwise.
This can be understood by noticing that for every \texttt{CY} in the decomposition of control-$S_i$, the initial $X$ stabilizer of the $\ket{+}$ state propagates as follows:
\begin{align}
    \tikzsetnextfilename{sec6-general-proof-2}
    \begin{quantikz}
                           &                           & \targ{}   & \ctrl{1}  & \\
        \lstick{$\ket{+}$} & \labeledwire{nicered}{X}  & \ctrl{-1} & \ctrl{-1} &
    \end{quantikz}
    \;\;\; \longrightarrow \;\;\;
    \tikzsetnextfilename{sec6-general-proof-3}
    \begin{quantikz}
                           & \targ{}   & \ctrl{1}  & \labeledwire{nicered}{Y} \\
        \lstick{$\ket{+}$} & \ctrl{-1} & \ctrl{-1} & \labeledwire{nicered}{Y}
    \end{quantikz}
\end{align}
Therefore, we need to measure the ancilla qubit in the $Y$ basis in order to measure $Y$ on the top qubit.
Since every \texttt{CY} contributes to a $Z$ on the ancilla qubit when writing down the propagation of the initial $X$ operator, we only measure it in the $Y$ basis when the number of \texttt{CY} is odd.

Decomposing every control Pauli gate into its $X$ and $Z$ parts gives us:
\begin{align*}
    \tikzsetnextfilename{sec6-general-proof-4}
    \begin{quantikz}[row sep=10pt, column sep=6pt]
                           & \phantomgate{H} & \gatebox{6}{6}  & \mygate{3}{S^Z_1} & \mygate{3}{S^X_1} &        & \mygate{3}{S^Z_{m}} & \mygate{3}{S^X_{m}} &  &            \\[-1.5em]
                           & \setwiretype{n} & \vdots          &                   &                   & \cdots &                     &                     &  &            \\[-1.1em]
                           &                 &                 &                   &                   &        &                     &                     &  &            \\
        \lstick{$\ket{+}$} &                 &                 & \ctrl{-1}         & \ctrl{-1}         &        &                     &                     &  & \meterD{M_1} \\[-1.2em]
                           & \setwiretype{n} & \vdots          &                   &                   & \ddots &                     &                     &  &            \\[-0.7em]
        \lstick{$\ket{+}$} &                 &                 &                   &                   &        & \ctrl{-3}           & \ctrl{-3}           &  & \meterD{M_m}
    \end{quantikz}
\end{align*}
We might now be tempted to turn every control $S^X_i$ into \texttt{CZ} and \texttt{H}, in order to then compile the \texttt{H} gates into their measurement-based version.
However, the presence of multiple layers of \texttt{H} gates within a blue box would prevent us from using \cref{lemma:h-cz-sandwich} to compile the \texttt{H} gates into \texttt{CZ} gates.
The key is to use the circuit identity
\begin{align} \label{eq:cx-cz-identity}
    \tikzsetnextfilename{sec6-general-proof-5}
    \begin{quantikz}
        & \targ{}   & \ctrl{2}  &  &  \\
        & \ctrl{-1} &           &  & \\
        &           & \ctrl{-2} &  &
    \end{quantikz}
    =
    \tikzsetnextfilename{sec6-general-proof-6}
    \begin{quantikz}
        & \ctrl{2}  & \targ{}   &           &  \\
        &           & \ctrl{-1} & \ctrl{1}  & \\
        & \ctrl{-2} &           & \ctrl{-1} &
    \end{quantikz}
\end{align}
which allows to push all the \texttt{CZ}s to the left and all the \texttt{CNOT}s to the right, at the cost of adding new connections between ancilla qubits.
We can use this identity to commute the control $S^X_i$ and $S^Z_{i+1}$ gates, inserting a \texttt{CZ} between their control qubits if $S^X_i$ and $S^Z_{i+1}$ intersect non-trivially on an odd number of qubits. For instance, if $S^X_1$ and $S^Z_2$ have an odd intersection, we obtain the following circuit:
\begin{align*}
    \tikzsetnextfilename{sec6-general-proof-7}
    \begin{quantikz}[row sep=11pt, column sep=6pt]
                           & \phantomgate{H} & \gatebox{7}{8} & \mygate{3}{S^Z_1} & \mygate{3}{S^Z_2} & \mygate{3}{S^X_1} & \mygate{3}{S^X_2} &        & \mygate{3}{S^Z_{m}} & \mygate{3}{S^X_{m}} &  &              \\[-1.5em]
                           & \setwiretype{n} & \vdots         &                   &                   &                   &                   & \cdots &                     &                     &  &              \\[-1.1em]
                           &                 &                &                   &                   &                   &                   &        &                     &                     &  &              \\
        \lstick{$\ket{+}$} &                 & \ctrl{1}       & \ctrl{-1}         &                   & \ctrl{-1}         &                   &        &                     &                     &  & \meterD{M_1} \\[-0.5em]
        \lstick{$\ket{+}$} &                 & \ctrl{-1}      &                   & \ctrl{-2}         &                   & \ctrl{-2}         &        &                     &                     &  & \meterD{M_2} \\[-1.2em]
                           & \setwiretype{n} & \vdots         &                   &                   &                   &                   & \ddots &                     &                     &  &              \\[-0.7em]
        \lstick{$\ket{+}$} &                 &                &                   &                   &                   &                   &        & \ctrl{-4}           & \ctrl{-4}           &  & \meterD{M_m}
    \end{quantikz}
\end{align*}
Repeating this commuting process for all the stabilizers, our circuit takes the form:
\begin{align*}
    \tikzsetnextfilename{sec6-general-proof-8}
    \begin{quantikz}[row sep=11pt, column sep=6pt]
                           & \phantomgate{H} & \gatebox{6}{8} & \mygate{3}{S^Z_1} &        & \mygate{3}{S^Z_{m}} & \mygate{3}{S^X_{1}} &        & \mygate{3}{S^X_{m}} &                  & &              \\[-1.5em]
                           & \setwiretype{n} & \vdots         &                   & \cdots &                     &                     & \cdots &                     &                  & &              \\[-1.1em]
                           &                 &                &                   &        &                     &                     &        &                     &                  & &              \\
        \lstick{$\ket{+}$} &                 &                & \ctrl{-1}         &        &                     & \ctrl{-1}           &        &                     & \mygate{3}{CZ_A} & & \meterD{M_1} \\[-1.2em]
                           & \setwiretype{n} & \vdots         &                   & \ddots &                     &                     & \ddots &                     &                  & &              \\[-0.7em]
        \lstick{$\ket{+}$} &                 &                &                   &        & \ctrl{-3}           &                     &        & \ctrl{-3}           &                  & & \meterD{M_m}
    \end{quantikz}
\end{align*}
where $CZ_A$ is a network of \texttt{CZ} gates with adjacency matrix $A$, defined by $A_{ij}=1$ if and only if $S^Z_i$ and $S^X_j$ have an odd intersection.
We can now compile the control $S^X_i$ into \texttt{H} and \texttt{CZ}
\begin{align*}
    \tikzsetnextfilename{sec6-general-proof-9}
    \begin{quantikz}[row sep=11pt, column sep=6pt]
                           & \phantomgate{H} & \gatebox{6}{9} & \mygate{3}{S^Z_1} &        & \mygate{3}{S^Z_{m}} & \mygate{1}{H} & \mygate{3}{\bar{S}^X_{1}} &        & \mygate{3}{\bar{S}^X_{m}} & \mygate{1}{H}    &  &              \\[-1.5em]
                           & \setwiretype{n} & \vdots         &                   & \cdots &                     &               &                           & \cdots &                           &                  &  &              \\[-1.1em]
                           &                 &                &                   &        &                     & \mygate{1}{H} &                           &        &                           & \mygate{1}{H}    &  &              \\
        \lstick{$\ket{+}$} &                 &                & \ctrl{-1}         &        &                     &               & \ctrl{-1}                 &        &                           & \mygate{3}{CZ_A} &  & \meterD{M_1} \\[-1.2em]
                           & \setwiretype{n} & \vdots         &                   & \ddots &                     &               &                           & \ddots &                           &                  &  &              \\[-0.7em]
        \lstick{$\ket{+}$} &                 &                &                   &        & \ctrl{-3}           &               &                           &        & \ctrl{-3}                 &                  &  & \meterD{M_m}
    \end{quantikz}
\end{align*}
where $\bar{S}_i^X$ is the operator built from $S_i^X$ by turning every $X$ into a $Z$.
Next we compile the \texttt{H} gates into their MBQC version using \cref{lemma:h-cz-sandwich}.
Over two rounds of measurements, this gives the following circuit, where we also annotate an example detector, associated to successive measurements of $S_1$:
\begin{align*}
    \tikzsetnextfilename{sec6-general-proof-10}
    \begin{quantikz}[row sep=10pt, column sep=2pt]
        \lstick{$\ket{+}$} & \phantomgate{H}           & & \gatebox{21}{29} &                   &        &                     &           &                       &           &                         &        &                             & &           &                       &           & &                     &                   &        &                     &           &                       &           &                         &        &                             & &           &                       & \ctrl{3}  &                 &                            &              \\[-1.2em]
                           & \setwiretype{n}           & & \vdots           &                   &        &                     &           &                       &           &                         &        &                             & &           &                       &           & &                     &                   &        &                     &           &                       &           &                         &        &                             & &           & \reflectbox{$\ddots$} &           &                 &                            &              \\[-0.7em]
        \lstick{$\ket{+}$} &                           & &                  &                   &        &                     &           &                       &           &                         &        &                             & &           &                       &           & &                     &                   &        &                     &           &                       &           &                         &        &                             & & \ctrl{3}  &                       &           &                 &                            &              \\
        \lstick{$\ket{+}$} &                           & &                  &                   &        &                     &           &                       &           &                         &        &                             & &           &                       &           & &                     &                   &        &                     &           &                       & \ctrl{3}  & \mygate{3}{\bar{S}_1^X} &        & \mygate{3}{\bar{S}_{m}^X}   & &           &                       & \ctrl{-3} & \phantomgate{H} & \phantomgate{H}            & \meterD{X}   \\[-1.5em]
                           & \setwiretype{n}           & & \vdots           &                   &        &                     &           &                       &           &                         &        &                             & &           &                       &           & &                     &                   &        &                     &           & \reflectbox{$\ddots$} &           &                         & \cdots &                             & &           & \reflectbox{$\ddots$} &           &                 &                            &              \\[-1.1em]
        \lstick{$\ket{+}$} &                           & &                  &                   &        &                     &           &                       &           &                         &        &                             & &           &                       &           & &                     &                   &        &                     & \ctrl{3}  &                       &           &                         &        &                             & & \ctrl{-3} &                       &           &                 &                            & \meterD{X}   \\
        \lstick{$\ket{+}$} & \labeledwire{gold}{X}     & &                  &                   &        &                     &           &                       &           &                         &        &                             & &           &                       & \ctrl{3}  & &                     & \mygate{3}{S_1^Z} &        & \mygate{3}{S_{m}^Z} &           &                       & \ctrl{-3} &                         &        &                             & &           &                       &           &                 &                            & \meterD{X}   \\[-1.5em]
                           & \setwiretype{n}           & & \vdots           &                   &        &                     &           &                       &           &                         &        &                             & &           & \reflectbox{$\ddots$} &           & &                     &                   & \cdots &                     &           & \reflectbox{$\ddots$} &           &                         &        &                             & &           &                       &           &                 &                            &              \\[-1.1em]
        \lstick{$\ket{+}$} & \labeledwire{gold}{X}     & &                  &                   &        &                     &           &                       &           &                         &        &                             & & \ctrl{3}  &                       &           & &                     &                   &        &                     & \ctrl{-3} &                       &           &                         &        &                             & &           &                       &           &                 &                            & \meterD{X}   \\
        \lstick{$\ket{+}$} & \labeledwire{niceblue}{X} & &                  &                   &        &                     &           &                       & \ctrl{3}  & \mygate{3}{\bar{S}_1^X} &        & \mygate{3}{\bar{S}_{m}^X}   & &           &                       & \ctrl{-3} & &                     &                   &        &                     &           &                       &           &                         &        &                             & &           &                       &           &                 &                            & \meterD{X}   \\[-1.5em]
                           & \setwiretype{n}           & & \vdots           &                   &        &                     &           & \reflectbox{$\ddots$} &           &                         & \cdots &                             & &           & \reflectbox{$\ddots$} &           & &                     &                   &        &                     &           &                       &           &                         &        &                             & &           &                       &           &                 &                            &              \\[-1.1em]
        \lstick{$\ket{+}$} & \labeledwire{niceblue}{X} & &                  &                   &        &                     & \ctrl{3}  &                       &           &                         &        &                             & & \ctrl{-3} &                       &           & &                     &                   &        &                     &           &                       &           &                         &        &                             & &           &                       &           &                 &                            & \meterD{X}   \\
                           &                           & &                  & \mygate{3}{S_1^Z} &        & \mygate{3}{S_{m}^Z} &           &                       & \ctrl{-3} &                         &        &                             & &           &                       &           & &                     &                   &        &                     &           &                       &           &                         &        &                             & &           &                       &           &                 &                            & \meterD{X}   \\[-1.5em]
                           & \setwiretype{n}           & & \vdots           &                   & \cdots &                     &           & \reflectbox{$\ddots$} &           &                         &        &                             & &           &                       &           & &                     &                   &        &                     &           &                       &           &                         &        &                             & &           &                       &           &                 &                            &              \\[-1.1em]
                           &                           & &                  &                   &        &                     & \ctrl{-3} &                       &           &                         &        &                             & &           &                       &           & &                     &                   &        &                     &           &                       &           &                         &        &                             & &           &                       &           &                 &                            & \meterD{X}   \\
        \lstick{$\ket{+}$} & \labeledwire{nicered}{X}  & &                  & \ctrl{-1}         &        &                     &           &                       &           & \ctrl{-4}               &        &                             & &           &                       &           & &                     &                   &        &                     &           &                       &           &                         &        &                             & &           & \mygate{3}{CZ_A}      &           &                 & \labeledwire{nicered}{M_1} & \meterD{M_1} \\[-1.2em]
                           & \setwiretype{n}           & & \vdots           &                   & \ddots &                     &           &                       &           &                         & \ddots &                             & &           &                       &           & &                     &                   &        &                     &           &                       &           &                         &        &                             & &           &                       &           &                 &                            &              \\[-0.7em]
        \lstick{$\ket{+}$} &                           & &                  &                   &        & \ctrl{-3}           &           &                       &           &                         &        & \ctrl{-6}                   & &           &                       &           & &                     &                   &        &                     &           &                       &           &                         &        &                             & &           &                       &           &                 &                            & \meterD{M_m} \\
        \lstick{$\ket{+}$} & \labeledwire{nicered}{X}  & &                  &                   &        &                     &           &                       &           &                         &        &                             & &           &                       &           & &                     & \ctrl{-10}        &        &                     &           &                       &           & \ctrl{-13}              &        &                             & &           & \mygate{3}{CZ_A}      &           &                 & \labeledwire{nicered}{M_1} & \meterD{M_1} \\[-1.2em]
                           & \setwiretype{n}           & & \vdots           &                   &        &                     &           &                       &           &                         &        &                             & &           &                       &           & &                     &                   & \ddots &                     &           &                       &           &                         & \ddots &                             & &           &                       &           &                 &                            &              \\[-0.7em]
        \lstick{$\ket{+}$} &                           & &                  &                   &        &                     &           &                       &           &                         &        &                             & &           &                       &           & &                     &                   &        & \ctrl{-12}          &           &                       &           &                         &        & \ctrl{-15}                  & &           &                       &           &                 &                            & \meterD{M_m}
    \end{quantikz}
\end{align*}
where we label in blue the Pauli operators acting on the support of $S_1^Z$ and in yellow those acting on the support of $S_1^X$.
The initial $X$ operators on the ancilla wires (in red) propagate only to $M_1$ on the same wires due to the offsetting effect of the propagation of the blue $X$ operators (whenever $S_1^Z$ anticommutes with an $S_i^X$) and the $CZ_A$ gate.
Moreover, it propagates to $M_1$ and not $X$ due to the possibility of $S_1^Z$ anticommuting with $S_1^X$, thereby propagating a $Z$ on this wire.

We can see that every round of measurements contributes to two teleportations of all the data qubits, and therefore to $2T+1$ layers in the cluster state complex. Layers $\ell=0$ and $\ell=2T+2$ are virtual layers representing $X$ errors on the input and output qubits.
Moreover, every round of stabilizer measurements gives rise to $m$ ancilla qubits, which we place on the odd layers of the cluster state complex by convention. Finally, the \texttt{CZ} network of this circuit has the same graph as the one described in the theorem statements.

Moving on to the detector nodes, we can see that the detector corresponding to the spacetime stabilizer displayed above (redundancy in two successive measurements of $S_1$), is connected to the ancilla nodes $a^{(\ell)}_1$ and $a^{(\ell+2)}_1$.
Furthermore, this detector node is also connected the data qubits $q^{(\ell+1)}_{j_1},\ldots,q^{(\ell+1)}_{j_\alpha}$ and $q^{(\ell+2)}_{k_1},\ldots,q^{(\ell+2)}_{k_\alpha}$, where $q^{(t)}_{j_i}$ ($q^{(t)}_{k_i}$) are again the data qubits in the support of $S^Z_1$ ($S^X_i$) at the layer $t$.
In other words, the detector node is connected to the data qubits one layer above those connecting the ancilla node, as stated in the theorem.

It remains to study the detector nodes corresponding to the input and output spacetime stabilizers.
The detector corresponding to the last $S_1$ measurement has the following form in the circuit picture:
\begin{align*}
    \tikzsetnextfilename{sec6-general-proof-11}
    \begin{quantikz}[row sep=10pt, column sep=2pt]
        \lstick{$\ket{+}$} & \phantomgate{H} \labeledwire{gold}{X}     & & \gatebox{12}{14} \push{\hspace{0.5em}\cdots\hspace{0.5em}} &                   &        &                     &           &                       &           &                         &        &                             & &           &                       & \ctrl{3}  & \phantomgate{H} & \labeledwire{gold}{X}       & \phantomgate{H} & \labeledwire{niceblue}{Z} & & \\[-1.5em]
                           & \setwiretype{n}                           & & \vdots                                                     &                   &        &                     &           &                       &           &                         &        &                             & &           & \reflectbox{$\ddots$} &           &                 &                             &                 &                           & &            \\[-1.1em]
        \lstick{$\ket{+}$} & \phantomgate{H} \labeledwire{gold}{X}     & & \push{\hspace{0.5em}\cdots\hspace{0.5em}}                  &                   &        &                     &           &                       &           &                         &        &                             & & \ctrl{3}  &                       &           &                 & \labeledwire{gold}{X}       &                 & \labeledwire{niceblue}{Z} & & \\
        \lstick{$\ket{+}$} & \phantomgate{H} \labeledwire{niceblue}{X} & & \push{\hspace{0.5em}\cdots\hspace{0.5em}}                  &                   &        &                     &           &                       & \ctrl{3}  & \mygate{3}{\bar{S}_1^X} &        & \mygate{3}{\bar{S}_{m}^X}   & &           &                       & \ctrl{-3} &                 & \labeledwire{niceblue}{X}   &                 &                           & & \meterD{X} \\[-1.5em]
                           & \setwiretype{n}                           & & \vdots                                                     &                   &        &                     &           & \reflectbox{$\ddots$} &           &                         & \cdots &                             & &           & \reflectbox{$\ddots$} &           &                 &                             &                 &                           & &            \\[-1.1em]
        \lstick{$\ket{+}$} & \phantomgate{H} \labeledwire{niceblue}{X} & & \push{\hspace{0.5em}\cdots\hspace{0.5em}}                  &                   &        &                     & \ctrl{3}  &                       &           &                         &        &                             & & \ctrl{-3} &                       &           &                 & \labeledwire{niceblue}{X}   &                 &                           & & \meterD{X} \\
        \lstick{$\ket{+}$} &                                           & & \push{\hspace{0.5em}\cdots\hspace{0.5em}}                  & \mygate{3}{S_1^Z} &        & \mygate{3}{S_{m}^Z} &           &                       & \ctrl{-3} &                         &        &                             & &           &                       &           &                 &                             &                 &                           & & \meterD{X} \\[-1.5em]
                           & \setwiretype{n}                           & & \vdots                                                     &                   & \cdots &                     &           & \reflectbox{$\ddots$} &           &                         &        &                             & &           &                       &           &                 &                             &                 &                           & &            \\[-1.1em]
        \lstick{$\ket{+}$} &                                           & & \push{\hspace{0.5em}\cdots\hspace{0.5em}}                  &                   &        &                     & \ctrl{-3} &                       &           &                         &        &                             & &           &                       &           &                 &                             &                 &                           & & \meterD{X} \\
        \lstick{$\ket{+}$} & \phantomgate{H} \labeledwire{nicered}{X}  & & \push{\hspace{0.5em}\cdots\hspace{0.5em}}                  & \ctrl{-1}         &        &                     &           &                       &           & \ctrl{-4}               &        &                             & &           & \mygate{3}{CZ_A}      &           &                 & \labeledwire{nicered}{M_1}  &                 &                           & & \meterD{M_1} \\[-1.2em]
                           & \setwiretype{n}                           & & \vdots                                                     &                   & \ddots &                     &           &                       &           &                         & \ddots &                             & &           &                       &           &                 &                             &                 &                           & &            \\[-0.7em]
        \lstick{$\ket{+}$} &                                           & & \push{\hspace{0.5em}\cdots\hspace{0.5em}}                  &                   &        & \ctrl{-3}           &           &                       &           &                         &        & \ctrl{-6}                   & &           &                       &           &                 &                             &                 &                           & & \meterD{M_m}
    \end{quantikz}
\end{align*}
where we again label in blue the Pauli operators acting on the support of $S_1^Z$ and in yellow those acting on the support of $S_1^X$.
We can see that the corresponding detector node is connected to the ancilla qubit $a^{(2T-1)}_1$, as well as to the data qubits $q^{(2T)}_{j_1},\ldots,q^{(2T)}_{j_\alpha}$, $q^{(2T+1)}_{k_1},\ldots,q^{(2T+1)}_{k_\beta}$ and $q^{(2T+2)}_{j_1},\ldots,q^{(2T+2)}_{j_\alpha}$, where $q^{(t)}_{j_i}$ ($q^{(t)}_{k_i}$) are again the data qubits in the support of $S^Z_1$ ($S^X_i$) at the layer $t$. The connection to data qubits at the layer $2T+2$ is due to the presence of $Z$s on the output qubits.

Finally, the detector corresponding to the first $S_1$ measurement has the following form in the circuit picture:
\begin{align*}
    \tikzsetnextfilename{sec6-general-proof-12}
    \begin{quantikz}[row sep=10pt, column sep=2pt]
        \lstick{$\ket{+}$} & \phantomgate{H}                           &                                       & & \gatebox{12}{15}  &                   &        &                     &           &                       &           &                         &        &                           & &           &                       & \ctrl{3}  & \push{\hspace{0.5em}\cdots\hspace{0.5em}} & \phantomgate{H}  &                             & \phantomgate{H} & \meterD{X} \\[-1.5em]
                           & \setwiretype{n}                           &                                       & & \vdots            &                   &        &                     &           &                       &           &                         &        &                           & &           & \reflectbox{$\ddots$} &           &                                           &                  &                             &                 &            \\[-1.1em]
        \lstick{$\ket{+}$} & \phantomgate{H}                           &                                       & &                   &                   &        &                     &           &                       &           &                         &        &                           & & \ctrl{3}  &                       &           & \push{\hspace{0.5em}\cdots\hspace{0.5em}} &                  &                             &                 & \meterD{X} \\
        \lstick{$\ket{+}$} &                                           &                                       & &                   &                   &        &                     &           &                       & \ctrl{3}  & \mygate{3}{\bar{S}_1^X} &        & \mygate{3}{\bar{S}_{m}^X} & &           &                       & \ctrl{-3} & \push{\hspace{0.5em}\cdots\hspace{0.5em}} &                  &                             &                 & \meterD{X} \\[-1.5em]
                           & \setwiretype{n}                           &                                       & & \vdots            &                   &        &                     &           & \reflectbox{$\ddots$} &           &                         & \cdots &                           & &           & \reflectbox{$\ddots$} &           &                                           &                  &                             &                 &            \\[-1.1em]
        \lstick{$\ket{+}$} &                                           &                                       & &                   &                   &        &                     & \ctrl{3}  &                       &           &                         &        &                           & & \ctrl{-3} &                       &           & \push{\hspace{0.5em}\cdots\hspace{0.5em}} &                  &                             &                 & \meterD{X} \\
                           & \phantomgate{H} \labeledwire{niceblue}{Z} & \phantomgate{H} \labeledwire{gold}{X} & &                   & \mygate{3}{S_1^Z} &        & \mygate{3}{S_{m}^Z} &           &                       & \ctrl{-3} &                         &        &                           & &           &                       &           & \push{\hspace{0.5em}\cdots\hspace{0.5em}} &                  & \labeledwire{gold}{X}       &                 & \meterD{X} \\[-1.5em]
                           & \setwiretype{n}                           &                                       & & \vdots            &                   & \cdots &                     &           & \reflectbox{$\ddots$} &           &                         &        &                           & &           &                       &           &                                           &                  &                             &                 &            \\[-1.1em]
                           & \phantomgate{H} \labeledwire{niceblue}{Z} & \phantomgate{H} \labeledwire{gold}{X} & &                   &                   &        &                     & \ctrl{-3} &                       &           &                         &        &                           & &           &                       &           & \push{\hspace{0.5em}\cdots\hspace{0.5em}} &                  & \labeledwire{gold}{X}       &                 & \meterD{X} \\
        \lstick{$\ket{+}$} & \phantomgate{H} \labeledwire{nicered}{X}  &                                       & &                   & \ctrl{-1}         &        &                     &           &                       &           & \ctrl{-4}               &        &                           & &           & \mygate{3}{CZ_A}      &           & \push{\hspace{0.5em}\cdots\hspace{0.5em}} &                  & \labeledwire{nicered}{M_1}  &                 & \meterD{M_1} \\[-1.2em]
                           & \setwiretype{n}                           &                                       & & \vdots            &                   & \ddots &                     &           &                       &           &                         & \ddots &                           & &           &                       &           &                                           &                  &                             &                 &            \\[-0.7em]
        \lstick{$\ket{+}$} &                                           &                                       & &                   &                   &        & \ctrl{-3}           &           &                       &           &                         &        & \ctrl{-6}                 & &           &                       &           & \push{\hspace{0.5em}\cdots\hspace{0.5em}} &                  &                             &                 & \meterD{M_m}
    \end{quantikz}
\end{align*}
where again we labeled in blue the Pauli operators acting on the support of $S_1^Z$ and in yellow those acting on the support of $S_1^X$. This corresponds to a detector node connected to the ancilla qubit $a^{(1)}_1$, as well as to the data qubits $q^{(0)}_{j_1},\ldots,q^{(0)}_{j_\alpha}$ and $q^{(1)}_{k_1},\ldots,q^{(1)}_{k_\beta}$, where $q^{(t)}_{j_i}$ ($q^{(t)}_{k_i}$) are again the data qubits in the support of $S^Z_1$ ($S^X_1$) at layer $t$.

The final detectors connected only to the virtual ancilla node and again come from the construction of the augmented cluster state complex in \cref{def:compressed-rep-cluster-state}.

\end{proof}

\section{Cluster state complex from a dynamical code}
\label{sec:from-floquet-codes-to-cluster-states}
As our final example of mapping from codes to MBQC patterns, we consider the class of dynamical codes \cite{hastings2021dynamically,fu2024error,Davydova_2024}.
In the context of this work, we broadly define a \textit{dynamical code} as a protocol running for $T$ rounds, such that a set $G^{(t)}=\{G_1^{(t)},\ldots,G_{m_t}^{(t)}\}$ of commuting Pauli operators is measured at each round $t \in \{1,\ldots,T\}$.
A similar definition of dynamical codes can be found in Fu and Gottesman \cite{fu2024error}.
We call the family of sets $\{G^{(t)}\}_{1 \leq t \leq T}$ the \textit{measurement schedule} of the code.

Dynamical codes with a periodic function $t \mapsto G^{(t)}$ are called Floquet codes, and were introduced by Hastings and Haah as a broad generalization of stabilizer and subsystem codes \cite{hastings2021dynamically}.
The honeycomb code constructed in their work is a Floquet code of period $3$ defined on a honeycomb lattice with qubits on vertices.
For this code,  $G^{(3t)}$, $G^{(3t+1)}$ and $G^{(3t+2)}$ are respectively the sets of $XX$, $YY$ and $ZZ$ Pauli operators associated with subsets of the edges the lattice.
While the stabilizers and logical operators of the code change at each round, the logical information is preserved during the protocol and the presence of spacetime stabilizers in the corresponding circuit enable fault-tolerant error correction.
Many variations of the honeycomb code have been constructed since then, including a CSS version of \cite{Kesselring_2024,Davydova_2023}, a planar version \cite{vuillot2021planar,haah2022boundaries,Kesselring_2024, sullivan2023floquet}, a hyperbolic version \cite{fahimniya2025faulttolerant,higgott2024constructions} and a 3D version \cite{Davydova_2023, zhang2023xcube,dua2024engineering,bauer2024topological}.
Other examples have been obtained by "Floquetifying" stabilizer and subsystem codes.
This works by starting from the measurement schedule of a given code and modifying it, often using ZX-calculus, to obtain new schedules with either lower-weight measurements or better general performance under some noise models \cite{Townsend_Teague_2023,rodatz2024floquetifying,delafuente2024xyzrubycode,delafuente2024dynamical,xu2025faulttolerant}.
Floquetifying the Bacon-Shor code has for instance led to variations of the code with a non-zero threshold \cite{alam2024dynamical,alam2025baconshor}.

However, our definition also allows the existence of non-periodic dynamical codes \cite{Davydova_2023,Kesselring_2024,Davydova_2024}.
A common case arises when the measurement schedule carries out a specific logical operation on the code---for instance, using dynamical automorphisms \cite{Davydova_2024}. More traditional protocols, such as lattice surgery \cite{horsman2012surface} or gauge fixing \cite{paetznick2013universal,bombin2015gauge,yoder2016universal,higgott2021subsystem}, can likewise be framed as non-periodic dynamical codes.

\begin{figure}
    \centering
    \subfloat[]{
        \includegraphics[width=0.34\linewidth]{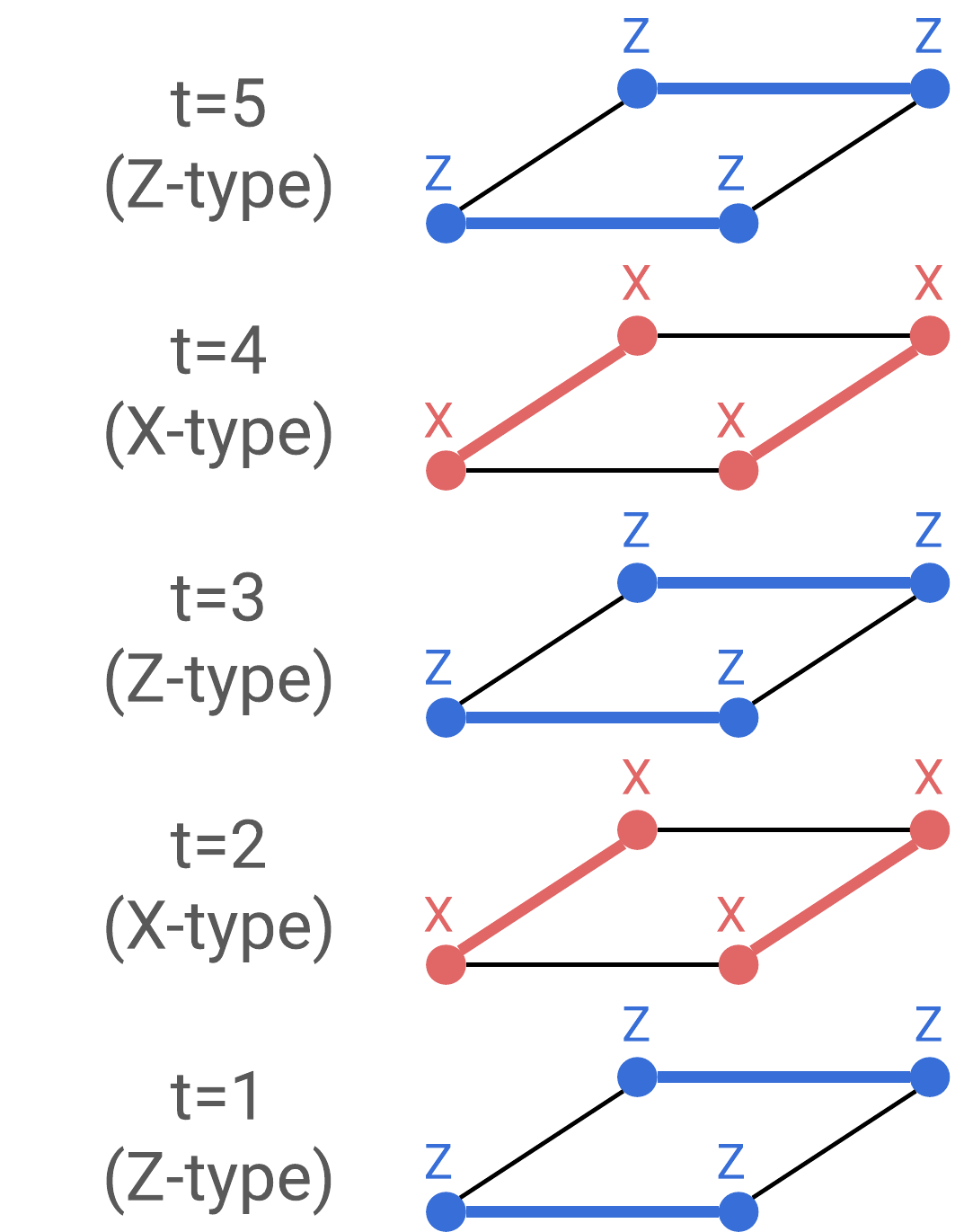}
    }
    \hspace{4em}
    \subfloat[]{
        \includegraphics[width=0.3\linewidth]{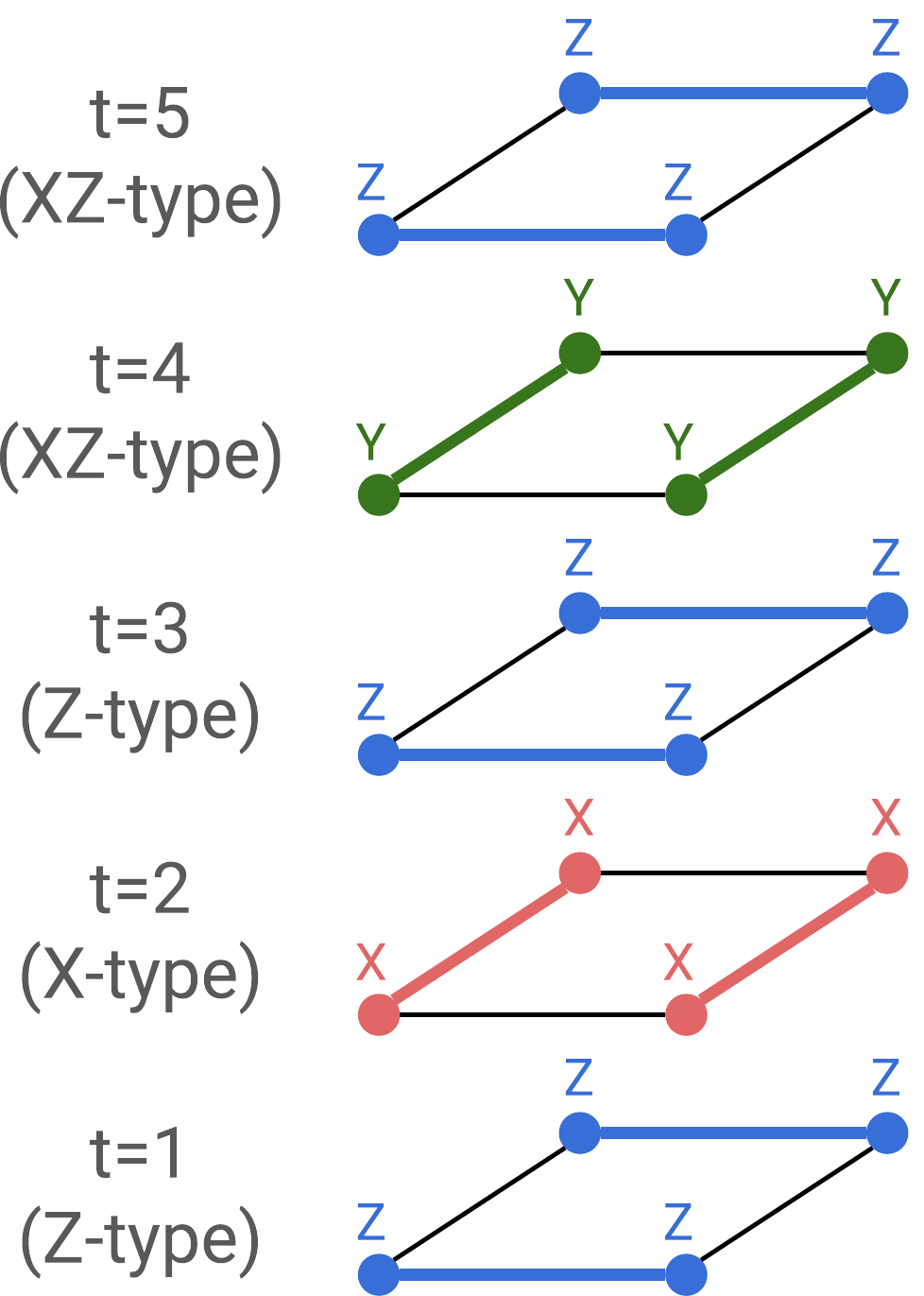}
    }
    \caption{ \label{fig:dynamical-code-examples}
        Schedule and round types for two dynamical variants of the $[\![4,2,2]\!]$-code.
        \textbf{(a)} $[\![4,1,2]\!]$ subsystem code, which can be seen as a period-2 Floquet code.
        \textbf{(b)} Four-qubit ladder code, as introduced by Hastings and Haah~\cite{hastings2021dynamically}. It varies from the $[\![4,1,2]\!]$-code schedule by the introduction of $Y$ measurements every four rounds. We note that the last round is of $XZ$-type due to the previous one being of $XZ$-type as well (a $Z$-type round cannot follow an $XZ$-type round).
    }
\end{figure}

\begin{figure}
    \centering
    \includegraphics[width=0.46\linewidth]{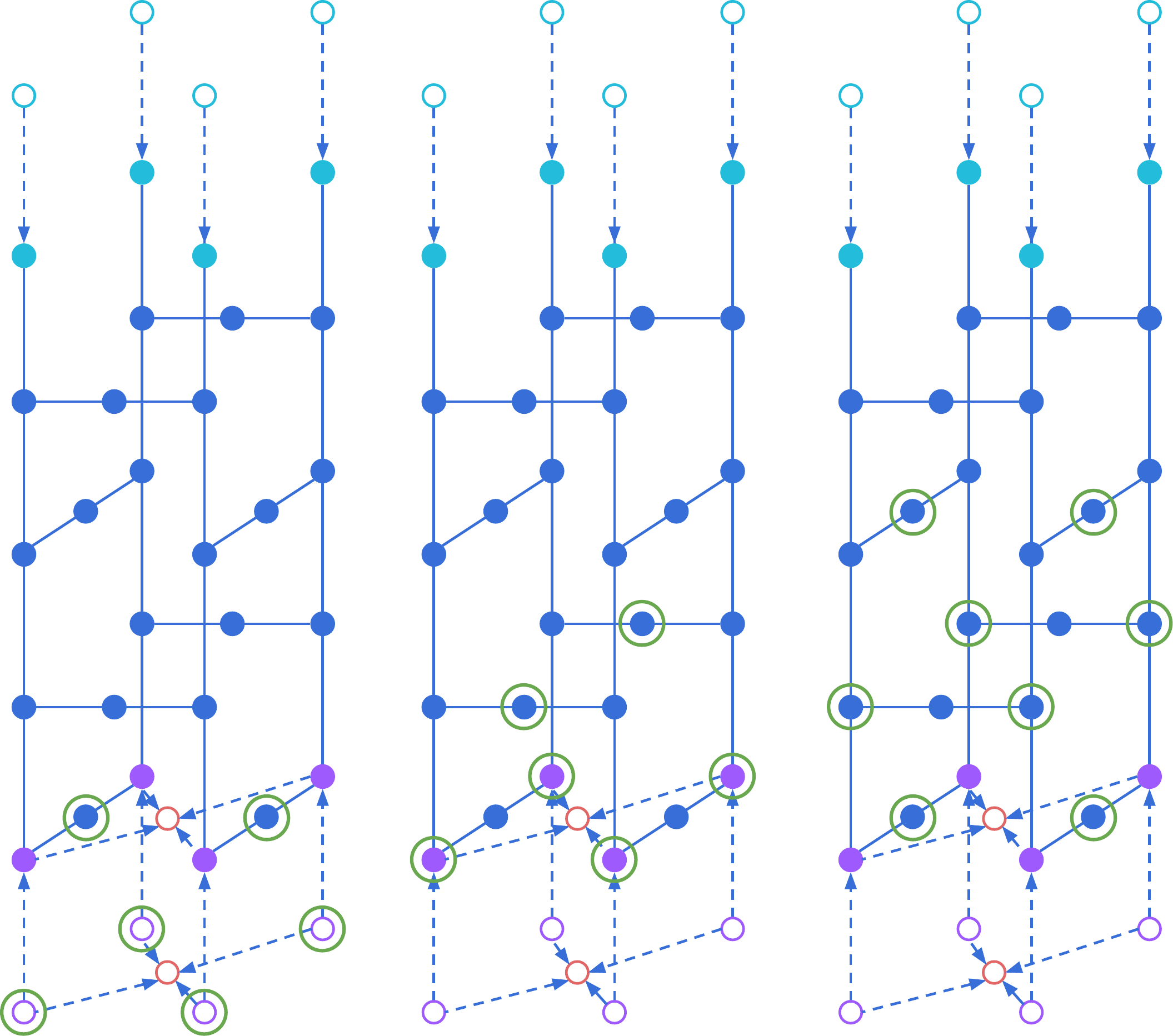}
    \hspace{0.8em}
    \includegraphics[width=0.46\linewidth]{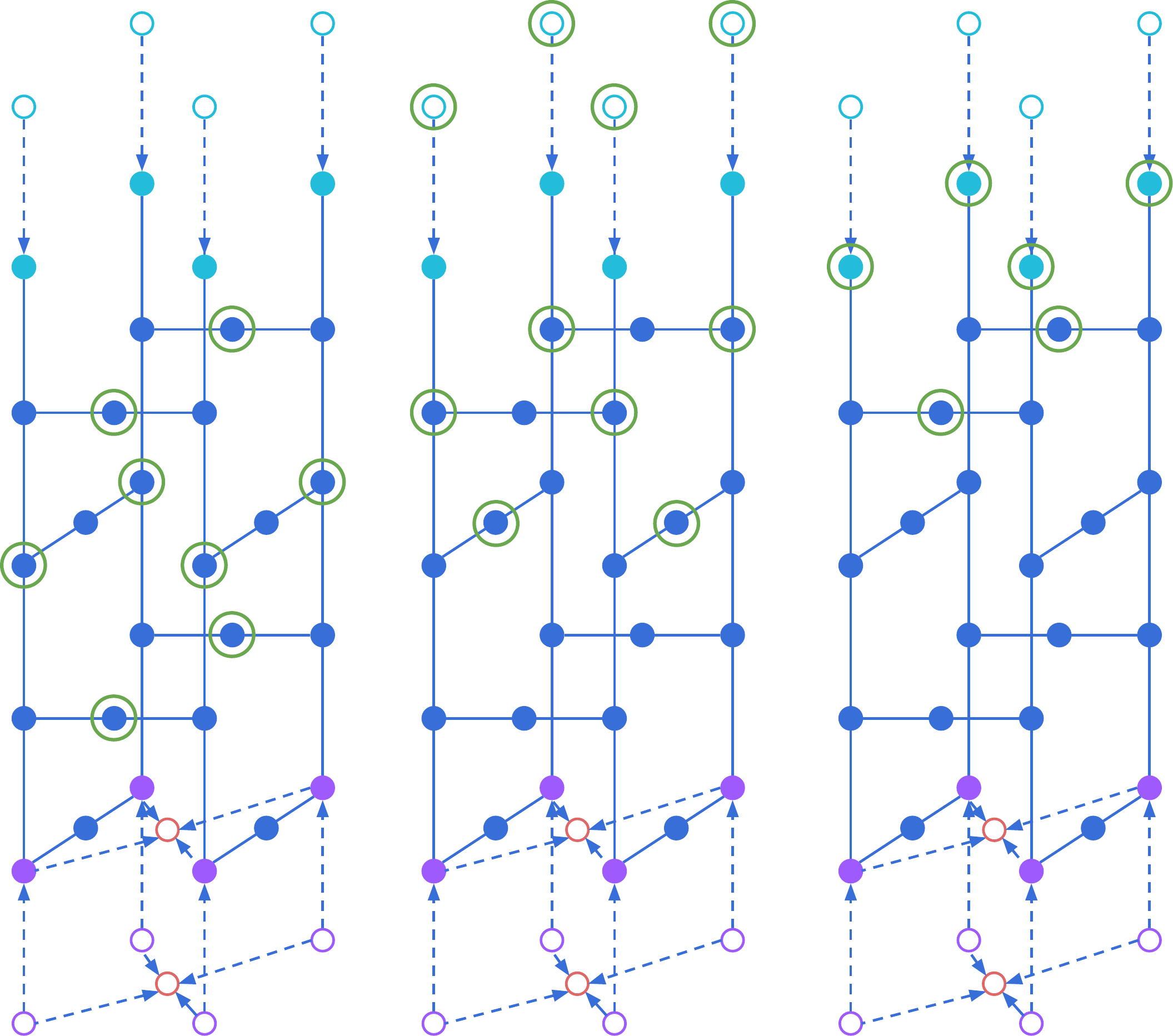}
    \caption{ \label{fig:412-code-cluster-state-detectors}
        Cluster state complex of the $[\![4,1,2]\!]$ subsystem code for $T=4$.
        For clarity, the six detectors are shown on separate graphs, by circling the nodes that are part of each detector.
    }
\end{figure}

To any dynamical code, we associate a circuit of the form:
\begin{equation*}
    \tikzsetnextfilename{sec7-dynamical-code-circuit}
    \begin{quantikz}[row sep=10pt, column sep=6pt]
                           &                          & \gatebox{9}{4} & \mygate{3}{G^{(1)}_1}  &        & \mygate{3}{G^{(1)}_{m_1}} & \phantomgate{H} & \ \ldots \ & \phantomgate{H} & \mygate{3}{G^{(T)}_1} \gatebox{9}{3} &        & \mygate{3}{G^{(T)}_{m_{T}}} & \phantomgate{H} &                        \\[-1.5em]
                           & \setwiretype{n}          & \vdots         &                        & \cdots &                           &                 &            &                 &                                      & \cdots &                             &                 &                        \\[-1.1em]
                           &                          &                &                        &        &                           &                 & \ \ldots \ &                 &                                      &        &                             &                 &                        \\
        \lstick{$\ket{+}$} &                          &                & \ctrl{-1}              &        &                           &                 & \ \ldots \ &                 &                                      &        &                             &                 & \meterD{M^{(1)}_1}     \\[-1.2em]
                           & \setwiretype{n}          & \vdots         &                        & \ddots &                           &                 &            &                 &                                      &        &                             &                 &                        \\[-0.7em]
        \lstick{$\ket{+}$} &                          &                &                        &        & \ctrl{-3}                 &                 & \ \ldots \ &                 &                                      &        &                             &                 & \meterD{M^{(1)}_{m_1}} \\
        \lstick{$\ket{+}$} &                          & \vdots         &                        &        &                           &                 & \ \ldots \ &                 & \ctrl{-4}                            &        &                             &                 & \meterD{M^{(T)}_1}     \\[-1.2em]
                           & \setwiretype{n}          & \vdots         &                        & \ddots &                           &                 &            &                 &                                      & \ddots &                             &                 &                        \\[-0.7em]
        \lstick{$\ket{+}$} &                          &                &                        &        &                           &                 & \ \ldots \ &                 &                                      &        & \ctrl{-6}                   &                 & \meterD{M^{(T)}_{m_T}}
    \end{quantikz}
\end{equation*}
where $M_i^{(t)}$ is either $X$ or $Y$ depending on whether the number of $Y$ terms in $G_i^{(t)}$ is even ($M_i^{(t)}=X$) or odd ($M_i^{(t)}=Y$).
To construct a cluster state complex that corresponds to this circuit, we need to progressively compile every round into a \texttt{CZ} network and merge those networks using \cref{lemma:cz-cz-merge}.
This can only be achieved if all the \texttt{CZ} networks are Hadamard-separated.
In order to ensure this property during the compilation and have an MBQC circuit that is as compact as possible, we will derive a specific recursive compilation of this circuit, where the compilation of one round depends on that of the previous round.
We start by labeling each round as $X$-\textit{type}, $Z$-\textit{type}, $XZ$-\textit{type} or $ZX$-\textit{type}.
The assignment is done recursively for each round $t$ as follows:
\begin{itemize}
    \item If all the measurement in $G^{(t)}$ are pure $X$ measurements, the round is labeled $X$-type
    \item If all the measurement in $G^{(t)}$ are pure $Z$ measurements, then
    \begin{itemize}
        \item if either $t=1$ or the previous round is of $X$-type or $ZX$-type, we label the round as $Z$-type,
        \item or if the previous round is of $Z$-type or $XZ$-type, we label the round as $XZ$-type.
    \end{itemize}
    \item If the measurement in $G^{(t)}$ are made of both $X$ and $Z$, then
    \begin{itemize}
        \item if either $t=1$ or the previous round is of $X$-type or $ZX$-type, we label the round as $ZX$-type,
        \item or if the previous round is of $Z$-type or $XZ$-type, we label the round as $XZ$-type.
    \end{itemize}
\end{itemize}
Every round will be compiled depending on its label, to ensure that rounds are always Hadamard-separated before merging and to avoid unnecessary double-Hadamard between two rounds.
Two examples of dynamical codes with a labelled schedule are shown in \cref{fig:dynamical-code-examples}.

\begin{figure}
    \centering
    \includegraphics[width=0.46\linewidth]{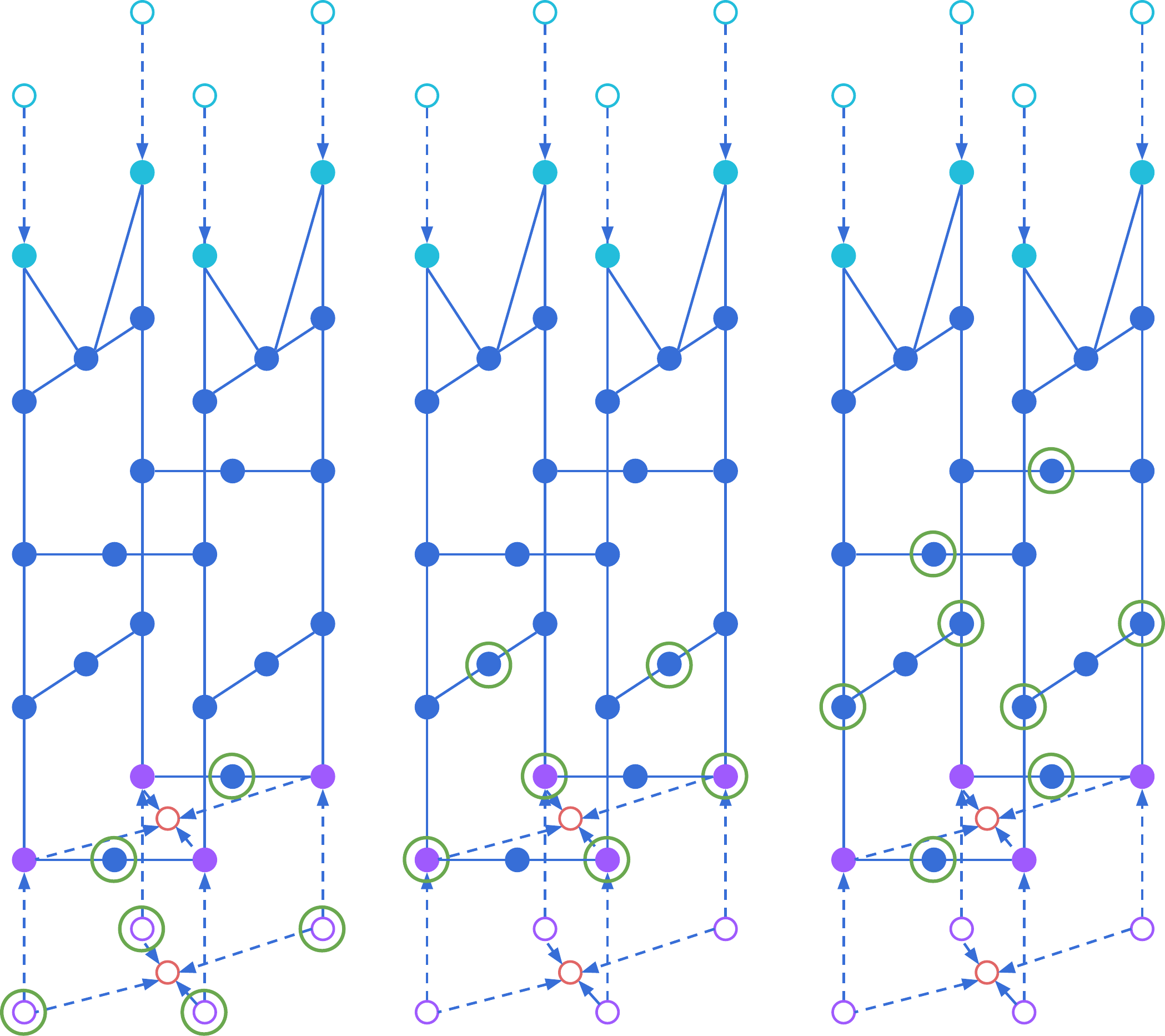}
    \hspace{0.8em}
    \includegraphics[width=0.46\linewidth]{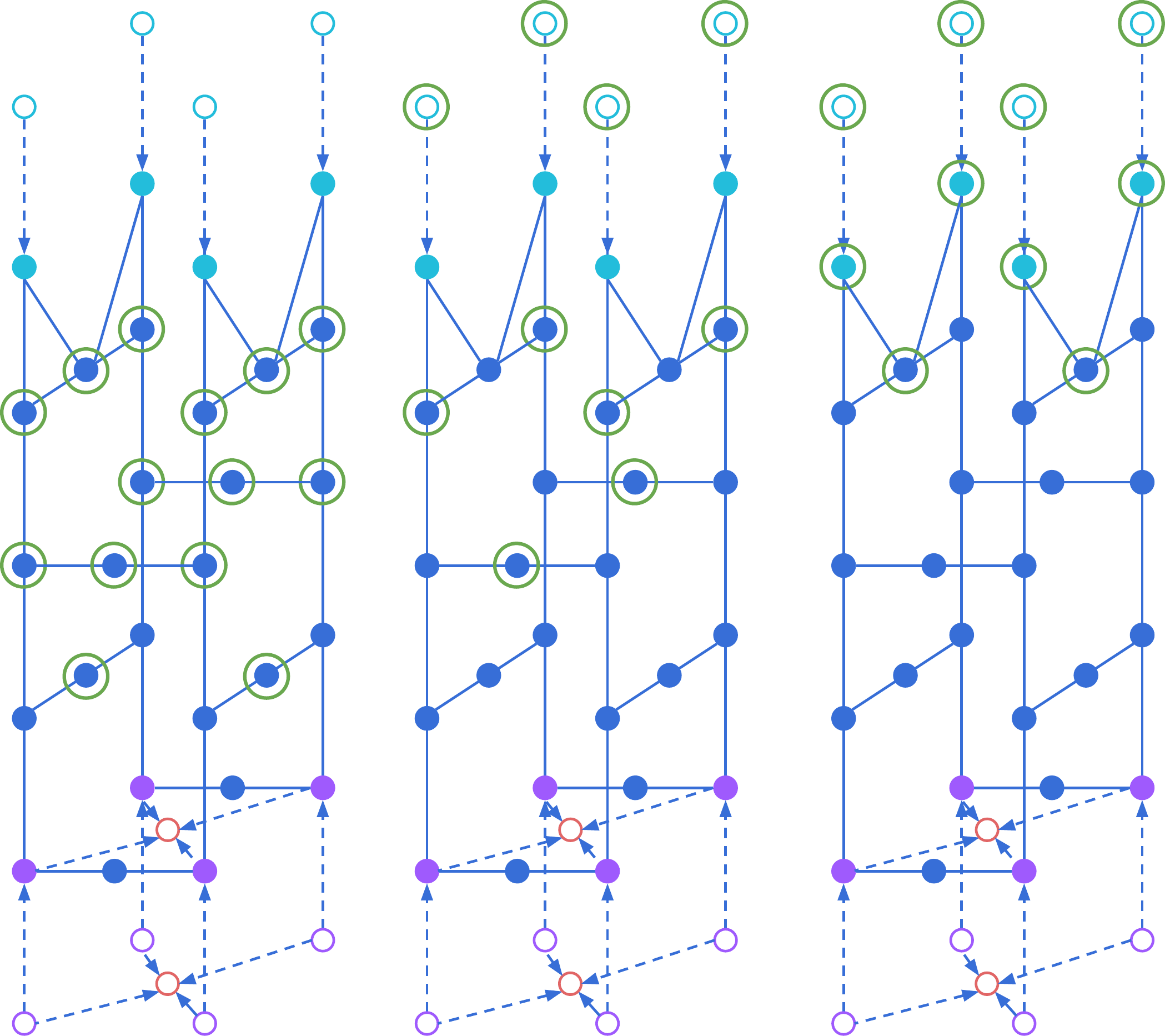}
    \caption{ \label{fig:ladder-code-cluster-state-detectors}
        Cluster state complex of the four-qubit ladder code for $T=4$, in its compressed co-representation.
        For clarity, the six detectors are shown on separate graphs, by circling the nodes that are part of each detector.
        Moreover, while the two ancilla nodes (red empty circles) are shown at two different layers in the figure for visual clarity, they formally both belong to layer $0$ in our construction.
    }
\end{figure}

We are now ready to describe the cluster state complex corresponding to this circuit.
\begin{theorem}
    Let us consider a dynamical code with measurement schedule $\{G^{(t)}\}_{1 \leq t \leq T}$, $G^{(t)}=\{G_1^{(t)},\ldots,G_{m_t}^{(t)}\}$, and input stabilizer group $\mathcal{S}=\langle S_1,\ldots,S_{m_0} \rangle$, as well as a noise model where errors only happen in-between rounds. Let $\mathcal{D}_S$ be the set of (full and incomplete) detectors that come from a spackle, where we see every element $D \in \mathcal{D}_S$ as a subset of $\bigsqcup_t G^{(t)}$ corresponding to all the measurements involved in the detector. Similarly, let $\mathcal{D}_B$ be the set of incomplete detectors that come from a backle. The circuit is equivalent to the following cluster state complex, described in the compressed representation:
    \begin{itemize}
        \item It has a layered structure, where layers are added recursively.
        We start by defining a layer $\ell=0$ of virtual data nodes $q_1^{(0)},\ldots,q_n^{(0)}$ and a layer $\ell=1$ with $n$ data nodes $q_1^{(1)},\ldots,q_n^{(1)}$, with a directed edge from $q_i^{(1)}$ to $q_i^{(0)}$ for all $i \in \{1,\ldots,n\}$. We also add ancilla nodes $a_1^{(0)},\ldots,a_m^{(0)}$, with a directed edge from each ancilla node $a_i^{(0)}$ to the data nodes at layer $0$ according to the $Z$ part of $S_i$, and to the data nodes of layer $1$ according to the $X$ part of $S_i$.
        Then, for each round $t$ of the dynamical code, let $\ell$ be the last layer added to the complex.
        We distinguish three cases, depending on the type of $G^{(t)}$:
        \begin{enumerate}
            \item $G^{(t)}$ is of $Z$-type.
            We add $m_t$ ancilla nodes $a_1^{(\ell)},\ldots,a_{m_t}^{(\ell)}$ to layer $\ell$, and connect them to the data nodes at this layer according to the support of each $G_i^{(t)}$.

            \item $G^{(t)}$ is of $X$-type.
            We add two new layers $\ell+1$ and $\ell+2$ with new data nodes $q_1^{(\ell+1)},\ldots,q_n^{(\ell+1)}$ and $q_1^{(\ell+2)},\ldots,q_n^{(\ell+2)}$.
            We connect each data node $q_i^{(\ell)}$ to $q_i^{(\ell+1)}$, and $q_i^{(\ell+1)}$ to $q_i^{(\ell+2)}$.
            We add $m_t$ ancilla nodes $a_1^{(\ell+1)},\ldots,a_{m_t}^{(\ell+1)}$ to layer $\ell+1$ and connect them to the data nodes at this layer according to the support of the $G_i^{(t)}$.

            \item $G^{(t)}$ is of $ZX$-type.
            We add two new layers $\ell+1$ and $\ell+2$ with new data nodes $q_1^{(\ell+1)},\ldots,q_n^{(\ell+1)}$ and $q_1^{(\ell+2)},\ldots,q_n^{(\ell+2)}$.
            We add $m_t$ ancilla nodes $a_1^{(\ell)},\ldots,a_{m_t}^{(\ell)}$ to layer $\ell$, and connect them to the data nodes at this layer according to the $Z$ part of each $G_i^{(t)}$, and to the data nodes of layer $\ell+1$ according to the $X$ part of each $G_i^{(t)}$.
            For all $i,j$, $i \neq j$, if the $X$ part of $G_i^{(t)}$ has an odd intersection with the $Z$ part of $G_j^{(t)}$, we add an edge between $a_i^{(\ell)}$ and $a_j^{(\ell)}$. This includes self-loops on in the case $i=j$.

            \item $G^{(t)}$ is of $XZ$-type.
            We add two new layers $\ell+1$ and $\ell+2$ with new data nodes $q_1^{(\ell+1)},\ldots,q_n^{(\ell+1)}$ and $q_1^{(\ell+2)},\ldots,q_n^{(\ell+2)}$.
            We add $m_t$ ancilla nodes $a_1^{(\ell+2)},\ldots,a_{m_t}^{(\ell+2)}$ to layer $\ell+2$, and connect them to the data nodes at this layer according to the $X$ part of $G_i^{(t)}$, and to the data nodes of layer $\ell+2$ according to the $Z$ part of $G_i^{(t)}$.
            For all $i,j$, if the $X$ part of $G_i^{(t)}$ has an odd intersection with the $Z$ part of $G_j^{(t)}$, we add an edge between $a_i^{(\ell+2)}$ and $a_j^{(\ell+2)}$. This includes self-loops on in the case $i=j$.
        \end{enumerate}
        Let $L$ be the last layer index $\ell$ added through this procedure. We then add a final output layer with data nodes $q_1^{(L+1)},\ldots,q_n^{(L+1)}$ with directed edges from $q_i^{(L)}$ to $q_i^{(L+1)}$ for all $i \in \{1,\ldots,n\}$.

        \item To construct the detectors, we define the \emph{spackle of an ancilla node} $a_i^{(\ell)}$, connected to the data nodes $q_{j_1}^{(\ell_1)},\ldots,q_{j_k}^{(\ell_k)}$ as the set of nodes consisting of $a_i^{(\ell)}$ and $q_{j_1}^{(\ell_1+2t+1)},\ldots,q_{j_k}^{(\ell_k+2t+1)}$ for all $t\geq 0$ such that $\ell + 2t +1 \leq L+1$.
        Similarly, the \textit{backle} of $a_i^{(\ell)}$ is the set of nodes consisting in $a_i^{(\ell)}$ and $q_{j_1}^{(\ell_1-2t-1)},\ldots,q_{j_k}^{(\ell_k-2t-1)}$ for all $t$ such that $\ell - 2t -1 \geq 0$.
        We now build one detector node per element $D \in \mathcal{D}_S$ ($D \in \mathcal{D}_B$), as the symmetric difference of the spackles (backles) of all the ancilla nodes corresponding to the measurements in $D$.
    \end{itemize}
\end{theorem}
\begin{proof}
We construct the cluster state complex by progressively compiling every round of the dynamical code into a \texttt{CZ} network and merging it into the previous one. The proof works by induction on the number of compiled rounds.

The $n$ qubits of the initial circuit are first added to the graph as nodes $q_1^{(1)},\ldots,q_n^{(1)}$ on layer $\ell=1$.

We now assume that the first $t-1$ rounds of the dynamical code have been compiled into a \texttt{CZ} network, which gives the first $\ell$ layers of the cluster state complex described in the theorem statement. We examine the different possibilities for the $t$-th round of the dynamical code.

If $G^{(t)}$ is of $X$-type, we compile the measurements into a \texttt{CZ} network with \texttt{H} gates on both sides, using \cref{lemma:pushing-h-away} to move the \texttt{H} gates outside the blue box:
\begin{align*}
    \tikzsetnextfilename{sec7-proof-1}
    \begin{quantikz}[row sep=10pt, column sep=6pt]
                           & \phantomgate{H} & \gatebox{6}{4}  & \mygate{3}{G^{(t)}_1}  &        & \mygate{3}{G^{(t)}_{m_t}} &  &                        \\[-1.5em]
                           & \setwiretype{n} & \vdots          &                        & \cdots &                           &  &                        \\[-1.1em]
                           &                 &                 &                        &        &                           &  &                        \\
        \lstick{$\ket{+}$} &                 &                 & \ctrl{-1}              &        &                           &  & \meterD{M_1^{(t)}}     \\[-1.2em]
                           & \setwiretype{n} & \vdots          &                        & \ddots &                           &  &                        \\[-0.7em]
        \lstick{$\ket{+}$} &                 &                 &                        &        & \ctrl{-3}                 &  & \meterD{M_{m_t}^{(t)}}
    \end{quantikz}
    \sim
    \tikzsetnextfilename{sec7-proof-2}
    \begin{quantikz}[row sep=10pt, column sep=6pt]
                           & \gate{H}        & & \gatebox{6}{4}  & \mygate{3}{\bar{G}^{(t)}_1}  &        & \mygate{3}{\bar{G}^{(t)}_{m_t}} & & \gate{H} &                        \\[-1.5em]
                           & \setwiretype{n} & & \vdots          &                              & \cdots &                                 & &          &                        \\[-1.1em]
                           & \gate{H}        & &                 &                              &        &                                 & & \gate{H} &                        \\
        \lstick{$\ket{+}$} &                 & &                 & \ctrl{-1}                    &        &                                 & &          & \meterD{M_1^{(t)}}     \\[-1.2em]
                           & \setwiretype{n} & & \vdots          &                              & \ddots &                                 & &          &                        \\[-0.7em]
        \lstick{$\ket{+}$} &                 & &                 &                              &        & \ctrl{-3}                       & &          & \meterD{M_{m_t}^{(t)}}
    \end{quantikz}
\end{align*}
Here, $\bar{G}^{(t)}_i$ denotes the Pauli operator coming from ${G}^{(t)}_i$, where every $X$ has been turned into a $Z$.
The circuit inside the blue box is now a \texttt{CZ} network, which we can merge with a preceding \texttt{CZ} network by compiling the \texttt{H} gates before it into their MBQC version and using \cref{lemma:cz-cz-merge} to merge the networks.
The new ancilla qubits used to compile the Hadamard gates give us a new layer of data nodes $q^{(\ell+1)}_1,\ldots,q^{(\ell+1)}_n$ in the cluster state complex, which are connected to the data qubits of the layer below.
The ancilla qubits used for measuring the operators $\bar{G}^{(t)}_i$ correspond to the ancilla nodes $a_i^{(\ell)}$ in the theorem statement.
Finally, we also compile and merge the \texttt{H} gates coming after the \texttt{CZ} network using \cref{lemma:h-and-hs-merge-right}.
This gives us a new layer of data nodes $q^{(\ell+2)}_1,\ldots,q^{(\ell+2)}_n$ in the cluster state complex.

If $G^{(t)}$ is of $Z$-type, the circuit corresponding to the measurement of the $Z$-stabilizers $G_1^{(t)},\ldots,G_{m_t}^{(t)}$ is already a \texttt{CZ} network. If $t>1$, since we are guaranteed that the previous round is of $ZX$- or $X$-type by construction of the labeling, it can be merged to the previous \texttt{CZ} network by \cref{lemma:cz-cz-merge}.
This compilation does not add any new layer to the cluster state complex, but only ancilla nodes connected to the data nodes of the current layer $\ell$ depending on the support of each $G_i^{(t)}$.

If $G^{(t)}$ is of $XZ$-type or $ZX$-type, we start by decomposing each measurement $G_{i}^{(t)}$ into its $X$ part---denoted $G_{i,X}^{(t)}$---and its $Z$ part---denoted $G_{i,Z}^{(t)}$:
\begin{align*}
    \tikzsetnextfilename{sec7-proof-3}
    \begin{quantikz}[row sep=10pt, column sep=6pt]
                           & \phantomgate{H} & \gatebox{6}{6}  & \mygate{3}{G_{1,Z}^{(t)}} & \mygate{3}{G_{1,Z}^{(t)}} &        & \mygate{3}{G_{n,Z}^{(t)}} & \mygate{3}{G_{n,X}^{(t)}} &  &                        \\[-1.5em]
                           & \setwiretype{n} & \vdots          &                           &                           & \cdots &                           &                           &  &                        \\[-1.1em]
                           &                 &                 &                           &                           &        &                           &                           &  &                        \\
        \lstick{$\ket{+}$} &                 &                 & \ctrl{-1}                 & \ctrl{-1}                 &        &                           &                           &  & \meterD{M_1^{(t)}}     \\[-1.2em]
                           & \setwiretype{n} & \vdots          &                           &                           & \ddots &                           &                           &  &                        \\[-0.7em]
        \lstick{$\ket{+}$} &                 &                 &                           &                           &        & \ctrl{-3}                 & \ctrl{-3}                 &  & \meterD{M_{m_t}^{(t)}}
    \end{quantikz}
\end{align*}
We might now be tempted to turn every control Pauli into \texttt{CZ} and \texttt{H}, in order to then compile the \texttt{H} gates into their measurement-based versions.
However, the presence of multiple layers of \texttt{H} gates within a blue box would prevent us from using \cref{lemma:h-cz-sandwich} to compile the \texttt{H} gates into \texttt{CZ} gates.
The key is to use one of the two circuit identities:
\begin{align}
    \label{eq:cx-cz-identity-1}
    \tikzsetnextfilename{sec7-proof-4}
    \begin{quantikz}
        & \targ{}   & \ctrl{2}  &  &  \\
        & \ctrl{-1} &           &  & \\
        &           & \ctrl{-2} &  &
    \end{quantikz}
    =
    \tikzsetnextfilename{sec7-proof-5}
    \begin{quantikz}
        & \ctrl{2}  & \targ{}   &           &  \\
        &           & \ctrl{-1} & \ctrl{1}  & \\
        & \ctrl{-2} &           & \ctrl{-1} &
    \end{quantikz}
    \\
    \label{eq:cx-cz-identity-2}
    \tikzsetnextfilename{sec7-proof-6}
    \begin{quantikz}
        & \ctrl{1}  & \targ{}  &  &  \\
        & \ctrl{-1} &           &  & \\
        &           & \ctrl{-2} &  &
    \end{quantikz}
    =
    \tikzsetnextfilename{sec7-proof-7}
    \begin{quantikz}
        & \targ{}   & \ctrl{1}  &           &  \\
        &           & \ctrl{-1} & \ctrl{1}  & \\
        & \ctrl{-2} &           & \ctrl{-1} &
    \end{quantikz}
\end{align}
which allows us to push all the \texttt{CZ}s to the left and all the \texttt{CNOT}s to the right (or the reverse), at the cost of adding new connections between ancilla qubits.
We can use this identity to commute the control $G_{i,X}^{(t)}$ and $G_{i+1,Z}^{(t)}$ gates, inserting a \texttt{CZ} between their control qubits if $G_{i,X}^{(t)}$ and $G_{i+1,Z}^{(t)}$ intersect non-trivially on an odd number of qubits.
For instance, if $G_{1,X}^{(t)}$ and $G_{2,Z}^{(t)}$ have an odd intersection, using \cref{eq:cx-cz-identity-1} on the first two measurements gives us the following the circuit:
\begin{align*}
    \tikzsetnextfilename{sec7-proof-8}
    \begin{quantikz}[row sep=11pt, column sep=6pt]
                           & \phantomgate{H} & \gatebox{7}{8} & \mygate{3}{G_{1,Z}^{(t)}} & \mygate{3}{G_{2,Z}^{(t)}} & \mygate{3}{G_{1,X}^{(t)}} & \mygate{3}{G_{2,X}^{(t)}} &        & \mygate{3}{G_{n,Z}^{(t)}} & \mygate{3}{G_{n,X}^{(t)}} &  &                        \\[-1.5em]
                           & \setwiretype{n} & \vdots         &                           &                           &                           &                           & \cdots &                           &                           &  &                        \\[-1.1em]
                           &                 &                &                           &                           &                           &                           &        &                           &                           &  &                        \\
        \lstick{$\ket{+}$} &                 & \ctrl{1}       & \ctrl{-1}                 &                           & \ctrl{-1}                 &                           &        &                           &                           &  & \meterD{M_1^{(t)}}     \\[-0.5em]
        \lstick{$\ket{+}$} &                 & \ctrl{-1}      &                           & \ctrl{-2}                 &                           & \ctrl{-2}                 &        &                           &                           &  & \meterD{M_2^{(t)}}     \\[-1.2em]
                           & \setwiretype{n} & \vdots         &                           &                           &                           &                           & \ddots &                           &                           &  &                        \\[-0.7em]
        \lstick{$\ket{+}$} &                 &                &                           &                           &                           &                           &        & \ctrl{-4}                 & \ctrl{-4}                 &  & \meterD{M_{m_t}^{(t)}}
    \end{quantikz}
\end{align*}
Repeating this commuting process for all the measurements, our circuit takes the form:
\begin{align*}
    \tikzsetnextfilename{sec7-proof-9}
    \begin{quantikz}[row sep=11pt, column sep=6pt]
                           & \phantomgate{H} & \gatebox{6}{8} & \mygate{3}{G_{1,Z}^{(t)}} &        & \mygate{3}{G_{n,Z}^{(t)}} & \mygate{3}{G_{1,X}^{(t)}} &        & \mygate{3}{G_{n,X}^{(t)}} &                    & &            \\[-1.5em]
                           & \setwiretype{n} & \vdots         &                           & \cdots &                           &                           & \cdots &                           &                    & &            \\[-1.1em]
                           &                 &                &                           &        &                           &                           &        &                           &                    & &            \\
        \lstick{$\ket{+}$} &                 &                & \ctrl{-1}                 &        &                           & \ctrl{-1}                 &        &                           & \mygate{3}{CZ_A}   & & \meterD{X} \\[-1.2em]
                           & \setwiretype{n} & \vdots         &                           & \ddots &                           &                           & \ddots &                           &                    & &            \\[-0.7em]
        \lstick{$\ket{+}$} &                 &                &                           &        & \ctrl{-3}                 &                           &        & \ctrl{-3}                 &                    & & \meterD{X}
    \end{quantikz}
\end{align*}
where $CZ_A$ is a network of \texttt{CZ} gates with adjacency matrix $A$, defined by $A_{ij}=1$ if and only if $G_{i,Z}^{(t)}$ and $G_{j,X}^{(t)}$ have an odd intersection.
We can now compile the control $G_{i,X}^{(t)}$ into \texttt{H} and \texttt{CZ}
\begin{align} \label{eq:zx-circuit}
    \tikzsetnextfilename{sec7-proof-10}
    \begin{quantikz}[row sep=11pt, column sep=6pt]
                           & \phantomgate{H} & \gatebox{6}{9} & \mygate{3}{G_{1,Z}^{(t)}} &        & \mygate{3}{G_{n,Z}^{(t)}} & \mygate{1}{H} & \mygate{3}{\bar{G}_{1,X}^{(t)}} &        & \mygate{3}{\bar{G}_{n,X}^{(t)}} &                  & & \mygate{1}{H} &                        \\[-1.5em]
                           & \setwiretype{n} & \vdots         &                           & \cdots &                           &               &                                 & \cdots &                                 &                  & &               &                        \\[-1.1em]
                           &                 &                &                           &        &                           & \mygate{1}{H} &                                 &        &                                 &                  & & \mygate{1}{H} &                        \\
        \lstick{$\ket{+}$} &                 &                & \ctrl{-1}                 &        &                           &               & \ctrl{-1}                       &        &                                 & \mygate{3}{CZ_A} & &               & \meterD{M_1^{(t)}}     \\[-1.2em]
                           & \setwiretype{n} & \vdots         &                           & \ddots &                           &               &                                 & \ddots &                                 &                  & &               &                        \\[-0.7em]
        \lstick{$\ket{+}$} &                 &                &                           &        & \ctrl{-3}                 &               &                                 &        & \ctrl{-3}                       &                  & &               & \meterD{M_{m_t}^{(t)}}
    \end{quantikz}
\end{align}
where we used \cref{lemma:pushing-h-away} to move the right-most \texttt{H} gates outside the blue box.
Similarly, using \cref{eq:cx-cz-identity-2}, we obtain the following circuit where the $X$ measurements come first:
\begin{align} \label{eq:xz-circuit}
    \tikzsetnextfilename{sec7-proof-11}
    \begin{quantikz}[row sep=11pt, column sep=6pt]
                           & \phantomgate{H} & \mygate{1}{H} &        & \mygate{3}{\bar{G}_{1,X}^{(t)}} \gatebox{6}{8} &        & \mygate{3}{\bar{G}_{n,X}^{(t)}} & \mygate{1}{H} & \mygate{3}{G_{1,Z}^{(t)}} &        & \mygate{3}{G_{n,Z}^{(t)}} &                  &  &                        \\[-1.5em]
                           & \setwiretype{n} &               & \vdots &                                                & \cdots &                                 &               &                           & \cdots &                           &                  &  &                        \\[-1.1em]
                           &                 & \mygate{1}{H} &        &                                                &        &                                 & \mygate{1}{H} &                           &        &                           &                  &  &                        \\
        \lstick{$\ket{+}$} &                 &               &        & \ctrl{-1}                                      &        &                                 &               & \ctrl{-1}                 &        &                           & \mygate{3}{CZ_A} &  & \meterD{M_1^{(t)}}     \\[-1.2em]
                           & \setwiretype{n} &               & \vdots &                                                & \ddots &                                 &               &                           & \ddots &                           &                  &  &                        \\[-0.7em]
        \lstick{$\ket{+}$} &                 &               &        &                                                &        & \ctrl{-3}                       &               &                           &        & \ctrl{-3}                 &                  &  & \meterD{M_{m_t}^{(t)}}
    \end{quantikz}
\end{align}
If the round $t$ is of $ZX$-type, we pick \cref{eq:zx-circuit} as our compilation, while if it is of $XZ$-type, we pick \cref{eq:xz-circuit}.
We can now compile the \texttt{H} gates in-between the \texttt{CZ} gates into their MBQC version using \cref{lemma:h-cz-sandwich}, giving us a \texttt{CZ} network within the blue box.

If the round is of $ZX$-type, we are guaranteed that the previous round is of $X$- or $ZX$-type, meaning that the compiled version finishes with an MBQC Hadamard gate.
We can therefore merge the  \texttt{CZ} network of \cref{eq:zx-circuit} with the previous \texttt{CZ} network by using \cref{lemma:cz-cz-merge}.
We also compile the right-most Hadamard gates of the circuit into their MBQC versions.
If the round is of $XZ$-type, we compile the left-most Hadamard gates of \cref{eq:xz-circuit} into their MBQC version and merge the \texttt{CZ} networks using \cref{lemma:cz-cz-merge}.
For both types of round, this new addition to the \texttt{CZ} network can be seen as adding two layers of data nodes to the cluster state complex. For a $ZX$-type round, the ancilla nodes connect the data nodes at layers $\ell$ and $\ell+1$, according to $G_{1,Z}^{(t)}$ and $G_{1,Z}^{(t)}$ respectively. For an $XZ$-type round, the ancilla nodes connect layers $\ell+1$ and $\ell+2$, according to $G_{1,X}^{(t)}$ and $G_{1,Z}^{(t)}$ respectively. Finally, due to the $\texttt{CZ}_A$ gate in the circuits, the ancilla nodes of index $i$ and $j$ are connected to each other if $G_{i,X}^{(t)}$ and $G_{j,Z}^{(t)}$ have an odd intersection.

We finish the construction by adding the virtual data nodes at layer $\ell=0$ and $\ell=L+1$, and their connections to layers $\ell=1$ and $\ell=L$ respectively, which are added following the definition of the compressed representation of the cluster state complex.

Let us now derive the detector nodes of the complex.
We start with the original circuit and study the evolution of the spacetime stabilizers as the circuit transforms in the process described above.
For each $D \in \mathcal{D}_S$, we have a spacetime stabilizer consisting of the spackle of all the measurements in $D$.
Since every measurement is implemented in the original circuit with a controlled-Pauli gate and an ancilla qubit initialized in the $\ket{+}$ state, the spackle of a measurement is here the spackle of the $X$ operator stabilizing the corresponding ancilla qubit.
For instance, the spackle of the measurement $G^{(t)}_1$ looks as follows
\begin{align*}
    \tikzsetnextfilename{sec7-proof-12}
    \begin{quantikz}[row sep=10pt, column sep=6pt]
                       &                          & \ \ldots \ & & \gatebox{9}{4} & \mygate{3}{G^{(t)}_1}  &        & \mygate{3}{G^{(t)}_{m_t}} & \phantomgate{H} & \labeledwire{nicered}{P_1}           & \phantomgate{H} & \mygate{3}{G^{(t+1)}_1} \gatebox{9}{3} &        & \mygate{3}{G^{(t+1)}_{m_{t+1}}} & \phantomgate{H} & \labeledwire{nicered}{P_1}                 & \phantomgate{H} & \ \ldots \ & \labeledwire{nicered}{P_1}                 & \phantomgate{H} &                        \\[-1.5em]
                       & \setwiretype{n}          &            & & \vdots         &                        & \cdots &                           &                 &                                      &                 &                                        & \cdots &                                 &                 &                                            &                 &            &                                            &                 &                        \\[-1.1em]
                       &                          & \ \ldots \ & &                &                        &        &                           &                 & \labeledwire{nicered}{P_n}           &                 &                                        &        &                                 &                 & \labeledwire{nicered}{P_n}                 &                 & \ \ldots \ & \labeledwire{nicered}{P_n}                 &                 &                        \\
    \lstick{$\ket{+}$} & \labeledwire{nicered}{X} & \ \ldots \ & &                & \ctrl{-1}              &        &                           &                 & \labeledwire{nicered}{Q^{(t)}_1}     &                 &                                        &        &                                 &                 & \labeledwire{nicered}{Q^{(t)}_1}           &                 & \ \ldots \ & \labeledwire{nicered}{Q^{(t)}_1}           &                 & \meterD{M^{(t)}_1}     \\[-1.2em]
                       & \setwiretype{n}          &            & & \vdots         &                        & \ddots &                           &                 &                                      &                 &                                        &        &                                 &                 &                                            &                 &            &                                            &                 &                        \\[-0.7em]
    \lstick{$\ket{+}$} &                          & \ \ldots \ & &                &                        &        & \ctrl{-3}                 &                 & \labeledwire{nicered}{Q^{(t)}_{m_t}} &                 &                                        &        &                                 &                 & \labeledwire{nicered}{Q^{(t)}_{m_t}}       &                 & \ \ldots \ & \labeledwire{nicered}{Q^{(t)}_{m_t}}       &                 & \meterD{M^{(t)}_{m_t}} \\
    \lstick{$\ket{+}$} &                          & \ \ldots \ & &                &                        &        &                           &                 &                                      &                 & \ctrl{-4}                              &        &                                 &                 & \labeledwire{nicered}{Q^{(t+1)}_1}         &                 & \ \ldots \ & \labeledwire{nicered}{Q^{(t+1)}_1}         &                 & \meterD{M^{(t+1)}_1}   \\[-1.2em]
                       & \setwiretype{n}          &            & & \vdots         &                        & \ddots &                           &                 &                                      &                 &                                        & \ddots &                                 &                 &                                            &                 &            &                                            &                 &                        \\[-0.7em]
    \lstick{$\ket{+}$} &                          & \ \ldots \ & &                &                        &        &                           &                 &                                      &                 &                                        &        & \ctrl{-6}                       &                 & \labeledwire{nicered}{Q^{(t+1)}_{m_{t+1}}} &                 & \ \ldots \ & \labeledwire{nicered}{Q^{(t+1)}_{m_{t+1}}} &                 & \meterD{M^{(t+1)}_{m_{t+1}}}
    \end{quantikz}
\end{align*}
where $P_1,\ldots,P_n$ are the Pauli operators that comprise $G^{(t)}_1$, and $Q^{(\ell)}_1,\ldots,Q^{(\ell)}_{m_\ell}$ for $\ell \geq t$ are residual Pauli operators due to $P_1,\ldots,P_n$ potentially not commuting with some $G^{(\ell)}_i$.
When multiplying the spackles of all the measurements involved in a spacetime stabilizer, all those residual operators will be equal to $I$ or $M^{(\ell)}_{i}$ (such that the operator commutes with the final measurement).
We therefore only need to look at the evolution of $P_1,\ldots,P_n$ at the different rounds when transforming the circuit.

If the round is of type $XZ$ or $ZX$, the first step was to compile each round of the circuit into \cref{eq:xz-circuit} or \cref{eq:zx-circuit}.
When compiling the Hadamard gates into their MBQC versions in the next step, any spacetime operator transforms as:
\begin{align} \label{eq:hadamard-detector-propagation}
    \begin{split}
        \tikzsetnextfilename{sec7-proof-13}
        \begin{quantikz}[row sep=10pt, column sep=8pt]
            & \labeledwire{nicered}{Z} & \gate{H} & \labeledwire{nicered}{X} &
        \end{quantikz}
        \longrightarrow
        \tikzsetnextfilename{sec7-proof-14}
        \begin{quantikz}[row sep=10pt, column sep=8pt]
            \lstick{$\ket{+}$} & \labeledwire{nicered}{X}  & \ctrl{1}  & \labeledwire{nicered}{X}  &            \\
                            & \labeledwire{nicered}{Z}  & \ctrl{-1} &                           & \meterD{X}
        \end{quantikz}
        \\
        \tikzsetnextfilename{sec7-proof-15}
        \begin{quantikz}[row sep=10pt, column sep=8pt]
            & \labeledwire{nicered}{X} & \gate{H} & \labeledwire{nicered}{Z} &
        \end{quantikz}
        \longrightarrow
        \tikzsetnextfilename{sec7-proof-16}
        \begin{quantikz}[row sep=10pt, column sep=8pt]
            \lstick{$\ket{+}$} &                           & \ctrl{1}  & \labeledwire{nicered}{Z}  &            \\
                            & \labeledwire{nicered}{X}  & \ctrl{-1} & \labeledwire{nicered}{X}  & \meterD{X}
        \end{quantikz}
    \end{split}
\end{align}
For instance, the spackle of a measurement at a $ZX$-type round transforms as:
\begin{align*} \label{eq:zx-circuit-detector}
    \tikzsetnextfilename{sec7-proof-17}
    \begin{quantikz}[row sep=11pt, column sep=6pt]
        \lstick{$\ket{+}$} & \phantomgate{H} \labeledwire{nicered}{P_1^X}       & \gatebox{12}{13} &                           &        &                           &               &                       &           &                                 &        &                                 &                  &                       & \ctrl{3}  &                 & \labeledwire{nicered}{P_1}           & &                        \\[-1.2em]
                           & \setwiretype{n}                                    & \vdots           &                           &        &                           &               &                       &           &                                 &        &                                 &                  & \reflectbox{$\ddots$} &           &                 &                                      & &                        \\[-0.7em]
        \lstick{$\ket{+}$} & \labeledwire{nicered}{P_n^X}                       &                  &                           &        &                           &               &                       &           &                                 &        &                                 & \ctrl{3}         &                       &           &                 & \labeledwire{nicered}{P_n}           & &                        \\
        \lstick{$\ket{+}$} & \phantomgate{H} \labeledwire{nicered}{\bar{P}_1^X} &                  &                           &        &                           &               &                       & \ctrl{3}  & \mygate{3}{\bar{G}_{1,X}^{(t)}} &        & \mygate{3}{\bar{G}_{n,X}^{(t)}} &                  &                       & \ctrl{-3} & \phantomgate{H} & \labeledwire{nicered}{\bar{P}_1^X}   & & \meterD{X}             \\[-1.5em]
                           & \setwiretype{n}                                    & \vdots           &                           &        &                           &               & \reflectbox{$\ddots$} &           &                                 & \cdots &                                 &                  & \reflectbox{$\ddots$} &           &                 &                                      & &                        \\[-1.1em]
        \lstick{$\ket{+}$} & \labeledwire{nicered}{\bar{P}_n^X}                 &                  &                           &        &                           & \ctrl{3}      &                       &           &                                 &        &                                 & \ctrl{-3}        &                       &           &                 & \labeledwire{nicered}{\bar{P}_n^X}   & & \meterD{X}             \\
                           & \phantomgate{H}                                    &                  & \mygate{3}{G_{1,Z}^{(t)}} &        & \mygate{3}{G_{n,Z}^{(t)}} &               &                       & \ctrl{-3} &                                 &        &                                 &                  &                       &           &                 &                                      & & \meterD{X}             \\[-1.5em]
                           & \setwiretype{n}                                    & \vdots           &                           & \cdots &                           &               & \reflectbox{$\ddots$} &           &                                 & \cdots &                                 &                  &                       &           &                 &                                      & &                        \\[-1.1em]
                           &                                                    &                  &                           &        &                           & \ctrl{-3}     &                       &           &                                 &        &                                 &                  &                       &           &                 &                                      & & \meterD{X}             \\
        \lstick{$\ket{+}$} & \labeledwire{nicered}{X}                           &                  & \ctrl{-1}                 &        &                           &               &                       &           & \ctrl{-4}                       &        &                                 &                  & \mygate{3}{CZ_A}      &           &                 & \labeledwire{nicered}{Q_1^{(t)}}     & & \meterD{M_1^{(t)}}     \\[-1.2em]
                           & \setwiretype{n}                                    & \vdots           &                           & \ddots &                           &               &                       &           &                                 & \ddots &                                 &                  &                       &           &                 &                                      & &                        \\[-0.7em]
        \lstick{$\ket{+}$} &                                                    &                  &                           &        & \ctrl{-3}                 &               &                       &           &                                 &        & \ctrl{-6}                       &                  &                       &           &                 & \labeledwire{nicered}{Q_{m_t}^{(t)}} & & \meterD{M_{m_t}^{(t)}}
    \end{quantikz}
\end{align*}
where $\bar{P}_i$ is the Pauli operator obtained from $P_i$ by exchanging $X$ and $Z$, and where $P_i^X$ (resp.\ $\bar{P}_i^X$) is the restriction of $P_i$ (resp.\ $\bar{P}_i$) to its $X$ part.
We can see that we have $X$ operators on all the data qubits located one layer above those connected to the ancilla qubit in the cluster state complex.
Compiling the subsequent Hadamard gates and propagating the $P_1,\ldots,P_n$ to the rest of the circuit leads to an $X$ operator at every other layer, for each data qubit that supports an $X$ at a given layer, as one can see from \cref{eq:hadamard-detector-propagation}.
This includes the last layer of virtual nodes, which represent the $Z$ part of $P_1,\ldots,P_n$ or $\bar{P}_1,\ldots,\bar{P}_n$ at the output of the circuit.
The same analysis can be done for ancilla qubits located at $XZ$-type, $X$-type and $Z$-type rounds, leading to $X$ operators on data qubits located every other layer, starting one layer above every data qubit connected to the ancilla qubit in the cluster state complex.

Consequently, the spackle of a measurement maps, up to the residual operators $Q_i^{(\ell)}$, to the spackle of the corresponding ancilla node in the cluster state complex.
Since those residual operators cancel when considering the products of spackles forming a spacetime stabilizer, such a spacetime stabilizer maps exactly to the product of spackles of the associated ancilla nodes in the cluster state complex.

It remains to show that the spacetime stabilizers built from the backle of measurements map to the backles of the corresponding ancilla nodes in the cluster state complex.
Let us look at the backle of the measurement $G_{1,Z}^{(t)}$ for a $ZX$-type round preceded by a Hadamard gate:
\begin{align*}
    \tikzsetnextfilename{sec7-proof-18}
    \begin{quantikz}[row sep=11pt, column sep=6pt]
        \lstick{$\ket{+}$} & \phantomgate{H}                      & & \gatebox{12}{14} &           &                       &           &                           &        &                           &               &                       & \ctrl{3}  & \mygate{3}{\bar{G}_{1,X}^{(t)}} &        & \mygate{3}{\bar{G}_{n,X}^{(t)}} &                   & \phantomgate{H} &                                    & &                        \\[-1.5em]
                           & \setwiretype{n}                      & & \vdots           &           &                       &           &                           &        &                           &               & \reflectbox{$\ddots$} &           &                                 & \cdots &                                 &                   &                 &                                    & &                        \\[-1.1em]
        \lstick{$\ket{+}$} &                                      & &                  &           &                       &           &                           &        &                           & \ctrl{3}      &                       &           &                                 &        &                                 &                   &                 &                                    & &                        \\
        \lstick{$\ket{+}$} & \phantomgate{H}                      & &                  &           &                       & \ctrl{3}  & \mygate{3}{G_{1,Z}^{(t)}} &        & \mygate{3}{G_{n,Z}^{(t)}} &               &                       & \ctrl{-3} &                                 &        &                                 &                   &                 & \labeledwire{nicered}{P_1^X}       & & \meterD{X}             \\[-1.5em]
                           & \setwiretype{n}                      & & \vdots           &           & \reflectbox{$\ddots$} &           &                           & \cdots &                           &               & \reflectbox{$\ddots$} &           &                                 & \cdots &                                 &                   &                 &                                    & &                        \\[-1.1em]
        \lstick{$\ket{+}$} &                                      & &                  & \ctrl{3}  &                       &           &                           &        &                           & \ctrl{-3}     &                       &           &                                 &        &                                 &                   &                 & \labeledwire{nicered}{P_n^X} & & \meterD{X}             \\
                           & \labeledwire{nicered}{\bar{P}_1}     & &                  &           &                       & \ctrl{-3} &                           &        &                           &               &                       &           &                                 &        &                                 &                   &                 & \labeledwire{nicered}{\bar{P}_1^X} & & \meterD{X}             \\[-1.5em]
                           & \setwiretype{n}                      & & \vdots           &           & \reflectbox{$\ddots$} &           &                           &        &                           &               &                       &           &                                 &        &                                 &                   &                 &                                    & &                        \\[-1.1em]
                           & \labeledwire{nicered}{\bar{P}_n}     & &                  & \ctrl{-3} &                       &           &                           &        &                           &               &                       &           &                                 &        &                                 &                   &                 & \labeledwire{nicered}{\bar{P}_n^X}       & & \meterD{X}             \\
        \lstick{$\ket{+}$} & \labeledwire{nicered}{Q_1^{(t)}}     & &                  &           &                       &           & \ctrl{-4}                 &        &                           &               &                       &           & \ctrl{-7}                       &        &                                 & \mygate{3}{CZ_A}  &                 & \labeledwire{nicered}{M_1^{(t)}}   & & \meterD{M_1^{(t)}}     \\[-1.2em]
                           & \setwiretype{n}                      & & \vdots           &           &                       &           &                           & \ddots &                           &               &                       &           &                                 & \ddots &                                 &                   &                 &                                    & &                        \\[-0.7em]
        \lstick{$\ket{+}$} & \labeledwire{nicered}{Q_{m_t}^{(t)}} & &                  &           &                       &           &                           &        & \ctrl{-6}                 &               &                       &           &                                 &        & \ctrl{-9}                       &                   &                 &                                    & & \meterD{M_{m_t}^{(t)}}
    \end{quantikz}
\end{align*}
We can see that we have $X$ operators on all data qubits located one layer below those connected to the ancilla qubit in the cluster state complex.
Compiling the preceding Hadamard gates and propagating the $\bar{P}_1,\ldots,\bar{P}_n$ to the rest of the circuit leads to $X$ operators every other layers, as one can see from \cref{eq:hadamard-detector-propagation}.
This includes the first layer of virtual nodes, which represent the $Z$ part of $P_1,\ldots,P_n$ or $\bar{P}_1,\ldots,\bar{P}_n$ at the input of the circuit.
The same analysis can be done for ancilla qubits located at $XZ$-type, $X$-type and $Z$-type rounds, leading to $X$ operators on data qubits located every other layer, starting one layer below every data qubit connected to the ancilla qubit in the cluster state complex.
\end{proof}

\section{Discussion}
\label{sec:discussion}

We introduced a new framework for spacetime codes based on chain complexes, enabling the compilation of Clifford circuits into cluster states while preserving their fault-tolerant properties. This framework can be applied to construct cluster states from non-CSS, subsystem, or dynamical codes. It can also be used to construct cluster states implementing logical operations by compiling the logical circuit into MBQC using our prescription.

We believe that our framework can be applied more broadly. For example, in \cref{sec:from-spacetime-codes-to-cluster-states}, we mapped certain Clifford circuits--—those built from single-qubit gates, controlled-Pauli gates, and single-qubit Pauli measurements—--into cluster states, leaving open whether more general Clifford circuits admit such a mapping.
Furthermore, while we only presented applications of our framework to MBQC compiling, we believe it could also be applied to other areas of fault-tolerant compilation, such as more general fusion-based quantum computing~\cite{bartolucci2021fusionbased}, hook error detection~\cite{beverland2024fault} and flag qubit design~\cite{chao2018quantum,chamberland2018flag}, or Floquetification of codes~\cite{Townsend_Teague_2023,delafuente2024xyzrubycode,delafuente2024dynamical,rodatz2024floquetifying}. Finally, recent discoveries show that cluster states can be constructed beyond foliation~\cite{nickerson2018measurement}.
Examples include crystal structures~\cite{newman2020generating}, which exhibit similar macroscopic behaviour to the foliated surface code but with different microscopic implementations.
Our framework could help us understand these structures either as the compilation of spacetime codes or as fault-tolerant transformations of the foliated surface code circuit

Another interesting research direction emerging from our work concerns its connection to other frameworks. For instance, it would be interesting to study the connection between fault-tolerant maps and distance-preserving rewrite rules in ZX-calculus, which have been widely used in the context of Floquetifying codes~\cite{rodatz2024floquetifying,delafuente2024xyzrubycode,delafuente2024dynamical,rodatz2025fault}, but also to unify different models of fault-tolerance~\cite{bombin2024unifying}.
Similarly, chain maps have been used to reduce the stabilizer weights of low-density parity-check codes~\cite{hastings2023quantumweightreduction,wills2025tradeoffConstructions,sabo2024weight}, to perform lattice surgery on quantum low-density parity-check codes~\cite{cowtan2024css,ide2025fault,poirson2025engineering}, and to map decoders between different codes~\cite{delfosse2014decoding,kubica2023efficient}, though it remains unclear whether these chain maps fit within our framework.

Finally, we identify several promising directions for extending our framework.
One natural extension is to explicitly track the noise model as chain complexes are transformed by fault-tolerant maps.
Doing so would ensure that simulations of different chain complexes yield identical logical error rates.
Concretely, this could be achieved by associating a probability distribution to the middle space of the chain complex and tracking how this distribution evolves under circuit transformations.
A second direction is to generalize fault-tolerant maps from individual chain complexes to families of chain complexes, where distance is preserved only asymptotically. Such a generalization could enable more powerful mapping results and, when combined with the probabilistic approach above, allow threshold theorems to be carried over from one family of fault-tolerant circuits to another.

\section*{Acknowledgments}

The authors would like to thank Dan Browne, Nicolas Delfosse, Oscar Higgott, Timo Hillmann, Julio Magdalena De La Fuente, Ewan Murphy, Clément Poirson and Abhishek Rajput for enlightening discussions.
Research at Perimeter Institute is supported in part by the Government of Canada through the Department of Innovation, Science and Economic Development Canada and by the Province of Ontario through the Ministry of Colleges and Universities.
This research was supported in part by grant NSF PHY-2309135 to the Kavli Institute for Theoretical Physics (KITP).
AP was supported in part by the Engineering and Physical Sciences Research Council (EP/S021582/1).

\clearpage

\bibliographystyle{quantum}
\bibliography{refs}

\begin{thebibliography}{100}

\bibitem{hastings2021dynamically}
Matthew~B. Hastings and Jeongwan Haah.
\newblock ``Dynamically {G}enerated {L}ogical {Q}ubits''.
\newblock \href{https://dx.doi.org/10.22331/q-2021-10-19-564}{{Quantum} {\bf 5}, 564}~(2021).

\bibitem{Kesselring_2024}
Markus~S. Kesselring, Julio~C. Magdalena de~la Fuente, Felix Thomsen, Jens Eisert, Stephen~D. Bartlett, and Benjamin~J. Brown.
\newblock ``Anyon condensation and the color code''.
\newblock \href{https://dx.doi.org/10.1103/prxquantum.5.010342}{PRX Quantum{\bf 5}}~(2024).

\bibitem{Townsend_Teague_2023}
Alex Townsend-Teague, Julio Magdalena de~la Fuente, and Markus Kesselring.
\newblock ``Floquetifying the colour code''.
\newblock \href{https://dx.doi.org/10.4204/eptcs.384.14}{Electronic Proceedings in Theoretical Computer Science {\bf 384}, 265–303}~(2023).

\bibitem{Davydova_2023}
Margarita Davydova, Nathanan Tantivasadakarn, and Shankar Balasubramanian.
\newblock ``Floquet codes without parent subsystem codes''.
\newblock \href{https://dx.doi.org/10.1103/prxquantum.4.020341}{PRX Quantum{\bf 4}}~(2023).

\bibitem{Davydova_2024}
Margarita Davydova, Nathanan Tantivasadakarn, Shankar Balasubramanian, and David Aasen.
\newblock ``Quantum computation from dynamic automorphism codes''.
\newblock \href{https://dx.doi.org/10.22331/q-2024-08-27-1448}{Quantum {\bf 8}, 1448}~(2024).

\bibitem{kribs2005unified}
David Kribs, Raymond Laflamme, and David Poulin.
\newblock ``Unified and generalized approach to quantum error correction''.
\newblock \href{https://dx.doi.org/10.1103/PhysRevLett.94.180501}{Phys. Rev. Lett. {\bf 94}, 180501}~(2005).

\bibitem{poulin2005stabilizer}
David Poulin.
\newblock ``Stabilizer {{Formalism}} for {{Operator Quantum Error Correction}}''.
\newblock \href{https://dx.doi.org/10.1103/PhysRevLett.95.230504}{Phys. Rev. Lett. {\bf 95}, 230504}~(2005).

\bibitem{kribs2006operator}
David~W. Kribs, Raymond Laflamme, David Poulin, and Maia Lesosky.
\newblock ``Operator quantum error correction''~(2006).
\newblock  \href{http://arxiv.org/abs/quant-ph/0504189}{arXiv:quant-ph/0504189}.

\bibitem{Higgott_2021}
Oscar Higgott and Nikolas~P. Breuckmann.
\newblock ``Subsystem codes with high thresholds by gauge fixing and reduced qubit overhead''.
\newblock \href{https://dx.doi.org/10.1103/physrevx.11.031039}{Physical Review X{\bf 11}}~(2021).

\bibitem{gidney2023less}
Craig Gidney and Dave Bacon.
\newblock ``Less {{Bacon More Threshold}}''~(2023).
\newblock  \href{http://arxiv.org/abs/2305.12046}{arXiv:2305.12046}.

\bibitem{alam2024dynamical}
M.~Sohaib Alam and Eleanor Rieffel.
\newblock ``Dynamical {{Logical Qubits}} in the {{Bacon-Shor Code}}''~(2024).
\newblock  \href{http://arxiv.org/abs/2403.03291}{arXiv:2403.03291}.

\bibitem{alam2025baconshor}
M.~Sohaib Alam, Jiajun Chen, and Thomas~R. Scruby.
\newblock ``Bacon-{{Shor Board Games}}''~(2025).
\newblock  \href{http://arxiv.org/abs/2504.02749}{arXiv:2504.02749}.

\bibitem{higgott2023improved}
Oscar Higgott, Thomas~C Bohdanowicz, Aleksander Kubica, Steven~T Flammia, and Earl~T Campbell.
\newblock ``Improved decoding of circuit noise and fragile boundaries of tailored surface codes''.
\newblock \href{https://dx.doi.org/10.1103/PhysRevX.13.031007}{Physical Review X {\bf 13}, 031007}~(2023).

\bibitem{mcewen2023relaxing}
Matt McEwen, Dave Bacon, and Craig Gidney.
\newblock ``Relaxing {H}ardware {R}equirements for {S}urface {C}ode {C}ircuits using {T}ime-dynamics''.
\newblock \href{https://dx.doi.org/10.22331/q-2023-11-07-1172}{{Quantum} {\bf 7}, 1172}~(2023).

\bibitem{gidney2023new}
Craig Gidney and Cody Jones.
\newblock ``New circuits and an open source decoder for the color code''~(2023).
\newblock  \href{http://arxiv.org/abs/2312.08813}{arXiv:2312.08813}.

\bibitem{shaw2024lowering}
Mackenzie~H. Shaw and Barbara~M. Terhal.
\newblock ``Lowering {{Connectivity Requirements For Bivariate Bicycle Codes Using Morphing Circuits}}''~(2024).
\newblock  \href{http://arxiv.org/abs/2407.16336}{arXiv:2407.16336}.

\bibitem{delfosse2024lowcostnoisereductionclifford}
Nicolas Delfosse and Edwin Tham.
\newblock ``Low-cost noise reduction for clifford circuits''~(2024).
\newblock  \href{http://arxiv.org/abs/2407.06583}{arXiv:2407.06583}.

\bibitem{beverland2024fault}
Michael~E. Beverland, Shilin Huang, and Vadym Kliuchnikov.
\newblock ``Fault tolerance of stabilizer channels''~(2024).
\newblock  \href{http://arxiv.org/abs/2401.12017}{arXiv:2401.12017}.

\bibitem{bacon2015sparse}
Dave Bacon, Steven~T. Flammia, Aram~W. Harrow, and Jonathan Shi.
\newblock ``Sparse quantum codes from quantum circuits''.
\newblock In Proceedings of the Forty-Seventh Annual ACM Symposium on Theory of Computing.
\newblock \href{https://dx.doi.org/10.1145/2746539.2746608}{Page 327–334}.
\newblock ~(2015).

\bibitem{pryadko2020maximumlikelihood}
Leonid~P. Pryadko.
\newblock ``On maximum-likelihood decoding with circuit-level errors''.
\newblock \href{https://dx.doi.org/10.22331/q-2020-08-06-304}{{Quantum} {\bf 4}, 304}~(2020).

\bibitem{gottesman2022opportunities}
Daniel Gottesman.
\newblock ``Opportunities and challenges in fault-tolerant quantum computation''~(2022).
\newblock  \href{http://arxiv.org/abs/2210.15844}{arXiv:2210.15844}.

\bibitem{gidney2021stim}
Craig Gidney.
\newblock ``Stim: a fast stabilizer circuit simulator''.
\newblock \href{https://dx.doi.org/10.22331/q-2021-07-06-497}{Quantum {\bf 5}, 497}~(2021).

\bibitem{delfosse2023spacetime}
Nicolas Delfosse and Adam Paetznick.
\newblock ``Spacetime codes of clifford circuits''~(2023).
\newblock  \href{http://arxiv.org/abs/2304.05943}{arXiv:2304.05943}.

\bibitem{yoshida2011classification}
Beni Yoshida.
\newblock ``Classification of quantum phases and topology of logical operators in an exactly solved model of quantum codes''.
\newblock \href{https://dx.doi.org/10.1016/j.aop.2010.10.009}{Annals of Physics {\bf 326}, 15--95}~(2011).

\bibitem{bombin2012universal}
H~Bomb{\'i}n, Guillaume {Duclos-Cianci}, and David Poulin.
\newblock ``Universal topological phase of two-dimensional stabilizer codes''.
\newblock \href{https://dx.doi.org/10.1088/1367-2630/14/7/073048}{New J. Phys. {\bf 14}, 073048}~(2012).

\bibitem{bombin2014structure}
H{\'e}ctor Bomb{\'i}n.
\newblock ``Structure of {{2D Topological Stabilizer Codes}}''.
\newblock \href{https://dx.doi.org/10.1007/s00220-014-1893-4}{Commun. Math. Phys. {\bf 327}, 387--432}~(2014).

\bibitem{kubica2015unfolding}
Aleksander Kubica, Beni Yoshida, and Fernando Pastawski.
\newblock ``Unfolding the color code''.
\newblock \href{https://dx.doi.org/10.1088/1367-2630/17/8/083026}{New Journal of Physics {\bf 17}, 083026}~(2015).

\bibitem{delfosse2014decoding}
Nicolas Delfosse.
\newblock ``Decoding color codes by projection onto surface codes''.
\newblock \href{https://dx.doi.org/10.1103/PhysRevA.89.012317}{Phys. Rev. A {\bf 89}, 012317}~(2014).
\newblock  \href{http://arxiv.org/abs/1308.6207}{arXiv:1308.6207}.

\bibitem{kubica2023efficient}
Aleksander Kubica and Nicolas Delfosse.
\newblock ``Efficient color code decoders in {$d\geq 2$} dimensions from toric code decoders''.
\newblock \href{https://dx.doi.org/10.22331/q-2023-02-21-929}{{Quantum} {\bf 7}, 929}~(2023).

\bibitem{khesin2024equivalence}
Andrey~Boris Khesin and Alexander Li.
\newblock ``Equivalence classes of quantum error-correcting codes''~(2024).
\newblock  \href{http://arxiv.org/abs/2406.12083}{arXiv:2406.12083}.

\bibitem{cross2025small}
Andrew Cross and Drew Vandeth.
\newblock ``Small binary stabilizer subsystem codes''~(2025).
\newblock  \href{http://arxiv.org/abs/2501.17447}{arXiv:2501.17447}.

\bibitem{delafuente2024xyzrubycode}
Julio C.~Magdalena de~la Fuente, Josias Old, Alex Townsend-Teague, Manuel Rispler, Jens Eisert, and Markus Müller.
\newblock ``The xyz ruby code: Making a case for a three-colored graphical calculus for quantum error correction in spacetime''~(2024).
\newblock  \href{http://arxiv.org/abs/2407.08566}{arXiv:2407.08566}.

\bibitem{delafuente2024dynamical}
Julio C.~Magdalena de~la Fuente.
\newblock ``Dynamical weight reduction of pauli measurements''~(2024).
\newblock  \href{http://arxiv.org/abs/2410.12527}{arXiv:2410.12527}.

\bibitem{rodatz2024floquetifying}
Benjamin Rodatz, Boldizsár Poór, and Aleks Kissinger.
\newblock ``Floquetifying stabiliser codes with distance-preserving rewrites''~(2024).
\newblock  \href{http://arxiv.org/abs/2410.17240}{arXiv:2410.17240}.

\bibitem{xu2025faulttolerant}
Yichen Xu and Arpit Dua.
\newblock ``Fault-tolerant protocols through spacetime concatenation''~(2025).
\newblock  \href{http://arxiv.org/abs/2504.08918}{arXiv:2504.08918}.

\bibitem{Coecke_2011}
Bob Coecke and Ross Duncan.
\newblock ``Interacting quantum observables: categorical algebra and diagrammatics''.
\newblock \href{https://dx.doi.org/10.1088/1367-2630/13/4/043016}{New Journal of Physics {\bf 13}, 043016}~(2011).

\bibitem{coecke2018picturing}
Bob Coecke and Aleks Kissinger.
\newblock ``Picturing quantum processes: A first course on quantum theory and diagrammatic reasoning''.
\newblock \href{https://dx.doi.org/10.1007/978-3-319-91376-6_6}{Cambridge University Press}. ~(2017).

\bibitem{bourassa2021blueprintscalable}
J.~Eli Bourassa, Rafael~N. Alexander, Michael Vasmer, Ashlesha Patil, Ilan Tzitrin, Takaya Matsuura, Daiqin Su, Ben~Q. Baragiola, Saikat Guha, Guillaume Dauphinais, Krishna~K. Sabapathy, Nicolas~C. Menicucci, and Ish Dhand.
\newblock ``Blueprint for a {S}calable {P}hotonic {F}ault-{T}olerant {Q}uantum {C}omputer''.
\newblock \href{https://dx.doi.org/10.22331/q-2021-02-04-392}{{Quantum} {\bf 5}, 392}~(2021).

\bibitem{tzitrin2021faulttolerant}
Ilan Tzitrin, Takaya Matsuura, Rafael~N. Alexander, Guillaume Dauphinais, J.~Eli Bourassa, Krishna~K. Sabapathy, Nicolas~C. Menicucci, and Ish Dhand.
\newblock ``Fault-{{Tolerant Quantum Computation}} with {{Static Linear Optics}}''.
\newblock \href{https://dx.doi.org/10.1103/PRXQuantum.2.040353}{PRX Quantum {\bf 2}, 040353}~(2021).

\bibitem{walshe2025linearoptical}
Blayney~W. Walshe, Ben~Q. Baragiola, Hugo Ferretti, Jos{\'e} Gefaell, Michael Vasmer, Ryohei Weil, Takaya Matsuura, Thomas Jaeken, Giacomo Pantaleoni, Zhihua Han, Timo Hillmann, Nicolas~C. Menicucci, Ilan Tzitrin, and Rafael~N. Alexander.
\newblock ``Linear-{{Optical Quantum Computation}} with {{Arbitrary Error-Correcting Codes}}''.
\newblock \href{https://dx.doi.org/10.1103/PhysRevLett.134.100602}{Phys. Rev. Lett. {\bf 134}, 100602}~(2025).

\bibitem{aghaeerad2025scaling}
H.~Aghaee~Rad, T.~Ainsworth, R.~N. Alexander, B.~Altieri, M.~F. Askarani, R.~Baby, L.~Banchi, B.~Q. Baragiola, J.~E. Bourassa, R.~S. Chadwick, I.~Charania, H.~Chen, M.~J. Collins, P.~Contu, N.~D'Arcy, G.~Dauphinais, R.~De~Prins, D.~Deschenes, I.~Di~Luch, S.~Duque, P.~Edke, S.~E. Fayer, S.~Ferracin, H.~Ferretti, J.~Gefaell, S.~Glancy, C.~{Gonz{\'a}lez-Arciniegas}, T.~Grainge, Z.~Han, J.~Hastrup, L.~G. Helt, T.~Hillmann, J.~Hundal, S.~Izumi, T.~Jaeken, M.~Jonas, S.~Kocsis, I.~Krasnokutska, M.~V. Larsen, P.~Laskowski, F.~Laudenbach, J.~Lavoie, M.~Li, E.~Lomonte, C.~E. Lopetegui, B.~Luey, A.~P. Lund, C.~Ma, L.~S. Madsen, D.~H. Mahler, L.~Mantilla~Calder{\'o}n, M.~Menotti, F.~M. Miatto, B.~Morrison, P.~J. Nadkarni, T.~Nakamura, L.~Neuhaus, Z.~Niu, R.~Noro, K.~Papirov, A.~Pesah, D.~S. Phillips, W.~N. Plick, T.~Rogalsky, F.~Rortais, J.~{Sabines-Chesterking}, S.~{Safavi-Bayat}, E.~Sazhaev, M.~Seymour, K.~Rezaei~Shad, M.~Silverman, S.~A. Srinivasan, M.~Stephan, Q.~Y. Tang, J.~F. Tasker, Y.~S. Teo, R.~B. Then, J.~E. Tremblay, I.~Tzitrin, V.~D. Vaidya, M.~Vasmer, Z.~Vernon, L.~F. S. S.~M. Villalobos, B.~W. Walshe, R.~Weil, X.~Xin, X.~Yan, Y.~Yao, M.~Zamani~Abnili, and Y.~Zhang.
\newblock ``Scaling and networking a modular photonic quantum computer''.
\newblock \href{https://dx.doi.org/10.1038/s41586-024-08406-9}{NaturePages 1--8}~(2025).

\bibitem{bartolucci2021fusionbased}
Sara Bartolucci, Patrick Birchall, Hector Bombin, Hugo Cable, Chris Dawson, Mercedes {Gimeno-Segovia}, Eric Johnston, Konrad Kieling, Naomi Nickerson, Mihir Pant, Fernando Pastawski, Terry Rudolph, and Chris Sparrow.
\newblock ``Fusion-based quantum computation''~(2021).
\newblock  \href{http://arxiv.org/abs/2101.09310}{arXiv:2101.09310}.

\bibitem{bombin2023fault}
H\'{e}ctor Bomb\'{i}n, Chris Dawson, Terry Farrelly, Yehua Liu, Naomi Nickerson, Mihir Pant, Fernando Pastawski, and Sam Roberts.
\newblock ``Fault-tolerant complexes''~(2023).
\newblock  \href{http://arxiv.org/abs/2308.07844}{arXiv:2308.07844}.

\bibitem{alexander2024manufacturable}
Koen Alexander, Andrea Bahgat, Avishai Benyamini, Dylan Black, Damien Bonneau, Stanley Burgos, Ben Burridge, Geoff Campbell, Gabriel Catalano, Alex Ceballos, Chia-Ming Chang, C.~J. Chung, Fariba Danesh, Tom Dauer, Michael Davis, Eric Dudley, Ping {Er-Xuan}, Josep Fargas, Alessandro Farsi, Colleen Fenrich, Jonathan Frazer, Masaya Fukami, Yogeeswaran Ganesan, Gary Gibson, Mercedes {Gimeno-Segovia}, Sebastian Goeldi, Patrick Goley, Ryan Haislmaier, Sami Halimi, Paul Hansen, Sam Hardy, Jason Horng, Matthew House, Hong Hu, Mehdi Jadidi, Henrik Johansson, Thomas Jones, Vimal Kamineni, Nicholas Kelez, Ravi Koustuban, George Kovall, Peter Krogen, Nikhil Kumar, Yong Liang, Nicholas LiCausi, Dan Llewellyn, Kimberly Lokovic, Michael Lovelady, Vitor Manfrinato, Ann Melnichuk, Mario Souza, Gabriel Mendoza, Brad Moores, Shaunak Mukherjee, Joseph Munns, Francois-Xavier Musalem, Faraz Najafi, Jeremy~L. O'Brien, J.~Elliott Ortmann, Sunil Pai, Bryan Park, Hsuan-Tung Peng, Nicholas Penthorn, Brennan Peterson, Matt Poush, Geoff~J. Pryde, Tarun Ramprasad, Gareth Ray, Angelita Rodriguez, Brian Roxworthy, Terry Rudolph, Dylan~J. Saunders, Pete Shadbolt, Deesha Shah, Hyungki Shin, Jake Smith, Ben Sohn, Young-Ik Sohn, Gyeongho Son, Chris Sparrow, Matteo Staffaroni, Camille Stavrakas, Vijay Sukumaran, Davide Tamborini, Mark~G. Thompson, Khanh Tran, Mark Triplet, Maryann Tung, Alexey Vert, Mihai~D. Vidrighin, Ilya Vorobeichik, Peter Weigel, Mathhew Wingert, Jamie Wooding, and Xinran Zhou.
\newblock ``A manufacturable platform for photonic quantum computing''~(2024).
\newblock  \href{http://arxiv.org/abs/2404.17570}{arXiv:2404.17570}.

\bibitem{raussendorf2003measurement}
Robert Raussendorf, Daniel~E. Browne, and Hans~J. Briegel.
\newblock ``Measurement-based quantum computation on cluster states''.
\newblock \href{https://dx.doi.org/10.1103/PhysRevA.68.022312}{Phys. Rev. A {\bf 68}, 022312}~(2003).

\bibitem{raussendorf2006fault}
Robert Raussendorf, Jim Harrington, and Kovid Goyal.
\newblock ``A fault-tolerant one-way quantum computer''.
\newblock \href{https://dx.doi.org/10.1016/j.aop.2006.01.012}{Annals of physics {\bf 321}, 2242--2270}~(2006).

\bibitem{broadbent2009parallelizing}
Anne Broadbent and Elham Kashefi.
\newblock ``Parallelizing quantum circuits''.
\newblock \href{https://dx.doi.org/10.1016/j.tcs.2008.12.046}{Theoretical computer science {\bf 410}, 2489--2510}~(2009).

\bibitem{gottesman1997stabilizer}
Daniel Gottesman.
\newblock ``Stabilizer {{Codes}} and {{Quantum Error Correction}}''.
\newblock PhD thesis.
\newblock Caltech.
\newblock ~(1997).
\newblock  \href{http://arxiv.org/abs/quant-ph/9705052}{arXiv:quant-ph/9705052}.

\bibitem{bolt2016foliated}
A.~Bolt, G.~Duclos-Cianci, D.~Poulin, and T.~M. Stace.
\newblock ``Foliated quantum error-correcting codes''.
\newblock \href{https://dx.doi.org/10.1103/PhysRevLett.117.070501}{Phys. Rev. Lett. {\bf 117}, 070501}~(2016).

\bibitem{BrownUniversalFTMBQC2020}
Benjamin~J. Brown and Sam Roberts.
\newblock ``Universal fault-tolerant measurement-based quantum computation''.
\newblock \href{https://dx.doi.org/10.1103/PhysRevResearch.2.033305}{Phys. Rev. Res. {\bf 2}, 033305}~(2020).

\bibitem{paesani2023highthreshold}
Stefano Paesani and Benjamin~J. Brown.
\newblock ``High-threshold quantum computing by fusing one-dimensional cluster states''.
\newblock \href{https://dx.doi.org/10.1103/physrevlett.131.120603}{Physical Review Letters{\bf 131}}~(2023).

\bibitem{negari2024spacetime}
Amir-Reza Negari, Tyler~D. Ellison, and Timothy~H. Hsieh.
\newblock ``Spacetime {{Markov}} length: A diagnostic for fault tolerance via mixed-state phases''~(2024).
\newblock  \href{http://arxiv.org/abs/2412.00193}{arXiv:2412.00193}.

\bibitem{newman2020generating}
Michael Newman, Leonardo~Andreta de~Castro, and Kenneth~R Brown.
\newblock ``Generating fault-tolerant cluster states from crystal structures''.
\newblock \href{https://dx.doi.org/10.22331/q-2020-07-13-295}{Quantum {\bf 4}, 295}~(2020).

\bibitem{hillmann2024single}
Timo Hillmann, Guillaume Dauphinais, Ilan Tzitrin, and Michael Vasmer.
\newblock ``Single-shot and measurement-based quantum error correction via fault complexes''~(2024).
\newblock  \href{http://arxiv.org/abs/2410.12963}{arXiv:2410.12963}.

\bibitem{nickerson2018measurement}
Naomi Nickerson and H{\'e}ctor Bomb{\'\i}n.
\newblock ``Measurement based fault tolerance beyond foliation''~(2018).
\newblock  \href{http://arxiv.org/abs/1810.09621}{arXiv:1810.09621}.

\bibitem{bombin2024unifying}
Hector Bombin, Daniel Litinski, Naomi Nickerson, Fernando Pastawski, and Sam Roberts.
\newblock ``Unifying flavors of fault tolerance with the {ZX} calculus''.
\newblock \href{https://dx.doi.org/10.22331/q-2024-06-18-1379}{{Quantum} {\bf 8}, 1379}~(2024).

\bibitem{Hatcher2000}
Allen Hatcher.
\newblock ``{Algebraic topology}''.
\newblock Cambridge Univ. Press. Cambridge~(2000).
\newblock  url:~\href{https://cds.cern.ch/record/478079}{cds.cern.ch/record/478079}.

\bibitem{calderbank1996good}
A.~R. Calderbank and Peter~W. Shor.
\newblock ``Good quantum error-correcting codes exist''.
\newblock \href{https://dx.doi.org/10.1103/PhysRevA.54.1098}{Phys. Rev. A {\bf 54}, 1098--1105}~(1996).

\bibitem{steane1996error}
A.~M. Steane.
\newblock ``Error correcting codes in quantum theory''.
\newblock \href{https://dx.doi.org/10.1103/PhysRevLett.77.793}{Phys. Rev. Lett. {\bf 77}, 793--797}~(1996).

\bibitem{Liu2024subsystemcsscodes}
Michael~Liaofan Liu, Nathanan Tantivasadakarn, and Victor~V. Albert.
\newblock ``Subsystem {CSS} codes, a tighter stabilizer-to-{CSS} mapping, and {G}oursat's {L}emma''.
\newblock \href{https://dx.doi.org/10.22331/q-2024-07-10-1403}{{Quantum} {\bf 8}, 1403}~(2024).

\bibitem{bacon2006operator}
Dave Bacon.
\newblock ``Operator quantum error-correcting subsystems for self-correcting quantum memories''.
\newblock \href{https://dx.doi.org/10.1103/PhysRevA.73.012340}{Phys. Rev. A {\bf 73}, 012340}~(2006).

\bibitem{bauer2025planar}
Andreas Bauer and Julio C.~Magdalena de~la Fuente.
\newblock ``Planar fault-tolerant circuits for non-{{Clifford}} gates on the {{2D}} color code''~(2025).
\newblock  \href{http://arxiv.org/abs/2505.05175}{arXiv:2505.05175}.

\bibitem{aliferis2007subsystem}
Panos Aliferis and Andrew~W. Cross.
\newblock ``Subsystem fault tolerance with the bacon-shor code''.
\newblock \href{https://dx.doi.org/10.1103/PhysRevLett.98.220502}{Phys. Rev. Lett. {\bf 98}, 220502}~(2007).

\bibitem{higgott2025sparse}
Oscar Higgott and Craig Gidney.
\newblock ``Sparse {B}lossom: correcting a million errors per core second with minimum-weight matching''.
\newblock \href{https://dx.doi.org/10.22331/q-2025-01-20-1600}{{Quantum} {\bf 9}, 1600}~(2025).

\bibitem{derks2024designing}
Peter-Jan H.~S. Derks, Alex Townsend-Teague, Ansgar~G. Burchards, and Jens Eisert.
\newblock ``Designing fault-tolerant circuits using detector error models''~(2024).
\newblock  \href{http://arxiv.org/abs/2407.13826}{arXiv:2407.13826}.

\bibitem{iyer2015hardnessDecoding}
Pavithran Iyer and David Poulin.
\newblock ``Hardness of decoding quantum stabilizer codes''.
\newblock \href{https://dx.doi.org/10.1109/TIT.2015.2422294}{IEEE Transactions on Information Theory {\bf 61}, 5209--5223}~(2015).

\bibitem{rodatz2025fault}
Benjamin Rodatz, Boldizsár Poór, and Aleks Kissinger.
\newblock ``Fault tolerance by construction''~(2025).
\newblock  \href{http://arxiv.org/abs/2506.17181}{arXiv:2506.17181}.

\bibitem{fu2024error}
Xiaozhen Fu and Daniel Gottesman.
\newblock ``Error correction in dynamical codes''~(2024).
\newblock  \href{http://arxiv.org/abs/2403.04163}{arXiv:2403.04163}.

\bibitem{fu2025subsystem}
Xiaozhen Fu and Daniel Gottesman.
\newblock ``Subsystem spacetime code''.
\newblock Forthcoming~(2025).

\bibitem{jozsa2005introduction}
Richard Jozsa.
\newblock ``An introduction to measurement based quantum computation''~(2005).
\newblock  \href{http://arxiv.org/abs/quant-ph/0508124}{arXiv:quant-ph/0508124}.

\bibitem{gottesman1999demonstrating}
Daniel Gottesman and Isaac~L Chuang.
\newblock ``Demonstrating the viability of universal quantum computation using teleportation and single-qubit operations''.
\newblock \href{https://dx.doi.org/10.1038/46503}{Nature {\bf 402}, 390--393}~(1999).

\bibitem{nielsen2003quantum}
Michael~A. Nielsen.
\newblock ``Quantum computation by measurement and quantum memory''.
\newblock \href{https://dx.doi.org/https://doi.org/10.1016/S0375-9601(02)01803-0}{Physics Letters A {\bf 308}, 96--100}~(2003).

\bibitem{leung2004quantum}
Debbie~W. Leung.
\newblock ``Quantum computation by measurements''~(2004).
\newblock  \href{http://arxiv.org/abs/quant-ph/0310189}{arXiv:quant-ph/0310189}.

\bibitem{raussendorf2001oneway}
Robert Raussendorf and Hans~J. Briegel.
\newblock ``A one-way quantum computer''.
\newblock \href{https://dx.doi.org/10.1103/PhysRevLett.86.5188}{Phys. Rev. Lett. {\bf 86}, 5188--5191}~(2001).

\bibitem{nielsen2005fault}
Michael~A. Nielsen and Christopher~M. Dawson.
\newblock ``Fault-tolerant quantum computation with cluster states''.
\newblock \href{https://dx.doi.org/10.1103/PhysRevA.71.042323}{Phys. Rev. A {\bf 71}, 042323}~(2005).

\bibitem{raussendorf2005longrange}
Robert Raussendorf, Sergey Bravyi, and Jim Harrington.
\newblock ``Long-range quantum entanglement in noisy cluster states''.
\newblock \href{https://dx.doi.org/10.1103/PhysRevA.71.062313}{Phys. Rev. A {\bf 71}, 062313}~(2005).

\bibitem{raussendorf2007topological}
R~Raussendorf, J~Harrington, and K~Goyal.
\newblock ``Topological fault-tolerance in cluster state quantum computation''.
\newblock \href{https://dx.doi.org/10.1088/1367-2630/9/6/199}{New Journal of Physics {\bf 9}, 199}~(2007).

\bibitem{browne2016one}
Dan Browne and Hans Briegel.
\newblock ``One-way quantum computation''.
\newblock \href{https://dx.doi.org/https://doi.org/10.1002/9783527805785.ch21}{Chapter~21, pages 449--473}.
\newblock John Wiley \& Sons, Ltd. ~(2016).

\bibitem{bolt2018decoding}
A.~Bolt, D.~Poulin, and T.~M. Stace.
\newblock ``Decoding schemes for foliated sparse quantum error-correcting codes''.
\newblock \href{https://dx.doi.org/10.1103/PhysRevA.98.062302}{Phys. Rev. A {\bf 98}, 062302}~(2018).

\bibitem{danos2005parsimonious}
Vincent Danos, Elham Kashefi, and Prakash Panangaden.
\newblock ``Parsimonious and robust realizations of unitary maps in the one-way model''.
\newblock \href{https://dx.doi.org/10.1103/PhysRevA.72.064301}{Phys. Rev. A {\bf 72}, 064301}~(2005).

\bibitem{backens2021there}
Miriam Backens, Hector Miller-Bakewell, Giovanni de~Felice, Leo Lobski, and John van~de Wetering.
\newblock ``There and back again: A circuit extraction tale''.
\newblock \href{https://dx.doi.org/10.22331/q-2021-03-25-421}{Quantum {\bf 5}, 421}~(2021).

\bibitem{mcelvanney2023flow}
Tommy McElvanney and Miriam Backens.
\newblock ``Flow-preserving {ZX}-calculus rewrite rules for optimisation and obfuscation''~(2023).
\newblock  \href{http://arxiv.org/abs/2304.08166}{arXiv:2304.08166}.

\bibitem{bombin2006topologicalquantum}
H.~Bomb{\'i}n and M.~A. {Martin-Delgado}.
\newblock ``Topological {{Quantum Distillation}}''.
\newblock \href{https://dx.doi.org/10.1103/PhysRevLett.97.180501}{Physical Review Letters {\bf 97}, 180501}~(2006).

\bibitem{kubica2015universaltransversal}
Aleksander Kubica and Michael~E. Beverland.
\newblock ``Universal transversal gates with color codes - a simplified approach''.
\newblock \href{https://dx.doi.org/10.1103/PhysRevA.91.032330}{Physical Review A {\bf 91}, 032330}~(2015).
\newblock  \href{http://arxiv.org/abs/1410.0069}{arXiv:1410.0069}.

\bibitem{vuillot2021planar}
Christophe Vuillot.
\newblock ``Planar floquet codes''~(2021).
\newblock  \href{http://arxiv.org/abs/2110.05348}{arXiv:2110.05348}.

\bibitem{haah2022boundaries}
Jeongwan Haah and Matthew~B. Hastings.
\newblock ``Boundaries for the {H}oneycomb {C}ode''.
\newblock \href{https://dx.doi.org/10.22331/q-2022-04-21-693}{{Quantum} {\bf 6}, 693}~(2022).

\bibitem{sullivan2023floquet}
Joseph Sullivan, Rui Wen, and Andrew~C. Potter.
\newblock ``Floquet codes and phases in twist-defect networks''.
\newblock \href{https://dx.doi.org/10.1103/PhysRevB.108.195134}{Phys. Rev. B {\bf 108}, 195134}~(2023).

\bibitem{fahimniya2025faulttolerant}
Ali Fahimniya, Hossein Dehghani, Kishor Bharti, Sheryl Mathew, Alicia~J. Kollár, Alexey~V. Gorshkov, and Michael~J. Gullans.
\newblock ``Fault-tolerant hyperbolic floquet quantum error correcting codes''~(2025).
\newblock  \href{http://arxiv.org/abs/2309.10033}{arXiv:2309.10033}.

\bibitem{higgott2024constructions}
Oscar Higgott and Nikolas~P. Breuckmann.
\newblock ``Constructions and performance of hyperbolic and semi-hyperbolic floquet codes''.
\newblock \href{https://dx.doi.org/10.1103/PRXQuantum.5.040327}{PRX Quantum {\bf 5}, 040327}~(2024).

\bibitem{zhang2023xcube}
Zhehao Zhang, David Aasen, and Sagar Vijay.
\newblock ``$x$-cube floquet code: A dynamical quantum error correcting code with a subextensive number of logical qubits''.
\newblock \href{https://dx.doi.org/10.1103/PhysRevB.108.205116}{Phys. Rev. B {\bf 108}, 205116}~(2023).

\bibitem{dua2024engineering}
Arpit Dua, Nathanan Tantivasadakarn, Joseph Sullivan, and Tyler~D. Ellison.
\newblock ``Engineering 3d floquet codes by rewinding''.
\newblock \href{https://dx.doi.org/10.1103/PRXQuantum.5.020305}{PRX Quantum {\bf 5}, 020305}~(2024).

\bibitem{bauer2024topological}
Andreas Bauer.
\newblock ``Topological error correcting processes from fixed-point path integrals''.
\newblock \href{https://dx.doi.org/10.22331/q-2024-03-20-1288}{{Quantum} {\bf 8}, 1288}~(2024).

\bibitem{horsman2012surface}
Dominic Horsman, Austin~G Fowler, Simon Devitt, and Rodney~Van Meter.
\newblock ``Surface code quantum computing by lattice surgery''.
\newblock \href{https://dx.doi.org/10.1088/1367-2630/14/12/123011}{New Journal of Physics {\bf 14}, 123011}~(2012).

\bibitem{paetznick2013universal}
Adam Paetznick and Ben~W. Reichardt.
\newblock ``Universal fault-tolerant quantum computation with only transversal gates and error correction''.
\newblock \href{https://dx.doi.org/10.1103/PhysRevLett.111.090505}{Phys. Rev. Lett. {\bf 111}, 090505}~(2013).

\bibitem{bombin2015gauge}
Héctor Bombín.
\newblock ``Gauge color codes: optimal transversal gates and gauge fixing in topological stabilizer codes''.
\newblock \href{https://dx.doi.org/10.1088/1367-2630/17/8/083002}{New Journal of Physics {\bf 17}, 083002}~(2015).

\bibitem{yoder2016universal}
Theodore~J. Yoder, Ryuji Takagi, and Isaac~L. Chuang.
\newblock ``Universal fault-tolerant gates on concatenated stabilizer codes''.
\newblock \href{https://dx.doi.org/10.1103/PhysRevX.6.031039}{Phys. Rev. X {\bf 6}, 031039}~(2016).

\bibitem{higgott2021subsystem}
Oscar Higgott and Nikolas~P. Breuckmann.
\newblock ``Subsystem codes with high thresholds by gauge fixing and reduced qubit overhead''.
\newblock \href{https://dx.doi.org/10.1103/PhysRevX.11.031039}{Phys. Rev. X {\bf 11}, 031039}~(2021).

\bibitem{chao2018quantum}
Rui Chao and Ben~W. Reichardt.
\newblock ``Quantum error correction with only two extra qubits''.
\newblock \href{https://dx.doi.org/10.1103/PhysRevLett.121.050502}{Phys. Rev. Lett. {\bf 121}, 050502}~(2018).

\bibitem{chamberland2018flag}
Christopher Chamberland and Michael~E. Beverland.
\newblock ``Flag fault-tolerant error correction with arbitrary distance codes''.
\newblock \href{https://dx.doi.org/10.22331/q-2018-02-08-53}{{Quantum} {\bf 2}, 53}~(2018).

\bibitem{hastings2023quantumweightreduction}
M.~B. Hastings.
\newblock ``On quantum weight reduction''~(2023).
\newblock  \href{http://arxiv.org/abs/2102.10030}{arXiv:2102.10030}.

\bibitem{wills2025tradeoffConstructions}
Adam Wills, Ting-Chun Lin, and Min-Hsiu Hsieh.
\newblock ``Tradeoff constructions for quantum locally testable codes''.
\newblock \href{https://dx.doi.org/10.1109/TIT.2024.3503500}{IEEE Transactions on Information Theory {\bf 71}, 426--458}~(2025).

\bibitem{sabo2024weight}
Eric Sabo, Lane~G. Gunderman, Benjamin Ide, Michael Vasmer, and Guillaume Dauphinais.
\newblock ``Weight reduced stabilizer codes with lower overhead''~(2024).
\newblock  \href{http://arxiv.org/abs/2402.05228}{arXiv:2402.05228}.

\bibitem{cowtan2024css}
Alexander Cowtan and Simon Burton.
\newblock ``{CSS} code surgery as a universal construction''.
\newblock \href{https://dx.doi.org/10.22331/q-2024-05-14-1344}{{Quantum} {\bf 8}, 1344}~(2024).

\bibitem{ide2025fault}
Benjamin Ide, Manoj~G. Gowda, Priya~J. Nadkarni, and Guillaume Dauphinais.
\newblock ``Fault-tolerant logical measurements via homological measurement''.
\newblock \href{https://dx.doi.org/10.1103/PhysRevX.15.021088}{Phys. Rev. X {\bf 15}, 021088}~(2025).

\bibitem{poirson2025engineering}
Clément Poirson, Joschka Roffe, and Robert~I. Booth.
\newblock ``Engineering css surgery: compiling any cnot in any code''~(2025).
\newblock  \href{http://arxiv.org/abs/2505.01370}{arXiv:2505.01370}.

\end{thebibliography}

\end{document}